\documentclass[a4paper]{article}

\title{A Decentralized Analysis of Multiparty  Protocols}

\author{Bas van den Heuvel and Jorge A. P\'{e}rez}

\usepackage{fullpage}

\usepackage{lscape}

\usepackage[T1]{fontenc}
\usepackage[utf8]{inputenc}
\usepackage[english]{babel}
\usepackage{empheq}
\usepackage{varwidth}
\usepackage{mathtools}
\usepackage{amsthm}
\usepackage{thmtools, thm-restate}
\usepackage{centernot}
\usepackage{environ}

\usepackage{amssymb}
\usepackage{stmaryrd}
\usepackage{relsize}
\usepackage{graphicx}
\usepackage{bm}
\usepackage{textgreek}
\usepackage{xspace}
\usepackage{easy-todo}
\usepackage{bussproofs}
\usepackage{changepage}
\usepackage{trimclip}
\usepackage[ruled,linesnumbered]{algorithm2e}
\usepackage{caption}
\usepackage{subcaption}
\usepackage{tabularx}
\usepackage{wrapfig}
\usepackage{enumitem}
\usepackage{minibox}
\usepackage{xparse}
\usepackage{tikz}
\usepackage{tikzit}

\tikzstyle{Process}=[fill=none, draw=none, shape=circle, inner sep=-.5]
\tikzstyle{block}=[fill=none, draw=black, shape=rectangle, rounded corners=.5mm, inner sep=.5mm]

\tikzstyle{crt}=[-, draw=cbBlue, line width=.8pt]
\tikzstyle{ci}=[-, draw=cbRed, line width=.8pt]
\tikzstyle{thick}=[-, line width=1pt]
\tikzstyle{thicker}=[-, line width=1.2pt]

\usetikzlibrary{positioning, shapes, calc, fit, shapes.callouts}
\usepackage[framemethod=tikz]{mdframed}
\usepackage{xcolor}
\usepackage{ifthen}
\usepackage[hypertexnames=false]{hyperref}
\usepackage{cleveref}

\tikzset{grid/.style={gray,very thin,opacity=.3}}

\newcommand{\lipicsEnd}{}

\definecolor{cblRed}{RGB}{154,0,3}
\definecolor{cblBlue}{RGB}{41,59,107}
\definecolor{cblBlueLt}{RGB}{0,144,156}
\definecolor{cblYellow}{RGB}{255,176,0}
\definecolor{cblPink}{RGB}{216,0,106}
\definecolor{cblPurple}{RGB}{25,0,138}
\definecolor{cblPurpleLt}{RGB}{193,99,228}
\definecolor{cblOrange}{RGB}{189,74,0}
\definecolor{cblGreen}{RGB}{47,80,18}

\hypersetup{
    hypertexnames=false,
    colorlinks=true,
    final,
    citecolor=cblGreen,
    linkcolor=cblRed,
    bookmarksnumbered=true,
    bookmarksopen=true,
    bookmarksopenlevel=1,
    pdfpagemode=UseOutlines
}

\newcommand{\todid}[1]{\ignorespaces}

\newcommand{\tsrep}[1]{\ignorespaces}

\SetKwProg{Def}{def}{ as}{}
\SetKwComment{Cmt}{{\normalfont \small //~}}{}

\crefalias{AlgoLine}{line}%
\makeatletter
\let\cref@old@stepcounter\stepcounter
\def\stepcounter#1{%
  \cref@old@stepcounter{#1}%
  \cref@constructprefix{#1}{\cref@result}%
  \@ifundefined{cref@#1@alias}%
    {\def\@tempa{#1}}%
    {\def\@tempa{\csname cref@#1@alias\endcsname}}%
  \protected@edef\cref@currentlabel{%
    [\@tempa][\arabic{#1}][\cref@result]%
    \csname p@#1\endcsname\csname the#1\endcsname}}
\makeatother
\crefname{line}{line}{lines}
\Crefname{line}{Line}{Lines}

\newtheorem{theorem}{Theorem}
\newtheorem{lemma}[theorem]{Lemma}
\newtheorem{proposition}[theorem]{Proposition}
\newtheorem{definition}{Definition}
\newtheorem{example}{Example}
\newtheorem*{example*}{Example}
\newtheorem{notation}{Notation}

\newcommand{\defref}[1]{Def.~\labelcref{#1}}

\newcommand{\secref}[1]{\S\,\labelcref{#1}}
\newcommand{\figref}[1]{Fig.~\labelcref{#1}}
\newcommand{\appref}[1]{App.~\labelcref{#1}}

\newcommand{\thex}{{\arabic{example}}}
\newcommand{\etal}{\emph{et al.}\xspace}

\newcommand{\rott}[1]{\mathpalette\rot{#1}}
\newcommand{\rot}[2]{\rotatebox[origin=c]{180}{$#1{#2}$}}
\newcommand{\ol}{\overline}
\newcommand{\sff}[1]{\relax\ifmmode\mathsf{#1}\else\textsf{#1}\fi}
\newcommand{\mathsc}[1]{\text{\normalfont\scshape#1}}
\newcommand{\scc}[1]{\relax\ifmmode\mathsc{#1}\else\textsc{#1}\fi}
\newcommand{\mbb}[1]{\mathbb{#1}}
\newcommand{\mcl}[1]{\mathcal{#1}}

\newenvironment{wfit}{\begin{varwidth}{\textwidth}}{\end{varwidth}}
\let\obigcap\bigcap
\renewcommand{\bigcap}{{\mathchoice{\textstyle}{}{}{}\obigcap}}
\let\obigcup\bigcup
\renewcommand{\bigcup}{{\mathchoice{\textstyle}{}{}{}\obigcup}}
\let\oprod\prod
\renewcommand{\prod}{{\mathchoice{\textstyle}{}{}{}\oprod}}
\newcommand{\sepr}{\; \mbox{\large{$\mid$}}\;}
\newcommand{\<}{\langle}
\renewcommand{\>}{\rangle}

\let\onu\nu
\renewcommand{\nu}[1]{(\bm{\onu} #1)}
\let\oparallel\|
\renewcommand{\|}{\mathbin{|}}
\newcommand{\fwd}{\mathbin{\leftrightarrow}}
\newcommand{\0}{\bm{0}}
\newcommand{\call}[1]{{\langle #1 \rangle}}
\newcommand{\subst}[1]{\{#1\}}
\newcommand{\puts}{\mathbin{\triangleleft}}
\renewcommand{\gets}{\mathbin{\triangleright}}

\newcommand{\redd}{\longrightarrow}
\newcommand{\xredd}[1]{\mathbin{\textcolor{cblPurpleLt}{\xrightharpoondown{\textcolor{black}{#1}}}}}

\newcommand{\nredd}{{\centernot\longrightarrow}}
\newcommand{\rred}[1]{\shortrightarrow_{#1}}
\newcommand{\brred}[1]{\beta_{#1}}
\newcommand{\krred}[1]{\kappa_{#1}}

\newcommand{\tensor}{\mathbin{\otimes}}

\newcommand{\parr}{\mathbin{\rott{\&}}}

\let\ooplus\oplus
\renewcommand{\oplus}{{\ooplus}}
\newcommand{\pri}{\sff{o}}
\newcommand{\opri}{\kappa}

\newcommand{\lift}[1]{{\uparrow^{#1}}}
\newcommand{\dual}[1]{\ol{#1}}

\newcommand{\fn}{\mathrm{fn}}
\newcommand{\bn}{\mathrm{bn}}
\newcommand{\pr}{\sff{pr}}
\newcommand{\dom}{\sff{dom}}
\newcommand{\frv}{\mathrm{frv}}
\newcommand{\an}{\mathrm{an}}
\newcommand{\live}{\mathrm{live}}

\newcommand{\infAx}[2]{\AxiomC{} \RightLabel{#2} \UnaryInfC{#1}}
\newcommand{\infAss}[1]{\AxiomC{#1}}

\newcommand{\infUn}[2]{\RightLabel{#2} \UnaryInfC{#1}}

\newcommand{\infBin}[2]{\RightLabel{#2} \BinaryInfC{#1}}

\newcommand{\infTer}[2]{\RightLabel{#2} \TrinaryInfC{#1}}

\newsavebox{\topprooftreebox}
\newlength{\topprooftreewidth}
\NewDocumentEnvironment{topprooftree}{m}%
    {\begin{lrbox}{\topprooftreebox}\ignorespaces}%
    {\DisplayProof\end{lrbox}\begin{center}\settowidth{\topprooftreewidth}{\topprooftreebox}\makebox[\topprooftreewidth]{\minibox{{#1}\\\usebox{\topprooftreebox}}}\end{center}}

\newcommand{\gend}{\sff{end}}
\newcommand{\gskip}{\sff{skip}}
\newcommand{\mto}{\mathbin{\twoheadrightarrow}}
\newcommand{\snd}{{!}}
\newcommand{\rcv}{{?}}
\newcommand{\sdot}{\mathbin{.}}

\renewcommand{\merge}{\sqcup}

\newcommand{\wrt}{\mathbin{\textcolor{cblPurple}{\upharpoonright}}}
\newcommand{\onto}{\mathbin{\textcolor{cblPink}{\downharpoonright}}}
\newcommand{\under}{\mathbin{\upharpoonright}}
\renewcommand{\part}{\sff{prt}}

\newcommand{\deps}{\sff{deps}}
\newcommand{\dep}{\sff{depon}}
\newcommand{\detdep}{\mathrm{ddep}}
\newcommand{\hdep}{\mathrm{hdep}}

\newcommand{\error}[1]{{\mathchoice{\textstyle}{}{}{}\mathsf{alarm}({#1})}}

\newcommand{\relprop}[4]{{\textcolor{cblBlueLt}{\llparenthesis}\mkern-1mu #4 \textcolor{cblBlueLt}{\rrparenthesis}}_{#1 \textcolor{cblBlueLt}{\rangle} #2}^{#3}}
\newcommand{\msgprop}[1]{{\textcolor{cblOrange}{\llparenthesis} #1 \textcolor{cblOrange}{\rrparenthesis}}}

\newcommand{\lmed}[3]{{\llbracket #3 \rrbracket}_{#1}^{#2}}
\newcommand{\mdm}[2]{{\sff{O}}_{#1}[#2]}
\newcommand{\rtr}{\mcl{R}}
\newcommand{\lsys}{\mathrm{ri}}
\newcommand{\sys}{\mathrm{net}}
\newcommand{\hub}{\mcl{H}}
\newcommand{\hole}{[\,]}

\newcommand{\rch}[1]{{\color{cblPurple}#1}}
\newcommand{\ich}[1]{{\color{cblPink}#1}}
\newcommand{\ci}[2]{{\ich{#1}}_{\ich{#2}}}
\newcommand{\crt}[2]{{\rch{#1}}_{\rch{#2}}}

\usepackage{fixme}
\fxsetup{theme=color,mode=multiuser,draft,inline,marginclue}
\definecolor{fxtarget}{rgb}{0.8300,0.1400,0.1400}
\FXRegisterAuthor{bas}{abas}{Bas}

\newcommand{\revstart}[1][]{%
    \ignorespaces%
    \ifthenelse{\equal{#1}{1}}{\color{cblPurple}}{%
    \ifthenelse{\equal{#1}{2}}{\color{cblPink}}{%
    \ifthenelse{\equal{#1}{3}}{\color{cblOrange}}%
    {\color{cblBlue}}%
    }}%
    \ignorespaces%
}
\newcommand{\revend}{\normalcolor\ignorespaces}
\renewcommand{\revstart}[1][]{\ignorespaces}
\renewcommand{\revend}{\ignorespaces}

\newcommand{\recdef}[2]{\mathrm{recdef}(#1,#2)}
\newcommand{\ctxpri}[2]{\mathrm{ctxpri}^{#1}(#2)}

\newcommand{\recpri}[2]{\mathrm{varpri}(#1,#2)}
\newcommand{\ctxbind}[1]{\mathrm{ctxbind}(#1)}
\newcommand{\subbind}[2]{\mathrm{subbind}(#1,#2)}
\newcommand{\deepunfold}[2]{\mathrm{deepUnfold}(#1,#2)}
\newcommand{\unfold}{\mathrm{unfold}}
\newcommand{\activ}[2]{\mathrm{active}(#1,#2)}
\newcommand{\recactiv}[2]{\mathrm{recactive}(#1,#2)}
\newcommand{\complet}{\circlearrowright}
\newcommand{\gsub}[1]{\leq_{#1}}
\newcommand{\recctx}[2]{\mathrm{recCtx}(#1,#2)}
\newcommand{\pn}{\mathrm{pn}}
\newcommand{\chr}{\mathrm{char}}
\newcommand{\tand}{\text{and}}

\newcommand{\B}{\mbb{B}}
\newcommand{\tst}{\text{s.t.}}

\let\oxrightarrow\xrightarrow
\renewcommand{\xrightarrow}[1]{\oxrightarrow{\vspace{-.8ex}#1}}

\def\ih#1{IH\textsubscript{#1}}

\newsavebox{\boxCyclic}

\newsavebox{\boxGauth}

\newsavebox{\boxGauthInterl}

\newsavebox{\boxDecentral}

\newsavebox{\boxCentral}

\newsavebox{\boxMedium}

\begin{document}

\maketitle

\begin{abstract}
    Protocols provide the unifying glue in concurrent and distributed software today; verifying that message-passing programs conform to such governing protocols is important but difficult.
    Static approaches based on multiparty session types (MPST) use protocols as types to avoid protocol violations and deadlocks in programs.
    An elusive problem for MPST is to ensure \emph{both} protocol conformance \emph{and} deadlock freedom for implementations with interleaved and delegated protocols.

    We propose a decentralized analysis of multiparty protocols, specified as global types and implemented as interacting processes in an asynchronous $\pi$-calculus.
    Our solution rests upon two novel notions: \emph{router processes} and \emph{relative types}.
    While router processes use the global type to enable the composition of participant implementations in arbitrary process networks, relative types extract from the global type the intended interactions and dependencies between \emph{pairs} of participants.
    In our analysis, processes are typed using APCP, a type system that ensures protocol conformance and deadlock freedom with respect to \emph{binary} protocols, developed in prior work.
    Our decentralized, router-based analysis enables the sound and complete transference of protocol conformance and deadlock freedom from APCP to multiparty protocols.
\end{abstract}

{\small
\tableofcontents
}

\newpage

\section{Introduction}
\label{s:introduction}

This paper presents a new approach to the analysis of the \emph{protocols} that pervade concurrent and distributed software.
Such protocols provide an essential unifying glue between communicating programs; ensuring that communicating programs implement protocols correctly, avoiding protocol violations and deadlocks, is an important but difficult problem.
Here, we study \emph{multiparty session types (MPST)}~\cite{journal/acm/HondaYC16}, an approach to correctness in message-passing programs that uses governing multiparty protocols as types in program verification.

As a motivating example, let us consider a \emph{recursive authorization protocol}, adapted from an example by Scalas and Yoshida~\cite{conf/popl/ScalasY19}.
It involves three participants: a Client, a Server, and an Authorization service.
Intuitively, the protocol proceeds as follows.
The Server requests the Client either to \emph{login} or to \emph{quit} the protocol.
In the case of login, the Client sends a password to the Authorization service, which then may authorize the login with the Server; subsequently, the protocol can be performed again: this is useful when, e.g., clients must renew their authorization privileges after some time.
In the case of quit, the protocol ends.

MPST use \emph{global types} to specify multiparty protocols.
The authorization protocol just described can be specified by the following global type between Client (`$c$'), Server (`$s$'), and Authorization service~(`$a$'):
\begin{equation}\label{eq:Gauth}
    G_{\sff{auth}} = \mu X \sdot s \mto c
    \left\{ \begin{array}{l}
            \sff{login} \sdot c \mto a \big\{ \sff{passwd}\<\sff{str}\> \sdot a \mto s \{ \sff{auth}\<\sff{bool}\> \sdot X \} \big\},
            \\
            \sff{quit} \sdot c \mto a \{ \sff{quit} \sdot \bullet \}
    \end{array} \right\}
\end{equation}
After declaring a recursion  on the variable $X$ (`$\mu X$'), the global type $G_{\sff{auth}}$ stipulates that $s$ sends to $c$ (`$s \mto c$') a label \sff{login} or \sff{quit}.
The rest of the protocol depends on this choice by $s$.
In the \sff{login}-branch, $c$ sends to $a$ a label \sff{passwd} along with a \sff{str}ing value (`$\<\sff{str}\>$') and $a$ sends to $s$ a label \sff{auth} and a \sff{bool}ean value, after which the protocol loops to the beginning of the recursion (`$X$').
In the \sff{quit}-branch, $c$ sends to $a$ a label \sff{quit} after which the protocol ends (`$\bullet$').

In MPST, participants are implemented as distributed processes that communicate asynchronously.
Each process must correctly implement its corresponding portion of the protocol; these individual guarantees ensure that the interactions between processes conform to the given global type.
Correctness follows from \emph{protocol fidelity} (processes interact as stipulated by the protocol), \emph{communication safety} (no errors or mismatches in messages), and \emph{deadlock freedom} (processes  never get stuck waiting for each other).
Ensuring that implementations satisfy these properties is a challenging problem, which is further compounded by two common and convenient features in interacting processes: \emph{delegation} and \emph{interleaving}.
We motivate them in the context of our example:
\begin{itemize}
    \item
        \emph{Delegation,} or higher-order channel passing, can effectively express that the Client transparently reconfigures its involvement by asking another participant (say, a  Password Manager) to act on its behalf;

    \item
        \emph{Interleaving} arises when a single process implements more than one role, as in, e.g., an implementation of both the Server and the Authorization service in a sequential process.
\end{itemize}
Note that while delegation is explicitly specified in a global type, interleaving arises in its implementation as interacting processes, not in its specification.

MPST have been widely studied from foundational and applied angles~\cite{conf/concur/BettiniCDLDY08,journal/lmcs/CastagnaDP12,journal/lmcs/HuBYD12,conf/ecoop/ScalasDHY17,conf/ice/BarbaneraD19,journal/asep/BejleriDVEM19,conf/popl/ScalasY19,journal/tcs/CastellaniDGH20,conf/ecoop/ImaiNYY20,journal/pacmpl/MajumdarYZ20}.
The original theory by Honda \etal~\cite{conf/popl/HondaYC08} defines a behavioral type system~\cite{journal/csur/HuttelLVCCDMPRT16,journal/ftpl/AnconaBB0CDGGGH16} for a $\pi$-calculus, which exploits linearity to ensure protocol fidelity and communication safety; most derived works retain this approach and target the same properties.
Deadlock freedom is hard to ensure by typing when implementations feature delegation and interleaving.
In simple scenarios without interleaving and/or delegation, deadlock freedom is easy, as it concerns a single-threaded protocol.
In contrast, deadlock freedom for processes running multiple, interleaved protocols (possibly delegated) is a much harder problem, addressed only by some advanced type systems~\cite{conf/concur/BettiniCDLDY08,conf/coordination/PadovaniVV14,journal/mscs/CoppoDYP16}.

In this paper, we tackle the problem of ensuring that networks of interacting processes correctly implement a given global type in a deadlock free manner, while supporting delegation and  interleaving.
Our approach is informed by the differences between \emph{orchestration} and \emph{choreography}, two salient approaches to the coordination and organization of interacting processes in service-oriented paradigms~\cite{journal/computer/Peltz03}:
\begin{itemize}
    \item
        In \emph{orchestration}-based approaches, processes interact through a coordinator process which ensures that they all follow the protocol as intended.
        Quoting Van der Aalst, in an orchestration ``the conductor tells everybody in the orchestra what to do and makes sure they all play in sync''~\cite{book/vdAalst09}.

    \item
        In \emph{choreography}-based approaches, processes interact directly following the protocol without external coordination.
        Again quoting Van der Aalst, in a choreography ``dancers dance following a global scenario without a single point of control''~\cite{book/vdAalst09}.
\end{itemize}

Specification and analysis techniques based on MPST fall under the choreography-based approach.
The global type provides the protocol's specification; based on the global type, implementations for each participant interact directly with each other, without an external coordinator.

As we will see, the contrast between orchestration and choreography is relevant here because it induces a different \emph{network topology} for interacting processes.
In an orchestration, the resulting process network is \emph{centralized}: all processes must connect to a central orchestrator process.
In a choreography, the process network is \emph{decentralized}, as processes can directly connect to each other.

\paragraph*{Contributions}

We develop a new decentralized analysis of multiparty  protocols.
\begin{itemize}
	\item
        Here `analysis' refers to (i)~ways of specifying such protocols as interacting processes \emph{and} (ii)~techniques to verify that those processes satisfy the intended correctness properties.

    \item
        Also, aligned with the above discussion, `decentralized' refers to the intended network topology for processes, which does not rely on an external coordinator.
\end{itemize}
Our decentralized analysis of global types enforces protocol fidelity, communication safety, and deadlock freedom for process implementations, while uniformly supporting delegation, interleaving, and asynchronous communication.

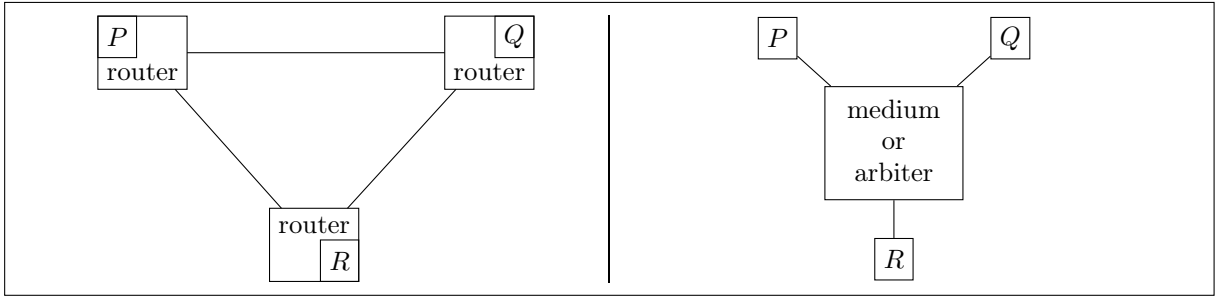
\begin{figure}
    \begin{mdframed}
        \mbox{}\hfill%
        \begin{minipage}{.45\textwidth}
            \centering
            \begin{tikzpicture}
    \begin{scope}[local bounding box=bndP]
        \node (P) [minimum height=5.5mm, draw=black, shape=rectangle] at (0, 0) {$P$};
        \node (rtrP) [below=4.7mm of P.west, anchor=west] {router};
    \end{scope}
    \node [fit=(bndP), draw=black, shape=rectangle, inner sep=0mm] (rtdP) {};

    \begin{scope}[local bounding box=bndQ]
        \node (Q) [minimum height=5.5mm, right=4.7cm of P, draw=black, shape=rectangle] {$Q$};
        \node (rtrQ) [below=4.7mm of Q.east, anchor=east] {router};
    \end{scope}
    \node [fit=(bndQ), draw=black, shape=rectangle, inner sep=0mm] (rtdQ) {};

    \begin{scope}[local bounding box=bndR]
        \node (R) [minimum height=5.5mm, below right=3.4cm of P, draw=black, shape=rectangle] {$R$};
        \node (rtrR) [above=4.7mm of R.east, anchor=east] {router};
    \end{scope}
    \node [fit=(bndR), draw=black, shape=rectangle, inner sep=0mm] (rtdR) {};

    \draw [-] (bndP) to (bndQ);
    \draw [-] (bndQ) to (bndR);
    \draw [-] (bndR) to (bndP);
\end{tikzpicture}
        \end{minipage}%
        \hfill\vline\hfill%
        \begin{minipage}{.45\textwidth}
            \centering
            \begin{tikzpicture}
    \node (O) [draw=black, shape=rectangle] at (0, 0) {
        $\begin{array}{c}
            \text{medium} \\ \text{or} \\ \text{arbiter}
        \end{array}$
    };
    \node (P) [above left=.5cm of O, minimum height=5.5mm, draw=black, shape=rectangle] {$P$};
    \node (Q) [above right=.5cm of O, minimum height=5.5mm, draw=black, shape=rectangle] {$Q$};
    \node (R) [below=.5cm of O, minimum height=5.5mm, draw=black, shape=rectangle] {$R$};

    \draw [-] (O) to (P);
    \draw [-] (O) to (Q);
    \draw [-] (O) to (R);
\end{tikzpicture}
        \end{minipage}%
        \hfill\mbox{}
    \end{mdframed}

    \caption{Given processes $P$, $Q$, and $R$  implementing the roles of $c$, $s$, and $a$, respectively, protocol $G_\sff{auth}$ can be realized as a choreography of routed implementations (our approach, left) and as an orchestration of implementations, with a medium or arbiter process (previous works, right).}
    \label{f:chorOrch}
\end{figure}

The \emph{key idea} of our analysis is to exploit global types to generate \emph{router processes} (simply \emph{routers}) that enable  participant implementations to communicate directly.
There is a router per participant; it is intended to serve as a ``wrapper'' for an individual participant's implementation.
The composition of an implementation with its corresponding router is called a \emph{routed implementation}.
Collections of routed implementations can then be connected in arbitrary \emph{process networks} that correctly realize the multiparty protocol, subsuming centralized and decentralized topologies.

Routers are \emph{synthesized} from global types, and do not change the behavior of the participant implementations they wrap; they merely ensure that networks of routed implementations correctly behave as described by the given multiparty protocol.
Returning to Van der Aalst's analogies quoted above, we may say that in our setting participant implementations are analogous to skilled but barefoot dancers, and that routers provide them with the appropriate shoes to dance without a central point of control.
To make this analogy a bit more concrete, \Cref{f:chorOrch} (left) illustrates the decentralized process network formed by routed implementations of the participants of $G_\sff{auth}$: once wrapped by an appropriate router, implementations $P$, $Q$, and $R$ can be composed directly in a decentralized process network.

A central technical challenge in our approach is to ensure that compositions of routed implementations conform to their global type.
The channels that enable the arbitrary composition of routed implementations need to be typed in accordance with the given multiparty protocol.
Unfortunately, the usual notion of projection in MPST, which obtains a single participant's perspective from a global type, does not suffice: we need a local perspective that is relative to the \emph{two participants} that the connected routed implementations represent.
To this end, we introduce a new notion, \emph{relative projection}, which isolates the exchanges of the global type that relate to pairs of participants.
In the case of $G_\sff{auth}$, for instance, we need three relative types, describing the protocol for $a$ and $c$, for $a$ and $s$, and for $c$ and $s$.

A derived challenge is that when projecting a global type onto a pair of participants, it is possible to encounter \emph{non-local choices}: choices by other participants that affect the protocol between the two participants involved in the projection.
To handle this, relative projection explicitly records non-local choices in the form of \emph{dependencies}, which inform the projection's participants that they need to coordinate on the results of the non-local choices.

To summarize, our decentralized analysis of global types relies on three intertwined novel notions:
\begin{itemize}
    \item
        \textbf{Routers} that wrap participant implementations in order to enable their composition in arbitrary network topologies, whilst guaranteeing that the resulting process networks correctly follow the given global type in a deadlock free manner.

    \item
        \textbf{Relative Types} that type the channels between routed implementations, obtained by means of a new notion of projection of global types onto pairs of participants.

    \item
        \textbf{Relative Projection} and \textbf{Dependencies} that make it explicit in relative types that participants need to coordinate on non-local choices.
\end{itemize}
\noindent
The key ingredients of our decentralized analysis for $G_{\sff{auth}}$ are jointly depicted in \Cref{l:sumdiag}.

\smallskip

With respect to prior analyses of multiparty protocols, a distinguishing feature of our work is its natural support of decentralized process networks, as expected in a choreography-based approach.
Caires and P\'{e}rez~\cite{conf/forte/CairesP16} connect participant implementations through a central coordinator, called \emph{medium process}.
This medium process is generated from a global type, and intervenes in all exchanges to ensure that the participant implementations follow the multiparty protocol.
The composition of the medium with the participant implementations can then be analyzed using a type system for binary sessions.
In a similar vein, Carbone \etal~\cite{conf/concur/CarboneLMSW16} define a type system in which they use global types to validate choreographies of participant implementations.
Their analysis of protocol implementations---in particular, deadlock freedom---relies on encodings into another type system where participant implementations connect to a central coordinator, called the \emph{arbiter process}.
Similar to mediums, arbiters are generated from the global type to ensure that participant implementations follow the protocol as intended.
Both these approaches are clear examples of orchestration, and thus do not support decentralized network topologies.

To highlight the differences between our decentralized analysis and prior approaches, compare the choreography of routed implementations in \Cref{f:chorOrch}~(left) with  an implementation of $G_\sff{auth}$  in the style of Caires and P\'{e}rez and of Carbone \etal, given in \Cref{f:chorOrch}~(right).
These prior works rely on orchestration because the type systems they use for verifying process implementations restrict connections between processes: they only admit a form of process composition that makes it impossible to simultaneously connect three or more participant implementations~\cite{conf/express/DardhaP15}.
In this paper, we overcome this obstacle by relying on APCP (Asynchronous Priority-based Classical Processes)~\cite{report/vdHeuvelP21B}, a type system that allows for more general forms of process composition.
By using annotations on types, APCP prevents \emph{circular dependencies}, i.e., cyclically connected processes that are stuck waiting for each other.
This is how our approach supports networks of routed participants in both centralized and decentralized topologies, thus subsuming  choreography and orchestration approaches.

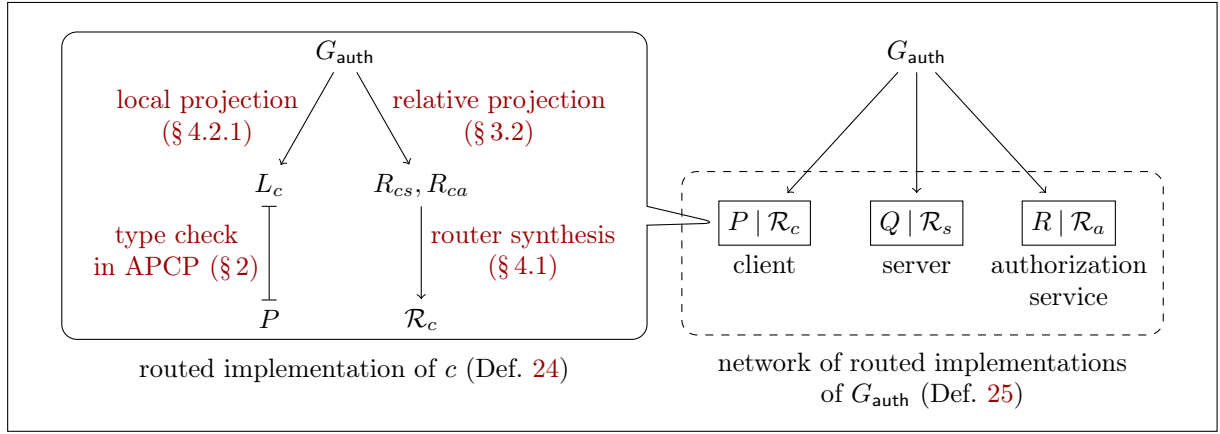
\begin{figure}[!t]
    \begin{mdframed}
        \medskip
        \begin{center}
            \begin{tikzpicture}
    \begin{scope}[local bounding box=BNDc]
        \node [align=center] (G) at (0,0) {$G_\sff{auth}$};

        \node [below=15mm of G,xshift=-10mm,anchor=center,align=center] (L) {$L_c$};
        \draw [-to] (G) -- (L) node [midway,left,text=cblRed,align=center] {local projection \\ (\secref{ss:locproj})};

        \node [below=15mm of G,xshift=10mm,anchor=center,align=center] (R) {$R_{cs},R_{ca}$};
        \draw [-to] (G) -- (R) node [midway,right,text=cblRed,align=center] {relative projection \\ (\secref{ss:relproj})};

        \node [below=15mm of L,anchor=center,align=center] (P) {$P$};
        \draw [|-|] (L) -- (P) node [midway,left,text=cblRed,align=center] {type check \\ in APCP (\secref{s:apcp})};

        \node [below=15mm of R,anchor=center,align=center] (RTR) {$\rtr_c$};
        \draw [-to] (R) -- (RTR) node [midway,right,text=cblRed,align=center] {router synthesis \\ (\secref{ss:routers})};
    \end{scope}

    \node [right=70mm of G,anchor=center,align=center] (G2) {$G_\sff{auth}$};

    \begin{scope}[local bounding box=BNDnet]
        \node [below=20mm of G2,xshift=-20mm,anchor=center,align=center,draw=black] (RIMPLc) {$P \| \rtr_c$};
        \node [below=3mm of RIMPLc.center,anchor=north,align=center] {client};
        \node [above=2mm of RIMPLc.center,anchor=south,align=center] {\vphantom{l}};
        \node [left=0mm of RIMPLc.west,anchor=east] {\hphantom{s}};

        \node [below=20mm of G2,anchor=center,align=center,draw=black] (RIMPLs) {$Q \| \rtr_s$};
        \node [below=3mm of RIMPLs.center,anchor=north,align=center] {server\vphantom{l}};
        \node [above=2mm of RIMPLs.center,anchor=south,align=center] {\vphantom{l}};

        \node [below=20mm of G2,xshift=20mm,anchor=center,align=center,draw=black] (RIMPLa) {$R \| \rtr_a$};
        \node [below=3mm of RIMPLa.center,anchor=north,align=center] {authorization \\ service};
        \node [above=2mm of RIMPLa.center,anchor=south,align=center] {\vphantom{l}};
        \node [right=0mm of RIMPLc.east,anchor=west] {\hphantom{s}};
        \node [below=10mm of RIMPLa.center,anchor=north,align=center] {\vphantom{l}};
    \end{scope}
    \node [fit=(BNDnet),draw=black,dashed,inner xsep=1mm,inner ysep=0mm,rounded corners=2mm] (BOXnet) {};
    \node [below=1mm of BOXnet.south,anchor=north,align=center] {network of routed implementations \\ of $G_\sff{auth}$ (\defref{d:networks})};

    \draw [-to] (G2) -- ([yshift=1mm,xshift=3mm]RIMPLc.north);
    \draw [-to] (G2) -- ([yshift=1mm,xshift=0mm]RIMPLs.north);
    \draw [-to] (G2) -- ([yshift=1mm,xshift=-3mm]RIMPLa.north);

    \node [fit=(BNDc),draw=black,inner xsep=3mm,inner ysep=0mm,rectangle callout,rounded corners=2mm,callout absolute pointer={(RIMPLc.west)}] (BOXc) {};
    \node [below=1mm of BOXc.south,anchor=north,align=center] {routed implementation of $c$ (\defref{d:routedImplementations})};
\end{tikzpicture}
        \end{center}
    \end{mdframed}
    \caption{Decentralized analysis of $G_\sff{auth}$ into a network of routed implementations. The definition of $G_\sff{auth}$ contains message types. Focusing on the client $c$ (on the left), $L_c$ denotes a session type, whereas $R_{cs}$ and $R_{ca}$ are relative types with respect to the server and the authorization service, respectively.  \label{l:sumdiag}}
\end{figure}

\paragraph*{Outline}

This paper is structured as follows.  Next, \Cref{s:apcp} recalls APCP as introduced by Van den Heuvel and P\'{e}rez~\cite{report/vdHeuvelP21B} and summarizes the correctness properties for asynchronous processes derived from typing. The following three sections develop and illustrate our contributions:
\begin{itemize}
    \item
        \Cref{s:mpst} introduces relative types and relative projection, and defines well-formed global types, a class of global types that includes protocols with non-local choices.
    \item
        \Cref{s:routers} introduces the synthesis of routers.
        A main result is their typability in APCP (\Cref{t:routerTypes}).
        We establish deadlock freedom for networks of routed implementations (\Cref{t:globalDlFree}), which we transfer to multiparty protocols via an operational correspondence result (\Cref{t:completeness,t:soundness}).
        Moreover, we show that our approach strictly generalizes prior analyses based on centralized topologies (\Cref{t:mediumBisim}).
    \item
        \Cref{s:routersInAction} demonstrates our contributions in action, with a full development of the routed implementations for $G_\sff{auth}$, and an example of the flexible support for \emph{delegation} and \emph{interleaving} enabled by our router-based approach and APCP.
\end{itemize}
We discuss further related works in \secref{s:relwork} and conclude the paper in \secref{s:concl}.
We use colors to improve readability.

\section{APCP: Asynchronous Processes, Deadlock Free by Typing}
\label{s:apcp}

We recall APCP as defined by Van den Heuvel and P\'{e}rez~\cite{report/vdHeuvelP21B}.
APCP is a type system for asynchronous $\pi$-calculus processes (with non-blocking outputs)~\cite{conf/ecoop/HondaT91,report/Boudol92}, with support for recursion and cyclic connections.
In this type system, channel endpoints are assigned linear types that represent two-party (binary) \emph{session types}~\cite{conf/concur/Honda93}.
Well-typed APCP processes  {preserve typing} (\Cref{t:subjectRed}) and are {deadlock free} (\Cref{t:closedDLFree}).

At its basis, APCP combines Dardha and Gay's Priority-based Classical Processes (PCP)~\cite{conf/fossacs/DardhaG18} with DeYoung \etal's continuation-passing semantics for asynchrony~\cite{conf/csl/DeYoungCPT12}, and adds recursion,  inspired by the work of Toninho \etal~\cite{conf/tgc/ToninhoCP14}.
We refer the interested reader to the work by Van den Heuvel and P\'{e}rez~\cite{report/vdHeuvelP21B} for a motivation of design choices and proofs of results.

\paragraph{Process Syntax}

\begin{figure}[t]
    \begin{mdframed}
        Process syntax:
        \begin{align*}
            P,Q ::=
            &~ x[y,z]
            & \text{output}
            & ~~~ \sepr ~~
              x(y,z) \sdot P
            & \text{input}
            \\[-3pt]
            \sepr\!
            &~ x[z] \triangleleft i
            & \text{selection}
            & ~~~ \sepr ~~
              x(z) \triangleright \{i{:}~ P\}_{i \in I}
            & \text{branching}
            & ~~~ \sepr ~~
              \nu{x y}P
            & \text{restriction}
            \\[-3pt]
            \sepr\!
            &~ P \| Q
            & \text{parallel}
            & ~~~ \sepr ~~
              \0
            & \text{inaction}
            & ~~~ \sepr ~~
              x \fwd y
            & \text{forwarder}
            \\[-3pt]
            \sepr\!
            &~ \mu X(\tilde{z}) \sdot P
            & \text{recursive loop}
            & ~~~ \sepr ~~
              X\call{\tilde{z}}
            & \text{recursive call}
        \end{align*}

        \vspace{-1em}
        \hbox to \textwidth{\leaders\hbox to 3pt{\hss . \hss}\hfil}

        \smallskip
        Structural congruence:
        \begin{align*}
            P \equiv_\alpha P' \implies {}
            &
            &
            P
            &\equiv P'
            &
            x \fwd y
            &\equiv y \fwd x
            \\
            &
            &
            P \| Q
            &\equiv Q \| P
            &
            \nu{x y} x \fwd y
            &\equiv \0
            \\
            &
            &
            P \| \0
            &\equiv P
            &
            P \| (Q \| R)
            &\equiv (P \| Q) \| R
            \\
            x,y \notin \fn(P) \implies {}
            &
            &
            P \| \nu{x y}Q
            &\equiv \nu{x y}(P \| Q)
            &
            \nu{x y}\0
            &\equiv \0
            \\
            |\tilde{z}| = |\tilde{y}| \implies {}
            &
            &
            \mu X(\tilde{z}) \sdot P
            &\equiv P \big\{(\mu X(\tilde{y}) \sdot P \subst{\tilde{y}/\tilde{z}}) / X\call{\tilde{y}}\big\}
            &
            \nu{x y}P
            &\equiv \nu{y x}P
            \\
            &
            &
            &
            &
            \nu{x y}\nu{z w} P
            &\equiv \nu{z w}\nu{x y} P
        \end{align*}

        \vspace{-1em}
        \hbox to \textwidth{\leaders\hbox to 3pt{\hss . \hss}\hfil}

        \smallskip
        Reduction:
        \begin{align*}
            & \brred{\scc{Id}}
            &
            z,y \neq x \implies {}
            &
            &
            \nu{y z}( x \fwd y \| P )
            &\redd P \subst{x/z}
            \\
            & \brred{\tensor\parr}
            &
            &
            &
            \nu{x y}( x[a,b] \| y(v,z) \sdot P )
            & \redd P \subst{a/v, b/z}
            \\
            & \brred{\oplus\&}
            &
            j \in I \implies {}
            &
            &
            \nu{x y}( x[b] \triangleleft j \| y(z) \triangleright \{i{:}~ P_i\}_{i \in I} )
            &\redd P_j \subst{b/z}
            \\[8pt]
            & \krred{\parr}
            &
            x \notin \tilde{v}, \tilde{w} \implies {}
            &
            &
            \nu{\tilde{v} \tilde{w}}(x(y,z) \sdot P \| Q)
            &\redd x(y,z) \sdot \nu{\tilde{v} \tilde{w}}(P \| Q)
            \\
            & \krred{\&}
            &
            x \notin \tilde{v}, \tilde{w} \implies {}
            &
            & \nu{\tilde{v} \tilde{w}}(x(z) \triangleright \{i{:}~ P_i\}_{i \in I} \| Q)
            & \redd x(z) \triangleright \{i{:}~ \nu{\tilde{v} \tilde{w}}(P_i \| Q)\}_{i \in I}
        \end{align*}

        \noindent
        \mbox{}\hfill%
        \begin{wfit}
            \begin{prooftree}
                \infAss{
                    $(P \equiv P') \wedge (P' \redd Q') \wedge (Q' \equiv Q)$
                }
                \infUn{
                    $P \redd Q$
                }{$\rred{\equiv}$}
            \end{prooftree}
        \end{wfit}%
        \hfill%
        \begin{wfit}
            \begin{prooftree}
                \infAss{
                    $\raisebox{9pt}{} P \redd Q$
                }
                \infUn{
                    $\nu{x y} P \redd \nu{x y} Q$
                }{$\rred{\onu}$}
            \end{prooftree}
        \end{wfit}%
        \hfill%
        \begin{wfit}
            \begin{prooftree}
                \infAss{
                    $\raisebox{9pt}{} P \redd Q$
                }
                \infUn{
                    $P \| R \redd Q \| R$
                }{$\rred{\|}$}
            \end{prooftree}
        \end{wfit}%
        \hfill\mbox{}
    \end{mdframed}

    \caption{Definition of the process language of APCP.}
    \label{f:procdef}
\end{figure}

We write $x, y, z, \ldots$ to denote (channel) \emph{endpoints} (also known as \emph{names}), and write $\tilde{x}, \tilde{y}, \tilde{z}, \ldots$ to denote sequences of endpoints.
Also, we write $i, j, k, \ldots$ to denote \emph{labels} for choices and $I, J, K, \ldots$ to denote sets of labels.
We write $X, Y, \ldots$ to denote \emph{recursion variables}, and $P,Q, \ldots$ to denote processes.

\Cref{f:procdef} (top) gives the syntax of processes, which communicate asynchronously by following a continuation-passing style.
The output action `$x[y,z]$' denotes the sending of endpoints $y$ and $z$ along $x$: while the former is the message, the latter is the protocol's continuation; both $y$ and $z$ are free.
The input prefix `$x(y,z) \sdot P$' blocks until a message $y$ and a continuation endpoint $z$ are received on $x$, binding $y$ and $z$ in $P$.
The selection action `$x[z] \puts i$' sends a label $i$ and a continuation endpoint $z$ along $x$.
The branching prefix `$x(z) \gets \{i{:}~ P_i\}_{i \in I}$' blocks until it receives a label $i \in I$ and a continuation endpoint $z$ on $x$, binding $z$ in each $P_i$.
Restriction `$\nu{x y} P$' binds $x$ and $y$ in $P$, thus declaring them as the two endpoints of the same channel and enabling communication, as in Vasconcelos~\cite{journal/ic/Vasconcelos12}.
The process `$\mkern-1mu P \| Q \mkern1mu$' denotes the parallel composition of $P$ and $Q$.
The process `$\0$' denotes inaction.
The forwarder process `$x \fwd y$' is a primitive copycat process that links together $x$ and $y$.
The prefix `$\mu X(\tilde{z}) \sdot P$' defines a recursive loop, where $\mu$ binds any free occurrences of $X$ in $P$ and the endpoints $\tilde{z}$ form a context for $P$.
The recursive call `$X\call{\tilde{z}}$' loops to its corresponding $\mu X$, providing the endpoints $\tilde{z}$ as context.
We only consider contractive recursion, disallowing processes with subexpressions of the form `$\mu X_1(\tilde{z}) \ldots \mu X_n(\tilde{z}) \sdot X_1\call{\tilde{z}}$'.

Endpoints and recursion variables are free unless they are bound somewhere.
We write `$\fn(P)$' and `$\frv(P)$' for the sets of free names and free recursion variables of $P$, respectively.
Also, we write `$P \subst{x/y}$' to denote the capture-avoiding substitution of the free occurrences of $y$ in $P$ for $x$.
The notation `$P\big\{(\mu X(\tilde{y}) \sdot Q) / X\call{\tilde{y}}\big\}$' denotes the substitution of occurrences of recursive calls `$X\call{\tilde{y}}$' for any sequence of names $\tilde{y}$  in $P$ with the recursive loop `$\mu X(\tilde{y}) \sdot Q$', which we call \emph{unfolding} recursion.
We write sequences of substitutions `$P \subst{x_1/y_1} \ldots \subst{x_n/y_n}$' as `$P \subst{x_1/y_1, \ldots, x_n/y_n}$'.

In an output `$x[y,z]$', both $y$ and $z$ are free, as mentioned above; they can be bound to a continuation process using parallel composition and restriction, as in, e.g., $\nu{y a}\nu{z b}(x[y,z] \| P_{a,b})$.
The same applies to selection `$x[z] \puts i$'.
We introduce useful notations that elide the restrictions and continuation endpoints:

\begin{notation}[Derivable Actions and Prefixes]\label{n:sugar}
    We use the following syntactic sugar:
    \begin{align*}
        \ol{x}[y] \cdot P
        & := \nu{y a}\nu{z b}(x[a,b] \| P\subst{z/x})
        &
        \ol{x} \puts \ell \cdot P
        & := \nu{z b}(x[b] \puts \ell \| P\subst{z/x})
        \\
        x(y) \sdot P
        & := x(y,z) \sdot P\subst{z/x}
        &
        x \gets \{i{:}~ P_i\}_{i \in I}
        & := x(z) \gets \{i{:}~ P_i\subst{z/x}\}_{i \in I}
    \end{align*}
\end{notation}

\noindent
Note the use of `${}\cdot{}$' instead of `${}\sdot{}$' in output and selection actions to stress that they are non-blocking.

\paragraph{Operational Semantics}

We define a reduction relation for processes ($P \redd Q$) that formalizes how complementary actions on connected endpoints may synchronize.
As usual for $\pi$-calculi, reduction relies on \emph{structural congruence} ($P \equiv Q$), which  equates the behavior of processes with minor syntactic differences; it is the smallest congruence relation satisfying the axioms in \Cref{f:procdef} (center).

Structural congruence defines the following properties of our process language.
Processes are equivalent up to $\alpha$-equivalence.
Parallel composition is associative and commutative, with unit `$\0$'.
The forwarder process is symmetric, and equivalent to inaction if both endpoints are bound together through restriction.
A parallel process may be moved into or out of a restriction as long as the bound channels do not appear free in the moved process: this is \emph{scope inclusion} and \emph{scope extrusion}, respectively.
Restrictions on inactive processes may be dropped, and the order of endpoints in restrictions and of consecutive restrictions does not matter.
Finally, a recursive loop is equivalent to its unfolding, replacing any recursive calls with copies of the recursive loop, where the call's endpoints are pairwise substituted for the contextual endpoints of the loop.

We define the reduction relation  by the axioms and closure rules in \Cref{f:procdef} (bottom).
Axioms labeled `$\beta$' are \emph{synchronizations} and those labeled `$\kappa$' are  \emph{commuting conversions}, which allow pulling prefixes on free channels out of restrictions; they are not necessary for deadlock freedom, but they are usually presented in Curry-Howard interpretations of linear logic as session types~\cite{conf/concur/CairesP10,conf/icfp/Wadler12,conf/fossacs/DardhaG18,conf/csl/DeYoungCPT12}.

Rule $\beta_{\scc{Id}}$ implements the forwarder as a substitution.
Rule $\beta_{\tensor \parr}$ synchronizes an output and an input on connected endpoints and substitutes the message and continuation endpoint.
Rule $\beta_{\oplus \&}$ synchronizes a selection and a branch:
the received label determines the continuation process, substituting the continuation endpoint appropriately.
Rule $\kappa_{\parr}$ (resp.\ $\kappa_{\&}$) pulls an input (resp.\ a branching) prefix on free channels out of enclosing restrictions.
Rules $\rightarrow_\equiv$, $\rightarrow_\onu$, and $\rightarrow_{\|}$ close reduction under structural congruence, restriction, and parallel composition, respectively.

\begin{notation}[Reductions]\label{n:redd}
    We write `$\redd_{\beta}$' for reductions derived from $\beta$-axioms, and `$\redd^\ast$' for the reflexive, transitive closure of `$\redd$'.
     Also, we write `$P \redd^\star Q$' if $P \redd^\ast Q$ in a finite number of steps, and `$P \nredd^\ast Q$' for the non-existence of a series of reductions from $P$ to $Q$.
\end{notation}

\paragraph{Session Types}

\begin{table}[t]
    \begin{mdframed}
        \begin{tabularx}{\textwidth}{l||X}
            \textbf{Session Type} & \textbf{Endpoint Behavior} \\ \hline
            \\[-7pt]
            $A \tensor^\pri B$ & output an endpoint of type $A$, then behave as $B$ \\
            $A \parr^\pri B$ & input an endpoint of type $A$, then behave as $B$ \\
            ${\oplus}^\pri \{ i{:}~ A_i \}_{i \in I}$ & select a label $i \in I$, then behave as $A_i$ \\
            $\&^\pri \{ i{:}~ A_i \}_{i \in I}$ & receive a choice for a label $i \in I$, then behave as $A_i$ \\
            $\bullet$ & closed session; no behavior
        \end{tabularx}
    \end{mdframed}

    \caption{Session types and their associated endpoint behaviors (cf.\ \Cref{d:props}).}
    \label{tbl:types}
\end{table}

The type system  assigns session types to channel endpoints.
We present session types as linear logic propositions following, e.g., Wadler~\cite{conf/icfp/Wadler12}, Caires and Pfenning~\cite{conf/esop/CairesP17}, and Dardha and Gay~\cite{conf/fossacs/DardhaG18}.
We extend these propositions with recursion and \emph{priority} annotations on connectives.
Intuitively, actions typed with lower priority should be performed before those with higher priority.
We write $\pri, \opri, \pi, \rho, \ldots$ to denote priorities, and `$\omega$' to denote the ultimate priority that is greater than all other priorities  and cannot be increased further.
That is, $\forall t \in \mbb{N}.~\omega > t$ and $\forall t \in \mbb{N}.~\omega + t = \omega$.

\begin{definition}[Session Types]\label{d:props}
    The following grammar defines the syntax of \emph{session types} $A,B$.
    Let $\pri \in \mbb{N} \cup \{\omega\}$.
    \begin{align*}
        A,B
        &::= A \tensor^\pri B \sepr A \parr^\pri B \sepr {\oplus}^\pri \{i: A\}_{i \in I} \sepr \&^\pri \{i: A\}_{i \in I} \sepr \bullet \sepr \mu X \sdot A \sepr X
    \end{align*}
\end{definition}

\noindent
\Cref{tbl:types} gives \emph{session types} and the behavior that is expected of an endpoint with each type  (recursive types entail a communication behavior only after unfolding).
Note that `$\bullet$' does not require a priority, as closed endpoints do not exhibit behavior and thus are non-blocking.
We define `$\bullet$' as a single, self-dual type for closed endpoints (cf.\ Caires~\cite{report/Caires14} and Atkey \etal~\cite{chapter/AtkeyLM16}).

Type `$\mu X \sdot A$' denotes a recursive type, in which $A$ may contain occurrences of the recursion variable `$X$'.
As customary, `$\mu$' is a binder: it induces the standard notions of $\alpha$-equivalence, substitution (denoted `$A\subst{B/X}$'), and free recursion variables (denoted `$\frv(A)$').
We work with tail-recursive, contractive types, disallowing types of the form `$\mu X_1 \ldots \mu X_n \sdot X_1$'.
 We postpone the formalization of the unfolding of recursive types, as it requires additional definitions to ensure consistency of priorities in types.

\emph{Duality}, the cornerstone of session types and linear logic, ensures that the two endpoints of a channel have matching actions.
Furthermore, dual types must have matching priority annotations.
The following inductive definition of duality suffices for our tail-recursive types (cf.\ Gay \etal~\cite{conf/places/GayTV20}).

\begin{definition}[Duality]\label{d:duality}
    The \emph{dual} of session type $A$, denoted `$\ol{A}$', is defined inductively as follows:
    \begin{align*}
        \ol{A \tensor^\pri B}
        &:= \ol{A} \parr^\pri \ol{B}
        &
        \ol{{\oplus}^\pri \{ i: A_i \}_{i \in I}}
        &:= \&^\pri \{ i: \ol{A_i} \}_{i \in I}
        &
        \ol{\bullet}
        &:= \bullet
        &
        \ol{\mu X \sdot A}
        &:= \mu X \sdot \ol{A}
        \\
        \ol{A \parr^\pri B}
        &:= \ol{A} \tensor^\pri \ol{B}
        &
        \ol{\&^\pri \{ i: A_i \}_{i \in I}}
        &:= {\oplus}^\pri \{ i: \ol{A_i} \}_{i \in I}
        &
        &
        &
        \ol{X}
        &:= X
    \end{align*}
\end{definition}

The priority of a type is determined by the priority of the type's outermost connective:

\begin{definition}[Priorities]\label{d:priority}
    For session type $A$, `$\pr(A)$' denotes its \emph{priority}:
    \begin{align*}
        \pr(A \tensor^\pri B)
        := \pr(A \parr^\pri B)
        &:= \pri
        &
        \pr(\mu X \sdot A)
        &:= \pr(A)
        \\
        \pr(\oplus^\pri\{i{:}~ A_i\}_{i \in I})
        := \pr(\&^\pri\{i{:}~ A_i\}_{i \in I})
        &:= \pri
        &
        \pr(\bullet)
        := \pr(X)
        &:= \omega
    \end{align*}
\end{definition}

\noindent
The priority of `$\bullet$' and `$X$' is $\omega$: they denote ``final'', non-blocking actions of protocols.
Although `$\tensor$' and `$\oplus$' also denote non-blocking actions, their priority is not constant:
duality ensures that the priority for `$\tensor$' (resp.\ `$\oplus$') matches the priority of a corresponding `$\parr$' (resp.\ `$\&$'), which  denotes a blocking action.

Having defined the priority of types, we now turn to formalizing the unfolding of recursive types.
Recall the intuition that actions typed with lower priority should be performed before those with higher priority.
Based on this rationale, we observe that unfolding should increase the priorities of the unfolded type.
This is because the actions related to the unfolded recursion should be performed \emph{after} the prefix.
The following definition \emph{lifts} priorities in types:

\begin{definition}[Lift]\label{d:lift}
    For proposition $A$ and $t \in \mbb{N}$, we define `$\mkern2mu\lift{t}A\mkern-3mu$' as the \emph{lift} operation:
    \begin{align*}
        \lift{t}(A \tensor^\pri B)
        &:= (\lift{t}A) \tensor^{\pri+t} (\lift{t}B)
        &
        \lift{t}(\oplus^\pri \{i{:}~ A_i\}_{i \in I})
        &:= \oplus^{\pri+t} \{i{:}~ \lift{t}A_i\}_{i \in I}
        &
        \lift{t}\bullet
        &:= \bullet
        \\
        \lift{t}(A \parr^\pri B)
        &:= (\lift{t}A) \parr^{\pri+t} (\lift{t}B)
        &
        \lift{t}(\&^\pri \{i{:}~ A_i\}_{i \in I})
        &:= \&^{\pri+t} \{i{:}~ \lift{t}A_i\}_{i \in I}
        \\
        \lift{t}(\mu X \sdot A)
        &:= \mu X \sdot \lift{t}(A)
        &
        \lift{t}X
        &:= X
    \end{align*}
\end{definition}

\begin{definition}\label{d:unf}
    The \emph{unfolding} of `$\mu X \sdot A$' is `$A\subst{\mu X \sdot (\lift{t} A)/X}$', denoted `$\unfold^t(\mu X \sdot A)$', where $t \in \mbb{N}$.
\end{definition}

When unfolding $\mu X \sdot A$ as $\unfold^t(\mu X \sdot A)$,
the ``lifter'' $t$ will depend on the highest priority of the types appearing in a typing context.
The highest priority of a type is defined as follows:

\begin{definition}[Highest Priority]\label{d:maxpr}
    For session type $A$, `$\max_\pr(A)$' denotes its \emph{highest priority}:
    \begin{align*}
        \max_\pr(A \tensor^\pri B) := \max_\pr(A \parr^\pri B)
        &:= \max(\max_\pr(A), \max_\pr(B), \pri)
        &
        \max_\pr(\mu X \sdot A)
        &:= \max_\pr(A)
        \\
        \max_\pr(\oplus^\pri \{i: A_i\}_{i \in I}) := \max_\pr(\&^\pri \{i: A_i\}_{i \in I})
        &:= \max(\max_{i \in I}(\max_\pr(A_i)), \pri)
        &
        \max_\pr(\bullet) := \max_\pr(X)
        &:= 0
    \end{align*}
\end{definition}

\noindent
Notice how, in contrast to \Cref{d:priority}, the highest priority of `$\bullet$' and `$X$' is 0: this is because they do not contribute to the increase in priority needed for unfolding recursive types.

\paragraph{Type Checking}
The typing (or, type checking) rules of APCP enforce that channel endpoints implement their ascribed session types, while ensuring that actions with lower priority are performed before those with higher priority (cf.\ Dardha and Gay~\cite{conf/fossacs/DardhaG18}).
They enforce the following laws:
\begin{enumerate}
    \item
        an action with priority $\pri$ must be prefixed only by inputs and branches with priority strictly smaller than $\pri$---this law only applies to inputs and branches, because outputs and selections are not prefixes;

    \item
        dual actions leading to synchronizations must have equal priorities (cf.\ Def.\ \labelcref{d:props}).
\end{enumerate}
Judgments are of the form `$P \vdash \Omega; \Gamma$':
\begin{itemize}
    \item
        $P$ is a process;
    \item
        $\Gamma$ is a context that assigns types to channels (`$x{:}~ A$');
    \item
        $\Omega$ is a context that assigns tuples of types to recursion variables (`$X{:}~ (A, B, \ldots)$').
\end{itemize}
A judgment `$P \vdash \Omega; \Gamma$' then means that $P$ can be typed in accordance with the type assignments for names recorded in $\Gamma$ and the recursion variables in $\Omega$.
Intuitively, the recursive context $\Omega$ ensures that the context endpoints concur between recursive definitions and calls.
Both contexts $\Gamma$ and $\Omega$ obey \emph{exchange}: assignments may be silently reordered.
$\Gamma$ is \emph{linear}, disallowing \emph{weakening} (i.e., all assignments must be used) and \emph{contraction} (i.e., assignments may not be duplicated).
$\Omega$ allows weakening and contraction, because a recursive definition may be called \emph{zero or more} times.

The empty context is written `$\emptyset$'.
We write `$\lift{t} \Gamma$' to denote the component-wise extension of lift (\Cref{d:lift}) to typing contexts.
Also, we write `$\pr(\Gamma)$' to denote the least priority of all types in $\Gamma$ (\Cref{d:priority}).
An assignment `$\tilde{z}{:}~ \tilde{A}$' means `$z_1{:}~ A_1, \ldots, z_k{:}~ A_k$'.

\begin{figure}[t]
    \begin{mdframed}
        {
            \mbox{}\hfill%
            \begin{wfit}
                \begin{prooftree}
                    \infAss{
                        $\vphantom{; \Gamma}$
                    }
                    \infUn{
                        $\0 \vdash \Omega; \emptyset$
                    }{\scc{Empty}}
                \end{prooftree}
            \end{wfit}%
            \hfill%
            \begin{wfit}
                \begin{prooftree}
                    \infAss{
                        $P \vdash \Omega; \Gamma$
                    }
                    \infUn{
                        $P \vdash \Omega; \Gamma, x{:}~ \bullet$
                    }{$\bullet$}
                \end{prooftree}
            \end{wfit}%
            \hfill%
            \begin{wfit}
                \begin{prooftree}
                    \infAss{
                        $\vphantom{; \Gamma}$
                    }
                    \infUn{
                        $x \fwd y \vdash \Omega; x{:}~ \ol{A}, y{:}~ A$
                    }{\scc{Id}}
                \end{prooftree}
            \end{wfit}%
            \hfill\mbox{}

            \smallskip
            \mbox{}\hfill%
            \begin{wfit}
                \begin{prooftree}
                    \infAss{
                        $P \vdash \Omega; \Gamma$
                        $\raisebox{1.8ex}{}$
                    }
                    \infAss{
                        $Q \vdash \Omega; \Delta$
                    }
                    \infBin{
                        $P \| Q \vdash \Omega; \Gamma, \Delta$
                    }{\scc{Mix}}
                \end{prooftree}
            \end{wfit}%
            \hfill%
            \begin{wfit}
                \begin{prooftree}
                    \infAss{
                        $P \vdash \Omega; \Gamma, x{:}~ A, y{:}~ \ol{A}$
                    }
                    \infUn{
                        $\nu{x y} P \vdash \Omega; \Gamma$
                    }{\scc{Cycle}}
                \end{prooftree}
            \end{wfit}%
            \hfill\mbox{}

            \smallskip
            \mbox{}\hfill%
            \begin{wfit}
                \begin{prooftree}
                    \infAss{
                        $\vphantom{\pr(\Gamma)}$
                    }
                    \infUn{
                        $x[y,z] \vdash \Omega; x{:}~ A \tensor^\pri B, y{:}~ \ol{A}, z{:}~ \ol{B}$
                    }{$\tensor$}
                \end{prooftree}
            \end{wfit}%
            \hfill%
            \begin{wfit}
                \begin{prooftree}
                    \infAss{
                        $P \vdash \Omega; \Gamma, y{:}~ A, z{:}~ B$
                    }
                    \infAss{
                        $\pri < \pr(\Gamma)$
                    }
                    \infBin{
                        $x(y, z) \sdot P \vdash \Omega; \Gamma, x{:}~ A \parr^\pri B$
                    }{$\parr$}
                \end{prooftree}
            \end{wfit}%
            \hfill\mbox{}

            \smallskip
            \mbox{}\hfill%
            \begin{wfit}
                \begin{prooftree}
                    \infAss{
                        $j \in I \vphantom{P_i \pr(\Gamma)}$
                    }
                    \infUn{
                        $x[z] \triangleleft j \vdash \Omega; x{:}~ {\oplus}^\pri\{i: A_i\}_{i \in I}, z{:}~ \ol{A_j} \vphantom{j}$
                    }{$\oplus$}
                \end{prooftree}
            \end{wfit}%
            \hfill%
            \begin{wfit}
                \begin{prooftree}
                    \infAss{
                        $\forall i \in I.~ P_i \vdash \Omega; \Gamma, z{:}~ A_i$
                    }
                    \infAss{
                        $\pri < \pr(\Gamma)$
                    }
                    \infBin{
                        $x(z) \triangleright \{i: P_i\}_{i \in I} \vdash \Omega; \Gamma, x{:}~ \&^\pri\{i: A_i\}_{i \in I}$
                    }{$\&$}
                \end{prooftree}
            \end{wfit}%
            \hfill\mbox{}

            \medskip
            \mbox{}\hfill%
            \begin{wfit}
                \begin{prooftree}
                    \infAss{
                        $t \in \mbb{N} > \max_\pr(\tilde{A})$
                    }
                    \infAss{
                        $P \vdash \Omega, X{:}~ {\tilde{A}}; \tilde{z}{:}~ \tilde{U}$ ~ where each $U_i = \unfold^t(\mu X \sdot A_i)$
                    }
                    \infAss{
                        $\forall A_i \in \tilde{A}.~ A_i \neq X$
                    }
                    \infTer{
                        $\mu X(\tilde{z}) \sdot P \vdash \Omega; \tilde{z}{:}~ \widetilde{\mu X \sdot A}$
                    }{\scc{Rec}}
                \end{prooftree}
            \end{wfit}%
            \hfill\mbox{}

            \medskip
            \mbox{}\hfill%
            \begin{wfit}
                \begin{prooftree}
                    \infAx{
                        $X\call{\tilde{z}} \vdash \Omega, X{:}~ \tilde{A}; \tilde{z}{:}~ \widetilde{\mu X \sdot A}$
                    }{\scc{Var}}
                \end{prooftree}
            \end{wfit}%
            \hfill\mbox{}

            \medskip
            \hbox to \textwidth{\leaders\hbox to 3pt{\hss . \hss}\hfil}

            \medskip
            \mbox{}\hfill%
            \begin{wfit}
                \begin{prooftree}
                    \infAss{
                        $P \vdash \Omega; \Gamma, y{:}~ A, x{:}~ B$
                    }
                    \infUn{
                        $\ol{x}[y] \cdot P \vdash \Omega; \Gamma, x{:}~ A \tensor^\pri B$
                    }{$\tensor^\star$}
                \end{prooftree}
            \end{wfit}%
            \hfill%
            \begin{wfit}
                \begin{prooftree}
                    \infAss{
                        $P \vdash \Omega; \Gamma, x{:}~ A_j$
                    }
                    \infAss{
                        $j \in I$
                    }
                    \infBin{
                        $\ol{x} \triangleleft j \cdot P \vdash \Omega; \Gamma, x{:}~ {\oplus}^\pri\{i: A_i\}_{i \in I}$
                    }{$\oplus^\star$}
                \end{prooftree}
            \end{wfit}%
            \hfill%
            \begin{wfit}
                \begin{prooftree}
                    \infAss{
                        $P \vdash \Omega; \Gamma$
                    }
                    \infAss{
                        $t \in \mbb{N}$
                    }
                    \infBin{
                        $P \vdash \Omega; \lift{t} \Gamma$
                    }{\scc{Lift}}
                \end{prooftree}
            \end{wfit}%
            \hfill\mbox{}%
        }%
    \end{mdframed}

    \caption{The typing rules of APCP (top) and admissible rules (bottom).}
    \label{f:apcpInf}
\end{figure}

\Cref{f:apcpInf} (top) gives the typing rules.
In typing rules, we often write `$\Gamma, x{:}~ A$' (or similarly for $\Omega$) to denote \emph{disjoint union}, i.e.\ $x \notin \dom(\Gamma)$.

Some type-preserving transformations of typing derivations correspond to process reductions (cf.\ \Cref{t:subjectRed}).
Other such transformations correspond to structural congruences (cf.\ \Cref{f:procdef} (middle)); we sometimes use this explicitly in typing derivations in the form of a rule `$\equiv$'.
If $P \equiv Q$ and $P \vdash \Omega; \Gamma$ and $Q \vdash \Omega; \Gamma'$ where $\Gamma$ and $\Gamma'$ are equal up to the unfolding of recursive types, then we say that $P \vdash \Omega; \Gamma'$ and $Q \vdash \Omega; \Gamma$; in the context of a typing derivation, we equate recursive types and their unfoldings.

We describe the typing rules from a \emph{bottom-up} perspective.
Axiom `\scc{Empty}' types an inactive process with no endpoints.
Rule `$\bullet$' silently removes a closed endpoint to the typing context.
Axiom `\scc{Id}' types forwarding between endpoints of dual type.
Rule `\scc{Mix}' types the parallel composition of two processes that do not share assignments on the same endpoints.
Rule `\scc{Cycle}' types a restriction, where the two restricted endpoints must be of dual type.
Note that a single application of `\scc{Mix}' followed by `\scc{Cycle}' coincides with the usual rule `\scc{Cut}' in type systems based on linear logic~\cite{conf/concur/CairesP10,conf/icfp/Wadler12}.
Axiom `$\tensor$' types an output action; this rule does not have premises to provide a continuation process, leaving the free endpoints to be bound to a continuation process using `\scc{Mix}' and `\scc{Cycle}'.
Similarly, axiom `$\oplus$' types an unbounded selection action.
Priority checks are confined to rules `$\parr$' and `$\&$', which type an input and a branching prefix, respectively.
In both cases, the used endpoint's priority must be lower than the priorities of the other types in the continuation's typing context.

Rule `\scc{Rec}' types a recursive definition by introducing a recursion variable to the recursion context whose tuple of types concurs with the contents of the recursive types in the typing context, where contractivity is guaranteed by requiring that the eliminated recursion variable may not appear unguarded in each of the context's types.
At the same time, the recursive types in the context are unfolded, and their priorities are lifted by a common value, denoted $t$ in the rule, that must be greater than the highest priority appearing in the original types (cf.\ \Cref{d:maxpr}).
Using a ``common lifter'', i.e., lifting the priorities of all types by the same amount is crucial: it maintains the relation between the priorities of the types in the context.

Axiom `\scc{Var}' types a recursive call on a variable in the recursive context.
The rule requires that all the types in the context are recursive on the recursion variable called, and that the types inside the recursive definitions concur with the respective types assigned to the recursion varialbe in the recursive context.
As mentioned before, the types associated to the introduced and consequently eliminated recursion variable is crucial in ensuring that a recursion is called with endpoints of the same type as required by its definition.

The binding of output and selection actions to continuation processes (\Cref{n:sugar}) is derivable in APCP.
The corresponding typing rules in \Cref{f:apcpInf} (bottom) are admissible using `\scc{Mix}' and `\scc{Cycle}' (cf.~\cite{report/vdHeuvelP21B}).
\Cref{f:apcpInf} (bottom) also includes an admissible rule `\scc{Lift}' that lifts a process' priorities.

The following result assures that, given a type, we can construct a process with an endpoint typable with the given type:

\begin{proposition}\label{p:charProc}
    Given a type $A$, there exists a $P$ such that $P \vdash \Omega; x{:}~A$.
\end{proposition}

\begin{proof}
    We inductively define a function `$\chr^x(A)$' that, given a type $A$ and an endpoint $x$, constructs a process that performs the behavior described by $A$:
    \begin{align*}
        \chr^x(A \tensor^\pri B)
        &:= \ol{x}[y] \cdot (\chr^y(A) \| \chr^x(B))
        &
        \chr^x(\oplus^\pri \{i:A_i\}_{i \in I})
        &:= \ol{x} \puts j \cdot \chr^x(A_j)
        \quad \text{[any $j \in I$]}
        \\
        \chr^x(A \parr^\pri B)
        &:= x(y) \sdot (\chr^y(A) \| \chr^x(B))
        &
        \chr^x(\&^\pri \{i:A_i\}_{i \in I})
        &:= x \gets \{i:\chr^x(A_i)\}_{i \in I}
        \\
        \chr^x(\bullet)
        &:= \0
        \qquad\qquad
        \chr^x(\mu X \sdot A)
        := \mu X(x) \sdot \chr^x(A)
        \qquad\qquad
        \chr^x(X)
        := X\call{x}
        \span\span
    \end{align*}
    For finite types,  we have: $\chr^x(A) \vdash \emptyset; x{:}~A$.
    For simplicity, we omit details about recursive types, which require unfolding.
    For closed, recursive types, we have: $\chr^x(\mu X \sdot A) \vdash \emptyset; x{:}~ \mu X \sdot A$.
\end{proof}

\paragraph{Type Preservation}

Well-typed processes satisfy protocol fidelity, communication safety, and deadlock freedom.
The first two properties follow directly from \emph{type preservation} (also known as \emph{subject reduction}), which ensures that reduction preserves typing.
In contrast to Caires and Pfenning~\cite{conf/concur/CairesP10} and Wadler~\cite{conf/icfp/Wadler12}, where type preservation corresponds to the elimination of (top-level) applications of rule \scc{Cut}, in APCP it corresponds to the more general elimination of (top-level) applications of rule \scc{Cycle}.

\begin{theorem}[Type Preservation~\cite{report/vdHeuvelP21B}]\label{t:subjectRed}
    If \,$P \vdash \Omega; \Gamma$ and $P \redd Q$, then $Q \vdash \Omega; \lift{t}\Gamma$ for $t \in \mbb{N}$.
\end{theorem}

\paragraph{Deadlock Freedom}

The deadlock freedom result for APCP adapts that for PCP~\cite{conf/fossacs/DardhaG18}.
As mentioned before, binding  asynchronous outputs and selections to continuations involves additional, low-level uses of \scc{Cycle}, which we cannot eliminate through process reduction.
Therefore,  top-level deadlock freedom holds for \emph{live processes} (\Cref{t:dlFree}).
A process is live if it is equivalent to a restriction on \emph{active names} that perform unguarded actions.
This way, e.g., in `$x[y,z]$' the name $x$ is active, but $y$ and $z$ are not.

\begin{definition}[Active Names]\label{d:an}
    The \emph{set of active names} of $P$, denoted `$\an(P)\mkern-2mu$', contains the (free) names that are used for unguarded actions (output, input, selection, branching):
    \begin{align*}
        \an(x[y,z])
        & := \{x\}
        &
        \an(x(y,z) \sdot P)
        & := \{x\}
        &
        \an(\0)
        & := \emptyset
        \\
        \an(x[z] \puts j)
        & := \{x\}
        &
        \an(x(z) \gets \{i{:}~ P_i\}_{i \in I})
        & := \{x\}
        &
        \an(x \fwd y)
        & := \{x,y\}
        \\
        \an(P \| Q)
        & := \an(P) \cup \an(Q)
        &
        \an(\mu X(\tilde{x}) \sdot P)
        & := \an(P)
        \\
        \an(\nu{x y}P)
        & := \an(P) \setminus \{x,y\}
        &
        \an(X\call{\tilde{x}})
        & := \emptyset
    \end{align*}
\end{definition}

\begin{definition}[Live Process]\label{d:live}
    A process $P$ is \emph{live}, denoted `$\mkern2mu\live(P)\mkern-3mu$', if there are names $x,y$ and process $P'$ such that $P \equiv \nu{x y}P'$ with $x,y \in \an(P')$.
\end{definition}

\noindent
We additionally need to account for recursion: as recursive definitions do not entail reductions, we must fully unfold them before eliminating \scc{Cycle}s.

\begin{lemma}[Unfolding]\label{l:unfold}
    If $P \vdash \Omega; \Gamma$, then there is  a  process $P^\star$ such that $P^\star \equiv P$ and $P^\star$ is not of the form `$\mu X(\tilde{z}); P'\mkern-2mu$' and $P^\star \vdash \Omega; \Gamma$.
\end{lemma}

\noindent
Deadlock freedom, given next, states that typable processes that are live can reduce.
It follows from an analysis of the priorities in the typing of the process, which makes it possible to find a pair of non-blocked, parallel, dual actions on connected endpoints, such that a communication can occur.
The analysis also considers the possibility that a blocking action is on an endpoint which is not connected (i.e., the endpoint is free), in which case a commuting conversion can be performed.
Confer the full proof by Van den Heuvel and Pérez~\cite[Theorem~5]{report/vdHeuvelP21B} for more details.

\begin{theorem}[Deadlock Freedom]\label{t:dlFree}
    If $P \vdash \emptyset; \Gamma$ and $\live(P)$, then there is process $Q$ such that $P \redd Q$.
\end{theorem}

\noindent
We now state the deadlock freedom result formalized by Van den Heuvel and Pérez~\cite{report/vdHeuvelP21B}.
Following, e.g., Caires and Pfenning~\cite{conf/concur/CairesP10} and Dardha and Gay~\cite{conf/fossacs/DardhaG18}, it concerns processes typable under empty contexts.
This way, the reduction guaranteed by \Cref{t:dlFree} corresponds to a synchronization ($\beta$-rule), rather than a commuting conversion ($\kappa$-rule).

\begin{theorem}[Deadlock Freedom for Processes Typable under Empty Contexts~\cite{report/vdHeuvelP21B}]\label{t:closedDLFree}
    If $P \vdash \emptyset; \emptyset$, then either $P \equiv \0$ or $P \redd_\beta Q\mkern2mu$ for some $Q$.
\end{theorem}

\paragraph{Fairness}

Processes typable under empty contexts are not only deadlock free, they are \emph{fair}: for each endpoint in the process, we can eventually observe a reduction involving that endpoint.
To formalize this property, we define \emph{labeled reductions}, which expose details about a communication:

\begin{definition}[Labeled Reductions]\label{d:procLred}
    Consider the labels
    \begin{align*}
        \alpha ::= x \fwd y ~\sepr~ x \rangle y{:}a ~\sepr~ x \rangle y{:}\ell
        \qquad\qquad
        \text{(forwarding, output/input, selection/branching)}
    \end{align*}
    where each label has subjects $x$ and $y$.
    The \emph{labeled reduction} `$\mkern1mu P \xredd{\alpha} Q \mkern-3mu$' is defined by the following rules:
    \begin{align*}
        \nu{y z}( x \fwd y \| P ) \xredd{x \fwd y} P \subst{x/z}
        \qquad
        \nu{x y}( x[a,b] \| y(v,z) \sdot P ) \xredd{x \rangle y{:}a} P \subst{a/v, b/z}
        \\
        \nu{x y}( x[b] \puts j \| y(z) \gets \{i{:}~ P_i\}_{i \in I} ) \xredd{x \rangle y{:}j} P_j \subst{b/z} \quad \text{(if $j \in I$)}
    \end{align*}

    \medskip
    \mbox{}\hfill%
    \begin{wfit}
        \begin{prooftree}
            \infAss{
                $(P \equiv P') \wedge (P' \xredd{\alpha} Q') \wedge (Q' \equiv Q)$
            }
            \infUn{
                $P \xredd{\alpha} Q$
            }{}
        \end{prooftree}
    \end{wfit}%
    \hfill%
    \begin{wfit}
        \begin{prooftree}
            \infAss{
                $P \xredd{\alpha} Q$
            }
            \infUn{
                $\nu{x y} P \xredd{\alpha} \nu{x y} Q$
            }{}
        \end{prooftree}
    \end{wfit}%
    \hfill%
    \begin{wfit}
        \begin{prooftree}
            \infAss{
                $P \xredd{\alpha} Q$
            }
            \infUn{
                $P \| R \xredd{\alpha} Q \| R$
            }{}
        \end{prooftree}
    \end{wfit}%
    \hfill\mbox{}
\end{definition}

\begin{proposition}
    For any $P$ and $P'$, $P \redd_\beta P'$ if and only if there exists a label $\alpha$ such that $P \xredd{\alpha} P'$.
\end{proposition}

\begin{proof}
    Immediate by definition, for each $\beta$-reduction in \Cref{f:procdef} (bottom) corresponds to a labeled reduction, and vice versa.
\end{proof}

Our fairness result states that processes typable under empty contexts have at least one finite reduction sequence (`$\redd^\star$') that enables a labeled reduction involving a \emph{pending} endpoint---an endpoint that occurs as the subject of an action, and is not bound by input or branching (see below).
Clearly, the typed process may have other reduction sequences, not necessarily finite.

\begin{definition}[Pending Names]\label{d:pending}
    Given a process $P$, we define the set of \emph{pending names} of $P$, denoted `$\mkern1mu\pn(P)\mkern-3mu$', as follows:
    \begin{align*}
        \pn(x[y,z])
        &:= \{x\}
        &
        \pn(x(y,z).P)
        &:= \{x\} \cup (\pn(P) \setminus \{y,z\})
        &
        \pn(\0)
        &:= \emptyset
        \\
        \pn(x[z] \puts j)
        &:= \{x\}
        &
        \pn(x(z) \gets \{i: P_i\}_{i \in I})
        &:= \{x\} \cup (\bigcup_{i \in I} \pn(P_i) \setminus \{z\})
        &
        \pn(x \fwd y)
        &:= \{x,y\}
        \\
        \pn(P \| Q)
        &:= \pn(P) \cup \pn(Q)
        &
        \pn(\mu X(\tilde{x}) \sdot P)
        &:= \pn(P)
        \\
        \pn(\nu{x y}P)
        &:= \pn(P)
        &
        \pn(X \call{\tilde{x}})
        &:= \emptyset
    \end{align*}
\end{definition}

\begin{theorem}[Fairness]\label{t:fairness}
    Suppose given a process $P \vdash \emptyset; \emptyset$.
    Then, for every $x \in \pn(P)$ there exists a process $P'$ such that $P \redd^\star P'$ and $P' \xredd{\alpha}\, Q$, for some process $Q$ and label $\alpha$ with subject $x$.
\end{theorem}

\begin{proof}
    Take any $x \in \pn(P)$.
    Because $P$ is typable under empty contexts, $x$ is bound to some $y \in \pn(P)$ by restriction.
    By typing, in $P$ there is exactly one action on $x$ and one action on $y$ (they may also appear in forwarder processes).
    Following the restrictions on priorities in the typing of $x$ and $y$ in $P$, the actions on $x$ and $y$ cannot appear sequentially in $P$ (cf.\ the proof by Van den Heuvel and Pérez~\cite{report/vdHeuvelP21B} for details on this reasoning).
    By typability, the action on $y$ is dual to the action on $x$.

    We apply induction on the number of inputs, branches, and recursive definitions in $P$ blocking the actions on $x$ and $y$, denoted $n$ and $m$, respectively.
    Because $P$ is typable under empty contexts, the blocking inputs and branches that are on names in $\pn(P)$ also have to be bound to pending names by restriction.
    The actions on these connected names may also be prefixed by inputs, branches, and recursive definitions, so we may need to unblock those actions as well.
    Since there can only be a finite number of names in any given process, we also apply induction on the number of prefixes blocking these connected actions.
    \begin{itemize}
        \item
            If $n = 0$ and $m = 0$, then the actions on $x$ and $y$ occur at the top-level; because they do not appear sequentially, the communication between $x$ and $y$ can take place immediately.
            Hence, $P \xredd{\alpha} Q$ where $x$ and $y$ are the subjects of $\alpha$.
            This proves the thesis, with $P' = P$.

        \item
            If $n > 0$ or $m > 0$, the analysis depends on the foremost prefix of the actions on $x$ and $y$.

            If the foremost prefix of either action is a recursive definition (`$\mu X(\tilde{y})$'), we unfold the recursion.
            Because a corresponding recursive call (`$X\call{\tilde{z}}$') cannot occur as a prefix, the effect of unfolding either (i)~triggers actions that occur in parallel to those on $x$ and $y$, or (ii)~the actions on $x$ or $y$ prefix the unfolded recursive call.
            In either case, the number of prefixes decreases, and the thesis follows from the IH.

            Otherwise, if neither foremost prefix is a recursive definition, then the foremost prefixes must be actions on names in $\pn(P)$.
            Consider the action that is typable with the least priority.
            W.l.o.g.\ assume that this is the foremost prefix of $x$.
            Suppose this action is on some endpoint $w$ connected to another endpoint $z \in \pn(P)$ by restriction.
            By typability, the priority of $w$ is less than that of $x$ and all of the prefixes in between.
            This means that the number of prefixes blocking the action on $z$ strictly decreases.
            Hence, by the IH, $P \redd^\star P'' \xredd{\alpha'}\, Q'$ in a finite number of steps, where $w$ and $z$ are the subjects of $\alpha'$.
            The communication between $w$ and $z$ can be performed, and $n$ decreases.
            By Type Preservation (\Cref{t:subjectRed}), $Q' \vdash \emptyset; \emptyset$.
            The thesis then follows from the IH: $P \redd^\star P'' \xredd{\alpha'}\, Q' \redd^\star P' \xredd{\alpha}\, Q$ in finite steps, where $x$ and $y$ are the subjects of $\alpha$.
            \qedhere
    \end{itemize}
\end{proof}

\subsection*{Examples}

To illustrate APCP processes and their session types,  we give implementations of the three participants in $G_\sff{auth}$ in \Cref{s:introduction}.

\begin{example}\label{ex:impl}
Processes $P$, $Q$, and $R$ are  typed implementations for  participants $c$, $s$, and $a$, respectively, where each process uses a single channel to perform the actions described by $G_\sff{auth}$.
\begin{align*}
    P
    &:= \mu X(\ci{c}{\mu}) \sdot \ci{c}{\mu} \gets \left\{ \hspace{-.5em} \begin{array}{ll}
            \sff{login}{:}
            & \ci{c}{\mu}(u) \sdot \ol{\ci{c}{\mu}} \puts \sff{passwd} \cdot \ol{\ci{c}{\mu}}[\bm{logmein345}] \cdot X \call{\ci{c}{\mu}} ,
            \\
            \sff{quit}{:}
            & \ci{c}{\mu}(w) \sdot \ol{\ci{c}{\mu}} \puts \sff{quit} \cdot \ol{\ci{c}{\mu}}[z] \cdot \0
    \end{array} \hspace{-.5em} \right\}
    \\
    &\vdash \ci{c}{\mu} {:}~ \mu X \sdot \&^{2} \left\{ \hspace{-.5em} \begin{array}{ll}
            \sff{login}{:}
            & \bullet \parr^{3} \oplus^{4} \{ \sff{passwd}{:}~ \bullet \tensor^{5} X \} ,
            \\
            \sff{quit}{:}
            & \bullet \parr^{3} \oplus^{4} \{ \sff{quit}{:}~ \bullet \tensor^{5} \bullet \}
    \end{array} \hspace{-.5em} \right\}
    \\[5pt]
    Q
    &:= \mu X(\ci{s}{\mu}) \sdot \ol{\ci{s}{\mu}} \puts \sff{login} \cdot \ol{\ci{s}{\mu}}[u] \cdot \ci{s}{\mu} \gets \{ \sff{auth}{:}~ \ci{s}{\mu}(v) \sdot X \call{\ci{s}{\mu}} \}
    \\
    &\vdash \ci{s}{\mu} {:}~ \mu X \sdot \oplus^{0} \{ \sff{login}{:}~ \bullet \tensor^{1} \&^{10} \{ \sff{auth}{:}~ \bullet \parr^{11} X \} , \sff{quit}{:}~ \bullet \tensor^{1} \bullet \}
    \\[5pt]
    R
    &:= \mu X(\ci{a}{\mu}) \sdot \ci{a}{\mu} \gets \left\{ \hspace{-.5em} \begin{array}{ll}
            \sff{login}{:}
            & \ci{a}{\mu} \gets \{ \sff{passwd}{:}~ \ci{a}{\mu}(u) \sdot \ol{\ci{a}{\mu}} \puts \sff{auth} \cdot \ol{\ci{a}{\mu}}[v] \cdot X \call{\ci{a}{\mu}} \} ,
            \\
            \sff{quit}{:}
            & \ci{a}{\mu} \gets \{ \sff{quit}{:}~ \ci{a}{\mu}(w) \sdot \0 \}
    \end{array} \hspace{-.5em} \right\}
    \\
    &\vdash \ci{a}{\mu} {:}~ \mu X \sdot \&^{2} \left\{ \hspace{-.5em} \begin{array}{ll}
            \sff{login}{:}
            & \&^{6} \{ \sff{passwd}{:}~ \bullet \parr^{7} \oplus^{8} \{ \sff{auth}{:}~ \bullet \tensor^{9} X \} \} ,
            \\
            \sff{quit}{:}
            & \&^{6} \{ \sff{quit}{:}~ \bullet \parr^{7} \bullet \}
    \end{array} \hspace{-.5em} \right\}
\end{align*}
Process $P$ is a specific implementation for $c$, where we use `$\bm{logmein345}$' to denote a closed channel endpoint representing the password string ``logmein345''.
Similarly, $Q$ is a specific implementation for $s$ that continuously chooses the \sff{login} branch.
\end{example}

Note that the processes above cannot be directly connected to each other to implement $G_\sff{auth}$. Our goal is to enable the composition of (typed) implementations such as  $P$, $Q$, and $R$ in a correct and deadlock free manner. We shall proceed as follows.
After setting up the routers that enable the composition of these processes according to $G_\sff{auth}$ (\Cref{s:routers}), we will return to this example in \Cref{s:routersInAction}. At that point, it will become clear that the priorities in the types of $P$, $Q$, and $R$ were chosen to ensure the correct composition  with their respective routers.

\section{Global Types and Relative Projection}
\label{s:mpst}

We analyze multiparty protocols specified as \emph{global types}.
We consider a standard syntax, with session delegation and recursion, subsuming the one given in the seminal paper by Honda \etal~\cite{journal/acm/HondaYC16}.
In the following, we write $p, q, r, s, \ldots$ to denote \emph{(protocol) participants}.

\begin{definition}[Types]\label{d:globtypes}
    \emph{Global types} $G$ and \emph{message types} $S,T$ are defined as:
    \begin{align*}
        G &::= p \mto q \{i\<S\> \sdot G\}_{i \in I} \sepr \mu X \sdot G \sepr X \sepr \bullet \sepr \gskip \sdot G
        \\
        S,T &::= {!}T \sdot S \sepr {?}T \sdot S \sepr {\oplus}\{i{:}~ S\}_{i \in I} \sepr \&\{i{:}~ S\}_{i \in I} \sepr \bullet
    \end{align*}
    We include basic types (e.g., \sff{unit}, \sff{bool}, \sff{int}), which are all syntactic sugar for $\bullet$.
    \lipicsEnd
\end{definition}

\noindent
The type `$p \mto q \{i\<S_i\> \sdot G_i\}_{i \in I}$' specifies a direct exchange from participant $p$ to participant~$q$, which precedes protocol $G_i$:
$p$ chooses a label $i \in I$ and sends it to $q$ along with a message of type $S_i$.
Message exchange is \emph{asynchronous}: the protocol can continue as $G_i$ before the message has been received by $q$.
The type `$\mu X \sdot G$' defines a recursive protocol: whenever a path of exchanges in $G$ reaches the recursion variable $X$, the protocol continues as `$\mu X \sdot G$'.
The type `$\bullet$' denotes the completed protocol.
For technical convenience, we introduce the construct `$\gskip \sdot G$', which denotes an unobservable step that precedes $G$.

Recursive definitions bind recursion variables, so recursion variables not bound by a recursive definition are free.
We write `$\frv(G)$' to denote the set of free recursion variables of $G$, and say $G$ is \emph{closed} if $\frv(G) = \emptyset$.
Recursion in global types is tail-recursive and \emph{contractive} (i.e.\ they contain no subexpressions of the form `$\mu X_1 \ldots \mu X_n \sdot X_1$').
As for the session types in \Cref{s:apcp}, we define the unfolding of a recursive global type by substituting copies of the recursive definition for recursive calls, i.e.\ `$\mu X \sdot G$' unfolds to `$G \subst{\mu X \sdot G / X}$'.

In approaches based on MPST, the grammar of global types specifies multiparty protocols but does not ensure their correct implementability; such guarantees are given in terms of \emph{well-formedness}, defined as projectability onto all participants (cf.\ \secref{ss:relproj}).

Message types $S,T$ define binary protocols, not to be confused with the types in \secref{s:apcp}.
Type `${!}T \sdot S$' (resp.\ `${?}T \sdot S$') denotes the output (resp.\ input) of a message of type $T$ followed by the continuation $S$.
Type `${\oplus}\{i{:}~ S_i\}_{i \in I}$' denotes \emph{selection}: the output of choice for a label $i \in I$ followed by the continuation $S_i$.
Type `$\&\{i{:}~ S_i\}_{i \in I}$' denotes \emph{branching}: the input of a label $i \in I$ followed by the continuation $S_i$.
Type `$\bullet$' denotes the end of the protocol.
 Note that, due to the tail-recursiveness of session and global types, there are no recursive message types.

It is useful to obtain the set of participants of a global type:

\begin{definition}[Participants]\label{d:globpart}
    We define the \emph{set of participants} of global type $G$, denoted `$\mkern1mu\part(G)\mkern-2mu$':
    \begin{align*}
        \part(p \mto q \{i\<S_i\> \sdot G_i\}_{i \in I})
        &:= \{p,q\} \cup (\bigcup_{i \in I}~ \part(G_i))
        &
        \part(\gskip \sdot G)
        &:= \part(G)
        &
        \part(\bullet)
        &:= \emptyset
        \\
        &
        &
        \part(\mu X \sdot G)
        &:= \part(G)
        &
        \part(X)
        &:= \emptyset
    \end{align*}
    \lipicsEnd
\end{definition}

\subsection{Relative Types}
\label{ss:reltypes}

While a global type such as $G_{\sff{auth}}$~\eqref{eq:Gauth} describes a protocol from a vantage point, we introduce \emph{relative types} that describe the interactions between \emph{pairs} of participants.
This way, relative types capture the peer-to-peer nature of multiparty protocols.
We develop projection from global types onto relative types (cf.\ \secref{ss:relproj}) and use it to establish a new class of \emph{well-formed} global types.

A choice between participants in a global type is \emph{non-local} if it influences future exchanges between other participants.
Our approach uses \emph{dependencies} to expose these non-local choices
in the relative types of these other participants.

Relative types express interactions between two participants.
Because we obtain a relative type through projection of a global type, we know which participants are involved.
Therefore, a relative type only mentions the sender of each exchange; we implicitly know that the recipient is the other participant.

\begin{definition}[Relative Types]\label{d:reltypes}
    \emph{Relative types} $R$ are defined as follows, where the $S_i$ are message types (cf.\ \defref{d:globtypes}):
    \[
        R ::= p \{i\<S_i\> \sdot R\}_{i \in I} \sepr p{?}r \{i \sdot R\}_{i \in I} \sepr p{!}r \{i \sdot R\}_{i \in I} \sepr \mu X \sdot R \sepr X \sepr \bullet \sepr \gskip \sdot R
    \]
\end{definition}

\noindent
We detail the syntax above, given participants $p$ and $q$.
\begin{itemize}
    \item
        Type `$p \{i\<S_i\> \sdot R_i\}_{i \in I}$' specifies that $p$ must choose a label $i \in I$ and send it to $q$ along with a message of type $S_i$ after which the protocol continues with $R_i$.
    \item
        Given an $r$ which is \emph{not} involved in the relative type (i.e., $p \neq r, q \neq r$), type `$p{?}r \{i \sdot R_i\}_{i \in I}$' expresses a {dependency}: a non-local choice between $p$ and $r$ which influences the protocol between $p$ and $q$.
        Here, the dependency indicates that after $p$ receives from $r$ the chosen label, $p$ must forward it to $q$, determining the protocol between $p$ and $q$.
    \item
        Similarly, type `$p{!}r \{i \sdot R_i\}_{i \in I}$' expresses a dependency, which indicates that after $p$ sends to $r$ the chosen label, $p$ must forward it to $q$.
    \item
        Types `$\mu X \sdot R$' and `$X$' define recursion, just as their global counterparts.
    \item
        The type `$\bullet$' specifies the end of the protocol between $p$ and $q$.
    \item
        The type `$\gskip \sdot R$' denotes an unobservable step that precedes $R$.
\end{itemize}

\begin{definition}[Participants of Relative Types]
    We define the \emph{set of participants} of relative type $R$, denoted `$\mkern1mu\part(R)\mkern-2mu$':
    \begin{align*}
        \part(p\{i\<S_i\> \sdot R_i\}_{i \in I})
        &:= \{p\} \cup (\bigcup_{i \in I}~ \part(R_i))
        &
        \part(\gskip \sdot R)
        &:= \part(R)
        &
        \part(\bullet)
        &:= \emptyset
        \\
        \part(p{?}r\{i \sdot R_i\}_{i \in I})
        &:= \{p\} \cup (\bigcup_{i \in I}~ \part(R_i))
        &
        \part(\mu X \sdot R)
        &:= \part(R)
        &
        \part(X)
        &:= \emptyset
        \\
        \part(p{!}r\{i \sdot R_i\}_{i \in I})
        &:= \{p\} \cup (\bigcup_{i \in I}~ \part(R_i))
    \end{align*}
    \lipicsEnd
\end{definition}

We introduce some useful notation:

\begin{notation}\leavevmode
    \begin{itemize}
        \item
            We write `$\mkern1mu p \mto q {:} i\<S\> \sdot G\mkern-3mu$' for a global type with a single branch `$\mkern1mu p \mto q \{i\<S\> \sdot G\}\mkern-3mu$ (and similarly for exchanges and dependencies in relative types).
        \item
            We omit `\sff{unit}' message types from global and relative types, writing `$\mkern1mu i \sdot G\mkern-2mu$' for `$\mkern1mu i\<\sff{unit}\> \sdot G\mkern-2mu$'.
        \item
            Given $k > 1$, we write `$\mkern1mu\gskip^k\mkern-3mu$' for a sequence of $k$ $\gskip$s.
            \lipicsEnd
    \end{itemize}
\end{notation}

\subsection{Relative Projection and Well-Formedness}\label{ss:relproj}

\begin{figure}[t]
    \begin{mdframed}
        \begin{align*}
            \span \detdep((p,q), s \mto r \{i\<S_i\> \sdot G_i\}_{i \in I}) := \begin{cases}
                \gskip \sdot (G_{i'} \wrt (p,q)) ~\text{[any $i' \in I$]}
                & \text{if $\forall i,j.$}
                \\
                & \text{$G_{i} \wrt (p,q) = G_{j} \wrt (p,q)$}
                \\
                p{!}r \{i \sdot (G_i \wrt (p,q))\}_{i \in I}
                & \text{if $p=s$}
                \\
                q{!}r \{i \sdot (G_i \wrt (p,q))\}_{i \in I}
                & \text{if $q=s$}
                \\
                p{?}s \{i \sdot (G_i \wrt (p,q))\}_{i \in I}
                & \text{if $p=r$}
                \\
                q{?}s \{i \sdot (G_i \wrt (p,q))\}_{i \in I}
                & \text{if $q=r$}
            \end{cases}
        \end{align*}

        \hbox to \textwidth{\leaders\hbox to 3pt{\hss . \hss}\hfil}

        \begin{align*}
            (s \mto r \{i\<S_i\> \sdot G_i\}_{i \in I}) \wrt (p,q)
            &:= \begin{cases}
                p \{i\<S_i\> \sdot (G_i \wrt (p,q))\}_{i \in I}
                & \text{if $p=s$ and $q=r$}
                \\
                q \{i\<S_i\> \sdot (G_i \wrt (p,q))\}_{i \in I}
                & \text{if $q=s$ and $p=r$}
                \\
                \detdep((p,q), s \mto r \{i\<S_i\> \sdot G_i\}_{i \in I})
                & \text{otherwise}
            \end{cases}
            \displaybreak[1] \\
            (\mu X \sdot G) \wrt (p,q)
            &:= \begin{cases}
                \mu X \sdot (G \wrt (p,q))
                & \text{if $G \wrt (p,q)$ defined and contractive on $X$}
                \\
                \bullet
                & \text{otherwise}
            \end{cases}
            \displaybreak[1] \\
            \span
            X \wrt (p,q)
            := X
            \qquad
            \bullet \wrt (p,q)
            := \bullet
            \qquad
            (\gskip \sdot G) \wrt (p,q)
            := \gskip \sdot (G \wrt (p,q))
        \end{align*}
        Above, `$\gskip^\ast$' denotes a sequence of zero or more $\gskip$.
    \end{mdframed}

    \caption{
        Dependency Detection (top), and Relative Projection (bottom, cf.\ \Cref{d:relproj}).
        \\
        When a side-condition does not hold, either is undefined.
    }
    \label{f:relproj}
\end{figure}

We define \emph{relative projection} for global types.
 We want relative projection to fail when it would return a non-contractive recursive type.
To this end, we define a notion of contractiveness on relative types:

\begin{definition}[Contractive Relative Types]\label{d:contrRel}
    Given a relative type $R$ and a recursion variable $X$, we say \emph{$R$ is contractive on $X$} if either of the following holds:
    \begin{itemize}
        \item
            $R$ contains an exchange, or

        \item
            $R$ ends in a recursive call on a variable other than $X$.
    \end{itemize}
\end{definition}

Relative projection then relies on the contractiveness of relative types.
It also relies on an auxiliary function to determine if a dependency message is needed and possible.

\begin{definition}[Relative Projection]\label{d:relproj}
    Given a global type $G$, we define its relative projection onto a pair of participants $p$ and $q$, denoted `\,$G \wrt (p,q)$\!', by induction on the structure of $G$ as given in \Cref{f:relproj} (bottom), using the auxiliary function $\detdep$ (cf.\ \Cref{f:relproj}, top).
    \lipicsEnd
\end{definition}

We discuss how \Cref{d:relproj} projects global types onto a pair of participants $(p,q)$, as per \Cref{f:relproj} (bottom).
The most interesting case is the projection of a direct exchange `$s \mto r \{i\<S_i\> \sdot G_i\}_{i \in I}$'.
When the exchange involves both $p$ and $q$, the projection yields an exchange between $p$ and $q$ with the appropriate sender.
Otherwise, the projection relies on the function `$\detdep$' in \Cref{f:relproj} (top), which determines whether the exchange is a non-local choice for $p$ and $q$ and yields an appropriate projection accordingly:
\begin{itemize}
    \item
        If the projections of all branches are equal, the exchange is not a non-local choice and $\detdep$ yields the unobservable step `$\gskip$' followed by the projection of any branch.
    \item
        If there are branches with different projections, the exchange is a non-local choice, so $\detdep$ yields a dependency if possible.
        If $p$ or $q$ is involved in the exchange, $\detdep$ yields an appropriate dependency (e.g., `$p{!}r$' if $p$ is the sender, or `$q{?}s$' if $q$ is the recipient).
        If neither $p$ nor $q$ are involved, then $\detdep$ cannot yield a dependency and projection is thus undefined.
\end{itemize}

The projection of `$\mu X \sdot G'$' considers the projection of the body `$G' \wrt (p,q)$' to see whether $p$ and $q$ interact in $G'$.
If $G' \wrt (p,q)$ is a (possibly empty) sequence of $\gskip$s followed by $\bullet$ or $X$, then $p$ and $q$ do not interact and the projection yields $\bullet$.
Otherwise, $p$ and $q$ do interact and projection preserves the recursive definition.
 Note that \Cref{d:contrRel} (contractiveness) is key here: e.g., $G' \wrt (p,q) = \gskip \sdot \mu Y \sdot \gskip \sdot X$ is not contractive on $X$, so $(\mu X \sdot G') \wrt (p,q) = \bullet$.
The projection of a recursive call `$X$' is simply `$X$'.

The projection of `$G_1 \| G_2$' is standard~\cite{conf/popl/HondaYC08}: it ensures that $G_1$ and $G_2$ do not share participants and only continues with either global type if both $p$ and $q$ are participants.
The projections of `$\bullet$' and `$\gskip$' are homomorphic.

\begin{example}[Projections of $G_\sff{auth}$]\label{ex:Gauthproj}
    To demonstrate relative projection, let us consider again $G_\sff{auth}$:
    \begin{align*}
        G_\sff{auth} = \mu X \sdot s \mto c \left\{ \begin{array}{@{}l@{}}
                \sff{login} \sdot c \mto a {:} \sff{passwd} \<\sff{str}\> \sdot a \mto s {:} \sff{auth} \<\sff{bool}\> \sdot X,
                \\
                \sff{quit} \sdot c \mto a {:} \sff{quit} \sdot \bullet
        \end{array} \right\}
    \end{align*}
    The relative projection onto $(s,c)$ is straightforward, as there are no non-local choices to consider:
    \begin{align*}
        G_\sff{auth} \wrt (s,c) = \mu X \sdot s \left\{ \begin{array}{@{}l@{}}
                \sff{login} \sdot \gskip^2 \sdot X,
                \\
                \sff{quit} \sdot \gskip \sdot \bullet
        \end{array} \right\}
    \end{align*}
    However, compare the projection of the initial \sff{login} branch onto $(s,a)$ and $(c,a)$ with the projection of the \sff{quit} branch: they are different.
    Therefore, the initial exchange between $s$ and $c$ is a non-local choice in the protocols relative to $(s,a)$ and $(c,a)$.
    Since $s$ is involved in this exchange, the non-local choice is detected by `$\mkern1mu\detdep\mkern-3mu$':
    \begin{align*}
        \detdep((s,a), s \mto c \{ \sff{login} \ldots,\quad \sff{quit} \ldots \}) = s{!}c \{ \sff{login} \ldots,\quad \sff{quit} \ldots \}
    \end{align*}
    Hence, this non-local choice can be included in the relative projection onto $(s,a)$ as a dependency:
    \begin{align*}
        G_\sff{auth} \wrt (s,a) = \mu X \sdot s{!}c \left\{ \begin{array}{@{}l@{}}
                \sff{login} \sdot \gskip \sdot a {:} \sff{auth} \<\sff{bool}\> \sdot X,
                \\
                \sff{quit} \sdot \gskip \sdot \bullet
        \end{array} \right\}
    \end{align*}
    Similarly, $c$ is involved in the initial exchange, so the non-local choice can also be included in the relative projection onto $(c,a)$ as a dependency:
    \begin{align*}
        G_\sff{auth} \wrt (c,a) = \mu X \sdot c{?}s \left\{ \begin{array}{@{}l@{}}
                \sff{login} \sdot c {:} \sff{passwd} \<\sff{str}\> \sdot \gskip \sdot X,
                \\
                \sff{quit} \sdot c {:} \sff{quit} \<\sff{unit}\> \sdot \bullet
        \end{array} \right\}
    \end{align*}
\end{example}

Since relative types are relative to pairs of participants, the input order of participants for projection does not matter:

\begin{proposition}\label{prop:relprojSensible}
    Suppose a global type $G$ and distinct participants $p,q \in \part(G)$.
    \begin{itemize}
        \item
            If \,$G \wrt (p,q)$ is defined, then $G \wrt (p,q) = G \wrt (q,p)$ and $\part(G \wrt (p,q)) \subseteq \{p,q\}$;
        \item
            $G \wrt (p,q)$ is undefined if and only if $G \wrt (q,p)$ is undefined.
            \lipicsEnd
    \end{itemize}
\end{proposition}

\paragraph*{Well-formed Global Types}

We may now define \emph{well-formedness} for global types.
Unlike usual MPST approaches, our definition relies exclusively on (relative) projection (\defref{d:relproj}), and does not appeal to external notions such as merge and subtyping~\cite{journal/lmcs/HuBYD12,conf/icdcit/YoshidaG20}.

\begin{definition}[Relative Well-Formedness]\label{d:rwf}
    A global type $G$ is \emph{relative well-formed} if, for every distinct $p,q \in \part(G)$, the  projection $G \wrt (p,q)$ is defined.
    \lipicsEnd
\end{definition}

The following contrasts our new notion of relative well-for\-med\-ness with notions of well-formedness based on the usual notion of local types~\cite{conf/popl/HondaYC08,conf/icalp/DenielouY13}. 

\begin{example}\label{ex:messages}
    Consider the following global type involving participants $p, q, r, s$:
    \[
        G_\thex := p \mto q \left\{
            \begin{array}{l}
                1\<S_a\> \sdot p \mto r {:} 1\<S_b\> \sdot p \mto s {:} 1\<S_c\> \sdot q \mto r {:} 1\<S_d\> \sdot q \mto s {:} 1\<S_e\> \sdot \bullet ,
                \\
                2\<S_f\> \sdot r \mto p {:} 2\<S_g\> \sdot s \mto p {:} 2\<S_h\> \sdot r \mto q {:} 2\<S_i\> \sdot s \mto q {:} 2\<S_j\> \sdot \bullet
            \end{array}
        \right\}
    \]
    The initial exchange between $p$ and $q$ is a non-local choice influencing the protocols between other pairs of participants.
    Well-formedness as in~\cite{conf/popl/HondaYC08,conf/icalp/DenielouY13} forbids non-local choices.
    In contrast, $G_\thex$ is relative well-formed: $p$ and $q$ must both forward the selected label to both $r$ and $s$.
    The dependencies in the following relative projections express precisely this:
    \begin{align*}
        G_\thex \wrt (p,r)
        &= p{!}q \{1 \sdot p {:} 1\<S_b\> \sdot \gskip^3 \sdot \bullet,\quad 2 \sdot r {:} 2\<S_g\> \sdot \gskip^3 \sdot \bullet \}
        \\
        G_\thex \wrt (p,s)
        &= p{!}q \{1 \sdot \gskip \sdot p {:} 1\<S_c\> \sdot \gskip^2 \sdot \bullet,\quad 2 \sdot \gskip \sdot s {:} 2\<S_h\> \sdot \gskip^2 \sdot \bullet \}
        \\
        G_\thex \wrt (q,r)
        &= q{?}p \{1 \sdot \gskip^2 \sdot q {:} 1\<S_d\> \sdot \gskip \sdot \bullet,\quad 2 \sdot \gskip^2 \sdot r {:} 2\<S_i\> \sdot \gskip \sdot \bullet \}
        \\
        G_\thex \wrt (q,s)
        &= q{?}p \{1 \sdot \gskip^3 \sdot q {:} 1\<S_e\> \sdot \bullet,\quad 2 \sdot \gskip^3 \sdot s {:} 2\<S_j\> \sdot \bullet \}
        \tag*{\lipicsEnd}
    \end{align*}
\end{example}

Dependencies in relative types follow the non-local choices in the given global type: by implementing such choices, dependencies ensure correct projectability.
They induce additional messages, but in our view this is an acceptable price to pay for an expressive notion of well-formedness based only on projection.
It is easy to see that in a global type with $n$ participants, the number of messages per communication is $\mcl{O}(n)$---an upper-bound following from the worst-case scenario in which both sender and recipient have to forward a label to $n-2$ participants due to dependencies, as in the example above.
However, in prac\-tice, sender and recipient will rarely both have to forward labels, let alone both to all participants.

\section{Analyzing Global Types using Routers}
\label{s:routers}

\begin{figure}[t]
    \begin{mdframed}
        \mbox{}\hfill%
        \begin{minipage}[t]{.45\textwidth}
            \vspace{1em}
            \hfill%
            \begin{tikzpicture}
    \begin{scope}[local bounding box=bndP]
        \node (P) [] at (0, 0) {$P$};
        \node (rtrP) [right=4mm of P] {$\lmed{c}{\{s,a\}}{G_\sff{auth}}$};
    \end{scope}
    \node [fit=(bndP), shape=rectangle, draw=black, inner xsep=0mm, inner ysep=1mm] (boxP) {};
    \node [below=0mm of boxP] {${} \in \lsys(G_\sff{auth}, \{c\})$};
    \draw[-, cblPink] (P) -- (rtrP);

    \begin{scope}[local bounding box=bndQ]
        \node (rtrQ) [right=1cm of rtrP, yshift=12mm] {$\lmed{s}{\{a,s\}}{G_\sff{auth}}$};
        \node (Q) [right=4mm of rtrQ] {$Q$};
    \end{scope}
    \node [fit=(bndQ), shape=rectangle, draw=black, inner xsep=1mm, inner ysep=0mm] (boxQ) {};
    \node [below=0mm of boxQ] {${} \in \lsys(G_\sff{auth}, \{s\})$};
    \draw[-, cblPink] (Q) -- (rtrQ);

    \begin{scope}[local bounding box=bndR]
        \node (rtrR) [right=1cm of rtrP, yshift=-12mm] {$\lmed{a}{\{s,c\}}{G_\sff{auth}}$};
        \node (R) [right=4mm of rtrR] {$R$};
    \end{scope}
    \node [fit=(bndR), shape=rectangle, draw=black, inner xsep=1mm, inner ysep=0mm] (boxR) {};
    \node [below=0mm of boxR] {${} \in \lsys(G_\sff{auth}, \{a\})$};
    \draw[-, cblPink] (R) -- (rtrR);

    \draw[-, cblPurple] (rtrP) -- (rtrQ.west);
    \draw[-, cblPurple] (rtrQ.west) to [out=225, in=135] (rtrR.west);
    \draw[-, cblPurple] (rtrR.west) -- (rtrP);
\end{tikzpicture}%
            \hfill
            \vspace{-1em}
        \end{minipage}%
        \hfill\vline\hfill%
        \begin{minipage}[t]{.45\textwidth}
            \vspace{1em}
            \hfill%
            \begin{tikzpicture}
    \begin{scope}[local bounding box=bndP]
        \node (P) [] at (0, 0) {$P$};
        \node (rtrP) [right=4mm of P] {$\lmed{c}{\{s,a\}}{G_\sff{auth}}$};
    \end{scope}
    \node [fit=(bndP), shape=rectangle, draw=black, inner xsep=0mm, inner ysep=1mm] (boxP) {};
    \node [below=0mm of boxP] {${} \in \lsys(G_\sff{auth}, \{c\})$};
    \draw[-, cblPink] (P) -- (rtrP);

    \begin{scope}[local bounding box=bndQ]
        \node (rtrQ) [right=10mm of rtrP, yshift=12mm] {$\lmed{s}{\{a,s\}}{G_\sff{auth}}$};
        \node (rtrR) [right=10mm of rtrP, yshift=-12mm] {$\lmed{a}{\{s,c\}}{G_\sff{auth}}$};
        \node (S) [right=17.5mm of rtrP] {$S$};
    \end{scope}
    \node [fit=(bndQ), shape=rectangle, draw=black, inner xsep=1mm, inner ysep=0mm] (boxQ) {};
    \node [below=0mm of boxQ] {${} \in \lsys(G_\sff{auth}, \{s,a\})$};
    \draw[-, cblPink] (S) -- (rtrQ);
    \draw[-, cblPink] (S) -- (rtrR);

    \draw[-, cblPurple] (rtrP) -- (rtrQ.west);
    \draw[-, cblPurple] (rtrP) -- (rtrR.west);

    \draw[-, cblPurple] (rtrQ.west) to [out=225, in=135] (rtrR.west);
\end{tikzpicture}%
            \hfill
            \vspace{-1em}
        \end{minipage}%
        \hfill\mbox{}
    \end{mdframed}

    \caption{
        Two different networks of routed implementations for $G_\sff{auth}$~\eqref{eq:Gauth}, without interleaving (left) and with interleaving (right).
        For participants $p$ and $\tilde{q}$, \Cref{d:router} gives the router process $\mkern1mu\lmed{p}{\tilde{q}}{G}\mkern-3mu$ and \Cref{d:routedImplementations} gives the set $\lsys(G, \tilde{q})$.
        Lines indicate channels and boxes are local compositions of processes.
    }
    \label{f:GauthNetwork}
\end{figure}

In this section, we develop our decentralized analysis of multiparty protocols (\secref{s:mpst}) using relative types (\secref{ss:reltypes}) and APCP (\secref{s:apcp}).
The intended setup is as follows.
Each participant's role in a global type $G$ is implemented by a process, which is connected to a \emph{router}: a process that orchestrates the participant's interactions in $G$.
The resulting \emph{routed implementations} (\defref{d:routedImplementations}) can then directly connect to each other to form a decentralized \emph{network of routed implementations} that implements $G$.
This way we realize the scenario sketched in \Cref{f:chorOrch} (left), which is featured in more detail in \Cref{f:GauthNetwork} (left).

Key in our analysis is the \emph{synthesis} of a participant's router from a global type (\secref{ss:routers}).
To assert well-typedness---and thus deadlock freedom---of networks of routed implementations (\Cref{t:routerTypes}), we extract binary session types from the global type and its associated relative types (\secref{ss:typing}):
\begin{itemize}
    \item
        from the global type we extract types for channels between implementations and routers;
    \item
        from the relative types we extract types for channels between pairs of routers.
\end{itemize}

After defining routers and showing their typability, we set up networks of routed implementations of global types (\secref{ss:networks}).
To enable the transference of deadlock freedom APCP to multiparty protocols, we then establish an operational correspondence between global types and networks of routed implementations (\Cref{t:completeness,t:soundness}).
Finally, to show that our routed approach strictly generalizes the prior centralized analyses~\cite{conf/forte/CairesP16,conf/concur/CarboneLMSW16}, we define an orchestrated analysis of global types and show that it is behaviorally equivalent to a centralized composition of routers (\secref{ss:mediums}).

In the following section (\secref{s:routersInAction}), we will show routers in action.

\subsection{Synthesis of Routers}
\label{ss:routers}

We synthesize routers by decomposing each exchange in the global type into four sub-steps, which we motivate by considering the initial exchange from $s$ to $c$ in $G_{\sff{auth}}$~\labelcref{eq:Gauth}: $s \mto c \{\sff{login} \ldots,\quad \sff{quit} \ldots\}$.
As explained in \Cref{ex:Gauthproj}, this exchange induces a dependency in the relative projections of $G_{\sff{auth}}$ onto $(s,a)$ and $(c,a)$.
We decompose this initial exchange as follows, where $P$, $Q$, and $R$ are the implementations of $c$, $s$, and $a$, respectively (given in \Cref{ex:impl}) and  $\rtr_x$ stands for the router of each $x \in \{s,c,a\}$. Below, multiple actions in one step happen concurrently:
\begin{enumerate}
    \item
        $Q$ sends $\ell \in \{\sff{login},\sff{quit}\}$ to $\rtr_s$.
    \item
        $\rtr_s$ sends $\ell$ to $\rtr_c$ (recipient) and $\rtr_a$ (output dependency). $Q$ sends \sff{unit} value $v$ to $\rtr_s$.
    \item
        $\rtr_c$ sends $\ell$ to $P$ and $\rtr_a$ (input dependency). $\rtr_s$ forwards $v$ to $\rtr_c$.
    \item
        $\rtr_c$ forwards $v$ to $P$. $\rtr_a$ sends $\ell$ to $R$.
\end{enumerate}
In \Cref{ss:typing}, we follow this decomposition to assign to each consecutive step a consecutive priority: this ensures the consistency of priority checks required to establish the deadlock freedom of networks of routed implementations.

We define router synthesis by means of an algorithm that returns a \emph{router process} for a given global type and participant.
More precisely: given $G$, a participant $p$, and $\tilde{q} = \part(G) \setminus \{p\}$, the algorithm generates a process, denoted `$\lmed{p}{\tilde{q}}{G}$', which connects with a process implementing $p$'s role in $G$ on channel $\ci{\mu}{p}$; we shall write such channels in \ich{pink}.
This router for $p$ connects with the routers of the other participants in $G$ ($q_i \in \tilde{q}$) on channels $\crt{p}{q_1}, \ldots, \crt{p}{q_n}$; we shall write such channels in \rch{purple}.
(This convention explains the colors of the lines in \Cref{f:GauthNetwork}.)

The router synthesis algorithm relies on relative projection to detect non-local choices; this way, the router can synchronize with the participant's implementation and with other routers appropriately.
To this end, we define the predicate `$\hdep$', which is true for an exchange and a pair of participants if the exchange induces a dependency for either participant.
Recall that relative projection produces a `$\gskip$' when an exchange is not non-local (cf.\ \Cref{f:relproj}).
Thus, `$\hdep$' only holds true if relative projection does not produce a `$\gskip$'.

\begin{definition}\label{d:hdep}
        The predicate `$\mkern2mu\hdep(q,p,G)\mkern-3mu$' is true if and only if
        \begin{itemize}
            \item
                $G = s \mto r \{i\<S_i\> \sdot G_i\}_{i \in I}$ and $q \notin \{s,r\}$ and $p \in \{s,r\}$, and
            \item
                $\detdep((p,q),G) \neq \gskip \sdot R$ for  all relative types $R$, where $\detdep$ is as in \figref{f:relproj} (top).
        \end{itemize}
\end{definition}

\begin{example}
    Consider the global type $G_\sff{h} := p \mto q \{ \sff{a} \sdot p \mto r {:} \sff{a} \sdot \bullet,\quad \sff{b} \sdot r \mto p {:} \sff{b} \sdot \bullet \}$.
    We have that $\hdep(q, p, G_\sff{h})$ is false because the initial exchange in $G_\sff{h}$ is not a dependency for $p$ and $q$, but $\hdep(r, p, G_\sff{h})$ is true because the initial exchange in $G_\sff{h}$ is indeed a dependency for $p$ and $r$.
\end{example}

\begin{algorithm}[t]
    \DontPrintSemicolon
    \SetAlgoNoEnd
    \SetInd{.2em}{.7em}
    \SetSideCommentRight
    \Def{$\lmed{p}{\tilde{q}}{G}$}{
        \Switch{$G$}{
            \uCase{$s \mto r \{i\<S_i\> \sdot G_i\}_{i \in I}$}{ \label{line:rtrComm}
                $\deps := \{q \in \tilde{q} \mid \hdep(q,p,G)\}$ \; \label{line:rtrDeps}

                \lIf{$p = s$}{ \label{line:rtrSend}
                    \KwRet
                    $\ci{\mu}{p} \gets \big\{i{:}~ \ol{\crt{p}{r}} \puts i \cdot {(\ol{\crt{p}{q}} \puts i)}_{q \in \deps}
                    \cdot \ci{\mu}{p}(v) \sdot \ol{\crt{p}{r}}[w] \cdot (v \fwd w \| \lmed{p}{\tilde{q}}{G_i}) \big\}_{i \in I}$
                }

                \lElseIf{$p = r$}{ \label{line:rtrRecv}
                    \KwRet
                    $\crt{p}{s} \gets \big\{i{:}~ \ol{\ci{\mu}{p}} \puts i \cdot {(\ol{\crt{p}{q}} \puts i)}_{q \in \deps}
                    \cdot \crt{p}{s}(v) \sdot \ol{\ci{\mu}{p}}[w] \cdot (v \fwd w \| \lmed{p}{\tilde{q}}{G_i}) \big\}_{i \in I}$
                }

                \uElseIf{$p \notin \{s,r\}$}{ \label{line:rtrElse}
                    $\dep_s := (s \in \tilde{q} \wedge \hdep(p,s,G))$ \; \label{line:rtrDepS}
                    $\dep_r := (r \in \tilde{q} \wedge \hdep(p,r,G))$ \; \label{line:rtrDepR}

                    \lIf{$\dep_s$ and $\neg \dep_r$}{ \label{line:rtrOnlyS}
                        \KwRet
                        $\crt{p}{s} \gets \big\{i{:}~ \ol{\ci{\mu}{p}} \puts i \cdot \lmed{p}{\tilde{q}}{G_i} \big\}_{i \in I}$
                    }

                    \lElseIf{$\dep_r$ and $\neg \dep_s$}{ \label{line:rtrOnlyR}
                        \KwRet
                        $\crt{p}{r} \gets \big\{i{:}~ \ol{\ci{\mu}{p}} \puts i \cdot \lmed{p}{\tilde{q}}{G_i} \big\}_{i \in I}$
                    }

                    \lElseIf{$\dep_s$ and $\dep_r$}{ \label{line:rtrBoth}
                        \KwRet
                        $ \crt{p}{s} \gets \big\{i{:}~ \ol{\ci{\mu}{p}} \puts i \cdot \crt{p}{r} \puts \{i{:}~ \lmed{p}{\tilde{q}}{G_i}\} \big\}_{i \in I} $
                    }

                    \lElse{ \label{line:rtrElseElse}
                        \KwRet $\lmed{p}{\tilde{q}}{G_{j}}$ for any $j \in I$
                    }
                }
            }

            \uCase{$\mu X \sdot G'$}{ \label{line:rtrRecDef}
                $\tilde{q}' := \{q \in \tilde{q} \mid G \wrt (p,q) \neq \bullet\}$ \; \label{line:rtrRecComp}
                \lIf{$\tilde{q}' \neq \emptyset$}{ \label{line:rtrRecRet}
                    \KwRet $\mu X(\ci{\mu}{p}, {(\crt{p}{q})}_{q \in \tilde{q}'}) \sdot \lmed{p}{\tilde{q}'}{G'}$
                }
                \lElse{ \label{line:rtrRecInact}
                    \KwRet $\0$
                }
            }

            \lCase{$X$}{ \label{line:rtrRecCall}
                \KwRet $X\call{\ci{\mu}{p}, {(\crt{p}{q})}_{q \in \tilde{q}}}$
            }

            \lCase{$\gskip \sdot G'$}{ \label{line:rtrSkip}
                \KwRet $\lmed{p}{\tilde{q}}{G'}$
            }

            \lCase{$\bullet$}{\label{line:rtrEnd}
                \KwRet $\0$
            }
        }
    }

    \caption{Synthesis of Router Processes (\defref{d:router}).}
    \label{alg:router}
\end{algorithm}

\begin{definition}[Router Synthesis]\label{d:router}
    Given a global type $G$, a participant $p$, and participants~$\tilde{q}$, \Cref{alg:router} defines the synthesis of a \emph{router process}, denoted `$\mkern1mu\lmed{p}{\tilde{q}}{G}\mkern-3mu$', that interfaces the interactions of $p\mkern-2mu$ with the other protocol participants according to $G$.
    \lipicsEnd
\end{definition}

\noindent
We often write `$\rtr_p$' for `$\lmed{p}{\part(G) \setminus \{p\}}{G}$' when $G$ is clear from the context.

\Cref{alg:router} distinguishes six cases depending on the syntax of $G$ (\defref{d:globtypes}).
The key case is `${s \mto r \{i\<U_i\> \sdot G_i\}_{i \in I}}$' (\cref{line:rtrComm}).
First, the algorithm computes a set $\deps$ of participants that depend on the exchange using $\hdep$ (cf.\ \defref{d:hdep}).
Then, the algorithm considers the three possibilities for $p$:
\begin{enumerate}
    \item
        If $p = s$ then $p$ is the sender (\cref{line:rtrSend}): the algorithm returns a process that receives a label $i \in I$ over $\ci{\mu}{p}$; sends $i$ over $\crt{p}{r}$ and over $\crt{p}{q}$ for every $q \in \deps$; receives a channel $v$ over $\ci{\mu}{p}$; forwards $v$ as $w$ over $\crt{p}{r}$; and continues as `$\lmed{p}{\tilde{q}}{G_i}$'.

    \item
        If $p = r$ then $p$ is the recipient (\cref{line:rtrRecv}): the algorithm returns a process that receives a label $i \in I$ over $\crt{p}{s}$; sends $i$ over $\ci{\mu}{p}$ and over $\crt{p}{q}$ for every $q \in \deps$; receives a channel $v$ over $\crt{p}{s}$; forwards $v$ as $w$ over $\ci{\mu}{p}$; and continues as `$\lmed{p}{\tilde{q}}{G_i}$'.

    \item\label{itm:rtrDeps}
        Otherwise, if $p$ is not involved (\cref{line:rtrElse}), we use `$\hdep$' to determine whether $p$ depends on an output from $s$, an input from $r$, or on both (\cref{line:rtrDepS,line:rtrDepR}).
        If $p$ only depends on the output from $s$, the algorithm returns a process that receives a label $i \in I$ over $\crt{p}{s}$; sends $i$ over $\ci{\mu}{p}$; and continues as `$\lmed{p}{\tilde{q}}{G_i}$' (\cref{line:rtrOnlyS}).
        If $p$ only depends on an input from $r$, the returned process is similar; the only difference is that $i$ is received over $\crt{p}{r}$ (\cref{line:rtrOnlyR}).

        When $p$ depends on \emph{both} the output from $s$ and on the input from $r$ (\cref{line:rtrBoth}), the algorithm returns a process that receives a label $i \in I$ over $\crt{p}{s}$; sends $i$ over $\ci{\mu}{p}$; receives the label $i$ over $\crt{p}{r}$; and continues as `$\lmed{p}{\tilde{q}}{G_i}$'.

        If there are no dependencies, the returned process is `$\lmed{p}{\tilde{q}}{G_j}$', for arbitrary $j \in I$ (\cref{line:rtrElseElse}).
\end{enumerate}

\noindent
In case `$\mu X \sdot G'$' (\cref{line:rtrRecDef}), the algorithm stores in `$\tilde{q}'$' those $q \in \tilde{q}$ that interact with $p$ in $G'$ (i.e.\ $\mu X \sdot G' \wrt (p,q) \neq \bullet$).
Then, if $\tilde{q}'$ is non-empty (\cref{line:rtrRecRet}), the algorithm returns a recursive definition with as context the channels $\crt{p}{q}$ for $q \in \tilde{q}'$ and $\ci{\mu}{p}$.
Otherwise, the algorithm returns `$\0$' (\cref{line:rtrRecInact}).
In case `$X$' (\cref{line:rtrRecCall}), the algorithm returns a recursive call with as context the channels $\crt{p}{q}$ for $q \in \tilde{q}$ and $\ci{\mu}{p}$.
In case `$\gskip \sdot G'$' (\cref{line:rtrSkip}), it continues with `$G'$' immediately.
Finally, in case `$\bullet$' (\cref{line:rtrEnd}), the algorithm returns `$\0$'.

Considering the number of steps required to return a process,  the {complexity} of \Cref{alg:router} is linear in the size of the given global type (defined as the sum of the number  of communications over all branches).

\subsection{Types for the Router's Channels}
\label{ss:typing}

Here, we obtain session types (cf.\ \defref{d:props}) for (i)~the channels between routers and implementations (\secref{ss:locproj}) and for (ii)~the channels between pairs of routers (\secref{ss:relprop}).
While the former are extracted from global types, the latter are extracted from relative types.

\subsubsection{The Channels between Routers and Implementations}
\label{ss:locproj}

We begin with the session types for the channels between routers and implementations (given in \ich{pink}), which we extract directly from the global type.
A participant's implementation performs on this channel precisely those actions that the participant must perform as per the global type.
Hence, we define this extraction as a form of \emph{local projection} of the global type onto a \emph{single participant}.
The resulting session type may used as a guidance for specifying a participant implementation, which can then connect to the router's dually typed channel endpoint.

\begin{figure}[t]
    \begin{mdframed}
        Below, $\pri \in \mbb{N}$ is arbitrary:
        \begin{align*}
              \msgprop{\bullet}
            & := \bullet
            & \msgprop{{!}T \sdot S}
            & := \msgprop{T} \tensor^\pri \msgprop{S}
            & \msgprop{{\oplus}\{i{:}~ S_i\}_{i \in I}}
            & := \oplus^\pri \{i{:}~ \msgprop{S_i} \}_{i \in I}
            \\[2pt]
            &
            & \msgprop{{?}T \sdot S}
            & := \msgprop{T} \parr^\pri \msgprop{S}
            & \msgprop{\&\{i{:}~ S_i\}_{i \in I}}
            & := \&^\pri\{i{:}~ \msgprop{S_i} \}_{i \in I}
        \end{align*}

        \hbox to \textwidth{\leaders\hbox to 3pt{\hss . \hss}\hfil}

        If $G=s \mto r \{i\<S_i\> \sdot G_i\}_{i \in I}$,
        \[
            G \onto^\pri p := \begin{cases}
                {\oplus}^{\pri} \{i{:}~ \msgprop{S_i} \tensor^{\pri+1} (G_i \onto^{\pri+4} p)\}_{i \in I}
                & \text{if $p=s$}
                \\
                \&^{\pri+2} \{i{:}~ \dual{\msgprop{S_i}} \parr^{\pri+3} (G_i \onto^{\pri+4} p)\}_{i \in I}
                & \text{if $p=r$}
                \\
                \&^{\pri+2} \{i{:}~ (G_i \onto^{\pri+4} p)\}_{i \in I}
                & \text{if $p \notin \{s,r\}$ and $\hdep(p,s,G)$}
                \\
                \&^{\pri+3} \{i{:}~ (G_i \onto^{\pri+4} p)\}_{i \in I}
                & \text{if $p \notin \{s,r\}$ and $\neg \hdep(p,s,G)$ and $\hdep(p,r,G)$}
                \\
                G_{i'} \onto^{\pri+4} p ~\text{[any $i' \in I$]}
                & \text{otherwise}
            \end{cases}
        \]

        Otherwise,
        \begin{align*}
            \span
            \bullet \onto^\pri p
            := \bullet
            \qquad
            (\gskip \sdot G') \onto^\pri p
            := G' \onto^{\pri+4} p
            \qquad
            X \onto^\pri p
            := X
            \\
            (\mu X \sdot G') \onto^\pri p
            &:= \begin{cases}
                \mu X \sdot (G' \onto^\pri p)
                & \text{if $G' \onto^\pri p$ defined and contractive on $X$}
                \\
                \bullet
                & \text{otherwise}
            \end{cases}
        \end{align*}
    \end{mdframed}

    \caption{Extracting Session Types from Message Types (top), and Local Projection: Extracting Session Types from a Global Type (bottom, cf.\ \Cref{d:locproj}).}
    \label{f:locproj}
\end{figure}

Global types contain message types (\defref{d:globtypes}), so we must first define how we extract session types from message types.
This is a straightforward definition, which leaves priorities unspecified: they do not matter for the typability of routers, which forward mes\-sages between implementations and other routers.
Note that one must still specify these priorities when type-checking implementations, making sure they concur between sender and recipient.

\begin{definition}[From Message Types to Session Types]\label{d:msgprop}
    We define the extraction of a session type from message type $S$, denoted `$\mkern1mu\msgprop{S\mkern1mu}\mkern-3mu$', by induction on the structure of $S$ as in \Cref{f:locproj} (top).
\end{definition}

We now define local projection.
To deal with non-local choices, local projection incorporates dependencies  by relying on the dependency detection of  relative projection (cf.\ \defref{d:relproj}).
 Also similar to relative projection, local projection relies on a notion of contractiveness for session types.

\begin{definition}[Contractive Session Types]\label{d:contrSes}
    Given a session type $A$ and a recursion variable $X$, we say \emph{$A$ is contractive on $X$} if either of the following holds:
    \begin{itemize}
        \item
            $A$ contains a connective in $\{\tensor,\parr,\oplus,\&\}$, or

        \item
            $A$ is a recursive call on a variable other than $X$.
    \end{itemize}
\end{definition}

\begin{definition}[Local Projection: From Global Types to Session Types]\label{d:locproj}
    We define the local projection of global type $G$ onto participant $p$ with priority $\pri$, denoted `$\mkern1mu G \onto^\pri p\mkern-2mu$', by induction on the structure of $G$ as in \Cref{f:locproj} (bottom), relying on message type extraction (\defref{d:msgprop}) and the predicate `$\mkern1mu\hdep\mkern-3mu$' (\defref{d:hdep}).
\end{definition}

\noindent
We consider the local projection of an exchange in a global type onto a participant $p$ with priority $\pri$.
The priorities in local projection reflect the four sub-steps into which we decompose exchanges in global types (cf.\ \Cref{ss:routers}).
There are three possibilities, depending on the involvement of $p$ in the exchange:
\begin{enumerate}
    \item
        If $p$ is the sender, local projection specifies a choice ($\oplus$) between the exchange's labels at priority $\pri$ and an output ($\tensor$) of the associated message type at priority $\pri+1$, followed by the projection of the chosen branch at priority $\pri+4$.
    \item
        If $p$ is the recipient, local projection specifies a branch ($\&$) on the exchange's labels at priority $\pri+2$ and an input ($\parr$) of the associated message type at priority $\pri+3$, followed by the projection of the chosen branch at priority $\pri+4$.
    \item
        If $p$ is neither sender nor recipient, local projection uses the predicate `$\hdep$' (\defref{d:hdep}) to detect a dependency on the sender's output or the recipient's input.
        If there is a dependency on the output, local projection specifies a branch on the exchange's labels at priority $\pri+2$.
        If there is a dependency on the input, local projection specifies a branch at priority $\pri+3$.
        Otherwise, when there is no dependency at all, local projection simply continues with the projection of any branch at priority $\pri+4$.
\end{enumerate}
Projection only preserves recursive definitions if they contain actual behavior (i.e.\ the projection of the recursive loop is  contractive, cf.\ \Cref{d:contrSes}).
The projections of `$\bullet$' and recursion variables are homomorphic.
The projection of `$\gskip$' simply projects the \gskip's continuation, at priority $\pri+4$ to keep the priority aligned with the priorities of the other types of the router.

\subsubsection{The Channels between Pairs of Routers}
\label{ss:relprop}

\begin{figure}[t]
    \begin{mdframed}
        \begin{align*}
            \relprop{p}{q}{\pri}{s\{i\<S_i\> \sdot R_i\}_{i \in I}}
            &:= \begin{cases}
                {\oplus}^{\pri+1} \left\{i{:}~ \msgprop{S_i} \tensor^{\pri+2} \relprop{p}{q}{\pri+4}{R_i} \right\}_{i \in I}
                & \text{if $p = s$}
                \\[6pt]
                \&^{\pri+1} \left\{i{:}~ \dual{\msgprop{S_i}} \parr^{\pri+2} \relprop{p}{q}{\pri+4}{R_i} \right\}_{i \in I}
                & \text{if $q = s$}
            \end{cases}
            \\[2pt]
            \relprop{p}{q}{\pri}{r{?}s\{i \sdot R_i\}_{i \in I}}
            &:= \begin{cases}
                {\oplus}^{\pri+2} \left\{i{:}~ \relprop{p}{q}{\pri+4}{R_i} \right\}_{i \in I}
                & \text{if $p = r$}
                \\[6pt]
                \&^{\pri+2} \left\{i{:}~ \relprop{p}{q}{\pri+4}{R_i} \right\}_{i \in I}
                & \text{if $q = r$}
            \end{cases}
            \\[2pt]
            \relprop{p}{q}{\pri}{s{!}r\{i \sdot R_i\}_{i \in I}}
            &:= \begin{cases}
                {\oplus}^{\pri+1} \left\{i{:}~ \relprop{p}{q}{\pri+4}{R_i} \right\}_{i \in I}
                & \text{if $p = s$}
                \\[6pt]
                \&^{\pri+1} \left\{i{:}~ \relprop{p}{q}{\pri+4}{R_i} \right\}_{i \in I}
                & \text{if $q = s$}
            \end{cases}
            \\[2pt]
            \span
            \relprop{p}{q}{\pri}{\bullet} := \bullet
            \qquad
            \relprop{p}{q}{\pri}{\gskip \sdot R} := \relprop{p}{q}{\pri+4}{R}
            \qquad
            \relprop{p}{q}{\pri}{\mu X \sdot R} := \mu X \sdot \relprop{p}{q}{\pri}{R}
            \qquad
            \relprop{p}{q}{\pri}{X} := X
        \end{align*}
    \end{mdframed}

    \caption{Extracting Session Types from Relative Types (cf.\ \Cref{d:relprop}).}
    \label{f:relprop}
\end{figure}

For the channels between pairs of routers (given in \rch{purple}), we extract session types from relative types (\defref{d:reltypes}).
Considering a relative type that describes the protocol between $p$ and $q$, this entails decomposing it into a type for $p$ and a dual type for $q$.

\begin{definition}[From Relative Types to Session Types]\label{d:relprop}
    We define the extraction of a session type from relative type $R$ between $p$ and $q$ at $p$'s perspective with priority $\pri$, denoted `$\mkern1mu\relprop{p}{q}{\pri}{R}\mkern-3mu$', by induction on the structure of $R$ as in \Cref{f:relprop}.
    \lipicsEnd
\end{definition}

\noindent
Here, extraction is \emph{directional}: in `$\relprop{p}{q}{\pri}{R}$', the annotation `$p \rangle q$' says that the session type describes the perspective of $p$'s router with respect to $q$'s.
Messages with sender $p$ are decomposed into selection ($\oplus$) at priority $\pri+1$ followed by output~($\tensor$) at priority $\pri+2$.
Dependencies on messages recieved by $p$ become selection types ($\oplus$) at priority $\pri+1$, and dependencies on messages sent by $p$ become selection types ($\oplus$) at priority $\pri+2$.
Messages from $q$ and dependencies on $q$ yield dual types.
Extraction from `$\bullet$' and recursion is homomorphic, and extraction from `$\gskip$' simply extracts from the \gskip's continuation at priority $\pri+4$.

This way, the channel endpoint of $p$'s router that connects to $q$'s router will be typed `$\relprop{p}{q}{\pri}{G \wrt (p,q)}$', i.e.\ the session type extracted from the relative projection of $G$ onto $p,q$ at $p$'s perspective.
Similarly, the endpoint of this channel at $q$'s router will have the type
`$\relprop{q}{p}{\pri}{G \wrt (p,q)}$', i.e.\ the same relative projection but at $q$'s perspective.
Clearly, these session types must be dual.

\begin{theorem}\label{t:relpropDual}
    Given a relative well-formed global type $G$ and $p,q \in \part(G)$,
    \[
        \relprop{p}{q}{\pri}{G \wrt (p,q)} = \dual{\relprop{q}{p}{\pri}{G \wrt (p,q)}}.
    \]
\end{theorem}

\begin{proof}
    By construction from \Cref{d:relproj} and \Cref{d:relprop}.
\end{proof}

\subsection{Networks of Routed Implementations}
\label{ss:networks}

Having defined routers and types for their channels, we now turn to defining \emph{networks of routed implementations}, i.e., process networks of routers and implementations that correctly represent a given multiparty protocol.
Then, we appeal to the types obtained in \secref{ss:typing} to establish the typability of routers (\Cref{t:routerTypes}).
Finally, we show that all networks of routed implementations of well-formed global types are deadlock free (\Cref{t:globalDlFree}), and that networks of routed implementations behave as depicted by the global types from which they are generated (\Cref{t:completeness,t:soundness}).

We begin by defining routed implementations, which connect implementations of subsets of protocol participants with routers:

\begin{definition}[Routed Implementations]\label{d:routedImplementations}
    Given a  closed,  relative well-formed global type $G$,
    for participants $\tilde{p} \subseteq \part(G)$, the \emph{set of routed implementations} of $\tilde{p}$ in $G$ is defined as follows (cf.\ \defref{d:locproj} for local projection `$\mkern1mu\onto\mkern-3mu$' and \defref{d:router} for router synthesis `$\mkern1mu\lmed{}{}{\ldots}\mkern-3mu$'):
    \[
        \lsys(G, \tilde{p}) := \left\{
            \nu{\ci{\mu}{p} \ci{p}{\mu}}_{p \in \tilde{p}}\, (Q \| \prod_{p \in \tilde{p}} \rtr_p) \mathrel{}\middle\vert\mathrel{} \begin{array}{l}
                Q \vdash  \emptyset; \Gamma,  {(\ci{p}{\mu}{:}~ G \onto^0 p)}_{p \in \tilde{p}}
                \\
                {} \wedge \forall p \in \tilde{p}.~ \rtr_p = \lmed{p}{\part(G) \setminus \{p\}}{G}
            \end{array}
        \right\}
    \]
    We write $\mcl{N}_{\tilde{p}}, \mcl{N}'_{\tilde{p}}, \ldots$ to denote elements of $\lsys(G, \tilde{p})$.
\end{definition}

\noindent
Thus, the composition of a collection of routers and an implementation $Q$ is a routed implementation as long as $Q$ can be typed in a context that includes the corresponding projected types.
Note that the parameter $\tilde{p}$ indicates the presence of \emph{interleaving}: when $\tilde{p}$ is a singleton, the set $\lsys(G, \tilde{p})$ contains processes in which there is a single router and the implementation $Q$ is single-threaded (non-interleaved); more interestingly, when $\tilde{p}$ includes two or more participants, the set $\lsys(G, \tilde{p})$ consists of processes in which the implementation $Q$ {interleaves} the roles of the multiple participants in $\tilde{p}$.

A network of routed implementations of a global type, or simply a \emph{network}, is then the composition of any combination of routed implementations that together account for all the protocol's participants.
Hence, we define sets of networks, quantified over all possible combinations of sets of participants and their respective routed implementations.
The definition relies on \emph{complete partitions} of the participants of a global type, i.e., a split of $\part(G)$ into non-empty, disjoint subsets whose union yields $\part(G)$.

\begin{definition}[Networks]\label{d:networks}
    Suppose given a  closed,  relative well-formed global type $G$.
    Let $\mbb{P}_G$ be the set of all complete partitions of $\part(G)$ with elements $\pi, \pi', \ldots$.
    The \emph{set of networks} of $G$ is defined as
    \[
        \sys(G) := \big\{ \nu{\crt{p}{q} \crt{q}{p}}_{p,q \in \part(G)} (\prod_{\tilde{p} \in \pi} \mcl{N}_{\tilde{p}}) ~\big\vert~  \pi \in \mbb{P}_G \wedge \forall \tilde{p} \in \pi.~ \mcl{N}_{\tilde{p}} \in \lsys(G,\tilde{p}) \big\}.
        \tag*{\lipicsEnd}
    \]
    We write $\mcl{N}, \mcl{N}', \ldots$ to denote elements of $\sys(G)$.
\end{definition}

\begin{example}
    \Cref{f:GauthNetwork} depicts two networks in $\sys(G_{\sff{auth}})$ related to different partitions of $\part(G_{\sff{auth}})$, namely $\big\{\{a\}, \{s\}, \{c\}\big\}$ (non-interleaved) on the left and $\big\{\{a,s\}, \{c\}\big\}$ (interleaved) on the right.
\end{example}

Because a network $\mcl{N}$ may not be typable under the empty typing context, we have the following definition to ``complete'' networks.

\begin{definition}[Completable Networks]\label{d:complNet}
    Suppose given a network $\mcl{N}$ such that $\mcl{N} \vdash \emptyset; \Gamma$.
    We say that $\mcl{N}$ is \emph{completable} if (i) $\Gamma$ is empty or (ii)~there exist $\tilde{v},\tilde{w}$ such that $\nu{\tilde{v} \tilde{w}} \mcl{N} \vdash \emptyset; \emptyset$.
    When $\mcl{N}$ is completable, we write `$\mkern1mu \mcl{N}^\complet\mkern-3mu$' to stand for $\mcl{N}$ (if $\mcl{N} \vdash \emptyset; \emptyset$) or $\nu{\tilde{v} \tilde{w}} \mcl{N}$ (otherwise).
\end{definition}

\begin{proposition}\label{p:complNetsExist}
    For any closed, relative well-formed global type $G$, there exists at least one completable network $\mcl{N} \in \sys(G)$.
\end{proposition}

\begin{proof}
    To construct a completable network in $\sys(G)$, we construct a routed implementation (\defref{d:routedImplementations}) for every $p \in \part(G)$.
    Given a $p \in \part(G)$, by \Cref{p:charProc}, there exists $Q \vdash \emptyset; \ci{p}{\mu}{:}~ G \onto^0 p$.
    Composing each such characteristic implementation process with routers, and then composing the routed implementations, we obtain a network $\mcl{N} \in \sys(G)$, where $\mcl{N} \vdash \emptyset; \emptyset$.
    Hence, $\mcl{N}$ is completable.
\end{proof}

\begin{figure}[t]
    \begin{mdframed}
        \centering
        \begin{tikzpicture}
    \begin{scope}[local bounding box=bndP]
        \node (P) [draw=black, shape=rectangle, inner sep=1mm] at (0, 0) {$P$};
        \node (struP) [right=2cm of P, yshift=6mm] {};
        \node (strdP) [right=2cm of P, yshift=-2mm] {};
    \end{scope}
    \node [fit=(bndP), shape=rectangle, inner sep=0mm] (boxP) {};
    \node [below=6mm of P] (Ptxt) {implementation};
    \draw[-to, cblRed] (Ptxt) -- (P);

    \begin{scope}[local bounding box=bndRtr]
        \node (rtr) [right=5cm of P, draw=black, shape=rectangle, inner sep=1mm] {$\lmed{c}{\{s,a\}}{G_\sff{auth}}$};
        \node (struRtr) [left=2cm of rtr, yshift=5mm] {};
        \node (strdRtr) [left=2cm of rtr, yshift=-3mm] {};
        \node (strS) [right=3.1cm of rtr, yshift=9mm] {};
        \node (strA) [right=3.1cm of rtr, yshift=-10mm] {};
    \end{scope}
    \node [fit=(bndRtr), shape=rectangle, inner sep=0mm] (boxRtr) {};
    \node [below=5.5mm of rtr] (rtrTxt) {router};
    \draw[-to, cblRed] (rtrTxt) -- (rtr);

    \draw [-, cblPink] (P) -- (rtr) node (cmu) [at start, above=-.5mm, xshift=1.1cm] {$\ci{c}{\mu} \color{black} {:}~ G_\sff{auth} \onto^\pri c$} node (muc) [at end, above=-.5mm, xshift=-1.1cm] {$\ci{\mu}{c} \color{black} {:}~ \ol{G_\sff{auth} \onto^\pri c}$};

    \node (s) [right=3.5cm of rtr, yshift=5mm, inner sep=.5mm] {};
    \node (s') [right=1cm of s, yshift=1.4mm, inner sep=.5mm] {};

    \draw[-, cblPurple] (rtr) -- (s) node (cs) [pos=.45, sloped, above=-1mm] {$\crt{c}{s} \color{black} {:}~ \relprop{c}{s}{\pri}{G_\sff{auth} \wrt (c,s)}$};
    \draw[loosely dotted, cblPurple] (s) -- (s');

    \node (a) [right=3.5cm of rtr, yshift=-5mm, inner sep=.5mm] {};
    \node (a') [right=1cm of a, yshift=-1.4mm, inner sep=.5mm] {};

    \draw[-, cblPurple] (rtr) -- (a) node [pos=.45, sloped, below=0mm] {$\crt{c}{a} \color{black} {:}~ \relprop{c}{a}{\pri}{G_\sff{auth} \wrt (c,a)}$};
    \draw[loosely dotted, cblPurple] (a) -- (a');

    \node (locproj) [above=1cm of cmu, xshift=1.5cm] {\Cref{d:locproj}};
    \draw[-to, cblRed] (locproj) -- (cmu);
    \draw[-to, cblRed] (locproj) -- (muc);

    \node (lmed) [above=1.1cm of rtr] {\Cref{d:router}};
    \draw[-to, cblRed] (lmed) -- (rtr);

    \node (relproj) [above=0.6cm of cs] {\Cref{d:relproj}};
    \draw[-to, cblRed] (relproj) -- (cs);

    \node (relprop) [above=0.6cm of cs, xshift=2.5cm] {\Cref{d:relprop}};
    \coordinate[right=0mm of cs, yshift=1mm, xshift=-1mm] (csc);
    \draw[-to, cblRed] (relprop) -- (csc);
\end{tikzpicture}
    \end{mdframed}

    \caption{
        Overview of \Cref{t:routerTypes}, with the definitions and notations for synthesizing and typing routers, using participant $c$ of $G_\sff{auth}$ implemented as $P$ (cf. \Cref{ex:impl}).
        Lines indicate channels and boxes indicate processes.
    }
    \label{f:typedRouter}
\end{figure}
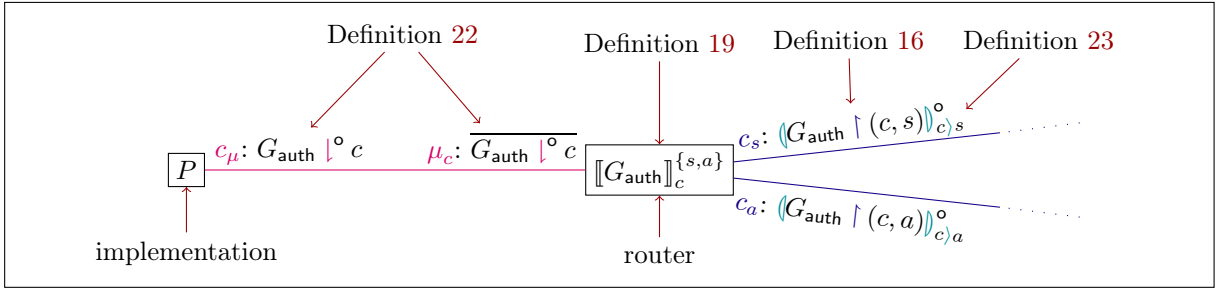

\subsubsection{The Typability of Routers}
\label{ss:routerTypes}

We wish to establish that the networks of a global type are deadlock free. This result, formalized by \Cref{t:globalDlFree} (\Cpageref{t:globalDlFree}), hinges on the typability of routers, which we address next.
\Cref{f:typedRouter} gives an overview of the definitions and notations involved in this theorem's statement.

\begin{restatable}{theorem}{routerTypes}\label{t:routerTypes}
    Suppose given a  closed,  relative well-formed global type $G$, and a $p \in \part(G)$.
    Then,
    \[
        \lmed{p}{\part(G) \setminus \{p\}}{G} \vdash \emptyset;~ \ci{\mu}{p}{:}~ \dual{G \onto^0 p},~ {\big(\crt{p}{q}{:}~ \relprop{p}{q}{0}{G \wrt (p,q)}\big)}_{q \in \part(G) \setminus \{p\}}.
    \]
\end{restatable}

\noindent
This result is a corollary of \Cref{t:routerTypesGen} (\Cpageref{t:routerTypesGen}), which we show next.
We give a full proof on \Cpageref{proof:routerTypes}, after the proof of \Cref{t:routerTypesGen}.

\paragraph{Alarm Processes}
We focus on networks of routed implementations---compositions of synthesized routers and well-typed processes.
However, in order to establish the typability of routers we must account for an edge case that goes beyond these assumptions, namely when a routed implementation is connected to some undesirable implementation, not synthesized by \Cref{alg:router}.
Consider the following example:

\begin{example}\label{ex:alarm}
    Consider again the global type $G_\sff{auth}$, which, for the purpose of this example, we write as follows:
    \begin{align*}
        G_\sff{auth} = s \mto c \left\{ \begin{array}{@{}l@{}}
                \sff{login}{:}~ G_\sff{login},
                \\
                \sff{quit}{:}~ G_\sff{quit}
        \end{array} \right\}
    \end{align*}
    As established in \Cref{ex:Gauthproj}, the initial exchange between $s$ and $c$ determines a dependency for the interactions of $a$ with both $s$ and $c$.
    Therefore, the implementation of $a$ needs to receive the choice between \sff{login} and \sff{quit} from the implementations of both $s$ and $c$.
    An undesirable implementation for $c$,  without  a router, could be for instance as follows:
    \begin{align*}
        R' := \crt{c}{s} \gets \left\{ \begin{array}{@{}l@{}}
                \sff{login}{:}~ \crt{c}{a} \puts \sff{quit} \cdot \ldots,
                \\
                \sff{quit}{:}~ \crt{c}{a} \puts \sff{quit} \cdot \ldots
        \end{array} \right\}
    \end{align*}
    Notice how $ R'$ always sends to $a$ the label \sff{quit}, even if the choice made by $s$ (and sent to $c$) is \sff{login}.
    Now, if $s$ chooses \sff{login}, the router of $a$ is in limbo: on the one hand, it expects $s$ to behave as specified in $G_\sff{login}$; on the other hand, it expects $c$ to behave as specified in $G_\sff{quit}$.
    Clearly, the router of $a$ is in an inconsistent state due to $c$'s implementation.
\end{example}

\noindent
Because routers always forward the chosen label correctly, this kind of undesirable behavior never occurs in the networks of \Cref{d:networks}---we state this formally in \secref{ss:transfer} (\Cref{t:noError}).
Still, in order to prove that our routers are well-typed, we must accommodate the possibility that a router ends up in an undesirable state due to inconsistent forwarding.
For this, we extend APCP with an \emph{alarm process} that signals an inconsistency on a given set of channel endpoints.

\begin{definition}[Alarm Process]\label{d:error}
    Given channel endpoints $\tilde{x} = x_1, \ldots, x_n$, we write `$\mkern2mu\error{\tilde{x}}\mkern-3mu$' to denote an inconsistent state on those endpoints.
\end{definition}

\noindent

In a way, $\error{\tilde{x}}$ is closer to an observable action (a ``barb'') than to an actual process term: $\error{\tilde{x}}$ does not have reductions, and no process from \Cref{f:procdef} (top) can reduce to $\error{\tilde{x}}$.
We assume that $\error{\tilde{x}}$  does not occur in participant implementations (cf.\ $Q$ in \Cref{d:routedImplementations});
we treat it as a process solely for the purpose of  refining the router synthesis algorithm (\Cref{alg:router}) with the possibility of  inconsistent forwarding. The refinement concerns the process on \cref{line:rtrBoth}:
\begin{align*}
    \crt{p}{s} \gets \big\{i{:}~ \ol{\ci{\mu}{p}} \puts i \cdot \crt{p}{r} \puts \{ i{:}~ \lmed{p}{\tilde{q}}{G_i} \} \big\}_{i \in I}
\end{align*}
We extend it with additional branches, as follows:
\begin{align}
    \crt{p}{s} \gets \big\{i{:}~ \ol{\ci{\mu}{p}} \puts i \cdot \crt{p}{r} \puts \left( \begin{array}{@{}l@{}}
                \{ i{:}~ \lmed{p}{\tilde{q}}{G_i} \}
                \\[5pt]
                ~\underline{\cup~ \{ i'{:}~ \error{\ci{\mu}{p}, {(\crt{p}{q})}_{q \in \tilde{q}}} \}_{i' \in I \setminus \{i\}}}~~
    \end{array} \right) \big\}_{i \in I}
    \label{eq:patched}
\end{align}
This new process for  \cref{line:rtrBoth}    captures the kind of inconsistency illustrated by \Cref{ex:alarm}, which occurs when a label $i \in I$ is received over $\crt{p}{s}$ after which a label $i' \in I \setminus \{i\}$ is received over $\crt{p}{r}$.
We account for this case by using the underlined alarm processes.

Routers are then made of processes as in \Cref{f:procdef} (top), selectively extended with alarms as just described.
Because $\error{\tilde{x}}$ merely acts as an observable that signals undesirable behavior, we find it convenient to type it using the following axiom:
\begin{prooftree}
    \infAx{
        $\error{x_1, \ldots, x_n} \vdash \Omega; x_1{:}~ A_1, \ldots, x_n{:}~ A_n$
    }{\scc{Alarm}}
\end{prooftree}
where the recursive context $\Omega$ and types $A_1, \ldots, A_n$ are arbitrary.

\paragraph{Context-based Typability}
Considering the refinement of \Cref{alg:router} with alarm processes, we prove \Cref{t:routerTypesGen} on \Cpageref{t:routerTypesGen}, from which \Cref{t:routerTypes} follows as a corollary.
It relies on some additional auxiliary definitions and results.

To type the router for a participant at any point in the protocol, we need the definition of the entire protocol.
It is not enough to only consider the current (partial) protocol at such points: we need information about bound recursion variables in order to perform unfolding in types.
To this end, we define \emph{global contexts}, that allow us to look at part of a protocol while retaining definitions that concern the entire protocol.

\begin{definition}[Global Contexts]\label{d:globctx}
    \emph{Global contexts} $\mcl{C}$ are given by the following grammar:
    \begin{align*}
        C ::= p \mto q \left( \begin{array}{@{}l@{}}
                \{i\<S\> \sdot G\}_{i \in I}
                \\
                {} \cup \{i'\<S\> \sdot C\}_{i' \notin I}
        \end{array} \right) \sepr \gskip \sdot C \sepr \mu X \sdot C \sepr []
    \end{align*}
    We often simply write `context' when it is clear that we are referring to a global context.
    Given a context $C$ and a global type $G$, we write `$\mkern1mu C[G]\mkern-3mu$' to denote the global type obtained by replacing the hole `$\mkern1mu[]\mkern-3mu$' in $C$ with $G$.
    If $G = C[G_s]$ for some context $C$ and global type $G_s$, then we write `$\mkern1mu G_s \gsub{C} G \mkern-3mu$'.
\end{definition}

As mentioned before, a context captures information about the recursion variables that are bound at any given point in a global type.
Our goal is to obtain a \emph{context-based} typability result for routers.

The order in which recursive variables are bound is important to correctly unfold types:
\begin{example}\label{ex:unfolding}
    Consider the following global type with three nested recursive definitions:
    \begin{align*}
        G_\sff{rec} = \mu X \sdot a \mto b : 1 \sdot \mu Y \sdot a \mto b : 2 \sdot \mu Z \sdot a \mto b \{\sff{x}: X,\quad \sff{y}: Y,\quad \sff{z}: Z\}
    \end{align*}
    To type the router for, e.g., $a$ at the final exchange between $a$ and $b$, we need to be aware of the unfolding of recursion.
    The recursion on $X$, $Y$, and $Z$ have all to be unfolded, and the recursion on $Z$ must include first the unfolding of $X$ and then the unfolding of $Y$, which must in turn include the prior unfolding of~$X$.
\end{example}

\noindent
To account for nested recursions, the following definition gives the bound variables of a context exactly in the order in which they appear:

\begin{definition}[Recursion Binders of Contexts]\label{d:binders}
    Given a global context $C$, the \emph{sequence of recursion binders to the hole} of $C$, denoted `$\mkern1mu \ctxbind{C}\mkern-3mu$', is defined as follows:
    \begin{align*}
        \span
        \ctxbind{\mu X \sdot C}
        := (X, \ctxbind{C})
        \qquad
        \ctxbind{\gskip \sdot C}
        := \ctxbind{C}
        \qquad
        \ctxbind{[]}
        := ()
        \\
        \ctxbind{p \mto q \left( \begin{array}{@{}l@{}}
                    \{i\<S_i\> \sdot G_i\}_{i \in I}
                    \\
                    {} \cup \{i'\<S_{i'}\> \sdot C\}_{i' \notin I}
        \end{array} \right)}
        &:= \ctxbind{C}
    \end{align*}
    Given $G_s \gsub{C} G$, the sequence of recursion binders of $G_s$, denoted `$\mkern1mu \subbind{G_s}{G}\mkern-3mu$', is defined as $\ctxbind{C}$.
\end{definition}

\noindent
The following retrieves the body of a recursive definition from a global context, informing us on how to unfold types:

\begin{definition}[Recursion Extraction]\label{d:recdef}
    The function `$\mkern1mu \recdef{X}{G}\mkern-3mu$' extracts the recursive definition on $X$ from $G$, i.e.\ $\recdef{X}{G} = G'$ if $\mu X \sdot G' \gsub{C} G$ for some context $C$.
    Also, `$\mkern1mu\recctx{X}{G}\mkern-3mu$' extracts the context of the recursive definition on $X$ in $G$, i.e.\ $\recctx{X}{G} = C$ if $\mu X \sdot \recdef{X}{G} \gsub{C} G$.
\end{definition}

\noindent
When unfolding bound recursion variables, we need the priorities of the unfolded types.
The following definition gives a priority that is expected at the hole in a context, as well as the priority expected at any recursive definition in a global type:

\begin{definition}[Absolute Priorities of Contexts]\label{d:ctxpri}
    Given a context $C$ and $\pri \in \mbb{N}$, we define $\ctxpri{\pri}{C}$ as follows:
    \begin{align*}
        \span
        \ctxpri{\pri}{[]}
        := \pri
        \qquad
        \ctxpri{\pri}{\gskip \sdot C}
        := \ctxpri{\pri+4}{C}
        \qquad
        \ctxpri{\pri}{\mu X \sdot C}
        := \ctxpri{\pri}{C}
        \\
        \ctxpri{\pri}{p \mto q \left( \begin{array}{@{}l@{}}
                    \{i\<S_i\> \sdot G_i\}_{i \in I}
                    \\
                    {} \cup \{i'\<S_{i'} \sdot C\}_{i' \notin I}
        \end{array} \right) }
        &:= \ctxpri{\pri+4}{C}
    \end{align*}
    Then, the \emph{absolute priority} of $C$, denoted `$\mkern1mu \ctxpri{}{C}\mkern-3mu$', is defined as $\ctxpri{0}{C}$.
    The absolute priority of $X$ in $G$, denoted `$\mkern1mu \recpri{X}{G}\mkern-3mu$', is defined as $\ctxpri{}{C}$ for some context $C$ such that $\mu X \sdot \recdef{X}{G} \gsub{C} G$.
\end{definition}

To avoid non-contractive recursive types, relative projection (cf.\ \Cref{f:relproj}) closes a type when the participants do not interact inside a recursive definition.
Hence, when typing a router for a recursive definition, we must determine which pairs of participants are ``active'' at any given point in a protocol, and close the connections with the ``inactive'' participants.

\begin{example}
    Consider the following global type, where a client (`$c$') requests two independent, infinite Fibonacci sequences (`$f_1$' and `$f_2$'):
    \begin{align*}
        G_\sff{fib} = c \mto f_1 : \sff{init}\<\sff{int} \times \sff{int}\> \sdot c \mto f_2 : \sff{init}\<\sff{int} \times \sff{int}\> \sdot \underbrace{\mu X \sdot f_1 \mto c : \sff{next}\<\sff{int}\> \sdot f_2 \mto c : \sff{next}\<\sff{int}\> \sdot X}_{G'_\sff{fib}}
    \end{align*}
    Participants $f_1$ and $f_2$ do not interact with each other in the body of the recursion, as formalized by their relative projection:
    \begin{align*}
        \recdef{X}{G_\sff{fib}} \wrt (f_1,f_2) = \gskip \sdot \gskip \sdot X
    \end{align*}
    Hence, $G'_\sff{fib} \wrt (f_1,f_2) = \bullet$, and $f_1$ and $f_2$ do not form an active pair of participants for the recursion in $G_\sff{fib}$.
    Therefore, $f_1$'s router closes its connection with $f_2$'s router at the start of the recursion on $X$, and vice versa.
\end{example}

\noindent
The following definition uses relative projection to determine the pairs of active participants at the hole of a context, as well as at any recursive definition in a global type.
We consider pairs of participants $(p,q)$ and $(q,p)$ to be equivalent.

\begin{definition}[Active Participants]\label{d:activeParts}
    Suppose given a relative well-formed global type $G$.
    The following mutually defined functions compute sets of \emph{pairs of active participants} for recursive definitions and contexts, denoted `$\mkern1mu\recactiv{X}{G}\mkern-3mu$' and `$\mkern1mu\activ{C}{G}\mkern-3mu$', respectively.
    \begin{align*}
        \recactiv{X}{G} &:= \{(p,q) \in \activ{\recctx{X}{G}}{G} \mid (\mu X \sdot \recdef{X}{G}) \wrt (p,q) \neq \bullet\}
        \\
        \activ{C}{G} &:= \begin{cases}
            \recactiv{Y}{G}
            & \text{if $\ctxbind{C} = (\tilde{X},Y)$}
            \\
            \part(G)^2
            & \text{otherwise}
        \end{cases}
    \end{align*}
    The interdependency between `$\mkern1mu\recactiv{X}{G}\mkern-3mu$' and `$\mkern1mu\activ{C}{G}\mkern-3mu$' is well-defined: the former function considers the active participants of a context, which contains less recursive definitions.
\end{definition}

\noindent
When typing a router for a given protocol, we have to keep track of assignments in the recursive context at any point in the protocol.
The following two lemmas ensure that the active participants of recursive definitions are consistent with the active participants of their bodies.

\begin{lemma}\label{l:binderCtxSubRec}
    Suppose given a closed, relative well-formed global type $G$, and a global type $G_s$ and context $C$ such that $G_s \gsub{C} G$.
    For any $Z \in \ctxbind{C}$, $\activ{C}{G} \subseteq \recactiv{Z}{G}$.
\end{lemma}

\begin{proof}
    \sloppy
    Take any $Z \in \ctxbind{C}$.
    Then $\ctxbind{C} = (\tilde{X},Y)$.
    By definition, ${\activ{C}{G} = \recactiv{Y}{G}}$.
    If $Y = Z$, the thesis is proven.
    Otherwise, by definition, ${\recactiv{Y}{G} \subseteq \activ{\recctx{Y}{G}}{G}}$.
    Since the recursive definition on $Z$ appears in $\recctx{Y}{G}$, it follows by induction on the size of $\tilde{X}$ that $\activ{\recctx{Y}{G}}{G} \subseteq \recactiv{Z}{G}$.
    This proves the thesis.
\end{proof}

\noindent
The following lemma ensures that when typing a recursive call, the endpoints given as context for the recursive call concur with the endpoints in the recursive context:

\begin{lemma}\label{l:binderParts}
    Suppose given a closed, relative well-formed global type $G$, a recursion variable $Z$, and a context $C$ such that $Z \gsub{C} G$.
    Then, $\activ{C}{G} = \recactiv{Z}{G}$.
\end{lemma}

\begin{proof}
    Because $G = C[Z]$ and $G$ is closed (i.e.\ $\frv(G) = \emptyset$), there is a recursive definition on $Z$ in $G$.
    Hence, $\ctxbind{C} \neq \emptyset$, i.e.\ $\ctxbind{C} = (\tilde{X}, Y)$ and $\activ{C}{G} = \recactiv{Y}{G}$.
    If $Y = Z$, the thesis is proven.
    Otherwise, the recursive definition on $Y$ in $G$ appears somewhere inside the recursive definition on $Z$.
    Suppose, for contradiction, that $\activ{C}{G} \neq \recactiv{Z}{G}$.
    There are two cases: there exists $(p,q) \in {\part(G)}^2$ s.t.\ (i) $(p,q) \in \activ{C}{G}$ and $(p,q) \notin \recactiv{Z}{G}$, or (ii) $(p,q) \in \recactiv{Z}{G}$ and $(p,q) \notin \activ{C}{G}$.
    Case (i) contradicts \Cref{l:binderCtxSubRec}.

    In case (ii), $(\mu Z \sdot \recdef{Z}{G}) \wrt (p,q) \neq \bullet$ and $(\mu Y \sdot \recdef{Y}{G}) \wrt (p,q) = \bullet$.
    The recursive call on $Z$ in $G$ appears somewhere inside the recursive definition on $Y$, and hence $\recdef{Y}{G} \wrt (p,q)$ contains the recursive call on $Z$.
    This means that $\recdef{Y}{G} \wrt (p,q)$ is contractive on $Y$ (\defref{d:contrRel}), and hence $(\mu Y \sdot \recdef{Y}{G}) \wrt (p,q) \neq \bullet$, contradicting the assumption.
\end{proof}

Our typability result for routers relies on relative and local projection.
Hence, we need to guarantee that all the projections we need at any given point of a protocol are defined.
The following result shows a form of compositionality for relative and local projection, guaranteeing the definedness of projections for all active participants of a given context:

\begin{proposition}\label{p:subWF}
    Suppose given a closed, relative well-formed global type $G$, and a global type $G_s$ such that $G_s \gsub{C} G$.
    Then, for every $(p,q) \in \activ{C}{G}$, the relative projection $G_s \wrt (p,q)$ is defined.
    Also, for every $p \in \{p \in \part(G) \mid \exists q \in \part(G).~ (p,q) \in \activ{C}{G}\}$, the local projection $G_s \onto^\pri p$ is defined for any priority $\pri$.
\end{proposition}

\begin{proof}
    Suppose that, for contradiction, $G_s \wrt (p,q)$ is undefined.
    We show by induction on the structure of $C$ that this means that $G \wrt (p,q)$ is undefined, contradicting the relative well-formedness of $G$.
    \begin{itemize}
        \item
            Hole: $C = []$.
            We have $G_s = G$, and the thesis follows immediately.
        \item
            Exchange: $C = r \mto s \left( \begin{array}{@{}l@{}}
                    \{i\<S_i\> \sdot G_i\}_{i \in I}
                    \\
                    {} \cup \{i'\<S_{i'}\> \sdot C'\}_{i' \notin I}
            \end{array} \right)$.
            By the IH, $C'[G_s] \wrt (p,q)$ is undefined.
            Since the relative projection of an exchange relies on the relative projection of each of the exchange's branches, $G \wrt (p,q)$ is undefined.
        \item
            Skip: $C = \gskip \sdot C'$.
            By the IH, $C'[G_s] \wrt (p,q)$ is undefined.
            Since the relative projection of a skip relies on the relative projection of the skip's continuation, $G \wrt (p,q)$ is undefined.
        \item
            Recursive definition: $C = \mu X \sdot C'$.
            It follows from \Cref{l:binderCtxSubRec} that $\activ{C}{G} \subseteq \recactiv{X}{G}$.
            Hence, $(p,q) \in \recactiv{X}{G}$, and thus $(\mu X \sdot \recdef{X}{G}) \wrt (p,q) = (\mu X \sdot C'[G_s]) \wrt (p,q) \neq \bullet$, which means that $C'[G_s] \wrt (p,q)$ is defined.
            This contradicts the IH.
    \end{itemize}
    The proof for the definedness of local projection is analogous.
\end{proof}

Recall \Cref{ex:unfolding}, where nested recursive definitions in a protocol require nested unfolding of recursive types.
The following definition gives us a concise way of writing such nested (or \emph{deep}) unfoldings:

\begin{definition}[Deep Unfolding]\label{d:deepUnfold}
    Suppose given a sequence of tuples $\tilde{U}$, with each tuple consisting of a recursion variable $X_i$, a lift $t_i \in \mbb{N}$, and a type $B_i$.
    The \emph{deep unfolding} of the type $A$ with $\tilde{U}$, denoted `$\mkern1mu \deepunfold{A}{\tilde{U}}\mkern-3mu$', is the type defined as follows:
    \begin{align*}
        \deepunfold{A}{()} &:= A
        \\
        \deepunfold{A}{(\tilde{U}, (X,t,B))} &:= \deepunfold{A}{\tilde{U}} \subst{\big(\mu X \sdot (\lift{t} \deepunfold{B}{\tilde{U}})\big) / X}
    \end{align*}
\end{definition}

\noindent
When typing a router's recursive call, the types of the router's endpoints are unfoldings of the types in the recursive context.
However, because of the deep unfolding in types, this is far from obvious.
The following result connects a particular form of deep unfolding with regular unfolding (cf.\ \Cref{d:unf}).

\begin{proposition}\label{p:deepUnfold}
    Suppose given a type $A$ and a sequence of tuples $\tilde{U}$ consisting of a recursion variable, a lift, and a substitution type.
    Then,
    \begin{align*}
        \deepunfold{A}{(\tilde{U}, (X, t, A))} &= \unfold^t(\mu X \sdot \deepunfold{A}{\tilde{U}}).
    \end{align*}
\end{proposition}

\begin{proof}
    By \Cref{d:deepUnfold}:
    \begin{align*}
        \deepunfold{A}{(\tilde{U}, (X, t, A))}
        &= \deepunfold{A}{\tilde{U}} \subst{\big(\mu X \sdot (\lift{t} \deepunfold{A}{\tilde{U}})\big) / X}
        \\
        &= \unfold^t(\mu X \sdot \deepunfold{A}{\tilde{U}})
        \qedhere
    \end{align*}
\end{proof}

Armed with these definitions and results, we can finally state our context-based typability result for routers:

\begin{theorem}\label{t:routerTypesGen}
    Suppose given a closed, relative well-formed global type $G$.
    Also, suppose given a global type $G_s$ such that $G_s \gsub{C} G$, and a $p \in \part(G)$ for which there is a $q \in \part(G)$ such that $(p,q) \in \activ{C}{G}$.
    Consider:
    \begin{itemize}
        \item
            the participants with whom $p$ interacts in $G_s$:
            $\tilde{q} = \{q \in \part(G) \mid (p,q) \in \activ{C}{G}\}$,

        \item
            the absolute priority of $G_s$:
            $\pri_C = \ctxpri{}{C}$,

        \item
            the sequence of bound recursion variables of $G_s$:
            $\widetilde{X_C} = \ctxbind{C}$,

        \item
            for every $X \in \widetilde{X_C}$:
            \begin{itemize}
                \item
                    the body of the recursive definition on $X$ in $G$:
                    $G_X = \recdef{X}{G}$,

                \item
                    the participants with whom $p$ interacts in $G_X$:
                    $\tilde{q}_X = \{q \in \part(G) \mid (p,q) \in \recactiv{X}{G}\}$,

                \item
                    the absolute priority of $G_X$:
                    $\pri_X = \recpri{X}{G}$,

                \item
                    the sequence of bound recursion variables of $G_X$ excluding $X$:
                    $\widetilde{Y_X} = \subbind{\mu X \sdot G_X}{G}$,

                \item
                    the type required for $\ci{\mu}{p}$ for a recursive call on $X$:
                    \begin{align*}
                        A_{X,p} = \deepunfold{\dual{G_X \onto^{\pri_X} p}}{{(Y,t_Y,\dual{G_Y \onto^{\pri_Y} p})}_{Y \in \widetilde{Y_X}}},
                    \end{align*}

                \item
                    the type required for $\crt{p}{q}$ for a recursive call on $X$:
                    \begin{align*}
                        B_{X,q} = \deepunfold{\relprop{p}{q}{\pri_X}{G_X \wrt (p,q)}}{{(Y,t_Y,\relprop{p}{q}{\pri_Y}{G_Y \wrt (p,q)})}_{Y \in \widetilde{Y_X}}},
                    \end{align*}

                \item
                    the minimum lift for typing a recursive definition on $X$:
                    $t_X = \max_\pr\left( A_X, {(B_{X,q})}_{q \in \tilde{q}_X} \right) + 1$,
            \end{itemize}

        \item
            the type expected for $\ci{\mu}{p}$ for $p$'s router for $G_s$:
            \begin{align*}
                D_p = \deepunfold{\dual{G_s \onto^{\pri_C} p}}{{(X,t_X,\dual{G_X \onto^{\pri_X} p})}_{X \in \widetilde{X_C}}},
            \end{align*}

        \item
            the type expected for $\crt{p}{q}$ for $p$'s router for $G_s$:
            \begin{align*}
                E_q = \deepunfold{\relprop{p}{q}{\pri_C}{G_s \wrt (p,q)}}{{(X,t_X,\relprop{p}{q}{\pri_X}{G_X \wrt (p,q)})}_{X \in \widetilde{X_C}}}.
            \end{align*}
    \end{itemize}

    \noindent
    Then, we have:
    \begin{align*}
        \lmed{p}{\tilde{q}}{G_s} \vdash
        {\Big(X{:}~ \big( A_X, {(B_{X,q})}_{q \in \tilde{q}_X} \big) \Big)}_{X \in \widetilde{X_C}};~
        \ci{\mu}{p}{:}~ D_p,~
        {(\crt{p}{q}{:}~ E_q)}_{q \in \tilde{q}}
    \end{align*}
\end{theorem}

\begin{proof}
    We apply induction on the structure of $G_s$, with six cases as in \Cref{alg:router}.
    We only detail the cases of exchange and recursion.
    Axiom $\scc{Alarm}$ is used in only one sub-case (case 3(c), cf.\ \Cref{f:tderii} below).

    \begin{itemize}
        \item
            \emph{Exchange}: $G_s = s \mto r \{i \<S_i\> \sdot G_i \}_{i \in I}$ (\cref{line:rtrComm}).

            In this case, we add connectives to the types obtained from the IH.
            Since we do not introduce any recursion variables to these types, the substitutions in the types from the IH are not affected.
            Hence, we can omit these substitutions from the types.
            Also, for each $i \in I$, we have $\frv(G_i) \subseteq \frv(G_s)$, i.e.\ the recursive context remains untouched in this derivation, so we also omit the recursive context.

            Let $\deps := \{q \in \tilde{q} \mid \hdep(q, p, G_s)\}$ (as on \cref{line:rtrDeps}).
            There are three cases depending on the involvement of $p$.
            \begin{enumerate}
                \item
                    If $p = s$, then $p$ is the sender (\cref{line:rtrSend}).

                    Let us consider the relative projections onto $p$ and the participants in $\tilde{q}$.
                    For the recipient $r$,
                    \begin{equation}
                        G_s \wrt (p,r) = p \{ i \sdot (G_i \wrt (p,r)) \}_{i \in I}.
                        \label{eq:outRecvRel}
                    \end{equation}
                    For each $q \in \deps$, by \Cref{d:hdep}, $\detdep((q,p), G) \neq \gskip \sdot R$ for some $R$.
                    That is, since $p$ is the sender of the exchange, for each $q \in \deps$, by the definitions in \Cref{f:relproj},
                    \begin{equation}
                        G_s \wrt (p,q) = p {!} r \{ i \sdot (G_i \wrt (p,q)) \}_{i \in I}.
                        \label{eq:outDepsRel}
                    \end{equation}
                    On the other hand, for each $q \in \tilde{q} \setminus \deps \setminus \{r\}$,
                    \begin{equation}
                        G_s \wrt (p,q) = \gskip \sdot (G_{i'} \wrt (p,q))
                        \label{eq:outSkipRel}
                    \end{equation}
                    for any $i' \in I$, because for each $i,j \in I$,
                    \begin{equation}
                        G_i \wrt (p,q) = G_j \wrt (p,q).
                        \label{eq:outSkipSame}
                    \end{equation}

                    Let us  take stock of  the types we expect for each of the router's channels.
                    \begin{flalign}
                        & \text{For $\ci{\mu}{p}$ we expect}
                        &
                        \ol{G_s \onto^{\pri_C} p}
                        &= \ol{{\oplus}^{\pri_C} \{ i{:}~ \msgprop{S_i} \tensor^{\pri_C+1} (G_i \onto^{\pri_C+4} p) \}_{i \in I}}
                        & \hfill
                        \nonumber \\
                        &
                        &
                        &= \&^{\pri_C} \{ i{:}~ \ol{\msgprop{S_i}} \parr^{\pri_C+1} \ol{(G_i \onto^{\pri_C+4} p)} \}_{i \in I}.
                        & \hfill
                        \label{eq:outCiType}
                        \\
                        & \text{For $\crt{p}{r}$ we expect}
                        &
                        \relprop{p}{r}{\pri_C}{G_s \wrt (p,q)}
                        &= \relprop{p}{r}{\pri_C}{p \{ i \sdot (G_i \wrt (p,r)) \}_{i \in I}}
                        \tag{cf.\ \eqref{eq:outRecvRel}}
                        & \hfill
                        \nonumber \\
                        &
                        &
                        &= {\oplus}^{\pri_C+1} \{i{:}~ \msgprop{S_i} \tensor^{\pri_C+2} \relprop{p}{r}{\pri_C+4}{G_i \wrt (p,r)} \}_{i \in I}.
                        & \hfill
                        \label{eq:outCrtRecvType}
                        \\
                        & \text{For each $q \in \deps$,}
                        \nonumber \\
                        & \text{for $\crt{p}{q}$ we expect}
                        &
                        \relprop{p}{q}{\pri_C}{G_s \wrt (p,q)}
                        &= \relprop{p}{q}{\pri_C}{p {!} r \{ i.~ (G_i \wrt (p,q)) \}_{i \in I}}
                        & \hfill
                        \tag{cf.\ \eqref{eq:outDepsRel}}
                        \nonumber \\
                        &
                        &
                        &= {\oplus}^{\pri_C+1} \{ i{:}~ \relprop{p}{q}{\pri_C+4}{G_i \wrt (p,q)} \}_{i \in I}.
                        & \hfill
                        \label{eq:outCrtDepsType}
                        \\
                        & \text{For each $q \in \tilde{q} \setminus \deps \setminus \{r\}$,}
                        \nonumber \\
                        & \text{for $\crt{p}{q}$ we expect}
                        &
                        \relprop{p}{q}{\pri_C}{G_s \wrt (p,q)}
                        &= \relprop{p}{q}{\pri_C}{\gskip \sdot (G_{i'} \wrt (p,q))}
                        & \hfill
                        \tag{cf.\ \eqref{eq:outSkipRel}}
                        \nonumber \\
                        &
                        &
                        &= \relprop{p}{q}{\pri_C+4}{G_{i'} \wrt (p,q)}
                        ~\text{for any $i' \in I$.}
                        & \hfill
                        \label{eq:outCrtSkipType}
                    \end{flalign}

                    Let us now consider the process returned by \Cref{alg:router}, with each prefix marked with a number:
                    \[
                        \lmed{p}{\tilde{q}}{G_s} =
                        \underbrace{\ci{\mu}{p} \gets \big\{i{:} \big.}_{1}~
                            \underbrace{\ol{\crt{p}{r}} \puts i}_{2_i} \cdot
                            \underbrace{{(\ol{\crt{p}{q}} \puts i)}_{q \in \deps}}_{3_i} \cdot
                            \underbrace{\ci{\mu}{p}(v)}_{4_i} \sdot
                            \underbrace{\ol{\crt{p}{r}}[w]}_{5_i} \cdot
                            (v \fwd w \| \lmed{p}{\tilde{q}}{G_i})
                        \big. \big\}_{i \in I}
                    \]
                    For each $i' \in I$, let $C_{i'} := C[s \mto r (\{i\<S_i\> \sdot G_i\}_{i \in I \setminus \{i'\}} \cup \{i'\<S_{i'}\> \sdot []\})]$.
                    Clearly, $G_{i'} \gsub{C_{i'}} G$.
                    Also, because we are not adding recursion binders, the current value of $\tilde{q}$ is appropriate for the IH.
                    With this context $C_{i'}$ and $\tilde{q}$, we apply the IH to obtain the typing of $\lmed{p}{\tilde{q}}{G_{i'}}$, where priorities start at $\ctxpri{}{C_{i'}} = \ctxpri{}{C}+4 = \pri_C+4$ (cf.\ \defref{d:ctxpri}).
                    Following these typings, \Cref{f:tder} gives the typing of $\lmed{p}{\tilde{q}}{G_s}$, referring to parts of the process by the number marking its foremost prefix above.

                    Clearly, the priorities in the derivation of \Cref{f:tder} meet all requirements.
                    The order of the applications of $\oplus^\star$ for each $q \in \deps$ does not matter, since the selection actions are asynchronous.

                    \begin{figure}[t]
                        \begin{mdframed}
                            {\small
                                \begin{prooftree}
                                    \infAx{
                                        $\forall i \in I.~ v \fwd w \vdash v{:}~ \ol{\msgprop{S_i}}, w{:}~ \msgprop{S_i}$
                                    }{\scc{Id}}
                                    \infAss{
                                        $\forall i \in I.~ \lmed{p}{\tilde{q}}{G_i} \vdash \ci{\mu}{p}{:}~ \ol{(G_i \onto^{\pri_C+4} p)}, {\big(\crt{p}{q}{:}~ \relprop{p}{q}{\pri_C+4}{G_i \wrt (p,q)}\big)}_{q \in \tilde{q}}$
                                    }
                                    \infBin{
                                        $\forall i \in I.~ v \fwd w \| \lmed{p}{\tilde{q}}{G_i} \vdash \begin{array}[t]{@{}lr@{}}
                                            \ci{\mu}{p}{:}~ \ol{(G_i \onto^{\pri_C+4} p)},
                                            v{:}~ \ol{\msgprop{S_i}},
                                            w{:}~ \msgprop{S_i},
                                            \\
                                            {\big(\crt{p}{q}{:}~ \relprop{p}{q}{\pri_C+4}{G_i \wrt (p,q)}\big)}_{q \in \tilde{q}}
                                        \end{array}$
                                    }{\scc{Mix}}
                                    \infUn{
                                        $\forall i \in I.~ 5_i \vdash \begin{array}[t]{@{}lr@{}}
                                            \ci{\mu}{p}{:}~ \ol{(G_i \onto^{\pri_C+4} p)},
                                            v{:}~ \ol{\msgprop{S_i}},
                                            \\
                                            \crt{p}{r}{:}~ \msgprop{S_i} \tensor^{\pri_C+2} \relprop{p}{r}{\pri_C+4}{G_i \wrt (p,r)},
                                            \\
                                            {\big(\crt{p}{q}{:}~ \relprop{p}{q}{\pri_C+4}{G_i \wrt (p,q)}\big)}_{q \in \tilde{q} \setminus \{r\}}
                                        \end{array}$
                                    }{$\tensor^\star$}
                                    \infUn{
                                        $\forall i \in I.~ 4_i \vdash \begin{array}[t]{@{}lr@{}}
                                            \ci{\mu}{p}{:}~ \ol{\msgprop{S_i}} \parr^{\pri_C+1} \ol{(G_i \onto^{\pri_C+4} p)},
                                            \\
                                            \crt{p}{r}{:}~ \msgprop{S_i} \tensor^{\pri_C+2} \relprop{p}{r}{\pri_C+4}{G_i \wrt (p,r)},
                                            \\
                                            {\big(\crt{p}{q}{:}~ \relprop{p}{q}{\pri_C+4}{G_i \wrt (p,q)}\big)}_{q \in \tilde{q} \setminus \{r\}}
                                        \end{array}$
                                    }{$\parr$}
                                    \infUn{
                                        $\forall i \in I.~ 3_i \vdash \begin{array}[t]{@{}lr@{}}
                                            \ci{\mu}{p}{:}~ \ol{\msgprop{S_i}} \parr^{\pri_C+1} \ol{(G_i \onto^{\pri_C+4} p)},
                                            \\
                                            \crt{p}{r}{:}~ \msgprop{S_i} \tensor^{\pri_C+2} \relprop{p}{r}{\pri_C+4}{G_i \wrt (p,r)},
                                            \\
                                            {\big(\crt{p}{q}{:}~ {\oplus}^{\pri_C+1} \{ i{:}~ \relprop{p}{q}{\pri_C+4}{G_i \wrt (p,q)} \}_{i \in I}\big)}_{q \in \deps},
                                            \\
                                            {\big(\crt{p}{q}{:}~ \relprop{p}{q}{\pri_C+4}{G_i \wrt (p,q)}\big)}_{q \in \tilde{q} \setminus \deps}
                                        \end{array}$
                                    }{$\forall q \in \deps.~ \oplus^\star$}
                                    \infUn{
                                        $\forall i \in I.~ 2_i \vdash \begin{array}[t]{@{}lr@{}}
                                            \ci{\mu}{p}{:}~ \ol{\msgprop{S_i}} \parr^{\pri_C+1} \ol{(G_i \onto^{\pri_C+4} p)},
                                            \\
                                            \crt{p}{r}{:}~ {\oplus}^{\pri_C+1} \{i{:}~ \msgprop{S_i} \tensor^{\pri_C+2} \relprop{p}{r}{\pri_C+4}{G_i \wrt (p,r)} \}_{i \in I},
                                            \\
                                            {\big(\crt{p}{q}{:}~ {\oplus}^{\pri_C+1} \{ i{:}~ \relprop{p}{q}{\pri_C+4}{G_i \wrt (p,q)} \}_{i \in I}\big)}_{q \in \deps},
                                            \\
                                            {\big(\crt{p}{q}{:}~ \relprop{p}{q}{\pri_C+4}{G_i \wrt (p,q)}\big)}_{q \in \tilde{q} \setminus \deps}
                                    & \text{(cf.\ \eqref{eq:outSkipSame})}
                                        \end{array}$
                                    }{$\oplus^\star$}
                                    \infUn{
                                        $\lmed{p}{\tilde{q}}{G_s} = 1 \vdash \begin{array}[t]{@{}lr@{}}
                                            \ci{\mu}{p}{:}~ \&^{\pri_C} \{ i{:}~ \ol{\msgprop{S_i}} \parr^{\pri_C+1} \ol{(G_i \onto^{\pri_C+4} p)} \}_{i \in I},
                                    & \text{(cf.\ \eqref{eq:outCiType})}
                                    \\
                                    \crt{p}{r}{:}~ {\oplus}^{\pri_C+1} \{i{:}~ \msgprop{S_i} \tensor^{\pri_C+2} \relprop{p}{r}{\pri_C+4}{G_i \wrt (p,r)} \}_{i \in I},
                                    & \text{(cf.\ \eqref{eq:outCrtRecvType})}
                                    \\
                                    {\big(\crt{p}{q}{:}~ {\oplus}^{\pri_C+1} \{ i{:}~ \relprop{p}{q}{\pri_C+4}{G_i \wrt (p,q)} \}_{i \in I}\big)}_{q \in \deps},
                                    & \text{(cf.\ \eqref{eq:outCrtDepsType})}
                                    \\
                                    {\big(\crt{p}{q}{:}~ \relprop{p}{q}{\pri_C+4}{G_{i'} \wrt (p,q)}\big)}_{q \in \tilde{q} \setminus \deps}
                                    & \text{(cf.\ \eqref{eq:outCrtSkipType})}
                                        \end{array}$
                                    }{$\&$}
                                \end{prooftree}
                            }
                        \end{mdframed}

                        \caption{Typing derivation used in the proof of \Cref{t:routerTypes}.}
                        \label{f:tder}
                    \end{figure}

                \item
                    If $p = r$, then $p$ is the recipient (\cref{line:rtrRecv}).
                    This case is analogous to the previous one.

                \item
                    If $p \notin \{r,s\}$ (\cref{line:rtrElse}), then further analysis depends on whether the exchange is a dependency for $p$.
                    Let
                    \begin{align*}
                        \dep_s
                        &:= (s \in \tilde{q} \wedge \hdep(p, s, G))
                        &
                        & \text{(as on \cref{line:rtrDepS}), and}
                        \\
                        \dep_r
                        &:= (r \in \tilde{q} \wedge \hdep(p, r, G))
                        &
                        & \text{(as on \cref{line:rtrDepR}).}
                    \end{align*}
                    To see what the truths of $\dep_s$ and $\dep_r$ mean, we follow \Cref{d:hdep} and the definitions in \Cref{f:relproj}.
                    \begin{align}
                        G_s \wrt (p,s)
                        &= \begin{cases}
                            s {!} r \{ i \sdot (G_i \wrt (p,s)) \}_{i \in I}
                            & \text{if $\dep_s$ is true}
                            \\
                            \gskip \sdot (G_{i'} \wrt (p,s)) ~\text{for any $i' \in I$}
                            & \text{otherwise}
                        \end{cases}
                        \label{eq:rDepSRel}
                        \\
                        G_s \wrt (p,r)
                        &= \begin{cases}
                            r {?} s \{ i \sdot (G_i \wrt (p,r)) \}_{i \in I}
                            & \text{if $\dep_r$ is true}
                            \\
                            \gskip \sdot (G_{i'} \wrt (p,r)) ~\text{for any $i' \in I$}
                            & \text{otherwise}
                        \end{cases}
                        \label{eq:rDepRRel}
                    \end{align}
                    Let us also consider the relative projections onto $p$ and the participants in $\tilde{q}$ besides $r$ and $s$, which follow by the relative well-formedness of $G_s$.
                    For each $q \in \tilde{q} \setminus \{r,s\}$,
                    \begin{equation}
                        G_s \wrt (p,q) = \gskip \sdot (G_{i'} \wrt (p,q))
                        \label{eq:rDepOtherRel}
                    \end{equation}
                    for any $i' \in I$.

                    The rest of the analysis depends on the truth of $\dep_s$ and $\dep_r$.
                    There are four cases.
                    \begin{enumerate}
                        \item
                            If $\dep_s$ is true and $\dep_r$ is false (\cref{line:rtrOnlyS}), let us  take stock of  the types we expect for each of the router's channels.
                            \begin{flalign}
                                & \text{For $\ci{\mu}{p}$ we expect}
                                &
                                \ol{G_s \onto^{\pri_C} p}
                                &= \ol{\&^{\pri_C+2} \{i{:}~ (G_i \onto^{\pri_C+4} p) \}_{i \in I}}
                                \nonumber
                                \\
                                &
                                &
                                &= {\oplus}^{\pri_C+2} \{i{:}~ \ol{(G_i \onto^{\pri_C+4} p)} \}_{i \in I}.
                                \label{eq:rDepSCiType}
                                \\
                                & \text{For $\crt{p}{s}$ we expect}
                                &
                                \relprop{p}{s}{\pri_C}{G_s \wrt (p,s)}
                                &= \relprop{p}{s}{\pri_C}{s {!} r \{ i \sdot (G_i \wrt (p,s)) \}_{i \in I}}
                                &
                                \text{(cf.\ \eqref{eq:rDepSRel})}
                                \nonumber \\
                                &
                                &
                                &= \&^{\pri_C+1} \{i{:}~ \relprop{p}{s}{\pri_C+4}{G_i \wrt (p,s)} \}_{i \in I}.
                                \label{eq:rDepSCrtSType}
                                \\
                                & \text{For each $q \in \tilde{q} \setminus \{s\}$,}
                                \nonumber
                                \\
                                & \text{for $\crt{p}{q}$ we expect}
                                &
                                \relprop{p}{q}{\pri_C}{G_s \wrt (p,q)}
                                &= \relprop{p}{q}{\pri_C}{\gskip \sdot (G_{i'} \wrt (p,q))}
                                &
                                \text{(cf.\ \eqref{eq:rDepRRel} and~\eqref{eq:rDepOtherRel})}
                                \nonumber \\
                                &
                                &
                                &= \relprop{p}{q}{\pri_C+4}{G_{i'} \wrt (p,q)} ~\text{for any $i' \in I$}.
                                \label{eq:rDepSCrtOtherType}
                            \end{flalign}

                            Similar to case (1), we apply the IH to obtain the typing of $\lmed{p}{\tilde{q}}{G_i}$ for each $i \in I$, starting at priority $\pri_C+4$.
                            We derive the typing of $\lmed{p}{\tilde{q}}{G_s}$:
                            \begin{prooftree}
                                \infAss{
                                    $\forall i \in I.~ \lmed{p}{\tilde{q}}{G_i} \vdash \ci{\mu}{p}{:}~ \ol{G_i \onto^{\pri_C+4} p}, {\big(\crt{p}{q}{:}~ \relprop{p}{q}{\pri_C+4}{G_i \wrt (p,q)}\big)}_{q \in \tilde{q}}$
                                }
                                \infUn{
                                    $\forall i \in I.~ \ol{\ci{\mu}{p}} \puts i \cdot \lmed{p}{\tilde{q}}{G_i} \vdash \ci{\mu}{p}{:}~ {\oplus}^{\pri_C+2} \{i{:}~ \ol{G_i \onto^{\pri_C+4} p} \}_{i \in I}, {\big(\crt{p}{q}{:}~ \relprop{p}{q}{\pri_C+4}{G_i \wrt (p,q)}\big)}_{q \in \tilde{q}}$
                                }{$\oplus^\star$}
                                \infUn{
                                    $\lmed{p}{\tilde{q}}{G_s} = \crt{p}{s} \gets \{i{:}~ \ol{\ci{\mu}{p}} \puts i \cdot \lmed{p}{\tilde{q}}{G_i} \}_{i \in I} \vdash \begin{array}[t]{@{}lr@{}}
                                        \ci{\mu}{p}{:}~ {\oplus}^{\pri_C+2} \{i{:}~ \ol{G_i \onto^{\pri_C+4} p} \}_{i \in I},
                                & \text{(cf.\ \eqref{eq:rDepSCiType})}
                                \\
                                \crt{p}{s}{:}~ \&^{\pri_C+1} \{i{:}~ \relprop{p}{q}{\pri_C+4}{G_i \wrt (p,q)} \}_{i \in I},
                                & \text{(cf.\ \eqref{eq:rDepSCrtSType})}
                                \\
                                {\big(\crt{p}{q}{:}~ \relprop{p}{q}{\pri_C+4}{G_i' \wrt (p,q)}\big)}_{q \in \tilde{q}}
                                & \text{(cf.\ \eqref{eq:rDepSCrtOtherType})}
                                    \end{array}$
                                }{$\&$}
                            \end{prooftree}

                        \item
                            The case where $\dep_s$ is false and $\dep_r$ is true (\cref{line:rtrOnlyR}) is analogous to the previous one.

                        \item
                            If both $\dep_s$ and $\dep_r$ are true (\cref{line:rtrBoth} and \eqref{eq:patched}), let us once again  take stock of  the types we expect for each of the router's channels.
                            \begin{flalign}
                                & \text{For $\ci{\mu}{p}$ we expect}
                                &
                                \ol{G_s \onto^{\pri_C} p}
                                &= \ol{\&^{\pri_C+2} \{i{:}~ (G_i \onto^{\pri_C+4} p) \}_{i \in I}}
                                \nonumber
                                \\
                                &
                                &
                                &= {\oplus}^{\pri_C+2} \{i{:}~ \ol{(G_i \onto^{\pri_C+4} p)} \}_{i \in I}
                                \label{eq:rDepCiType}
                                \\
                                & \text{For $\crt{p}{s}$ we expect}
                                &
                                \relprop{p}{s}{\pri_C}{G_s \wrt (p,s)}
                                &= \relprop{p}{s}{\pri_C}{s {!} r \{ i \sdot (G_i \wrt (p,s)) \}_{i \in I}}
                                &
                                \text{(cf.\ \eqref{eq:rDepSRel})}
                                \nonumber \\
                                &
                                &
                                &= \&^{\pri_C+1} \{i{:}~ \relprop{p}{s}{\pri_C+4}{G_i \wrt (p,s)} \}_{i \in I}
                                \label{eq:rDepCrtSType}
                                \\
                                & \text{For $\crt{p}{r}$ we expect}
                                &
                                \relprop{p}{r}{\pri_C}{G_s \wrt (p,r)}
                                &= \relprop{p}{r}{\pri_C}{r {?} s \{ i \sdot (G_i \wrt (p,r)) \}_{i \in I}}
                                &
                                \text{(cf.\ \eqref{eq:rDepRRel})}
                                \nonumber \\
                                &
                                &
                                &= \&^{\pri_C+2} \{i{:}~ \relprop{p}{r}{\pri_C+4}{G_i \wrt (p,r)} \}_{i \in I}
                                \label{eq:rDepCrtRType}
                                \\
                                & \text{For each $q \in \tilde{q} \setminus \{s,r\}$,}
                                \nonumber
                                \\
                                & \text{for $\crt{p}{q}$ we expect}
                                &
                                \relprop{p}{q}{\pri_C}{G_s \wrt (p,q)}
                                &= \relprop{p}{q}{\pri_C}{\gskip \sdot (G_{i'} \wrt (p,q))}
                                &
                                \text{(cf.\ \eqref{eq:rDepOtherRel})}
                                \nonumber \\
                                &
                                &
                                &= \relprop{p}{q}{\pri_C+4}{G_{i'} \wrt (p,q)} ~\text{for any $i' \in I$}
                                \label{eq:rDepCrtOtherType}
                            \end{flalign}

                            It is clear from~\eqref{eq:rDepCrtSType} and~\eqref{eq:rDepCrtRType} that the router will receive label $i \in I$ first on $\crt{p}{s}$ and then $i' \in I$ on $\crt{p}{r}$.
                            We rely on alarm processes (\Cref{d:error}) to handle the case $i' \neq i$.

                            Similar to case (1), we apply the IH to obtain the typing of $\lmed{p}{\tilde{q}}{G_i}$ for each $i \in I$, starting at priority $\pri_C+4$.
                            \Cref{f:tderii} gives the typing of $\lmed{p}{\tilde{q}}{G_s}$.

                            \begin{figure}[t!]
                                \begin{mdframed}
                                    {\small
                                        \begin{prooftree}
                                            \infAss{
                                                $\begin{array}[b]{@{}l@{}}
                                                    \forall i \in I.
                                                    \\
                                                    \lmed{p}{\tilde{q}}{G_i} \vdash \begin{array}[t]{@{}l@{}}
                                                        \ci{\mu}{p}{:} \ol{G_i \onto^{\pri_C+4} p},
                                                        \\
                                                        {\big(\crt{p}{q}{:} \relprop{p}{q}{\pri_C+4}{G_i \wrt (p,q)}\big)}_{q \in \tilde{q}}
                                                    \end{array}
                                                \end{array}$
                                            }
                                            \infAx{
                                                $\begin{array}[b]{@{}l@{}}
                                                    \forall i \in I. \\
                                                    \forall i' \in I \setminus \{i\}.~ \error{\sff{chs}}
                                                    \vdash \begin{array}[t]{@{}l@{}}
                                                        \ci{\mu}{p}{:}~ \ol{G_i \onto^{\pri_C+4} p},
                                                        \\
                                                        \crt{p}{s}{:}~ \relprop{p}{s}{\pri_C+4}{G_i \wrt (p,s)},
                                                        \\
                                                        \crt{p}{r}{:}~ \relprop{p}{r}{\pri_C+4}{G_{i'} \wrt (p,r)},
                                                        \\
                                                        {\big(\crt{p}{q}{:}~ \relprop{p}{q}{\pri_C+4}{G_i \wrt (p,q)}\big)}_{q \in \tilde{q} \setminus \{s,r\}}
                                                    \end{array}
                                                \end{array}$
                                            }{\scc{Alarm}}
                                            \infBin{
                                                $\forall i \in I.~ \crt{p}{r} \gets \{i{:}~ \lmed{p}{\tilde{q}}{G_i} \} \cup \{ i'{:}~ \error{\sff{chs}} \}_{i' \in I \setminus \{i\}} \vdash \begin{array}[t]{@{}l@{}}
                                                    \ci{\mu}{p}{:}~ \ol{G_i \onto^{\pri_C+4} p},
                                                    \\
                                                    \crt{p}{s}{:}~ \relprop{p}{s}{\pri_C+4}{G_i \wrt (p,s)},
                                                    \\
                                                    \crt{p}{r}{:}~ \&^{\pri_C+2} \{i{:}~ \relprop{p}{r}{\pri_C+4}{G_{i'} \wrt (p,r)} \}_{i \in I},
                                                    \\
                                                    {\big(\crt{p}{q}{:}~ \relprop{p}{q}{\pri_C+4}{G_i \wrt (p,q)}\big)}_{q \in \tilde{q} \setminus \{s,r\}}
                                                \end{array}$
                                            }{$\&$}
                                            \infUn{
                                                $\forall i \in I.~ \ol{\ci{\mu}{p}} \puts i \cdot \crt{p}{r} \gets \{i{:}~ \lmed{p}{\tilde{q}}{G_i} \} \cup \{ i'{:}~ \error{\sff{chs}} \}_{i' \in I \setminus \{i\}} \vdash \begin{array}[t]{@{}l@{}}
                                                    \ci{\mu}{p}{:}~ {\oplus}^{\pri_C+2} \{i{:}~ \ol{G_i \onto^{\pri_C+4} p} \}_{i \in I},
                                                    \\
                                                    \crt{p}{s}{:}~ \relprop{p}{s}{\pri_C+4}{G_i \wrt (p,s)},
                                                    \\
                                                    \crt{p}{r}{:}~ \&^{\pri_C+2} \{i{:}~ \relprop{p}{r}{\pri_C+4}{G_{i'} \wrt (p,r)} \}_{i \in I},
                                                    \\
                                                    {\big(\crt{p}{q}{:}~ \relprop{p}{q}{\pri_C+4}{G_i \wrt (p,q)}\big)}_{q \in \tilde{q} \setminus \{s,r\}}
                                                \end{array}$
                                            }{$\oplus^\star$}
                                            \infUn{
                                                $\underbrace{\crt{p}{s} \gets \{i{:} \ol{\ci{\mu}{p}} \puts i \cdot \crt{p}{r} \gets \{ i{:} \lmed{p}{\tilde{q}}{G_i} \} {\cup} \{ i'{:} \error{\sff{chs}} \}_{i' \in I \setminus \{i\}} \}_{i \in I}}_{\lmed{p}{\tilde{q}}{G_s}} \vdash \begin{array}[t]{@{}lr@{}}
                                                    \ci{\mu}{p}{:}~ {\oplus}^{\pri_C+2} \{i{:}~ \ol{G_i \onto^{\pri_C+4} p} \}_{i \in I},
                                            & \text{(cf.\ \eqref{eq:rDepCiType})}
                                            \\
                                            \crt{p}{s}{:}~ \&^{\pri_C+1} \{i{:}~ \relprop{p}{s}{\pri_C+4}{G_i \wrt (p,s)} \}_{i \in I},
                                            & \text{(cf.\ \eqref{eq:rDepCrtSType})}
                                            \\
                                            \crt{p}{r}{:}~ \&^{\pri_C+2} \{i{:}~ \relprop{p}{r}{\pri_C+4}{G_{i'} \wrt (p,r)} \}_{i \in I},
                                            & \text{(cf.\ \eqref{eq:rDepCrtRType})}
                                            \\
                                            {\big(\crt{p}{q}{:}~ \relprop{p}{q}{\pri_C+4}{G_{i'} \wrt (p,q)}\big)}_{q \in \tilde{q} \setminus \{s,r\}}
                                            & \text{(cf.\ \eqref{eq:rDepCrtOtherType})}
                                                \end{array}$
                                            }{$\&$}
                                        \end{prooftree}
                                    }
                                \end{mdframed}

                                \caption{Typing derivation used in the proof of \Cref{t:routerTypes}, where $\sff{chs} = \{\ci{\mu}{p}\} \cup \{\crt{p}{q} \mid q \in \tilde{q}\}$.}
                                \label{f:tderii}
                            \end{figure}

                        \item
                            If both $\dep_s$ and $\dep_r$ are false, let us again  take stock of  the types we expect for each of the router's channels.
                            \begin{flalign*}
                                & \text{For $\ci{\mu}{p}$ we expect}
                                &
                                \ol{G_s \onto^{\pri_C} p}
                                &= \ol{G_{i'} \onto^{\pri_C+4} p} ~\text{for any $i' \in I$}.
                                \\
                                & \text{For each $q \in \tilde{q}$,}
                                \nonumber
                                \\
                                & \text{for $\crt{p}{q}$ we expect}
                                &
                                \relprop{p}{q}{\pri_C}{G_s \wrt (p,q)}
                                &= \relprop{p}{q}{\pri_C}{\gskip \sdot (G_{i'} \wrt (p,q))}
                                &
                                \text{(cf.\ \eqref{eq:rDepSRel},~\eqref{eq:rDepRRel} and~\eqref{eq:rDepOtherRel})}
                                \nonumber \\
                                &
                                &
                                &= \relprop{p}{q}{\pri_C+4}{G_{i'} \wrt (p,q)} ~\text{for any $i' \in I$.}
                            \end{flalign*}
                            Similar to case (1), we apply the IH to obtain the typing of $\lmed{p}{\tilde{q}}{G_{i'}}$, starting at priority $\pri_C+4$.
                            This directly proves the thesis.
                    \end{enumerate}
            \end{enumerate}

        \item
            \emph{Recursive definition}: $G_s = \mu Z \sdot G'$ (\cref{line:rtrRecDef}).

            Let
            \begin{align}
                \tilde{q}' := \{q \in \tilde{q} \mid G_s \wrt (p,q) \neq \bullet\}
                \label{eq:rRecDefqtildeprime}
            \end{align}
            (as on \cref{line:rtrRecComp}).
            We consider the relative projections onto $p$ and the participants in $\tilde{q}$.
            For each $q \in \tilde{q}'$, we know $G_s \wrt (p,q) \neq \bullet$, while for each $q \in \tilde{q} \setminus \tilde{q}'$, we know $G_s \wrt (p,q) = \bullet$.
            More precisely, by \Cref{d:relproj}, for each $q \in \tilde{q}'$,
            \begin{align}
                G_s \wrt (p,q) &= (\mu Z \sdot G') \wrt (p,q) = \mu Z \sdot (G' \wrt (p,q)).
                \label{eq:rDefRelPrime}
            \end{align}
            and thus
            \begin{align*}
                G' \wrt (p,q) &\neq \gskip^\ast \sdot \bullet ~\text{and}~ G' \wrt (p,q) \neq \gskip^\ast \sdot Z.
            \end{align*}
            For each $q \in \tilde{q} \setminus \tilde{q}'$,
            \begin{align}
                G_s \wrt (p,q) &= (\mu Z \sdot G') \wrt (p,q) = \bullet,
                \label{eq:rDefRelOth}
            \end{align}
            and thus
            \begin{align*}
                G' \wrt (p,q) &= \gskip^\ast \sdot \bullet ~\text{or}~ G' \wrt (p,q) = \gskip^\ast \sdot Z.
            \end{align*}

            Further analysis depends on whether $\tilde{q}' = \emptyset$ or not. We thus examine two cases:
            \begin{itemize}
                \item
                    If $\tilde{q}' = \emptyset$ (\cref{line:rtrRecRet}), let us consider the local projection $G_s \onto^{\pri_C} p$.
                    We prove that $G_s \onto^{\pri_C} p = \bullet$.

                    Suppose, for contradiction, that $G_s \onto^{\pri_C} p \neq \bullet$.
                    Then, by the definitions in \Cref{f:locproj}, $G' \onto^{\pri_C} p \neq X$ and $G' \onto^{\pri_C} p \neq \bullet$.
                    That is, $G' \onto^{\pri_C} p$ contains communication actions or some recursion variable other than $Z$.
                    However, communication actions in $G' \onto^{\pri_C} p$ originate from exchanges in $G'$, either involving $p$ and some $q \in \tilde{q}$, or as a dependency on an exchange involving some $q \in \tilde{q}$.
                    Moreover, recursion variables in $G' \onto^{\pri_C} p$ originate from recursion variables in $G'$.
                    But this would mean that for this $q$, $G' \wrt (p,q)$ contains interactions or recursion variables, contradicting \eqref{eq:rDefRelOth}.
                    Therefore, it cannot be the case that $G_s \onto^{\pri_C} p \neq \bullet$.

                    Let us  take stock of  the types we expect for each of the router's channels.
                    For now, we omit the substitutions in the types.
                    \begin{flalign*}
                        & \text{For $\ci{\mu}{p}$ we expect}
                        &
                        \ol{G_s \onto^{\pri_C} p}
                        &= \ol{\bullet} = \bullet.
                        & \hfill
                        \\
                        & \text{For each $q \in \tilde{q}$, for $\crt{p}{q}$ we expect}
                        &
                        \relprop{p}{q}{\pri_C}{G_s \wrt (p,q)}
                        &= \relprop{p}{q}{\pri_C}{\bullet} = \bullet.
                        & \hfill
                        \tag{cf.\ \eqref{eq:rDefRelOth}}
                    \end{flalign*}
                    Because all expected types are $\bullet$, the substitutions do not affect the types, so we can omit them altogether.

                    First we apply \scc{Empty}, giving us an arbitrary recursive context, and thus the recursive context we need.
                    Then, we apply $\bullet$ for $\ci{\mu}{p}$ and for $\crt{p}{q}$ for each $q \in \tilde{q}$, and obtain the typing of $\lmed{p}{\tilde{q}}{G_s}$ (omitting the recursive context):
                    \[
                        \lmed{p}{\tilde{q}}{G_s} = \0 \vdash \ci{\mu}{p}{:}~ \bullet, {(\crt{p}{q}{:}~ \bullet)}_{q \in \tilde{q}}
                    \]

                \item
                    \begin{sloppypar}
                        If $\tilde{q}' \neq \emptyset$ (\cref{line:rtrRecInact}), then, following similar reasoning as in the previous case, ${G_s \onto^{\pri_C} p = \mu Z \sdot (G' \onto^{\pri_C} p)}$.
                        We  take stock of  the types we expect for each of the router's channels.
                        Note that, because of the recursive definition on $Z$ in $G_s$, there cannot be another recursive definition in the context $C$ capturing the recursion variable $Z$.
                        Therefore, by \Cref{d:binders}, $Z \notin \widetilde{X_C}$.
                    \end{sloppypar}
                    \begin{flalign}
                        & \text{For $\ci{\mu}{p}$ we expect}
                        &
                        &\deepunfold{\ol{G_s \onto^{\pri_C} p}}{\ldots}
                        \nonumber \\
                        &
                        &
                        &{} = \deepunfold{\ol{\mu Z \sdot (G' \onto^{\pri_C} p)}}{\ldots}
                        \nonumber \\
                        &
                        &
                        &{} = \deepunfold{\mu Z \sdot \ol{G' \onto^{\pri_C} p}}{\ldots}
                        \nonumber \\
                        &
                        &
                        &{} = \mu Z \sdot \deepunfold{\ol{G' \onto^{\pri_C} p}}{{(X,t_X,\ol{G_X \onto^{\pri_X} p})}_{X \in \widetilde{X_C}}}.
                        \span
                        \label{eq:rRecDefMu}
                        \allowdisplaybreaks \\
                        & \text{For each $q \in \tilde{q}'$,}
                        \nonumber \\
                        & \text{for $\crt{p}{q}$ we expect}
                        &
                        &\deepunfold{\relprop{p}{q}{\pri_C}{G_s \wrt (p,q)}}{\ldots}
                        \nonumber \\
                        &
                        &
                        &{} = \deepunfold{\relprop{p}{q}{\pri_C}{\mu Z \sdot (G' \wrt (p,q))}}{\ldots}
                        \nonumber \\
                        &
                        &
                        &{} = \deepunfold{\mu Z \sdot \relprop{p}{q}{\pri_C}{G' \wrt (p,q)}}{\ldots}
                        &
                        \text{(cf.\ \eqref{eq:rDefRelPrime})}
                        \nonumber \\
                        &
                        &
                        &= \mu Z \sdot \deepunfold{\relprop{p}{q}{\pri_C}{G' \wrt (p,q)}}{{(X,t_X,\relprop{p}{q}{\pri_X}{G_X \wrt (p,q)})}_{X \in \widetilde{X_C}}}.
                        \span
                        \label{eq:rRecDefqprime}
                        \allowdisplaybreaks \\
                        & \text{For each $q \in \tilde{q} \setminus \tilde{q}'$,}
                        \nonumber \\
                        & \text{for $\crt{p}{q}$ we expect}
                        &
                        &\deepunfold{\relprop{p}{q}{\pri_C}{G_s \wrt (p,q)}}{\ldots}
                        \nonumber \\
                        &
                        &
                        &{} = \deepunfold{\relprop{p}{q}{\pri_C}{\bullet}}{\ldots}
                        \nonumber \\
                        &
                        &
                        &{} = \deepunfold{\bullet}{\ldots}
                        &
                        \text{(cf.\ \eqref{eq:rDefRelOth})}
                        \nonumber \\
                        &
                        &
                        &{} = \bullet.
                        \label{eq:rRecDefOthers}
                    \end{flalign}
                    We also need an assignment in the recursive context for every $X \in \widetilde{X_C}$, but not for $Z$.

                    Let $C' = C[\mu Z \sdot []]$.
                    Clearly, $G' \gsub{C'} G$.
                    Let us first establish some facts about the recursion binders, priorities, and active participants related to $C'$, $G'$, and $Z$:
                    \begin{itemize}
                        \item
                            $\widetilde{X_{C'}} = \ctxbind{C'} = (\ctxbind{C}, Z) = (\widetilde{X_C}, Z)$ (cf.\ \defref{d:binders}).
                        \item
                            $G_Z = \recdef{Z}{G} = G'$, as proven by the context $C'$ (cf.\ \defref{d:recdef}).
                        \item
                            $\widetilde{Y_Z} = \subbind{\mu Z \sdot G_Z}{G} = \ctxbind{C} = \widetilde{X_C}$.
                        \item
                            $\pri_{C'} = \ctxpri{}{C'} = \ctxpri{}{C} = \pri_{C}$, and $\pri_Z = \recpri{Z}{G} = \ctxpri{}{C} = \pri_{C}$, and hence $\pri_{C'} = \pri_Z$ (cf.\ \defref{d:ctxpri}).
                        \item
                            $\tilde{q}_Z = \tilde{q}'$ (cf.\ \defref{d:activeParts} and \eqref{eq:rRecDefqtildeprime}).
                    \end{itemize}

                    Because $\widetilde{X_{C'}} = (\widetilde{X_C}, Z)$ and $\tilde{q}' = \tilde{q}_Z$, $\tilde{q}'$ is appropriate for the IH.
                    We apply the IH on $C'$, $G'$, and $\tilde{q}'$ to obtain a typing for $\lmed{p}{\tilde{q}'}{G'}$, where we immediately make use of the facts established above.
                    We give the assignment to $Z$ in the recursive context separate from those for the recursion variables in $\widetilde{X_C}$.
                    Also, by \Cref{p:deepUnfold}, we can write the final unfolding on $Z$ in the types separately.
                    For example, the type for $\ci{\mu}{p}$ is
                    \begin{align*}
                        & \deepunfold{\dual{G' \onto^{\pri_{C'}} p}}{{(X,t_X,\dual{G_X \onto^{\pri_X} p})}_{X \in \widetilde{X_{C'}}}}
                        \\
                        &= \deepunfold{\dual{G' \onto^{\pri_{C}} p}}{{(X,t_X,\dual{G_X \onto^{\pri_X} p})}_{X \in (\widetilde{X_C}, Z)}}
                        \\
                        &= \deepunfold{\dual{G' \onto^{\pri_{C}} p}}{\big({(X,t_X,\dual{G_X \onto^{\pri_X} p})}_{X \in \widetilde{X_C}}, (Z,t_Z,\dual{G_Z \onto^{\pri_Z} p})\big)}
                        \\
                        &= \deepunfold{\dual{G' \onto^{\pri_{C}} p}}{\big({(X,t_X,\dual{G_X \onto^{\pri_X} p})}_{X \in \widetilde{X_C}}, (Z,t_Z,\dual{G' \onto^{\pri_C} p})\big)}
                        \\
                        &= \unfold^{t_Z}\big( \mu Z \sdot \deepunfold{\dual{G' \onto^{\pri_{C}} p}}{{(X,t_X,\dual{G_X \onto^{\pri_X} p})}_{X \in \widetilde{X_C}}} \big).
                    \end{align*}

                    The resulting typing is as follows:
                    \begin{align*}
                        \lmed{p}{\tilde{q}'}{G'} \vdash \begin{array}[t]{@{}l@{}}
                            {\left(X{:}~ \left( \begin{array}{@{}l@{}}
                                            \deepunfold{\dual{G_X \onto^{\pri_X} p}}{{(Y,t_Y,\dual{G_Y \onto^{\pri_Y} p})}_{Y \in \widetilde{Y_X}}},
                                            \\[4pt]
                                            {\Big(\deepunfold{\relprop{p}{q}{\pri_X}{G_X \wrt (p,q)}}{{(Y,t_Y,\relprop{p}{q}{\pri_Y}{G_Y \wrt (p,q)})}_{Y \in \widetilde{Y_X}}}\Big)}_{q \in \tilde{q}_X}
                            \end{array} \right) \right)}_{X \in \widetilde{X_C}},
                            \\[16pt]
                            Z{:}~ \left( \begin{array}{@{}l@{}}
                                            \deepunfold{\dual{G' \onto^{\pri_C} p}}{{(X,t_X,\dual{G_X \onto^{\pri_X} p})}_{X \in \widetilde{X_C}}},
                                            \\[4pt]
                                            {\Big(\deepunfold{\relprop{p}{q}{\pri_C}{G' \wrt (p,q)}}{{(X,t_X,\relprop{p}{q}{\pri_X}{G_X \wrt (p,q)})}_{X \in \widetilde{X_C}}}\Big)}_{q \in \tilde{q}'}
                            \end{array} \right);
                            \\[16pt]
                            \ci{\mu}{p}{:}~ \unfold^{t_Z}\big(\mu Z \sdot \deepunfold{\dual{G' \onto^{\pri_C} p}}{{(X,t_X,\dual{G_X \onto^{\pri_X} p})}_{X \in \widetilde{X_C}}}\big),
                            \\[6pt]
                            {\Big(\crt{p}{q}{:}~ \unfold^{t_Z}\big(\mu Z \sdot \deepunfold{\relprop{p}{q}{\pri_C}{G' \wrt (p,q)}}{{(X,t_X,\relprop{p}{q}{\pri_X}{G_X \wrt (p,q)})}_{X \in \widetilde{X_C}}}\big)\Big)}_{q \in \tilde{q}'}
                        \end{array}
                    \end{align*}

                    By assumption, we have
                    \begin{align*}
                        t_Z &= \max_\pr\left( \begin{array}{@{}l@{}}
                                \deepunfold{\dual{G' \onto^{\pri_C} p}}{{(X,t_X,\dual{G_X \onto^{\pri_X} p})}_{X \in \widetilde{X_C}}}
                                \\[4pt]
                                {\Big(\deepunfold{\relprop{p}{q}{\pri_C}{G' \wrt (p,q)}}{{(X,t_X,\relprop{p}{q}{\pri_X}{G_X \wrt (p,q)})}_{X \in \widetilde{X_C}}}\Big)}_{q \in \tilde{q}'}
                        \end{array} \right) + 1,
                    \end{align*}
                    so $t_Z$ is clearly greater than the maximum priority appearing in the types before unfolding.
                    Hence, we can apply \scc{Rec} to eliminate $Z$ from the recursive context, and to fold the types, giving the typing of $\lmed{p}{\tilde{q}}{G_s} = \mu Z(\ci{\mu}{p}, {(\crt{p}{q})}_{q \in \tilde{q}'}) \sdot \lmed{p}{\tilde{q}'}{G'}$:
                    \begin{align*}
                        \lmed{p}{\tilde{q}}{G_s} \vdash \begin{array}[t]{@{}l@{}}
                            {\left(X{:}~ \left( \begin{array}{@{}l@{}}
                                            \deepunfold{\dual{G_X \onto^{\pri_X} p}}{{(Y,t_Y,\dual{G_Y \onto^{\pri_Y} p})}_{Y \in \widetilde{Y_X}}},
                                            \\[4pt]
                                            {\Big(\deepunfold{\relprop{p}{q}{\pri_X}{G_X \wrt (p,q)}}{{(Y,t_Y,\relprop{p}{q}{\pri_Y}{G_Y \wrt (p,q)})}_{Y \in \widetilde{Y_X}}}\Big)}_{q \in \tilde{q}_X}
                            \end{array} \right) \right)}_{X \in \widetilde{X_C}};
                            \\[16pt]
                            \ci{\mu}{p}{:}~ \mu Z \sdot \deepunfold{\dual{G' \onto^{\pri_C} p}}{{(X,t_X,\dual{G_X \onto^{\pri_X} p})}_{X \in \widetilde{X_C}}},
                            \\[6pt]
                            {\Big(\crt{p}{q}{:}~ \mu Z \sdot \deepunfold{\relprop{p}{q}{\pri_C}{G' \wrt (p,q)}}{{(X,t_X,\relprop{p}{q}{\pri_X}{G_X \wrt (p,q)})}_{X \in \widetilde{X_C}}}\Big)}_{q \in \tilde{q}'}
                        \end{array}
                    \end{align*}

                    In this typing, the type for $\ci{\mu}{p}$ concurs with~\eqref{eq:rRecDefMu}, and, for every $q \in \tilde{q}'$, the type for $\crt{p}{q}$ concurs with~\eqref{eq:rRecDefqprime}.
                    For every $q \in \tilde{q} \setminus \tilde{q}'$, we can add the type for $\crt{p}{q}$ in~\eqref{eq:rRecDefOthers} by applying $\bullet$.
                    This proves the thesis.
            \end{itemize}

        \item
            \emph{Recursive call}: $G_s = Z$ (\cref{line:rtrRecCall}).

            Clearly, because $G$ is closed (i.e.\ $\frv(G) = \emptyset$), $Z \in \widetilde{X_C}$.
            More precisely, $\widetilde{X_C} = (\tilde{X}_1, Z, \tilde{X}_2)$.

            Note that the recursive definitions on the variables in $\tilde{X}_1$ appear in $G$ after the recursive definitions on the variables in $(Z, \tilde{X}_2)$.
            Because the unfoldings of $(Z, \tilde{X}_2)$ occur before the unfoldings of $\tilde{X}_1$, the recursive definitions on the variables in $\tilde{X}_1$ are renamed in order to avoid capturing these variables when performing the unfoldings of $(Z, \tilde{X}_2)$.
            So, after the unfoldings of $(Z, \tilde{X}_2)$, there are no recursive calls on the variables in $\tilde{X}_1$ anymore, so the unfoldings on $\tilde{X}_1$ do not have any effect on the types.

            Also, note that $\tilde{X}_2 = \widetilde{Y_Z}$ (cf.\ \defref{d:binders}).

            Let us take stock of the types we expect for our router's channels.
            \begin{flalign}
                & \text{For $\ci{\mu}{p}$ we expect}
                &
                &\deepunfold{\ol{G_s \onto^{\pri_C} p}}{\ldots}
                \nonumber \\
                &
                &
                &{} = \deepunfold{\ol{Z \onto^{\pri_C} p}}{\ldots}
                \nonumber \\
                &
                &
                &{} = \deepunfold{\ol{Z}}{\ldots}
                \nonumber \\
                &
                &
                &{} = \deepunfold{Z}{{(X,t_X,\ol{G_X \onto^{\pri_X} p})}_{X \in (\tilde{X}_1, Z, \widetilde{Y_Z})}}
                \nonumber \\
                &
                &
                &{} = \deepunfold{Z}{{(X,t_X,\ol{G_X \onto^{\pri_X} p})}_{X \in (Z, \widetilde{Y_Z})}}
                \nonumber \\
                &
                &
                &{} = \mu Z \sdot (\lift{t_Z} \deepunfold{\ol{G_Z \onto^{\pri_Z} p}}{{(X,t_X,\ol{G_X \onto^{\pri_X} p})}_{X \in \widetilde{Y_Z}}})
                \label{eq:rCallMu}
                \\
                & \text{For each $q \in \tilde{q}$,}
                \nonumber \\
                & \text{for $\crt{p}{q}$ we expect}
                &
                &\deepunfold{\relprop{p}{q}{\pri_C}{G_s \wrt (p,q)}}{\ldots}
                \nonumber \\
                &
                &
                &{} = \deepunfold{\relprop{p}{q}{\pri_C}{Z}}{\ldots}
                \nonumber \\
                &
                &
                &{} = \deepunfold{Z}{{(X,t_X,\relprop{p}{q}{\pri_X}{G_X \wrt (p,q)})}_{X \in (\tilde{X}_1, Z, \widetilde{Y_Z})}}
                \nonumber \\
                &
                &
                &{} = \deepunfold{Z}{{(X,t_X,\relprop{p}{q}{\pri_X}{G_X \wrt (p,q)})}_{X \in (Z, \widetilde{Y_Z})}}
                \nonumber \\
                &
                &
                &{} = \mu Z \sdot (\lift{t_Z} \deepunfold{\relprop{p}{q}{\pri_Z}{G_Z \wrt (p,q)}}{{(X,t_X,\relprop{p}{q}{\pri_X}{G_X \wrt (p,q)})}_{X \in \widetilde{Y_Z}}})
                \label{eq:rCallTildeq}
            \end{flalign}

            Also, we need an assignment in the recursive context for every $X \in \widetilde{X_C}$.
            By \Cref{l:binderParts}, $\tilde{q} = \tilde{q}_Z$.
            Hence, for $Z$, the assignment should be as follows:
            \begin{align}
                Z{:}~ \left( \begin{array}{@{}l@{}}
                    \deepunfold{\ol{G_Z \onto^{\pri_Z} p}}{{(X,t_X,\ol{G_X \onto^{\pri_X} p})}_{X \in \widetilde{Y_Z}}},
                    \\
                    {\Big(\deepunfold{\relprop{p}{q}{\pri_Z}{G_Z \wrt (p,q)}}{{(X,t_X,\relprop{p}{q}{\pri_X}{G_X \wrt (p,q)})}_{X \in \widetilde{Y_Z}}}\Big)}_{q \in \tilde{q}}
                \end{array} \right)
                \label{eq:rCallRec}
            \end{align}

            We apply \scc{Var} to obtain the typing of $\lmed{\tilde{q}}{p}{G_s}$, where we make us the rule's allowance for an arbitrary recursive context up to the assignment to $Z$.
            \scc{Var} is applicable, because the types are recursive definitions on $Z$, concurring with the types assigned to $Z$, and lifted by a common lifter $t_Z$.
            \begin{prooftree}
                \infAx{
                    $\begin{array}{@{}l@{}}
                        \lmed{\tilde{q}}{p}{G_s} = X\call{\ci{\mu}{p}, {(\crt{p}{q})}_{q \in \tilde{q}}}
                        \\[6pt]
                        {} \vdash \begin{array}[t]{@{}l@{}}
                            {\left(X{:}~ \left( \begin{array}{@{}l@{}}
                                            \deepunfold{\dual{G_X \onto^{\pri_X} p}}{{(Y,t_Y,\dual{G_Y \onto^{\pri_Y} p})}_{Y \in \widetilde{Y_X}}},
                                            \\[4pt]
                                            {\Big(\deepunfold{\relprop{p}{q}{\pri_X}{G_X \wrt (p,q)}}{{(Y,t_Y,\relprop{p}{q}{\pri_Y}{G_Y \wrt (p,q)})}_{Y \in \widetilde{Y_X}}}\Big)}_{q \in \tilde{q}_X}
                            \end{array} \right) \right)}_{X \in \widetilde{X_C} \setminus (Z)},
                            \\[8pt]
                            Z{:}~ \left( \begin{array}{@{}l@{}}
                                \deepunfold{\ol{G_Z \onto^{\pri_Z} p}}{{(X,t_X,\ol{G_X \onto^{\pri_X} p})}_{X \in \widetilde{Y_Z}}},
                                \\
                                {\Big(\deepunfold{\relprop{p}{q}{\pri_Z}{G_Z \wrt (p,q)}}{{(X,t_X,\relprop{p}{q}{\pri_X}{G_X \wrt (p,q)})}_{X \in \widetilde{Y_Z}}}\Big)}_{q \in \tilde{q}}
                            \end{array} \right);
                            \\[8pt]
                            \ci{\mu}{p}{:}~ \mu Z \sdot (\lift{t_Z} \deepunfold{\ol{G_Z \onto^{\pri_Z} p}}{{(X,t_X,\ol{G_X \onto^{\pri_X} p})}_{X \in \widetilde{Y_Z}}}),
                            \\[6pt]
                            {\left( \crt{p}{q}{:}~ \mu Z \sdot (\lift{t_Z} \deepunfold{\relprop{p}{q}{\pri_Z}{G_Z \wrt (p,q)}}{{(X,t_X,\relprop{p}{q}{\pri_X}{G_X \wrt (p,q)})}_{X \in \widetilde{Y_Z}}}) \right)}_{q \in \tilde{q}}
                        \end{array}
                    \end{array}$
                }{\scc{Var}}
            \end{prooftree}
            In this typing, the type of $\ci{\mu}{p}$ concurs with the expected type in~\eqref{eq:rCallMu}, the types of $\crt{p}{q}$ for each $q \in \tilde{q}$ concur with the expected types in~\eqref{eq:rCallTildeq}, and the assignment to $Z$ in the recursive context concurs with~\eqref{eq:rCallRec}.
            This proves the thesis.
            \qedhere
    \end{itemize}
\end{proof}

Now, we can prove \Cref{t:routerTypes} as a corollary of \Cref{t:routerTypesGen}:
\begin{proof}[Proof of \Cref{t:routerTypes} on \Cpageref{t:routerTypes}]\label{proof:routerTypes}
    We have been given a closed, relative well-formed global type $G$, and a participant $p \in \part(G)$.
    Let $C := []$ and $G_s := G$.
    Clearly, $G_s \gsub{C} G$.
    By \Cref{d:activeParts}, ${\activ{C}{G} = {\part(G)}^2}$.
    For $p$ to be a participant of $G$, there must be an exchange involving $p$ and some other participant $q$, i.e.\ there exists a $q \in \part(G)$ such that $(p,q) \in \activ{C}{G}$.
    Moreover, $\tilde{q}$ as defined in \Cref{t:routerTypesGen} is $\{q \in \part(G) \mid (p,q) \in \activ{C}{G}\} = q \in \part(G) \setminus \{p\}$.
    Hence, \Cref{t:routerTypesGen} allows us to find a typing for $\lmed{p}{\part(G) \setminus \{p\}}{G}$.

    Let us consider the precise values of the ingredients of \Cref{t:routerTypesGen} in our application:
    \begin{enumerate}
        \item
            $\pri_C = \ctxpri{}{C} = 0$,

        \item
            $\widetilde{X_C} = \ctxbind{C} = ()$,

        \item
            $\begin{array}[t]{@{}rll@{}}
                D_p
                &= \deepunfold{\dual{G_s \onto^{\pri_C} p}}{{(X,t_X,\dual{G_x \onto^{\pri_X} p})}_{X \in \widetilde{X_C}}}
                \\
                &= \dual{G \onto^0 p}
                &\text{(cf.\ \Cref{d:deepUnfold}),}
            \end{array}$

        \item
            $\begin{array}[t]{@{}rll@{}}
                E_q
                &= \deepunfold{\relprop{p}{q}{\pri_C}{G_s \wrt (p,q)}}{{(X,t_X,\relprop{p}{q}{\pri_X}{G_X \wrt (p,q)})}_{X \in \widetilde{X_C}}}
                \\
                &= \relprop{p}{q}{0}{G \wrt (p,q)}
                &\text{(cf.\ \Cref{d:deepUnfold}).}
            \end{array}$
    \end{enumerate}

    Finally, the result of \Cref{t:routerTypesGen} is as follows:
    \begin{align*}
        \lmed{p}{\tilde{q}}{G_s} \vdash
        {\Big(X{:}~ \big( A_X, {(B_{X,q})}_{q \in \tilde{q}_X} \big) \Big)}_{X \in \widetilde{X_C}};~
        \ci{\mu}{p}{:}~ D_p,~
        {(\crt{p}{q}{:}~ E_q)}_{q \in \tilde{q}}
    \end{align*}
    Applying (1)--(4) above, we get the following:
    \begin{align*}
        \lmed{p}{\part(G) \setminus \{p\}}{G} \vdash
        \emptyset;~
        \ci{\mu}{p}{:}~ \dual{G \onto^0 p},~
        {\big(\crt{p}{q}{:}~ \relprop{p}{q}{0}{G \wrt (p,q)}\big)}_{q \in \part(G) \setminus \{p\}}
    \end{align*}
    This coincides exactly with the result of \Cref{t:routerTypes}.
\end{proof}

\subsubsection{Transference of Results (Operational Correspondence)}
\label{ss:transfer}

Given a global type $G$, we now formalize the transference of correctness properties such as deadlock freedom from `$\sys(G)$' (cf.\ \Cref{d:networks}) to `$G$'.
Here, we define an \emph{operational correspondence} between networks and global types, in both directions.
That is, we show that a network performs interactions between implementations and routers and between pairs of routers if and only if that communication step is stipulated in the corresponding global type (\Cref{t:completeness,t:soundness}).

Before formalizing the operational correspondence, we show that networks of routed implementations never reduce to alarm processes.
To be precise, because alarm processes only can occur in routers (not in implementations), we show that
none of the routers of a network reduces to an alarm process, formalized using evaluation contexts:

\begin{definition}[Evaluation Context]\label{d:evalcontext}
    We define an \emph{evaluation context} as a process with a single hole `$\mkern1mu\hole\mkern-3mu$', not prefixed by input or branching:
    \begin{align*}
        E ::= \nu{x y}\, E \sepr P \| E \sepr \mu X (\tilde{z}) \sdot E \sepr \hole
    \end{align*}
    Given an evaluation context $E$, we write `$\mkern1mu E[P]\mkern-3mu$' to denote the process obtained by replacing the hole in $E$ with $P$.
\end{definition}

\begin{theorem}\label{t:noError}
    Given a relative well-formed global type $G$ and a network of routed implementations $\mcl{N} \in \sys(G)$, then
    \[
        \mcl{N} \nredd^\ast E[\error{\tilde{x}}],
    \]
    for any evaluation context $E$ and set of endpoints $\tilde{x}$.
\end{theorem}

\begin{proof}
    By definition (\Cref{d:networks}), $\mcl{N}$ consists only of routers (\Cref{d:router}) and well-typed processes not containing the alarm process (cf.\ the assumption below \Cref{d:error}).

    Suppose, for contradiction, that there are $E \in \mcl{E}$ and $\tilde{x}$ such that $\mcl{N} \redd^\ast E[\error{\tilde{x}}]$.
    Since only routers can contain the alarm process, there is a router $\rtr_p$ in $\mcl{N}$ for participant $p \in \part(G)$ that reduces to the alarm process.
    Since it is the only possibility for a router synthesized by \Cref{alg:router} to contain the alarm process, it must contain the process in \eqref{eq:patched}.
    This process is synthesized on \cref{line:rtrBoth} of \Cref{alg:router}, so there is an exchange in $G$ with sender $s \in \part(G) \setminus \{p\}$ and recipient $r \in \part(G) \setminus \{p\}$ that is a dependency for the interactions of $p$ with both $s$ and $r$.

    For this exchange, the router $\rtr_s$ for $s$ contains the process returned on \cref{line:rtrSend} of \Cref{alg:router}, and the router $\rtr_r$ for $r$ contains the process returned on \cref{line:rtrRecv}.
    Suppose $s$ has a choice between the labels in $I$, and the implementation of $s$ chooses $i \in I$.
    Then, $\rtr_s$ sends $i$ to $\rtr_r$ and $\rtr_p$.

    Now, for $\rtr_p$ to reduce to the alarm process, it has to receive from $\rtr_r$ a label $i' \in I \setminus \{i\}$.
    However, this contradicts \cref{line:rtrRecv} of \Cref{alg:router}, which clearly defines $\rtr_r$ to send $i$ to $\rtr_p$.
    Hence, $\mcl{N} \nredd^\ast E[\error{\tilde{x}}]$.
\end{proof}

It follows from this and the typability of routers (\Cref{t:routerTypes}) that networks of routed implementations are deadlock free:

\begin{theorem}\label{t:globalDlFree}
    For relative well-formed global type $G$, every $\mcl{N} \in \sys(G)$ is deadlock free.
    \lipicsEnd
\end{theorem}

\begin{proof}
    By the typability of routers (\Cref{t:routerTypes}) and the duality of the types of router channels (\Cref{t:relpropDual}), $\mcl{N} \vdash \emptyset; \emptyset$.
    Hence, by \Cref{t:closedDLFree}, $\mcl{N}$ is deadlock free, and by \Cref{t:noError}, $\mcl{N}$ never reduces to the alarm process.
\end{proof}

To formalize our operational correspondence result, we apply the labeled reductions for processes `$\mkern1mu P \xredd{\alpha} Q$' (cf.\ \Cref{d:procLred}) and define a labeled transition system (LTS) for global types.

\begin{definition}[LTS for Global Types]\label{d:globalLts}
    We define the relation `$\mkern1mu G \xrightarrow{\alpha} G'\mkern-2mu$', with labels `$\mkern1mu\beta\mkern-3mu$' of the form `$\mkern1mu p \rangle q{:}\ell\<S\>\mkern-3mu$' (sender, recipient, label, and message type), by the following rules:

    \medskip
    \noindent
    \mbox{}\hfill%
    \begin{wfit}
        \begin{prooftree}
            \infAss{
                $\raisebox{8pt}{} j \in I$
            }
            \infUn{
                $\mkern-6mu p \mkern1mu{\mto}\mkern1mu q \{i\<S_i\> \sdot G_i\}_{i \in I} \mkern1mu{\xrightarrow{p \rangle q{:}j\<S_j\>}}\mkern1mu G_j \mkern-6mu$
            }{}
        \end{prooftree}
    \end{wfit}%
    \hfill%
      \begin{wfit}
        \begin{prooftree}
            \infAss{
                $G \xrightarrow{\alpha} G'$
            }
            \infUn{
                $\mkern-6mu \raisebox{12pt}{} \gskip \sdot G \mkern1mu{\xrightarrow{\alpha}}\mkern1mu G' \mkern-6mu$
            }{}
        \end{prooftree}
    \end{wfit}%
    \hfill%
    \begin{wfit}
        \begin{prooftree}
            \infAss{
                $G \subst{\mu X \sdot G / X} \xrightarrow{\alpha} G'$
            }
            \infUn{
                $\mu X \sdot G \xrightarrow{\alpha} G'$
            }{}
        \end{prooftree}
    \end{wfit}%
    \hfill\mbox{}
\end{definition}

Intuitively, operational correspondence states:
\begin{enumerate}
    \item
        every transition of a global type is mimicked by a precise sequence of labeled reductions originating from an associated completable network (\emph{completeness}; \Cref{t:completeness}), and
    \item
        for every labeled reduction originated in a completable network there is a corresponding global type transition (\emph{soundness}; \Cref{t:soundness}).
\end{enumerate}

\noindent
We write `$\rho_1 \rho_2$' for the composition of relations `$\rho_1$' and `$\rho_2$'.
Recall that the notation `$\redd^\star$' stands for finite sequences of reductions, as defined in \Cref{n:redd}.

\begin{theorem}[Operational Correspondence: Completeness]\label{t:completeness}
    Suppose given a relative well-formed global type $G$.
    Also, suppose given $p,q \in \part(G)$ and a set of labels $J$ such that $j \in J$ if and only if $G \xrightarrow{p \rangle q: j\<S_j\>} G_j$ for some $S_j$.
    Then,
    \begin{enumerate}
        \item
            for any completable $\mcl{N} \in \sys(G)$, there exists a $j' \in J$ such that $\mcl{N}^\complet \redd^\star \xredd{\ci{p}{\mu} \rangle \ci{\mu}{p}:j'}\, \mcl{N}_0$;
        \item
            for any $j' \in J$, there exists a completable $\mcl{N} \in \sys(G)$ such that $\mcl{N}^\complet \redd^\star \xredd{\ci{p}{\mu} \rangle \ci{\mu}{p}:j'}\, \mcl{N}_0$;
        \item
            for any completable $\mcl{N} \in \sys(G)$ and any $j' \in J$, if $\mcl{N}^\complet \redd^\star \xredd{\ci{p}{\mu} \rangle \ci{\mu}{p}:j'}\, \mcl{N}_0$, then there exists a comple\-table $\mcl{N}_{j'} \in \sys(G_{j'})$ such that,
            \begin{align*}
                \mcl{N}_0 \,\xredd{\crt{p}{q} \rangle \crt{q}{p}{:}j'} \redd^\star \xredd{\ci{\mu}{q} \rangle \ci{q}{\mu}{:}j'} \redd^\star \xredd{\ci{p}{\mu} \rangle \ci{\mu}{p}{:}v}\, \xredd{\crt{p}{q} \rangle \crt{q}{p}{:}w}\, \xredd{v \fwd w} \redd^\star \xredd{\ci{\mu}{q} \rangle \ci{q}{\mu}{:}w}\, \xredd{v \fwd w}\, \mcl{N}_{j'}^\complet.
            \end{align*}
    \end{enumerate}
\end{theorem}

\begin{proof}
    By the labelled transitions of global types (\defref{d:globalLts}) and relative well-formedness, $G$ is a sequence of $\gskip$s followed by an exchange from $p$ to $q$ over the labels in $J$.
    Since the $\gskip$s do not influence the behavior of routers, let us assume simply that
    \begin{align*}
        G = p \mto q \{j \<S_j\> \sdot G_j\}_{j \in J}.
    \end{align*}
    We prove each Subitem separately.

    \begin{enumerate}
        \item[(a)]
            Take any completable $\mcl{N} \in \sys(G)$.
            By definition (\defref{d:complNet}), $\mcl{N}^\complet \vdash \emptyset; \emptyset$.
            By the construction of networks of routed implementations (\defref{d:networks}), $\ci{p}{\mu} \in \bn(\mcl{N}^\complet)$, and $\ci{p}{\mu}$ is connected to $\ci{\mu}{p}$.

            Also by construction, the type of $\ci{p}{\mu}$ in the typing derivation of $\mcl{N}^\complet$ is
            \begin{align*}
                G \onto^0 p = \oplus^0 \{j: \msgprop{S_j} \tensor^1 (G_j \onto^4 p)\})_{j \in J}.
            \end{align*}
            By the well-typedness of $\mcl{N}^\complet$, we can infer the kind of action that is defined on $\ci{p}{\mu}$: a selection, or a forwarder.
            By induction on the number of connected forwarders (which is finite by the finiteness of process terms), eventually a forwarder has to be connected to a selection.
            So, after reducing the forwarders, we have a selection on $\ci{p}{\mu}$, of some $j' \in J$.

            Hence, by Fairness (\Cref{t:fairness}), after a finite number of steps, we can observe a communication of the label $j'$ from $\ci{p}{\mu}$ to $\ci{\mu}{p}$.
            This proves the thesis: $\mcl{N}^\complet \redd^\star \xredd{\ci{p}{\mu} \rangle \ci{\mu}{p}: j'}\, \mcl{N}_0$.

        \item[(b)]
            Following the proof of the existence of completable networks (\Cref{p:complNetsExist}), we can generate an implementation process for all of $G$'s participants from local projections (cf.\ \Cref{p:charProc}).
            Take any $j' \in J$.
            For the implementation process of $p$, we specifically generate an implementation process that sends the label $j'$.
            These implementation processes allow us to construct $\mcl{N}$, which by construction is in $\sys(G)$ and is completable.
            Following the reasoning as in Subitem (a), ${\mcl{N}^\complet \redd^\star \xredd{\ci{p}{\mu} \rangle \ci{\mu}{p} : j'}\, \mcl{N}_0}$.

        \item[(c)]
            By definition (\defref{d:complNet}), $\mcl{N}^\complet \vdash \emptyset; \emptyset$.
            Hence, by Fairness (\Cref{t:fairness}), for any of the pending names of $\mcl{N}^\complet$, we can observe a communication after a finite number of steps.
            By construction (\defref{d:networks}), the endpoints that we are required to observe by thesis are bound in $\mcl{N}^\complet$.
            From the shape of $G$, the definition of routed implementations (\defref{d:routedImplementations}), and the typability of routers (\Cref{t:routerTypes}), we know the types of all the required endpoints in $\mcl{N}^\complet$.
            We can deduce the required labeled reductions following the reasoning as in Subitem (a).
            Let us summarize the origin of each of the network's steps:
            \begin{enumerate}[label=\arabic*.]
                \item
                    $\mcl{N}^\complet \redd^\star \xredd{\ci{p}{\mu} \rangle \ci{\mu}{p} : j'}\, \mcl{N}_0$:
                    The implementation of $p$ selects label $j'$ with $p$'s router.

                \item
                    $\mcl{N}_0 \xredd{\crt{p}{q} \rangle \crt{q}{p} : j'} \mcl{N}_1$:
                    The router of $p$ forwards $j'$ to $q$'s router.

                \item
                    $\mcl{N}_1 \redd^\star \mcl{N}_2$:
                    The router of $p$ forwards $j'$ to the routers of the participant that depend on the output by $p$, and these routers forward $j'$ to their respective implementations.

                \item
                    $\mcl{N}_2 \,\xredd{\ci{\mu}{q} \rangle \ci{q}{\mu} : j'} \redd^\star \mcl{N}_3$:
                    The router of $q$ forwards $j'$ to $q$'s implementation, and to the routers of the participants that depend on the input by $q$, and these routers forward $j'$ to their respective implementation (if they have not done so already for the output dependency on $p$).

                \item
                    $\mcl{N}_3 \xredd{\ci{p}{\mu} \rangle \ci{\mu}{p} : v} \xredd{\crt{p}{q} \rangle \crt{q}{p} : w} \xredd{v \fwd w} \mcl{N}_{4,v}$:
                    The implementation of $p$ sends an endpoint $v$ to $p$'s router, which sends a fresh endpoint $w$ to $q$'s router, and $v$ is forwarded to $w$.

                \item
                    $\mcl{N}_{4,v} \redd^\star \xredd{\ci{\mu}{q} \rangle \ci{q}{\mu} : w}\, \xredd{v \fwd w}\, \mcl{N}_{j'}^\complet$:
                    The router of $q$ sends a fresh endpoint $w$ to $q$'s implementations, and $v$ is forwarded to $w$.

                    In $\mcl{N}_{j'}^\complet$, all routers have transitioned to routers for $G_{j'}$.
                    Moreover, by Type Preservation (\Cref{t:subjectRed}), $\mcl{N}_{j'}^\complet \vdash \emptyset; \emptyset$.
                    By isolating restrictions on endpoints that belong only to implementation processes, we can find $\mcl{N}_{j'} \in \sys(G_{j'})$ such that $\mcl{N}_{j'}^\complet$ is its completion.
                    This proves the thesis.
            \end{enumerate}
    \end{enumerate}

    Note that $G$ can also contain recursive definitions before the initial exchange; this case can be dealt with by unfolding.
\end{proof}

\revstart[1]

Our soundness result, given below as~\Cref{t:soundness}, will  capture the notion that after any sequence of reductions from the network of a global type $G$, a network of another global type $G'$ can be reached. Crucially,  $G'$ can be reached from $G$ through a series of transitions.
Networks are inherently concurrent, whereas global types are built out of sequential compositions; as a result, the network could have enabled (asynchronous) actions that correspond to exchanges that are not immediately enabled in the global type.

For example, consider the global types $G = a \mto b \big\{1\<S_1\>.c \mto d \{1\<S'\>.\gend\}, 2\<S_2\>.c \mto d \{1\<S'\>.\gend\}\big\}$
and
$G' = a \mto b \big\{1\<S\>.b \mto c \{1\<S'\>.\gend\}\big\}$.
Clearly, the initial exchange in $G$ between $a$ and $b$ is not a dependency for the following exchange between $c$ and $d$.
The routers of $c$ and $d$ synthesized from $G$ thus  start with their exchange, without awaiting the initial exchange between $a$ and $b$ to complete.
Hence, in a network of $G$, both exchanges in $G$ may be enabled simultaneously.
We further refer to exchanges that may be simultaneously enabled in networks as \emph{independent (global) exchanges}.
While all exchanges appearing in $G$ are independent, the two exchanges in $G'$ are not.

In the proof of soundness, we may encounter in a network reductions related to independent exchanges, so we have to be able to identify the independent exchanges in the global type to which the network belongs.
\Cref{l:delayedPrefix} states that independent exchanges related to observed reductions in a network of a global type $G$ can be reached from $G$ after any sequence of transitions in a finite number of steps.
The proof of this lemma relies on \Cref{l:routerNodepEqual}, which ensures that if a participant does not depend on a certain exchange, then the routers synthesized at each of the branches of the exchange are equal.

\begin{lemma}\label{l:routerNodepEqual}
    Suppose given a relative well-formed global type $G = s \mto r \{i\<S_i\>.G_i\}_{i \in I}$, and take any ${p \in \part(G) \setminus \{s,r\}}$ and $\tilde{q} \subseteq \part(G) \setminus \{p\}$.
    If neither $\hdep(p,s,G)$ nor $\hdep(p,r,G)$ holds, then ${\lmed{p}{\tilde{q}}{G_i} = \lmed{p}{\tilde{q}}{G_j}}$ for every $i,j \in I$.
\end{lemma}

\begin{proof}
    The analysis proceeds by cases on the structure of $G$.
    As a representative case we consider ${G = s \mto r \{1\<S_1\>.G_1, 2\<S_2\>.G_2\}}$.
    Towards a contradiction, we assume $\lmed{p}{\tilde{q}}{G_1} \neq \lmed{p}{\tilde{q}}{G_2}$.
    There are many cases where \Cref{alg:router} generates differents routers for $p$ at $G_1$ and at $G_2$.
    We discuss the interesting case where $\lmed{p}{\tilde{q}}{G_1} = \ci{\mu}{p} \gets \ldots$ (\cref{line:rtrSend}) and $\lmed{p}{\tilde{q}}{G_2} = \crt{p}{q_2} \gets \ldots$ (\cref{line:rtrRecv}).
    Then $G_1 = p \mto q_1 \{\ldots\}$ and $G_2 = q_2 \mto p \{\ldots\}$.
    We have $G_1 \wrt (p,q_1) = p \{\ldots\}$ and $G_2 \wrt (p,q_1) = \gskip \ldots$ or $G_2 \wrt (p,q_1) = p{?}q_2 \{\ldots\}$ (w.l.o.g., assume the former).
    Since $G$ is relative well-formed, the projection $G \wrt (p,q_1)$ must exist.
    Hence, since $p \notin \{s,r\}$ and $G_1 \wrt (p,q_1) \neq G_2 \wrt (p,q_1)$, it must be the case that $q_1 \in \{s,r\}$---w.l.o.g., assume $q_1 = s$.
    Then $G \wrt (p,q_1) = q_1{!}r \{1.p \{\ldots\}, 2.\gskip \ldots\}$, and thus $\hdep(p,q_1,G) = \hdep(p,s,G)$ is true.
    This contradicts the assumption that $\hdep(p,s,G)$ is false.
\end{proof}

\begin{lemma}\label{l:delayedPrefix}
    Suppose given a relative well-formed global type $G$ and a completable $\mcl{N} \in \sys(G)$ such that $\mcl{N}^\complet \xredd{\ci{c}{\mu} \> \ci{\mu}{c}:\ell}$, for some $c \in \part(G)$.
    For every $G'$ and $\beta_1,\ldots,\beta_n$ ($n \geq 0$) such that
    $G \xrightarrow{\beta_1} \ldots \xrightarrow{\beta_n} G'$ where $c$ is not involved in any $\beta_k$ (with $G = G'$ if $n=0$),
    there exist $G''$, $d \in \part(G)$, and $\beta'_1,\ldots,\beta'_m$ ($m \geq 0$)
    such that ${G' \xrightarrow{\beta'_1} \ldots \xrightarrow{\beta'_m} G'' = c \mto d \{i\<S_i\>.G_i\}_{i \in I}}$ where $c$ is not involved in any $\beta'_k$  (with $G'' = c \mto d \{i\<S_i\>.G_i\}_{i \in I}$ if $m=0$).
\end{lemma}

\begin{proof}
By induction on $n$ (\ih1).
We first observe that the behavior on $\ci{\mu}{c}$ in $\mcl{N}^\complet$ can only arise from the router generated for $c$ at $G$, following \Cref{alg:router} (\cref{line:rtrSend}) after finitely many passes through lines~\labelcref{line:rtrElseElse} (no dependency) and \labelcref{line:rtrSkip} (skip); for simplicity, assume only \cref{line:rtrElseElse} applies.

    \begin{itemize}
        \item Case $n = 0$.
            Let $x \geq 0$ denote the number of passes through \cref{line:rtrElseElse} to generate the router for $c$ at~$G$.
            We apply induction on $x$ (\ih2):
            \begin{itemize}
                \item Case $x = 0$.
                    The router for $c$ at $G$ is generated through \cref{line:rtrSend}, so $G = c \mto d \{i\<S_i\>.G_i\}_{i \in I}$, proving the thesis.

                \item Case $x = x'+1$.
                    Then $G = a \mto b \{i\<S_i\>.G_i\}_{i \in I}$ and \cref{line:rtrElseElse} returns the router for $c$ at $G_j$ for any $j \in I$.
                    We have $G \xrightarrow{a \> b: j\<S_j\>} G_j$.
                    Given the same implementation process for $c$ as in $\mcl{N}$, we can construct a completable $\mcl{M} \in \sys(G_j)$ such that $\mcl{M}^\complet \xredd{\ci{c}{\mu} \> \ci{\mu}{c}:\ell}$.
                    Hence, the thesis follows from \ih2.
            \end{itemize}

        \item Case $n = n'+1$.
            By assumption, $G \xrightarrow{\beta_1} G'_1$ where $c$ is not the sender or recipient in $\beta_1$.
            Hence, $G = a \mto b \{i.\<S_i\>.G_i\}_{i \in I}$ where $G'_1 = G_j$ for some $j \in I$.
            The router for $c$ at $G$ is thus generated through \cref{line:rtrElseElse} of \Cref{alg:router}.
            It follows from \Cref{l:routerNodepEqual} that this router is equal to the router for $c$ at $G'_1$, but with one less pass through \cref{line:rtrElseElse}.
            Given the same implementation process for $c$ as in $\mcl{N}$, we can construct a completable $\mcl{M} \in \sys(G'_1)$ such that $\mcl{M}^\complet \xredd{\ci{c}{\mu} \> \ci{\mu}{c}:\ell}$.
            Hence, the thesis follows from \ih1.
            \qedhere
    \end{itemize}
\end{proof}

The proof of soundness relies on \Cref{p:indep}: if different reductions are enabled for a given process, then they do not exclude each other.
That is, the same process is reached no matter the order in which those reductions are executed.
We refer to simultaneously enabled reductions as \emph{independent reductions}.

\begin{proposition}[Independent Reductions]\label{p:indep}
    Suppose given a process $P \vdash \Omega; \Gamma$ and reduction labels $\alpha$ and $\alpha'_1,\ldots,\alpha'_n$ ($n \geq 1$) where $\alpha \notin \{\alpha'_1,\ldots,\alpha'_n\}$ (cf.\ \Cref{d:procLred}).
    If $P \xredd{\alpha}$ and $P \xredd{\alpha'_1} \ldots \xredd{\alpha'_n}$, then there exists a process $Q$ such that $P \xredd{\alpha} \xredd{\alpha'_1} \ldots \xredd{\alpha'_n} Q$ and $P \xredd{\alpha'_1} \ldots \xredd{\alpha'_n} \xredd{\alpha} ~Q$.
\end{proposition}

\begin{proof}
    By induction on $n$:
    \begin{itemize}
        \item $n = 1$.
            By assumption, $P \xredd{\alpha}$ and $P \xredd{\alpha'_1}$.
            The proof proceeds by considering all possible combinations of shapes for $\alpha$ and $\alpha'_1$ (forwarder, output/input, and selection/branching).

            Consider the case where ${\alpha = x \> y: a}$ and $\alpha'_1 = w \> z: b$.
            Because  $P$ is well-typed, we infer that there are evaluation contexts $E_1$ and $E_2$ such that $P \equiv E_1[\nu{xy}(x[a,c] \| y(a,c).P_1)] \equiv E_2[\nu{wz}(w[b,d] \| z(b,d).P_2])$ (\Cref{d:evalcontext}).
            Since the reductions labeled $\alpha$ and $\alpha'_1$ are both enabled in $P$, it cannot be the case that $x,y \in \fn(P_2)$ and $w,z \in \fn(P_1)$.
            Hence, there exists an evaluation context $E_3$ such that $P \equiv E_3[\nu{xy}(x[a,c] \| y(a,c).P_1) \| \nu{wz}(w[b,d] \| z(b,d).P_2])$.
            Then $P \xredd{\alpha} Q_1 \equiv E_3[P_1 \| \nu{wz}(w[b,d] \| z(b,d).P_2])$ and $P \xredd{\alpha'_1} Q_2 \equiv E_3[\nu{xy}(x[a,c] \| y(a,c).P_1) \| P_2]$.
            Let $Q = E_3[P_1 \| P_2]$; then $Q_1 \xredd{\alpha'_1} Q$ and $Q_2 \xredd{\alpha} Q$.
            Hence, $P \xredd{\alpha} \xredd{\alpha'_1} Q$ and $P \xredd{\alpha'_1} \xredd{\alpha} Q$.

			All other cases proceed similarly.
            Note that when one of the reductions (say, $\alpha$) has a selection/branching label, such a reduction  would discard some branches and thus possible behaviors.
            This is not an issue for establishing the thesis, because typability guarantees that the sub-process that enables the $\alpha'$-labeled reduction does not appear under the to-be-discarded branches. Hence, the execution of $\alpha$ will not jeopardize the $\alpha'$-labeled reduction.

        \item $n = n'+1$ for $n' \geq 1$.
            By the IH, ${P \xredd{\alpha} \xredd{\alpha'_1} \ldots \xredd{\alpha'_{n'}} Q'_1}$ and $P \xredd{\alpha'_1} \ldots \xredd{\alpha'_{n'}} P' \xredd{\alpha} Q'_1$.
            By assumption, $P' \xredd{\alpha'_n} Q'_2$.
            Since $P$ is well-typed, by \Cref{t:subjectRed} (Subject Reduction), $P'$ is well-typed.
            Since $P' \xredd{\alpha} Q'_1$ and $P' \xredd{\alpha'_n} Q'_2$, we can follow the same argumentation as in the base case to show that $Q'_1 \xredd{\alpha'_n} Q$ and $Q'_2 \xredd{\alpha} Q$.
            Hence, $P \xredd{\alpha} \xredd{\alpha'_1} \ldots \xredd{\alpha'_{n'}} Q'_1 \xredd{\alpha'_n} Q$ and $P \xredd{\alpha'_1} \ldots \xredd{\alpha'_{n'}} P' \xredd{\alpha'_n} Q'_2 \xredd{\alpha} Q$.
            \qedhere
    \end{itemize}
\end{proof}

To understand the proof of soundness and the r\^{o}le of independent reductions therein, consider the following example.
We first introduce some notation which we also use in the proof of soundness: given an ordered sequence of reduction labels $A = (\alpha_1,\ldots,\alpha_k)$, we write $P \xredd{A} Q$ to denote $P \xredd{\alpha_1} \ldots \xredd{\alpha_k} Q$.

\begin{example}\label{ex:presound}
    The recursive global type $G = \mu X.a \mto b:1\<S\>.c \mto d:1\<S\>.X$
    features two independent exchanges.
    Consider a network $\mcl{N} \in \sys(G)$.
    Let $A$ denote the sequence of labeled reductions necessary to complete the exchange in $G$ between $a$ and $b$, and $C$ similarly for the exchange between $c$ and $d$.
    Assuming that communication with routers is not blocked by implementation processes, we have $\mcl{N}^\complet \xredd{A}$ and $\mcl{N}^\complet \xredd{C}$, because the exchanges are independent.

    Now, suppose that from $\mcl{N}^\complet$ we observe $m$ times the sequence of $C$ reductions: $\mcl{N}^\complet \underbrace{\xredd{C} \ldots \xredd{C}}_{\text{$m$ times}} N'$. We see that $N'$ is not a network of a global type reachable from $G$: there are still $m$ exchanges between $a$~and~$b$ pending.
    Still, we can exhibit a series of transitions from $G$ that includes $m$ times the exchange between $c$ and $d$: $$G \underbrace{\xrightarrow{a \> b:1\<S\>} \xrightarrow{c \> d:1\<S\>} \ldots \xrightarrow{a \> b:1\<S\>} \xrightarrow{c \> d:1\<S\>}}_{\text{$m$ times}} G$$
    Following these transitions, we can exhibit a corresponding sequence of reductions from $\mcl{N}^\complet$ that includes $m$ times the sequence $C$ and ends up in another network $\mcl{M} \in \sys(G)$: $$\mcl{N}^\complet \underbrace{\xredd{A} \xredd{C} \ldots \xredd{A} \xredd{C}}_{\text{$m$ times}} \mcl{M}^\complet$$
    At this point it is crucial that from $\mcl{N}^\complet$ the sequences of reductions $A$ and $C$ can be performed independently.
    Hence, by \Cref{p:indep}, $\mcl{N}^\complet \underbrace{\xredd{C} \ldots \xredd{C}}_{\text{$m$ times}} N' \underbrace{\xredd{A} \ldots \xredd{A}}_{\text{$m$ times}} \mcl{M}^\complet$.
\end{example}

In the proof of soundness, whenever we assure that certain reductions are independent, we refer to those assurances as \emph{independence facts} (\textbf{IFacts}).
Also, in the proof we consider labeled reductions, and distinguish between \emph{protocol} and \emph{implementation} reductions: the former are reductions with labels that indicate any interaction with a router, and the latter are any other reductions (which, by the definition of networks, can only occur within participant implementation processes).
By a slight abuse of notation, given ordered sequences of reduction labels $A$ and $A'$, we write $A' \subseteq A$ to denote that $A'$ is a subsequence of $A$, where the labels in $A'$ appear in the same order in $A$ but not necessarily in sequence (and similarly for $A' \subset A$).
With $A \setminus A'$ we denote the sequence obtained from $A$ by removing all the labels in $A'$, and $A \cup A'$ denotes the sequence obtained by adding the labels from $A'$ to the end of $A$.

\begin{theorem}[Operational Correspondence: Soundness]\label{t:soundness}
    Suppose given
    a relative well-formed global type $G$ and
    a completable $\mcl{N} \in \sys(G)$.
    For every ordered sequence of $k \geq 0$ reduction labels $A = (\alpha_1, \ldots, \alpha_k)$ and $N'$ such that $\mcl{N}^\complet \xredd{A} N'$, there exist $G'$ and
    $\beta_1,\ldots,\beta_n$ (with $n \geq 0$) such that (i)~$G \xrightarrow{\beta_1} \ldots \xrightarrow{\beta_n} G'$
    and (ii)~$N' \redd^\ast \mcl{M}^\complet$, with $\mcl{M} \in \sys(G')$.
\end{theorem}

\begin{proof}
    By induction on the structure  of $G$;  we detail the interesting cases of labeled exchanges with implicitly unfolded recursive definitions.
    We exhibit transitions $G \xrightarrow{\beta_1} \ldots \xrightarrow{\beta_n} G'$ and establish a corresponding sequence of reductions $\mcl{N}^\complet \redd^\ast \mcl{M}^\complet$ that includes all the labels in $A$, with $\mcl{M} \in \sys(G')$.
    During this step, we assure the independence between the observed reductions $A$ and the reductions we establish (\textbf{IFacts}).
    Using these independence assurances, we show that also $N' \redd^\ast \mcl{M}^\complet$.

    We apply induction on the size of $A$ (\ih1) to show the existence of (i) $G'$ and $\beta_1,\ldots,\beta_n$ such that (i) $G \xrightarrow{\beta_1} \ldots \xrightarrow{\beta_n} G'$ and (ii) $\mcl{N}^\complet \redd^\ast \mcl{M}^\complet$ including all reductions in $A$, with $\mcl{M} \in \sys(G')$:
    \begin{itemize}
        \item
            Base case: then $A$ is empty, and the thesis holds trivially, with $G' = G$ and $\mcl{M} = \mcl{N}$.

        \item
            Inductive case: then $A$ is non-empty.

            By the definition of networks (\Cref{d:networks}), we know that reductions starting at $\mcl{N}^\complet$ are protocol reductions related to an independent exchange in $G$, or implementation reductions.
            Every protocol reduction in $A$ is related to some exchange in $G$, and so we can group sequences of protocol reductions related to the same exchange.
            By construction, every such sequence of protocol reductions $A_\ast \subseteq A$ starts with an implementation sending a label to a router, i.e., with a label of the form $\alpha_\ast = \ci{c}{\mu} \> \ci{\mu}{c}: \ell$.
            For each such $\alpha_\ast$, the router in $\mcl{N}$ of the sender $c$ has been synthesized from $G$ in a finite number of inductive steps.
            We take the $\alpha_\ast$ that originates from the router synthesized in the least number of steps.
            This gives us the $A_\ast$ starting with $\alpha_\ast$ that relates to an exchange in $G$ which is not prefixed by exchanges relating to any of the other $A'_\ast \subseteq A \setminus A_\ast$.

            Networks are well-typed by definition.
            None of the reductions in $A_\ast$ are blocked by protocol reductions appearing earlier in $A$ (\textbf{IFact~1}): they originate from exchanges in $G$ appearing after the exchange related to $A_\ast$, and the priorities in their related types are thus higher than those in the types related to $A_\ast$, i.e., blocking by input or branching would contradict the well-typedness of $\mcl{N}^\complet$.
            However, it may be that some implementation reductions $A_+ \subseteq A \setminus A_\ast$ do block the reductions in $A_\ast$; they are also not blocked by any prior protocol reductions due to priorities (\textbf{IFact~2}).
            Hence, from $\mcl{N}^\complet$ we can perform the implementation reductions in $A_+$.
            By Subject Reduction (\Cref{t:subjectRed}), this results in another completed network $\mcl{N}_0^\complet$ of $G$.
            This establishes the reduction sequence $\mcl{N}^\complet \xredd{A_+} \mcl{N}_0^\complet \xredd{A_\ast}$.

            By \Cref{l:delayedPrefix}, there are $m \geq 0$ transitions $G \xrightarrow{\beta_1} G_1 \ldots \xrightarrow{\beta_m} G_m$ where the initial prefix of $G_m$ corresponds to the labeled choice by the implementation of $c$: $G_m = c \mto d \{i\<S_i\>.G'_i\}_{i \in I}$, with $\ell \in I$.
            Additionally, $G_m$ contains exchanges related to every sequence of protocol reductions in $A \setminus A_+ \setminus A_\ast$: all these sequences start with a selection from implementation to router, and thus the involved participants do not depend on any of the exchanges between $G$ and $G_m$, such that \Cref{l:routerNodepEqual} applies.
            To establish a sequence of reductions from $\mcl{N}_0^\complet$ to the completion of a network $\mcl{N}_m \in \sys(G_m)$, we
            apply induction on $m$ (\ih2):
            \begin{itemize}
                \item
                    The base case where $m = 0$ is trivial, with $G_m = G$ and thus $\mcl{N}_m^\complet = \mcl{N}_0^\complet$.

                \item
                    In the inductive case, following the same approach as in the proof of completeness (\Cref{t:completeness}), we reduce $\mcl{N}_0^\complet \redd^\ast \mcl{N}_1^\complet$ such that $\mcl{N}_1 \in \sys(G_1)$.
                    Then, by \ih2, $\mcl{N}_1^\complet \redd^\ast \mcl{N}_m^\complet$ where $\mcl{N}_m \in \sys(G_m)$.
                    Note that these reductions may require implementation reductions to unblock protocol reductions, and these implementation reductions may appear in $A$.
                    None of the reductions from $\mcl{N}_0^\complet$ to $\mcl{N}_m^\complet$ can be blocked by any of the other protocol reductions in $A$, following again from priorities in types; hence, the leftover reductions in $A$ are independent from these reductions (\textbf{IFact~3}).
                    Additionally, the sequence of protocol reductions $A_\ast$ was already enabled from $\mcl{N}_0^\complet$, so those reductions are also independent (\textbf{IFact~4}).
            \end{itemize}

            We know that $\mcl{N}_m^\complet \xredd{\ci{c}{\mu} \> \ci{\mu}{c}: \ell}$ and $G_m \xrightarrow{c \> d:\ell\<S_\ell\>} G'_\ell$.
            From $\mcl{N}_m^\complet$, we again follow the proof of completeness to show that $\mcl{N}_m^\complet  \redd^\ast \mcl{M}_\ell^\complet$, where $\mcl{M}_\ell \in \sys(G'_\ell)$.
            Given the definition of routers, it must be that all the reductions in $A_\ast$ appear in this sequence of reductions.
            Let $A' \subset A$ denote the leftover reductions from $A$ (i.e., $A$ except all reductions that occurred between $\mcl{N}^\complet$ and $\mcl{M}_\ell^\complet$, including $A_\ast$ and $A_+$).
            By \textbf{IFacts~1--4}, $\mcl{M}_\ell^\complet \xredd{A'} M'$.
            Then by \ih1, there exist $G'$ and $\beta_{m+2},\ldots,\beta_n$ (with $n \geq m+1$) such that (i)~$G'_\ell \xrightarrow{\beta_{m+2}} \ldots \xrightarrow{\beta_n} G'$ and (ii) $\mcl{M}_\ell^\complet \redd^\ast \mcl{M}^\complet$ including all reductions in~$A'$, with $\mcl{M} \in \sys(G')$.
            Let $\beta_{m+1} = c \> d:\ell\<S_\ell\>$.
            We have shown the existence of $G'$ and $\beta_1,\ldots,\beta_n$ such that (i) $G \xrightarrow{\beta_1} \ldots \xrightarrow{\beta_m} G_m \xrightarrow{\beta_{m+1}} G'_\ell \xrightarrow{\beta_{m+2}} \ldots \xrightarrow{\beta_n} G'$ and (ii) $\mcl{N}^\complet \redd^\ast \mcl{N}_m^\complet \redd^\ast \mcl{M}_\ell^\complet \redd^\ast \mcl{M}^\complet$ including all reductions in $A$, with $\mcl{M} \in \sys(G')$.
    \end{itemize}

    We are left to show that from $\mcl{N}^\complet \redd^\ast \mcl{M}^\complet$ and $\mcl{N}^\complet \xredd{A} N'$, we can conclude that $N' \redd^\ast \mcl{M}^\complet$.
    We apply induction on the size of $A$ (\ih3), using  \textbf{IFacts~1--4} and \Cref{p:indep}:
    \begin{itemize}
        \item
            Base case: Then $A$ is empty, there is nothing to do, and the thesis is proven.

        \item
            Inductive case: Then $A = A' \cup (\alpha')$.
            By \ih3, $\mcl{N}^\complet \xredd{A'} N'' \redd^\ast N''' \xredd{\alpha'} \redd^\ast \mcl{M}^\complet$.
            Moreover, by assumption, $N'' \xredd{\alpha'} N'$.
            \textbf{IFacts~1--4} show that the $\alpha'$-labeled reduction is independent from the reductions between $N''$ and $N'''$.
            Hence, by \Cref{p:indep}, we have $\mcl{N}^\complet \xredd{A'} N'' \xredd{\alpha'} N' \redd^\ast \mcl{M}^\complet$.
            That is, $\mcl{N}^\complet \xredd{A} N' \redd^\ast \mcl{M}^\complet$, proving the thesis.
            \qedhere
    \end{itemize}
\end{proof}

In the light of \Cref{t:soundness}, let us revisit \Cref{ex:presound}:
\begin{example}[Revisiting \Cref{ex:presound}]
    Recall the global type $G = \mu X.a \mto b:1\<S\>.c \mto d:1\<S\>.X$ from \Cref{ex:presound}, with two independent exchanges.
    We take some $\mcl{N} \in \sys(G)$ such that $\mcl{N}^\complet \underbrace{\xredd{C} \ldots \xredd{C}}_{\text{$m$ times}} N'$, where $C$ denotes the sequence of reduction labels corresponding to the  exchange between $c$ and $d$.
    By \Cref{t:soundness}, there indeed are $G'$ and $\beta_1,\ldots,\beta_n$ such that $G \xrightarrow{\beta_1} \ldots \xrightarrow{\beta_n} G'$ and $N' \redd^\ast \mcl{M}^\complet$, with $\mcl{M} \in \sys(G')$.
    To be precise, following \Cref{t:soundness}, indeed
    $$G \underbrace{\xrightarrow{a \> b:1\<S\>} \xrightarrow{c \> d:1\<S\>} \ldots \xrightarrow{a \> b:1\<S\>} \xrightarrow{c \> d:1\<S\>}}_{\text{$m$ times}} G \quad \text{ and } \quad  \mcl{N}^\complet \underbrace{\xredd{C} \ldots \xredd{C}}_{\text{$m$ times}} N' \underbrace{\xredd{A} \ldots \xredd{A}}_{\text{$m$ times}} \mcl{M}^\complet$$
    where $A$ is the sequence of reduction labels corresponding to the  exchange between $a$ and $b$ and ${\mcl{M} \in \sys(G)}$.
\end{example}
\revend

\subsection{Routers Strictly Generalize Centralized Orchestrators}
\label{ss:mediums}

Unlike our decentralized analysis, previous analyses of global types using binary session types rely on centralized orchestrators (called mediums~\cite{conf/forte/CairesP16} or arbiters~\cite{conf/concur/CarboneLMSW16}).
Here, we show that our approach strictly generalizes these centralized approaches.
Readers interested in our decentralized approach in action may safely skip this section and go directly to \Cref{s:routersInAction}.

We introduce an algorithm that synthesizes an orchestrator---a single process that orchestrates the interactions between a protocol's participants  (\secref{ss:mediumSynth}).
We show that the composition of this orchestrator with a context of participant implementations is behaviorally equivalent to the specific case in which routed implementations are organized in a \emph{centralized composition} (\Cref{t:mediumBisim} in \secref{ss:bisim}).

\subsubsection{Synthesis of Orchestrators}
\label{ss:mediumSynth}

\begin{algorithm}[t]
    \DontPrintSemicolon
    \SetAlgoNoEnd
    \SetInd{.2em}{.7em}
    \SetSideCommentRight
    \Def{$\mdm{\tilde{q}}{G}$}{
        \Switch{$G$}{
            \uCase{$s \mto r \{i\<S_i\> \sdot G_i\}_{i \in I}$}{ \label{line:mdmComm}
                $\deps := \{q \in \tilde{q} \mid \hdep(q,s,G) \vee \hdep(q,r,G)\}$ \; \label{line:mdmDep}
                \KwRet \label{line:mdmCommRet}
                $\ci{\mu}{s} \triangleright \{i{:}~ \ol{\ci{\mu}{r}} \triangleleft i \cdot \underline{{(\ol{\ci{\mu}{q}} \triangleleft i)}_{q \in \deps}} \cdot \ci{\mu}{s}(v) \sdot \ol{\ci{\mu}{r}}[w] \cdot (v \fwd w \| \mdm{\tilde{q}}{G_i}) \}_{i \in I}$
            }

            \uCase{$\mu X \sdot G'$}{ \label{line:mdmRecDef}
                $\tilde{q}' := \{q \in \tilde{q} \mid G \onto^0 q \neq \bullet\}$ \; \label{line:mdmRecComp}
                \lIf{$\tilde{q}' \neq \emptyset$}{ \label{line:mdmRecRet}
                    \KwRet $\mu X({(\ci{\mu}{q})}_{q \in \tilde{q}'}) \sdot \mdm{\tilde{q}'}{G'}$
                }
                \lElse{ \label{line:mdmRecInact}
                    \KwRet $\0$
                }
            }

            \lCase{$X$}{ \label{line:mdmRecCall}
                \KwRet $X\call{{(\ci{\mu}{q})}_{q \in \tilde{q}}}$
            }

            \lCase{$\gskip \sdot G'$}{ \label{line:mdmSkip}
                \KwRet $\mdm{\tilde{q}}{G'}$
            }

            \lCase{$\gend$}{\label{line:mdmEnd}
                \KwRet $\0$
            }
        }
    }
    \caption{Synthesis of Orchestrator Processes (\defref{d:medium}).}
    \label{alg:medium}
\end{algorithm}

We define the synthesis of an orchestrator from a global type.
The orchestrator of $G$ will have a channel endpoint $\ci{\mu}{p_i}$ for connecting to the process implementation of every $p_i \in \part(G)$.

\begin{definition}[Orchestrator]\label{d:medium}
    Given a global type $G$ and participants $\tilde{q}$, \Cref{alg:medium} defines the synthesis of an \emph{orchestrator process}, denoted `$\mkern1mu\mdm{\tilde{q}}{G}\mkern-3mu$', that orchestrates interactions according to $G$.
    \lipicsEnd
\end{definition}

\Cref{alg:medium} follows a similar structure as the router synthesis algorithm (\Cref{alg:router}).
The input parameter `$\tilde{q}$' keeps track of active participants, making sure recursions are well-defined; it should be initialized as `$\part(G)$'.

We briefly discuss how the orchestrator process is generated.
The interesting case is an exchange `$p \mto q \{i\<U_i\> \sdot G_i\}_{i \in I}$' (\cref{line:mdmComm}), where the algorithm combines the several cases of the router's algorithm (that depend on the involvement of the router's participant).
First, the sets of participants `$\deps$' that depend on the sender and on the recipient are computed (\cref{line:mdmDep}) using the auxiliary predicate `$\hdep$' (cf.\ \defref{d:hdep}).
Then, the algorithm returns a process (\cref{line:mdmCommRet}) that receives a label $i \in I$ over $\ci{\mu}{s}$; forwards it over $\ci{\mu}{r}$ and over $\ci{\mu}{q}$ for all $q \in \deps$; receives a channel over $\ci{\mu}{s}$; forwards it over $\ci{\mu}{r}$; and continues as `$\mdm{\tilde{q}}{G_i}$'.

The synthesis of a recursive definition `$\mu X \sdot G'$' (\cref{line:mdmRecDef}) requires care, as the set of active participants $\tilde{q}$ may change.
In order to decide which $q \in \tilde{q}$ are active in $G'$, the algorithm computes the local projection of $G$ onto each $q \in \tilde{q}$ to determine the orchestrator's future behavior on $\ci{\mu}{q}$, creating a new set $\tilde{q}'$ with those $q \in \tilde{q}$ for which the projection is different from `$\bullet$' (\cref{line:mdmRecComp}).
Then, the algorithm returns a recursive process with as context the channel endpoints $\ci{\mu}{q}$ for $q \in \tilde{q}'$, with `$\mdm{\tilde{q}'}{G'}$' as the body.

The synthesis of a recursive call `$X$' (\cref{line:mdmRecCall}) yields a recursive call with as context the channels $\ci{\mu}{q}$ for $q \in \tilde{q}$.
Finally, for `$\gskip \sdot G'$' (\cref{line:mdmSkip}) the algorithm returns the orchestrator for $G'$, and for `$\bullet$' (\cref{line:mdmEnd}) the algorithm returns `$\0$'.

There is a minor difference between the orchestrators synthesized by \Cref{alg:medium} and the mediums defined by Caires and P\'{e}rez~\cite{conf/forte/CairesP16}.
The difference is in the underlined portion in \cref{line:mdmCommRet}, which denotes explicit messages (obtained via dependency detection) needed to deal with non-local choices.
The mediums by Caires and P\'{e}rez do not include such communications, as their typability is based on local types, which rely on a merge operation at projection time.
The explicit actions in \cref{line:mdmCommRet} make the orchestrator compatible with participant implementations that connect with routers.
Aside from these actions, our concept of orchestrator is essentially the same as that of the mediums by Caires and P\'{e}rez.

Crucially, orchestrators can be typed using local projection (cf.\ \defref{d:locproj}) similar to the typing of routers using relative projection (cf.\ \Cref{t:routerTypes}).
This result follows by construction:

\begin{theorem}\label{t:mediumTypes}
    Given a  closed,  relative well-formed global type $G$,
    \begin{align*}
        \mdm{\part(G)}{G} \vdash \emptyset; {(\ci{\mu}{p}{:}~ \dual{(G \onto^0 p)})}_{p \in \part(G)}.
    \end{align*}
\end{theorem}

\begin{proof}
    We prove a more general statement.
    Suppose given a closed, relative well-formed global type $G$.
    Also, suppose given a global type $G_s \gsub{C} G$.
    Consider:
    \begin{itemize}
        \item
            the participants that are active in $G_s$:
            $\tilde{q} = \{q \in \part(G) \mid \exists p \in \part(G).~ (p,q) \in \activ{C}{G}\}$,

        \item
            the absolute priority of $G_s$:
            $\pri_C = \ctxpri{}{C}$,

        \item
            the sequence of bound recursion variables of $G_s$:
            $\widetilde{X_C} = \ctxbind{C}$,

        \item
            for every $X \in \widetilde{X_C}$:
            \begin{itemize}
                \item
                    the body of the recursive definition on $X$ in $G$:
                    $G_X = \recdef{X}{G}$,

                \item
                    the participants that are active in $G_X$:
                    $\tilde{q}_X = \{q \in \part(G) \mid \exists p \in \part(G).~ (p,q) \in \recactiv{X}{G}\}$,

                \item
                    the absolute priority of $G_X$:
                    $\pri_X = \recpri{X}{G}$,

                \item
                    the sequence of bound recursion variables of $G_X$ excluding $X$:
                    $\widetilde{Y_X} = \subbind{\mu X \sdot G_X}{G}$,

                \item
                    the type required for $\ci{\mu}{q}$ for a recursive call on $X$:
                    \begin{align*}
                        A_{X,q} = \deepunfold{\ol{G_X \onto^{\pri_X} q}}{{(Y,t_Y,\ol{G_Y \onto^{\pri_Y} q})}_{Y \in \widetilde{Y_X}}}
                    \end{align*}

                \item
                    the minimum lift for typing a recursive definition on $X$:
                    $t_X = \max_\pr \left( {(A_{X,q})}_{q \in \tilde{q}_X} \right) + 1$,
            \end{itemize}

        \item
            the type expected for $\ci{\mu}{q}$ for the orchestrator for $G_s$:
            \begin{align*}
                D_q = \deepunfold{\ol{G_s \onto^{\pri_C} q}}{{(X,t_X,\ol{G_X \onto^{\pri_X} q})}_{X \in \widetilde{X_C}}}.
            \end{align*}
    \end{itemize}
    Then, we have:
    \begin{align*}
        \mdm{\tilde{q}}{G_s} \vdash
        {\left( X{:}~ {\big(A_{X,q}\big)}_{q \in \tilde{q}_X} \right)}_{X \in \widetilde{X_C}};~
        {\left( \ci{\mu}{q}{:}~ D_q \right)}_{q \in \tilde{q}}
    \end{align*}
    Similar to how \Cref{t:routerTypes} follows from \Cref{t:routerTypesGen}, the thesis follows as a corollary from this more general statement (cf.\ the proof of \Cref{t:routerTypes} on \Cpageref{proof:routerTypes}).

    We apply induction on the structure of $G_s$, with six cases as in \Cref{alg:medium}.
    We only detail the cases of exchange and recursion.
    \begin{itemize}
        \item
            \emph{Exchange}: $G_s = s \mto r \{i \<S_i\> \sdot G_i \}_{i \in I}$ (\cref{line:mdmComm}).

            Following similar reasoning as in the case for exchange in the proof of \Cref{t:routerTypesGen}, we can omit the unfoldings on types, as well as the recursive context.

            Let $\deps_s := \{q \in \tilde{q} \mid \hdep(q,s,G_s)\}$ and $\deps_r := \{q \in \tilde{q} \setminus \deps_s \mid \hdep(q,r,G_s)\}$.
            Note that $\deps_s \cup \deps_r$ coincides with $\deps$ as defined on \cref{line:mdmDep} and that $s,r \notin \deps_s \cup \deps_r$.

            Let us  take stock of  the types we expect for each of the orchestrator's channels.
            \begin{flalign}
                & \text{For $\ci{\mu}{s}$ we expect}
                &
                \ol{G_s \onto^\pri s}
                &= \ol{{\oplus}^\pri \{i{:}~ \msgprop{S_i} \tensor^{\pri+1} (G_i \onto^{\pri+4} s) \}_{i \in I}}
                \nonumber
                \\
                &
                &
                &= \&^\pri \{i{:}~ \ol{\msgprop{S_i}} \parr^{\pri+1} \ol{(G_i \onto^{\pri+4} s)} \}_{i \in I}.
                \label{eq:mSType}
                \\
                & \text{For $\ci{\mu}{r}$ we expect}
                &
                \ol{G_s \onto^\pri r}
                &= \ol{{\&}^{\pri+2} \{i{:}~ \ol{\msgprop{S_i}} \parr^{\pri+3} (G_i \onto^{\pri+4} r) \}_{i \in I}}
                \nonumber
                \\
                &
                &
                &= {\oplus}^{\pri+2} \{i{:}~ \msgprop{S_i} \tensor^{\pri+3} \ol{(G_i \onto^{\pri+4} r)} \}_{i \in I}.
                \label{eq:mRType}
                \\
                & \text{For each $q \in \deps_s$,}
                \nonumber \\
                & \text{for $\ci{\mu}{q}$ we expect}
                &
                \ol{G_s \onto^\pri q}
                &= \ol{\&^{\pri+2} \{i{:}~ (G_i \onto^{\pri+4} q) \}_{i \in I}}
                & \hfill
                \nonumber \\
                &
                &
                &= {\oplus}^{\pri+2} \{i{:}~ \ol{(G_i \onto^{\pri+4} q)} \}_{i \in I}.
                \label{eq:mDepSType}
                \\
                & \text{For each $q \in \deps_r$,}
                \nonumber \\
                & \text{for $\ci{\mu}{q}$ we expect}
                &
                \ol{G_s \onto^\pri q}
                &= \ol{\&^{\pri+3} \{i{:}~ (G_i \onto^{\pri+4} q) \}_{i \in I}}
                \nonumber \\
                &
                &
                &= {\oplus}^{\pri+3} \{i{:}~ \ol{(G_i \onto^{\pri+4} q)} \}_{i \in I}.
                \label{eq:mDepRType}
                \\
                & \text{For each $q \in \tilde{q} \setminus \deps_s \setminus \deps_r \setminus \{s,r\}$,}
                \span
                \nonumber \\
                & \text{for $\ci{\mu}{q}$ we expect}
                &
                \ol{G_s \onto^\pri q}
                &= \ol{G_{i'} \onto^{\pri+4} q} ~\text{for any $i' \in I$}.
                \label{eq:mDepOtherType}
            \end{flalign}

            Let us now consider the process returned by \Cref{alg:router}, with each prefix marked with a number.
            \[
                \mdm{\tilde{q}}{G} =
                \underbrace{\ci{\mu}{s} \triangleright \{i{:}}_{1}~
                    \underbrace{\ol{\ci{\mu}{r}} \triangleleft i}_{2_i} \cdot
                    \underbrace{{(\ol{\ci{\mu}{q}} \triangleleft i)}_{q \in \deps}}_{3_i} \cdot
                    \underbrace{\ci{\mu}{s}(v)}_{4_i} \sdot
                    \underbrace{\ol{\ci{\mu}{r}}[w]}_{5_i} \cdot
                    (v \fwd w \| \mdm{\tilde{q}}{G_i})
                \}_{i \in I}
            \]

            For each $i' \in I$, let $C_{i'} := C[s \mto r (\{i\<S_i\> \sdot G_i\}_{i \in I \setminus \{i'\}} \cup \{i'\<S_{i'}\> \sdot []\})]$.
            Clearly, $G_{i'} \gsub{C_{i'}} G$.
            Also, because we are not adding recursion binders, the current value of $\tilde{q}$ is appropriate for the IH.
            With $C_{i'}$ and $\tilde{q}$, we apply the IH to obtain the typing of $\mdm{\tilde{q}}{G_{i'}}$, where priorities start at $\ctxpri{}{C_{i'}} = \ctxpri{}{C} + 4$ (cf.\ \defref{d:ctxpri}).
            Following these typings, \Cref{f:tderiii} gives the typing of $\mdm{\tilde{q}}{G_s}$, referring to parts of the process by the number marking its foremost prefix above.

            Clearly, the priorities in the derivation in \Cref{f:tderiii} meet all requirements.
            The order of the applications of $\oplus^\star$ for each $q \in \deps_s \cup \deps_r$ does not matter, since the selection actions are asynchronous.

            \begin{figure}[t!]
                \begin{mdframed}
                    {\small
                        \begin{prooftree}
                            \infAx{
                                $\forall i \in I.~ v \fwd w \vdash \begin{array}[t]{@{}l@{}}
                                    v{:}~ \ol{S_i},
                                    w{:}~ S_i
                                \end{array}$
                            }{\scc{Id}}
                            \infAss{
                                $\forall i \in I.~ \mdm{\tilde{q}}{G_i} \vdash \begin{array}[t]{@{}l@{}}
                                    {(\ci{\mu}{q}{:}~ \ol{G_i \onto^{\pri+4} q})}_{q \in \tilde{q}}
                                \end{array}$
                            }
                            \infBin{
                                $\forall i \in I.~ v \fwd w \| \mdm{\tilde{q}}{G_i} \vdash \begin{array}[t]{@{}l@{}}
                                    v{:}~ \ol{S_i},
                                    w{:}~ S_i,
                                    \\
                                    {(\ci{\mu}{q}{:}~ \ol{G_i \onto^{\pri+4}} q)}_{q \in \tilde{q}}
                                \end{array}$
                            }{\scc{Mix}}
                            \infUn{
                                $\forall i \in I.~ 5_i \vdash \begin{array}[t]{@{}l@{}}
                                    v{:}~ \ol{S_i},
                                    \\
                                    \ci{\mu}{r}{:}~ \msgprop{S_i} \tensor^{\pri+3} \ol{(G_i \onto^{\pri+4} r)},
                                    \\
                                    {(\ci{\mu}{q}{:}~ \ol{G_i \onto^{\pri+4} q})}_{q \in \tilde{q} \setminus \{r\}}
                                \end{array}$
                            }{$\tensor^\star$}
                            \infUn{
                                $\forall i \in I.~ 4_i \vdash \begin{array}[t]{@{}l@{}}
                                    \ci{\mu}{s}{:}~ \ol{\msgprop{S_i}} \parr^{\pri+1} \ol{(G_i \onto^{\pri+4} s)},
                                    \\
                                    \ci{\mu}{r}{:}~ \msgprop{S_i} \tensor^{\pri+3} \ol{(G_i \onto^{\pri+4} r)},
                                    \\
                                    {(\ci{\mu}{q}{:}~ \ol{G_i \onto^{\pri+4} q})}_{q \in \tilde{q} \setminus \{s,r\}}
                                \end{array}$
                            }{$\parr$}
                            \infUn{
                                $\forall i \in I.~ 3_i \vdash \begin{array}[t]{@{}l@{}}
                                    \ci{\mu}{s}{:}~ \ol{\msgprop{S_i}} \parr^{\pri+1} \ol{(G_i \onto^{\pri+4} s)},
                                    \\
                                    \ci{\mu}{r}{:}~ \msgprop{S_i} \tensor^{\pri+3} \ol{(G_i \onto^{\pri+4} r)},
                                    \\
                                    {(\ci{\mu}{q}{:}~ {\oplus}^{\pri+2} \{i{:}~ \ol{(G_i \onto^{\pri+4} q)} \}_{i \in I})}_{q \in \deps_s}
                                    \\
                                    {(\ci{\mu}{q}{:}~ {\oplus}^{\pri+3} \{i{:}~ \ol{(G_i \onto^{\pri+4} q)} \}_{i \in I})}_{q \in \deps_r}
                                    \\
                                    {(\ci{\mu}{q}{:}~ \ol{G_i \onto^{\pri+4} q})}_{q \in \tilde{q} \setminus \deps_s \setminus \deps_r \setminus \{s,r\}}
                                \end{array}$
                            }{$\forall q \in \deps_s \cup \deps_r.~ \oplus^\star$}
                            \infUn{
                                $\forall i \in I.~ 2_i \vdash \begin{array}[t]{@{}l@{}}
                                    \ci{\mu}{s}{:}~ \ol{\msgprop{S_i}} \parr^{\pri+1} \ol{(G_i \onto^{\pri+4} s)},
                                    \\
                                    \ci{\mu}{r}{:}~ {\oplus}^{\pri+2} \{i{:}~ \msgprop{S_i} \tensor^{\pri+3} \ol{(G_i \onto^{\pri+4} r)} \}_{i \in I},
                                    \\
                                    {(\ci{\mu}{q}{:}~ {\oplus}^{\pri+2} \{i{:}~ \ol{(G_i \onto^{\pri+4} q)} \}_{i \in I})}_{q \in \deps_s}
                                    \\
                                    {(\ci{\mu}{q}{:}~ {\oplus}^{\pri+3} \{i{:}~ \ol{(G_i \onto^{\pri+4} q)} \}_{i \in I})}_{q \in \deps_r}
                                    \\
                                    {(\ci{\mu}{q}{:}~ \ol{G_i \onto^{\pri+4} q})}_{q \in \tilde{q} \setminus \deps_s \setminus \deps_r \setminus \{s,r\}}
                                \end{array}$
                            }{$\oplus^\star$}
                            \infUn{
                                $\mdm{\tilde{q}}{G_s} = 1 \vdash \begin{array}[t]{@{}lr@{}}
                                    \ci{\mu}{s}{:}~ \&^\pri \{i{:}~ \ol{\msgprop{S_i}} \parr^{\pri+1} \ol{(G_i \onto^{\pri+4} s)} \}_{i \in I},
                                    & \text{(cf.\ \eqref{eq:mSType})}
                                    \\
                                    \ci{\mu}{r}{:}~ {\oplus}^{\pri+2} \{i{:}~ \msgprop{S_i} \tensor^{\pri+3} \ol{(G_i \onto^{\pri+4} r)} \}_{i \in I},
                                    & \text{(cf.\ \eqref{eq:mRType})}
                                    \\
                                    {(\ci{\mu}{q}{:}~ {\oplus}^{\pri+2} \{i{:}~ \ol{(G_i \onto^{\pri+4} q)} \}_{i \in I})}_{q \in \deps_s}
                                    & \text{(cf.\ \eqref{eq:mDepSType})}
                                    \\
                                    {(\ci{\mu}{q}{:}~ {\oplus}^{\pri+3} \{i{:}~ \ol{(G_i \onto^{\pri+4} q)} \}_{i \in I})}_{q \in \deps_r}
                                    & \text{(cf.\ \eqref{eq:mDepRType})}
                                    \\
                                    {(\ci{\mu}{q}{:}~ \ol{G_{i'} \onto^{\pri+4} q})}_{q \in \tilde{q} \setminus \deps_s \setminus \deps_r \setminus \{s,r\}}
                                    & \text{(cf.\ \eqref{eq:mDepOtherType})}
                                \end{array}$
                            }{$\&$}
                        \end{prooftree}
                    }
                \end{mdframed}
                \caption{Typing derivation used in the proof of \Cref{t:mediumTypes}.}
                \label{f:tderiii}
            \end{figure}

        \item
            \emph{Recursive definition}: $G_s = \mu Z \sdot G'$ (\cref{line:mdmRecDef}).
            Let
            \begin{align}
                \tilde{q}' := \{q \in \tilde{q} \mid G_s \onto^\pri q \neq \bullet\}
                \label{eq:mRecActive}
            \end{align}
            (as on \cref{line:mdmRecComp}).
            The analysis depends on whether $\tilde{q}' = \emptyset$ or not.
            \begin{itemize}
                \item
                    If $\tilde{q}' = \emptyset$ (\cref{line:mdmRecInact}),  let us take stock of  the types expected for each of the orchestrator's channels.
                    For now, we omit the substitutions in the types.
                    \begin{flalign}
                        & \text{For each $q \in \tilde{q}$, for $\ci{\mu}{q}$ we expect}
                        &
                        \dual{G_s \onto^{\pri_C} q}
                        &= \bullet.
                        & \hfill
                        \label{eq:mRecEmptyType}
                    \end{flalign}
                    Because all expected types are $\bullet$, the substitutions do not affect the types, so we can omit them altogether.

                    First we apply \scc{Empty}, giving us an arbitrary recursive context, thus the recursive context we need.
                    Then, we apply $\bullet$ for $\ci{\mu}{q}$ for each $q \in \tilde{q}$ (cf.\ \eqref{eq:mRecEmptyType}), and obtain the typing of $\mdm{\tilde{q}}{G_s}$ (omitting the recursive context):
                    \[
                        \mdm{\tilde{q}}{G_s} = \0 \vdash {(\ci{\mu}{q}{:}~ \bullet)}_{q \in \tilde{q}}.
                    \]

                \item
                    If $\tilde{q}' \neq \emptyset$ (\cref{line:mdmRecRet}),  let us take stock of the types expected for each of the orchestrator's channels.
                    Note that, because of the recursive definition on $Z$ in $G_s$, there cannot be another recursive definition in the context $C$ capturing the recursion variable $Z$.
                    Therefore, by \Cref{d:binders}, $Z \notin \widetilde{X_C}$.
                    \begin{flalign}
                        & \text{For each $q \in \tilde{q}'$,}
                        &
                        &
                        & \hfill
                        \nonumber \\
                        & \text{for $\ci{\mu}{q}$ we expect}
                        &
                        & \deepunfold{\dual{G_s \onto^{\pri_C} q}}{\ldots}
                        \nonumber \\
                        &
                        &
                        &= \deepunfold{\dual{\mu Z \sdot (G' \onto^{\pri_C} q)}}{\ldots}
                        \nonumber \\
                        &
                        &
                        &= \deepunfold{\mu Z \sdot \dual{G' \onto^{\pri_C} q}}{\ldots}
                        \nonumber \\
                        &
                        &
                        &= \mu Z \sdot \deepunfold{\dual{G' \onto^{\pri_C} q}}{{(X,t_X,\dual{G_X \onto^{\pri_X} q})}_{X \in \widetilde{X_C}}}.
                        \label{eq:mRecRecType}
                        \\
                        & \text{For each $q \in \tilde{q} \setminus \tilde{q}'$,}
                        \nonumber \\
                        & \text{for $\ci{\mu}{q}$ we expect}
                        &
                        & \deepunfold{\dual{G_s \onto^{\pri_C} q}}{\ldots}
                        \nonumber \\
                        &
                        &
                        &= \deepunfold{\bullet}{\ldots}
                        = \bullet.
                        \label{eq:mRecOtherType}
                    \end{flalign}
                    We also need an assignment in the recursive context for every $X \in \widetilde{X_C}$, but not for $Z$.

                    Let $C' = C[\mu Z \sdot []]$.
                    Clearly, $G' \gsub{C'} G$.
                    Let us establish some facts about the recursion binders, priorities, and active participants related to $C'$, $G'$, and $Z$:
                    \begin{itemize}
                        \item
                            $\widetilde{X_{C'}} = \ctxbind{C'} = (\ctxbind{C}, Z) = (\widetilde{X_C}, Z)$ (cf.\ \defref{d:binders}).
                        \item
                            $G_Z = \recdef{Z}{G} = G'$, as proven by the context $C'$ (cf.\ \defref{d:recdef}).
                        \item
                            $\widetilde{Y_Z} = \subbind{\mu Z \sdot G_Z}{G} = \ctxbind{C} = \widetilde{X_C}$.
                        \item
                            $\pri_{C'} = \ctxpri{}{C'} = \ctxpri{}{C} = \pri_C$, and $\pri_Z = \recpri{Z}{G} = \ctxpri{C} = \pri_C$, and hence $\pri_{C'} = \pri_Z$ (cf.\ \defref{d:ctxpri}).
                        \item
                            $\tilde{q}_Z = \tilde{q}'$ (cf.\ \defref{d:activeParts} and~\eqref{eq:mRecActive}).
                    \end{itemize}

                    Because $\widetilde{X_{C'}} = (\widetilde{X_C}, Z)$ and $\tilde{q}' = \tilde{q}_Z$, $\tilde{q}'$ is appropriate for the IH.
                    We apply the IH on $C'$, $G'$, and $\tilde{q}'$ to obtain a typing for $\mdm{\tilde{q}'}{G'}$, where we immediately make use of the facts established above.
                    We given the assignment to $Z$ in the recursive context separate from those for the recursion variables in $\widetilde{X_C}$.
                    Also, by \Cref{p:deepUnfold}, we can write the final unfolding on $Z$ in the types separately.
                    \begin{align*}
                        \mdm{\tilde{q}'}{G'} \vdash \begin{array}[t]{@{}l@{}}
                            {\left( X{:}~ {\big( \deepunfold{\dual{G_X \onto^{\pri_X} q}}{{(Y,t_Y,\dual{G_Y \onto^{\pri_Y} q})}_{Y \in \widetilde{Y_X}}} \big)}_{q \in \tilde{q}_X} \right)}_{X \in \widetilde{X_C}},
                            \\[6pt]
                            Z{:}~ {\big( \deepunfold{\dual{G' \onto^{\pri_C} q}}{{(X,t_X,\dual{G_X \onto^{\pri_X} q})}_{X \in \widetilde{X_C}}} \big)}_{q \in \tilde{q}'};
                            \\[6pt]
                            {\left( \ci{\mu}{q}{:}~ \unfold^{t_Z}(\mu Z \sdot \deepunfold{\dual{G' \onto^{\pri_C} q}}{{(X,t_X,\dual{G_X \onto^{\pri_X} q})}_{X \in \widetilde{X_C}}}) \right)}_{q \in \tilde{q}'}
                        \end{array}
                    \end{align*}

                    By assumption, we have
                    \begin{align*}
                        t_Z &= \max_\pr {\left( \deepunfold{\dual{G' \onto^{\pri_C} q}}{{(X,t_X,\dual{G_X \onto^{\pri_X} q})}_{X \in \widetilde{X_C}}} \right)}_{q \in \tilde{q}'} + 1,
                    \end{align*}
                    so $t_Z$ is clearly bigger than the maximum priority appearing in the types before unfolding.
                    Hence, we can apply \scc{Rec} to eliminate $Z$ from the recursive context, and to fold the types, giving the typing of $\mdm{\tilde{q}}{G_s} = \mu Z({(\ci{\mu}{q})}_{q \in \tilde{q}'}) \sdot \mdm{\tilde{q}'}{G'}$:
                    \begin{align*}
                        \mdm{\tilde{q}}{G_s} \vdash \begin{array}[t]{@{}l@{}}
                            {\left( X{:}~ {\big( \deepunfold{\dual{G_X \onto^{\pri_X} q}}{{(Y,t_Y,\dual{G_Y \onto^{\pri_Y} q})}_{Y \in \widetilde{Y_X}}} \big)}_{q \in \tilde{q}_X} \right)}_{X \in \widetilde{X_C}};
                            \\[6pt]
                            {\left( \ci{\mu}{q}{:}~ \mu Z \sdot \deepunfold{\dual{G' \onto^{\pri_C} q}}{{(X,t_X,\dual{G_X \onto^{\pri_X} q})}_{X \in \widetilde{X_C}}} \right)}_{q \in \tilde{q}'}
                        \end{array}
                    \end{align*}

                    In this typing, the type for $\ci{\mu}{q}$ for every $q \in \tilde{q}'$ concurs with~\eqref{eq:mRecRecType}.
                    For every $q \in \tilde{q} \setminus \tilde{q}'$, we can add the type for $\ci{\mu}{q}$ in~\eqref{eq:mRecOtherType} by applying $\bullet$.
                    This proves the thesis.
            \end{itemize}

        \item
            \emph{Recursive call}: $G_s = Z$ (\cref{line:mdmRecCall}).

            Following similar reasoning as in the case of recursive call in the proof of \Cref{t:routerTypesGen}, let us take stock of the types we expect for our orchestrator's channels.
            \begin{align}
                & \text{For each $q \in \tilde{q}$,}
                \nonumber \\
                & \text{for $\ci{\mu}{q}$ we expect}
                &
                & \deepunfold{\dual{G_s \onto^{\pri_C} q}}{\ldots}
                \nonumber \\
                &
                &
                &= \deepunfold{\dual{Z \onto^{\pri_C} q}}{\ldots}
                \nonumber \\
                &
                &
                &= \deepunfold{\dual{Z}}{\ldots}
                \nonumber \\
                &
                &
                &= \deepunfold{Z}{{(X,t_X,\dual{G_X \onto^{\pri_X} q})}_{X \in (\tilde{X}_1,Z,\widetilde{Y_Z})}}
                \nonumber \\
                &
                &
                &= \deepunfold{Z}{{(X,t_X,\dual{G_X \onto^{\pri_X} q})}_{X \in (Z,\widetilde{Y_Z})}}
                \nonumber \\
                &
                &
                &= \mu Z \sdot (\lift{t_Z} \deepunfold{\dual{G_Z \onto^{\pri_Z} q}}{{(X,t_X,\dual{G_X \onto^{\pri_X} q})}_{X \in \widetilde{Y_Z}}})
                \label{eq:mRecCallType}
            \end{align}
            Also, we need an assignment in the recursive context for every $X \in \widetilde{X_C}$.
            By \Cref{l:binderParts}, $\tilde{q} = \tilde{q}_Z$.
            Hence, for $Z$, the assignment should be as follows:
            \begin{align}
                Z{:}~ {\left( \deepunfold{\dual{G_Z \onto^{\pri_Z} q}}{{(X,t_X,\dual{G_X \onto^{\pri_X} q)}_{X \in \widetilde{Y_Z}}}} \right)}_{q \in \tilde{q}}
                \label{eq:mRecCallRec}
            \end{align}

            We apply \scc{Var} to obtain the typing of $\mdm{\tilde{q}}{G_s}$, where we make us the rule's allowance for an arbitrary recursive context up to the assignment to $Z$.
            \scc{Var} is applicable, because the types are recursive definitions on $Z$, concurring with the types assigned to $Z$, and lifted by a common lifter~$t_Z$.
            \begin{prooftree}
                \infAx{
                    $\mdm{\tilde{q}}{G_s} = X\call{{(\ci{\mu}{q})}_{q \in \tilde{q}}} \vdash \begin{array}[t]{@{}l@{}}
                        {\left( X{:}~ {\big( \deepunfold{\dual{G_X \onto^{\pri_X} q}}{{(Y,t_Y,\dual{G_Y \onto^{\pri_Y} q})}_{Y \in \widetilde{Y_X}}} \big)}_{q \in \tilde{q}_X} \right)}_{X \in \widetilde{X_C} \setminus (Z)},
                        \\
                        Z{:}~ {\left( \deepunfold{\dual{G_Z \onto^{\pri_Z} q}}{{(X,t_X,\dual{G_X \onto^{\pri_X} q)}_{X \in \widetilde{Y_Z}}}} \right)}_{q \in \tilde{q}};
                        \\
                        {\big( \ci{\mu}{q}{:}~ \mu Z \sdot (\lift{t_Z} \deepunfold{\dual{G_Z \onto^{\pri_Z} q}}{{(X,t_X,\dual{G_X \onto^{\pri_X} q})}_{X \in \widetilde{Y_Z}}}) \big)}_{q \in \tilde{q}}
                    \end{array}$
                }{\scc{Var}}
            \end{prooftree}
            In this typing, the types of $\ci{\mu}{q}$ for each $q \in \tilde{q}$ concur with the expected types in~\eqref{eq:mRecCallType}, and the assignment to $Z$ in the recursive context concurs with~\eqref{eq:mRecCallRec}.
            This proves the thesis.
            \qedhere
    \end{itemize}
\end{proof}

\subsubsection{Orchestrators and Centralized Compositions of Routers are Behaviorally Equivalent}
\label{ss:bisim}

First, we formalize what we mean with a centralized composition of routers, which we call a \emph{hub of routers}.
A hub of routers is just a specific composition of routers, formalized as the centralized composition of the routers of all a global type's participants synthesized from the global type:

\begin{definition}[Hub of a Global Type]\label{d:hub}
    Given global type $G$, we define the \emph{hub of routers} of $G$ as follows:
    \[
        \hub_G := \nu{\crt{p}{q} \crt{q}{p}}_{p,q \in \part(G)} \big(\prod_{p \in \part(G)} \rtr_p\big)
        \tag*{\lipicsEnd}
    \]
\end{definition}

Hubs of routers can be typed using local projection (cf.\ \defref{d:locproj}), identical to the typing of orchestrators (cf.\ \Cref{t:mediumTypes}):

\begin{theorem}\label{t:hubTypes}
    For relative well-formed global type $G$ and priority $\pri$,
    \[
        \mcl{H}_G \vdash \emptyset; {(\ci{\mu}{p}{:}~ \dual{(G \onto^\pri p)})}_{p \in \part(G)}.
    \]
\end{theorem}

\begin{proof}
    By the typability of routers (\Cref{t:routerTypes}) and the duality of the types of the endpoints connecting pairs of routers (\Cref{t:relpropDual}).
\end{proof}

\begin{figure}[t]
    \begin{mdframed}
        \mbox{}\hfill%
        \begin{wfit}
            \begin{prooftree}
                \infAss{
                    $\raisebox{3.0ex}{}$
                }
                \infUn{
                    $x[a,b] \xrightarrow{x[a,b]} \0$
                }{out}
            \end{prooftree}
        \end{wfit}%
        \hfill%
        \begin{wfit}
            \begin{prooftree}
                \infAss{
                    $P \xrightarrow{x[a,b]} P'$
                }
                \infUn{
                    $\nu{y a}\nu{z b} P \xrightarrow{\nu{y a}\nu{z b}x[a,b]} P'$
                }{out-open}
            \end{prooftree}
        \end{wfit}%
        \hfill\mbox{}

        \smallskip
        \mbox{}\hfill%
        \begin{wfit}
            \begin{prooftree}
                \infAss{
                    $\raisebox{4.0ex}{}$
                }
                \infUn{
                    $x(v,w) \sdot P \xrightarrow{x(v,w)} P$
                }{in}
            \end{prooftree}
        \end{wfit}%
        \hfill%
        \begin{wfit}
            \begin{prooftree}
                \infAss{
                    $P \xrightarrow{\nu{y a}\nu{z b}x[a,b]} P'$
                }
                \infAss{
                    $Q \xrightarrow{x(v,w)} Q'$
                }
                \infBin{
                    $P \| Q \xrightarrow{\tau} \nu{y v}\nu{z w}( P' \| Q' )$
                }{out-close}
            \end{prooftree}
        \end{wfit}%
        \hfill\mbox{}

        \smallskip
        \mbox{}\hfill%
        \begin{wfit}
            \begin{prooftree}
                \infAss{
                    $\raisebox{3.0ex}{}$
                }
                \infUn{
                    $x[b] \triangleleft j \xrightarrow{x[b] \triangleleft j} \0$
                }{sel}
            \end{prooftree}
        \end{wfit}%
        \hfill%
        \begin{wfit}
            \begin{prooftree}
                \infAss{
                    $P \xrightarrow{x[b] \triangleleft j} P'$
                }
                \infUn{
                    $\nu{z b} P \xrightarrow{\nu{z b}x[b] \triangleleft j} P'$
                }{sel-open}
            \end{prooftree}
        \end{wfit}%
        \hfill\mbox{}

        \smallskip
        \mbox{}\hfill%
        \begin{wfit}
            \begin{prooftree}
                \infAss{
                    $j \in I \raisebox{3.9ex}{}$
                }
                \infUn{
                    $x(w) \triangleright \{i: P_i\}_{i \in I} \xrightarrow{x(w) \triangleright j} P_j$
                }{bra}
            \end{prooftree}
        \end{wfit}%
        \hfill%
        \begin{wfit}
            \begin{prooftree}
                \infAss{
                    $P \xrightarrow{\nu{z b}x[b] \triangleleft j} P'$
                }
                \infAss{
                    $Q \xrightarrow{x(w) \triangleright j} Q'$
                }
                \infBin{
                    $P \| Q \xrightarrow{\tau} \nu{z w}( P' \| Q' )$
                }{sel-close}
            \end{prooftree}
        \end{wfit}%
        \hfill\mbox{}

        \smallskip
        \mbox{}\hfill%
        \begin{wfit}
            \begin{prooftree}
                \infAss{
                    $P \xrightarrow{\alpha} Q$
                }
                \infAss{
                    $\bn(\alpha) \cap \fn(R) = \emptyset$
                }
                \infBin{
                    $P \| R \xrightarrow{\alpha} Q \| R$
                }{par-L}
            \end{prooftree}
        \end{wfit}%
        \hfill%
        \begin{wfit}
            \begin{prooftree}
                \infAss{
                    $P \xrightarrow{\alpha} Q$
                }
                \infAss{
                    $\bn(\alpha) \cap \fn(R) = \emptyset$
                }
                \infBin{
                    $R \| P \xrightarrow{\alpha} R \| Q$
                }{par-R}
            \end{prooftree}
        \end{wfit}%
        \hfill\mbox{}

        \smallskip
        \mbox{}\hfill%
        \begin{wfit}
            \begin{prooftree}
                \infAss{
                    $\raisebox{3.2ex}{}$
                }
                \infUn{
                    $\nu{y z}( x \fwd y \| P ) \xrightarrow{\tau} P \subst{x/z}$
                }{id}
            \end{prooftree}
        \end{wfit}%
        \hfill%
        \begin{wfit}
            \begin{prooftree}
                \infAss{
                    $P \xrightarrow{\alpha} Q$
                }
                \infAss{
                    $\{y,y'\} \cap \fn(\alpha) = \emptyset$
                }
                \infBin{
                    $\nu{y y'} P \xrightarrow{\alpha} \nu{y y'} Q$
                }{res}
            \end{prooftree}
        \end{wfit}%
        \hfill\mbox{}%
    \end{mdframed}

    \caption{Labeled transition system for APCP (cf.\ \Cref{d:lts}).}
    \label{fig:lts}
\end{figure}

In order to state the behavioral equivalence of orchestrators and hubs of routers, we first define the specific behavioral equivalence we desire.
To this end, we first define a labeled transition system (LTS) for APCP:

\begin{definition}[LTS for APCP]\label{d:lts}
    We define the labels $\alpha$ for transitions for processes as follows:
    \begin{align*}
        \alpha ::=
        &~ \tau
        & \text{communication}
        \\[-4pt]
        \sepr\!
        &~ x[a,b]
        & \text{output}
        & \qquad \sepr
        \nu{y a}\nu{z b} x[a,b]
        & \text{bound output}
        \\[-4pt]
        \sepr\!
        &~ x[b] \triangleleft j
        & \text{selection}
        & \qquad \sepr
        \nu{z b} x[b] \triangleleft j
        & \text{bound selection}
        \\[-4pt]
        \sepr\!
        &~ x(v,w)
        & \text{input}
        & \qquad \sepr
        x(w) \triangleright j
        & \text{branch}
    \end{align*}
    The relation \emph{labeled transition} ($P \xrightarrow{\alpha} Q$) is then defined by the rules in \Cref{fig:lts}.
    \lipicsEnd
\end{definition}

\begin{proposition}\label{p:redIsTau}
    $P \redd_\beta Q$ if and only if $P \xrightarrow{\tau} Q$.
\end{proposition}

\noindent
As customary, we write `$\Rightarrow$' for the reflexive, transitive closure of $ \xrightarrow{\tau}$, and we write `$\xRightarrow{\alpha}$' for $\Rightarrow \xrightarrow{\alpha} \Rightarrow$ if $\alpha \neq \tau$ and for $\Rightarrow$ otherwise.

We can now define the behavioral equivalence we desire:

\begin{definition}[Weak bisimilarity]\label{d:weakBisim}
    A binary relation $\mathbb{B}$ on processes is a \emph{weak bisimulation} if whenever $(P,Q) \in \mathbb{B}$,
    \begin{itemize}
        \item $P \xrightarrow{\alpha} P'$ implies that there is $Q'$ such that $Q \xRightarrow{\alpha} Q'$ and $(P',Q') \in \mbb{B}$, and
        \item $Q \xrightarrow{\alpha} Q'$ implies that there is $P'$ such that $P \xRightarrow{\alpha} P'$ and $(P',Q') \in \mbb{B}$.
    \end{itemize}
    Two processes $P$ and $Q$ are \emph{weakly bisimilar} if there exists a weak bisimulation $\mbb{B}$ such that $(P,Q) \in \mbb{B}$.
    \lipicsEnd
\end{definition}

Our equivalence result shall relate the behavior of an orchestrator and a hub on a single but arbitrary channel.
More specifically, our result will demonstrate that both settings exhibit the same actions on a channel endpoint connect to a particular participant's implementation.
In order to isolate such a channel, we place the orchestrator and hub of routers in an evaluation context consisting of restrictions and parallel compositions with arbitrary processes, such that it connects all but one of the orchestrator's or hub's channels.
For example, given a global type $G$ and implementations $P_q \vdash \emptyset; \ci{\mu}{q}{:}~ G \onto^0 q$ for every participant $q \in \part(G) \setminus \{p\}$, we could use the following evaluation context:
\begin{align*}
    E := \nu{\ci{\mu}{q} \ci{q}{\mu}}_{q \in \part(G) \setminus \{p\}} \big( {\prod}_{q \in \part(G) \setminus \{p\}} P_q \| \hole \big)
\end{align*}
Replacing the hole in this evaluation context with the orchestrator or hub of routers of $G$ leaves one channel free: the channel $\ci{\mu}{p}$ for the implementation of $p$.
Now, we can observe the behavior of these two processes on $\ci{\mu}{p}$.

In what follows we write $\mcl{O}_G^{\tilde{q}}$ instead of $\mdm{\tilde{q}}{G}$.
When we appeal to router and orchestrator synthesis, we often omit the parameter $\tilde{q}$. That is, we write $\lmed{p}{}{G}$ instead of $\lmed{p}{\tilde{q}}{G}$, and
$\mcl{O}_G$ instead of $\mcl{O}_G^{\tilde{q}}$.

\begin{theorem}\label{t:mediumBisim}
    Suppose given a relative well-formed global type $G$.
    Let $\mcl{H}_G$ be the hub of routers of $G$ (\defref{d:hub}) and take the orchestrator $\mcl{O}_G^{\part(G)}$ of $G$ (\defref{d:medium}).
    Let $p \in \part(G)$, and let $E$ be an evaluation context such that $E[\mcl{H}_G] \vdash \emptyset; \ci{\mu}{p}{:}~ \dual{(G \onto^\pri p)}$.
    Then, $E[\mcl{H}_G]$ and $E[\mcl{O}_G]$ are weakly bisimilar (\defref{d:weakBisim}).
\end{theorem}


\noindent
We first give an intuition for the proof of \Cref{t:mediumBisim} and its ingredients, after which we give the proof using these ingredients; then, we detail the ingredients.
The proof is by coinduction, i.e., by exhibiting a weak bisimulation $\B$ that contains the pair $(E[\hub_G],E[\mcl{O}_G])$. To construct $\B$ and prove that it is a weak bisimulation we require the following:
\begin{itemize}
    \item
        We define a function that, given a global type $G$ and a \emph{starting relation} $\B_0$, computes a corresponding \emph{candidate relation}. This function is denoted $\B(G,\B_0)$ (\defref{d:candidateBisim}).

    \item
        Suppose $G \xrightarrow{\beta_1} \ldots \xrightarrow{\beta_k} G'$, with $k \geq 0$.
        Given some starting relation $\B_0$, we want to show that the relation obtained from $\B(G',\B_0)$ is a weak bisimulation, for which we need to assert that $\B_0$ is an appropriate starting relation.
        To this end, we define a function that computes a \emph{consistent} starting relation for a bisimulation relation, given a pair $(P,Q)$ of processes and a participant $p$ of $G$.
        This function is denoted $\<G \xrightarrow{\beta_1} \ldots \xrightarrow{\beta_k} G', (P,Q), p\>$ (\defref{d:consStartingRel}).

    \item
        The property that processes in such a consistent starting relation follow a pattern of specific labeled transitions, passing through a context containing the router of $p$ or the orchestrator (\Cref{l:startTrans}).

    \item
        The property that the relation obtained from $\B(G',\B_0,p)$ is a weak bisimulation, given the consistent starting relation $\B_0 = \<G \xrightarrow{\beta_1} \ldots \xrightarrow{\beta_k} G', (E[\hub_G],E[\mcl{O}_G]), p\>$ (\Cref{l:Bbisim}).
\end{itemize}
\Cref{t:mediumBisim} follows from these definitions and results:
\begin{proof}[Proof of \Cref{t:mediumBisim}]
    Let $\B = \B(G,\B_0)$, where $\B_0 = \<G, (E[\hub_G],E[\mcl{O}_G]), p\>$.
    By \Cref{l:Bbisim}, $\B$ is a weak bisimulation.
    Because $(E[\hub_G],E[\mcl{O}_G]) \in \B_0 \subseteq \B$, it then follows that $E[\hub_G]$ and $E[\mcl{O}_G]$ are weakly bisimilar.
\end{proof}

We setup some notations:

\begin{notation}
    We adopt the following notational conventions.
    \begin{itemize}
        \item
            We write $\mathsf{Proc}$ to denote the set of all typable APCP processes.
        \item
            In the LTS for APCP (\defref{d:lts}), we simplify labels: we write an overlined variant for output and selection (e.g., for $\nu{ab}\ci{\mu}{p}[a] \puts \ell$ we write $\ol{\ci{\mu}{p}} \puts \ell$), and omit continuation channels for input and branching (e.g., for $\ci{\mu}{p}(a) \gets \ell$ we write $\ci{\mu}{p} \gets \ell$).
        \item
            Also, we write $P \xRightarrow{\alpha_1 \ldots \alpha_n} Q$ rather than $P\xRightarrow{\alpha_1} P_1 \xRightarrow{\alpha_2} P_2 \ldots \xRightarrow{\alpha_n} Q$.
        \item
            We write $\tilde{\alpha}$ to denote a sequence of labels, e.g., if $\tilde{\alpha} = \alpha_1 \ldots \alpha_n$ then ${\xRightarrow{\tilde{\alpha}}} = {\xRightarrow{\alpha_1 \ldots \alpha_n}}$.
            If $\tilde{\alpha} = \epsilon$ (empty sequence), then ${\xRightarrow{\tilde{\alpha}}} = {\Rightarrow}$.
        \end{itemize}
\end{notation}

The following function defines a relation on processes, which we will use as the weak bisimulation between $E[\hub_G]$ and $E[\mcl{O}_G]$:

\begin{definition}[Candidate Relation]\label{d:candidateBisim}
    Let $G$ be a global type
    and let $p$ be a participant of $G$.
    Also, let $\B_0 \subseteq \mathsf{Proc} \times \mathsf{Proc}$ denote a relation on processes.
    We define a \emph{candidate relation} for a weak bisimulation of the hub and orchestrator of $G$ observed on $\ci{\mu}{p}$ starting at $\B_0$, by abuse of notation denoted $\B(G, \B_0, p)$. The definition is inductive on the structure of $G$:
    \begin{itemize}
            \item
            $G = \bullet$.
            Then $\B(G,\B_0,p) = \B_0$.

        \item
            $G = s \mto r \{i\<S_i\> \sdot G_i\}_{i \in I}$.
            We distinguish four cases, depending on the involvement of $p$:
            \begin{itemize}
                \item
                    $p = s$.
                    For every $i \in I$, let
                    \begin{align*}
                        \B_1^i &= \{(P_1,Q_1) \mid \exists (P_0,Q_0) \in \B_0 ~\tst~ P_0 \xrightarrow{\ci{\mu}{p} \gets i} \Rightarrow P_1 ~\tand~ Q_0 \xrightarrow{\ci{\mu}{p} \gets i} \Rightarrow Q_1\};
                        \\
                        \B_2^i &= \{(P_2,Q_2) \mid \exists (P_1,Q_1) \in \B_1^i ~\tst~ P_1 \xrightarrow{\ci{\mu}{p}(y)} \Rightarrow P_2 ~\tand~ Q_1 \xrightarrow{\ci{\mu}{p}(y)} \Rightarrow Q_2\}
                    \end{align*}
                    Then
                    \[
                        \B(G,\B_0,p) = \B_0 \cup \bigcup_{i \in I} (\B_1^i \cup \B(G_i,\B_2^i,p)).
                    \]

                \item
                    $p = r$.
                    For every $i \in I$, let
                    \begin{align*}
                        \B_1^i &= \{(P_1,Q_1) \mid \exists (P_0,Q_0) \in \B_0 ~\tst~ P_0 \xrightarrow{\ol{\ci{\mu}{p}} \puts i} \Rightarrow P_1 ~\tand~ Q_0 \xrightarrow{\ol{\ci{\mu}{p}} \puts i} \Rightarrow Q_1\};
                        \\
                        \B_2^i &= \{(P_2,Q_2) \mid \exists (P_1,Q_1) \in \B_1 ~\tand~ y ~\tst~ P_1 \xrightarrow{\ol{\ci{\mu}{p}}[y]} \Rightarrow P_2 ~\tand~ Q_1 \xrightarrow{\ol{\ci{\mu}{p}}[y]} \Rightarrow Q_2\}.
                    \end{align*}
                    Then
                    \[
                        \B(G,\B_0,p) = \B_0 \cup \bigcup_{i \in I} (\B_1^i \cup \B(G_i,\B_2^i,p)).
                    \]

                \item
                    $p \notin \{s,r\}$ and $\hdep(p,s,G)$ or $\hdep(p,r,G)$.
                    For every $i \in I$, let
                    \[
                        \B_1^i = \{(P_1,Q_1) \mid \exists (P_0,Q_0) \in \B_0 ~\tst~ P_0 \xrightarrow{\ol{\ci{\mu}{p}} \puts i} \Rightarrow P_1 ~\tand~ Q_0 \xrightarrow{\ol{\ci{\mu}{p}} \puts i} \Rightarrow Q_1\}
                    \]
                    Then
                    \[
                        \B(G,\B_0,p) = \B_0 \cup \bigcup_{i \in I} \B(G_i,\B_1^i,p).
                    \]

                \item
                    $p \notin \{s,r\}$ and neither $\hdep(p,s,G)$ nor $\hdep(p,r,G)$.
                    Then
                    \[
                        \B(G,\B_0,p) = \B(G_j,\B_0,p)
                    \]
                    for any $j \in I$.
            \end{itemize}

        \item
            $G = \mu X \sdot G'$.
            Then $\B(G,\B_0,p) = \B(G'\{\mu X \sdot G'/X\},\B_0,p)$.

        \item
            $G = \gskip \sdot G'$.
            Then $\B(G,\B_0,p) = \B(G',\B_0,p)$.
    \end{itemize}
\end{definition}

\noindent
The function $\B(G,\B_0,p)$ constructs a relation between processes by following labeled transitions on $\ci{\mu}{p}$ that concur with the expected behavior of $p$'s router and the orchestrator depending on the shape of $G$.
For example, for $G = s \mto p \{i\<S_i\> \sdot G_i\}_{i \in I}$, for each $i \in I$, the function constructs $\B_1^i$ containing the processes reachable from $\B_0$ through a transition labeled $\ol{\ci{\mu}{p}} \puts i$ (selection of the label chosen by $s$), and $\B_2^i$ containing the processes reachable from $\B_0$ through a transition labeled $\ol{\ci{\mu}{p}}[y]$ (output of the endpoint sent by $s$); the resulting relation then consists of $\B_0$ and, for each $i \in I$, $\B_1^i$ and $\B(G_i,\B_2^i,p)$ (i.e., the candidate relation for $G_i$ starting with $B_2^i$).
Since we are interested in a \emph{weak} bisimulation, the $\tau$-transitions of one process do not need to be simulated by related processes.
Hence, e.g., if $(P,Q) \in \B_0$ and $P \xrightarrow{\tau} P'$ and $Q \xrightarrow{\tau} Q'$, then $\{(P,Q), (P',Q), (P,Q'), (P',Q')\} \subseteq \B(G,\B_0,p)$.
This way, we only \emph{synchronize} related processes when they can both take the same labeled transition.

We intend to show that, if $G \xrightarrow{\beta_1} \ldots \xrightarrow{\beta_k} G'$, the function $\B(G',\B_0,p)$ constructs a weak bisimulation.
However, for this to hold, the starting relation $\B_0$ cannot be arbitrary: the pairs of processes in $\B_0$ have to be reachable from $E[\hub_G]$ and $E[\mcl{O}_G]$ through labeled transitions that concur with the transitions from $G$ to $G'$.
Moreover, the processes must have ``passed through'' evaluation contexts containing the router for $p$ at $G'$ and the orchestrator at $G'$.
The following defines a \emph{consistent starting relation}, parametric on $k$, that satisfies these requirements.
Note that for constructing the relation $\B$, we only need the following definition for $k = 0$.
However,
in the proof that $\B$ is a weak bisimulation
we need to generalize it to $k \geq 0$ to assure that the starting relation of coinductive steps is consistent.

\begin{definition}[Consistent Starting Relation]\label{d:consStartingRel}
    Let $G \xrightarrow{\beta_1} \ldots \xrightarrow{\beta_k} G'$  (with $k \geq 0$) be a sequence of labeled transitions from $G$ to $G'$ including the intermediate global types (cf. \Cref{d:globalLts}) and let $p$ be a participant of $G$.
    Also, let $(P,Q)$ be a pair of initial processes.
    We define the \emph{consistent starting relation} for observing the hub and orchestrator of $G'$ on $\ci{\mu}{p}$ starting at $(P,Q)$ after the transitions from $G$ to $G'$, denoted $\<G \xrightarrow{\beta_1} \ldots \xrightarrow{\beta_k} G', (P,Q), p\>$.
    The definition is inductive on the number $k$ of transitions:
    \begin{itemize}
        \item
            $k = 0$.
            Then $\<G, (P,Q), p\> = \{(P',Q') \mid P \Rightarrow P' ~\tand~ Q \Rightarrow Q'\}$.

        \item
            $k = k'+1$.
            Then
            \begin{align*}
                & \<G \xrightarrow{\beta_1} \ldots \xrightarrow{\beta_{k'}} G_{k'} \xrightarrow{\beta_k} G_k, (P,Q), p\> = {}
                \\
                &\quad \{(P_k,Q_k) \mid \begin{array}[t]{@{}l@{}}
                        \exists (P_{k'},Q_{k'}) \in \<G \xrightarrow{\beta_1} \ldots \xrightarrow{\beta_{k'}} G_{k'}, (P,Q), p\>
                        \\
                        \tst~ (\begin{array}[t]{@{}l@{}l@{}}
                            &~ (\exists C ~\tst~ P_{k'} \xRightarrow{\tilde{\alpha}} C[\lmed{p}{}{G_k}] \Rightarrow P_k)
                            \\
                            \tand &~ (\exists D ~\tst~ Q_{k'} \xRightarrow{\tilde{\alpha}} D[\mcl{O}_{G_k}] \Rightarrow Q_k))
                \}, \end{array} \end{array}
            \end{align*}
            where $\tilde{\alpha}$ depends on $\beta_k = s \> r: j\<S_j\>$ and $G_{k'}$ (in unfolded form if $G_{k'} = \mu X \sdot G'_{k'}$):
            \begin{itemize}
                \item
                    If $p = s$, then $\tilde{\alpha} = \ci{\mu}{p} \gets j \, \ci{\mu}{p}(y)$.
                \item
                    If $p = r$, then $\tilde{\alpha} = \ol{\ci{\mu}{p}} \puts j \, \ol{\ci{\mu}{p}}[y]$.
                \item
                    If $p \notin \{s,r\}$ and $\hdep(p,s,G_k)$ or $\hdep(p,r,G_k)$, then $\tilde{\alpha} = \ol{\ci{\mu}{p}} \puts j$.
                \item
                    If $p \notin \{s,r\}$ and neither $\hdep(p,s,G_k)$ nor $\hdep(p,r,G_k)$, then $\tilde{\alpha} = \epsilon$.
            \end{itemize}
    \end{itemize}
\end{definition}

\begin{lemma}\label{l:startTrans}
    Let $G$ be a relative well-formed global type such that $G \xrightarrow{\beta_1} \ldots \xrightarrow{\beta_k} G'$ for $k \geq 0$ and let $p$ be a participant of $G$.
    Also, let $E$ be an evaluation context such that $\fn(E) = \{\ci{\mu}{p}\}$.
    Then there exists $\tilde{\alpha}$ such that, for every $(P,Q) \in \<G \xrightarrow{\beta_1} \ldots \xrightarrow{\beta_k} G', (E[\hub_G],E[\mcl{O}_{G}]), p\>$,
    \begin{itemize}
        \item
            $E[\hub_G] \xRightarrow{\tilde{\alpha}} C\big[\lmed{p}{}{G'}\big] \Rightarrow P$ where $C$ is an evaluation context without an output or selection on $\ci{\mu}{p}$; and
        \item
            $E[\mcl{O}_{G}] \xRightarrow{\tilde{\alpha}} D\big[\mcl{O}_{G'}\big] \Rightarrow Q$ where $D$ is an evaluation context without an output or selection on $\ci{\mu}{p}$.
    \end{itemize}
\end{lemma}

\begin{proof}
    By induction on $k$.
    In the base case ($k=0$), we have $G = G'$, so $E[\hub_G] = C[\lmed{p}{}{G'}] \Rightarrow P$ and $E[\mcl{O}] = D[\mcl{O}_{G'}] \Rightarrow Q$.

    For the inductive case ($k=k'+1$), we detail the representative case where $$G \xrightarrow{\beta_1} \ldots \xrightarrow{\beta_{k'}} G_{k'} = p \mto s \{i\<S_i\> \sdot G'_i\}_{i \in I} \xrightarrow{p \> s: i'\<S_{i'}\>} G'$$ for some $i' \in I$.
    By the IH, for every $(P_{k'},Q_{k'}) \in \<G \xrightarrow{\beta_1} \ldots \xrightarrow{\beta_{k'}} G_{k'}, (E[\hub_G],E[\mcl{O}_G]),p\>$, there exists $\tilde{\alpha}'$ such that $E[\hub_G] \xRightarrow{\tilde{\alpha}'} C'[\lmed{}{p}{G_{k'}}] \Rightarrow P_{k'}$ and $E[\mcl{O}_G] \xRightarrow{\tilde{\alpha}'} D'[\mcl{O}_{G_{k'}}] \Rightarrow Q_{k'}$ where $C'$ and $D'$ are without output or selection on $\ci{\mu}{p}$.
    Take any $(P,Q) \in \<G \xrightarrow{\beta_1} \ldots \xrightarrow{\beta_{k'}} G_{k'} \xrightarrow {s \> p:i'\<S_{i'}\>} G', (E[\hub_G], E[\mcl{O}_G]), p\>$.
    By definition, there exists $(P_{k'},Q_{k'}) \in \<G \xrightarrow{\beta_1} \ldots \xrightarrow{\beta_{k'}} G_{k'}, (E[\hub_G],E[\mcl{O}_G]),p\>$ such that $$P_{k'} \xRightarrow{\ol{\ci{\mu}{p}} \puts i' \, \ol{\ci{\mu}{p}}[y]} C[\lmed{}{p}{G'}] \Rightarrow P\text{ and }Q_{k'} \xRightarrow{\ol{\ci{\mu}{p}} \puts i' \, \ol{\ci{\mu}{p}}[y]} D[\mcl{O}_{G'}] \Rightarrow Q$$ where there are no outputs or selection on $\ci{\mu}{p}$ in $C$ and $D$.
    Let $\tilde{\alpha} = \tilde{\alpha}' \, \ol{\ci{\mu}{p} \puts i'} \, \ol{\ci{\mu}{p}}[y]$.
    \sloppy
    Then ${E[\hub_G] \xRightarrow{\tilde{\alpha}} C[\lmed{p}{}{G'}] \Rightarrow P}$ and $E[\mcl{O}_G] \xRightarrow{\tilde{\alpha}} D[\mcl{O}_{G'}] \Rightarrow Q$.
\end{proof}

\begin{lemma}\label{l:Bbisim}
    Let $G$ be a relative well-formed global type such that $G \xrightarrow{\beta_1} \ldots \xrightarrow{\beta_k} G'$ (with $k
    \geq 0$) and let $p$ be a participant of $G$.
    Also, let $E$ be an evaluation context such that $\fn(E) = \{\ci{\mu}{p}\}$.
    Then the relation $\B(G',\B_0)$, with $\B_0 = \<G \xrightarrow{\beta_1} \ldots \xrightarrow{\beta_k} G', (E[\hub_G],E[\mcl{O}_{G}]), p\>$, is a weak bisimulation (cf. \Cref{d:weakBisim}).
\end{lemma}

\begin{proof}
    By coinduction on the structure of $G'$; there are four cases (communication, recursion, $\gskip$, and $\bullet$).
    We only detail the interesting case of communication, which is the only case which involves transitions with labels other than $\tau$.
    There are four subcases depending on the involvement of $p$ in the communication ($p$ is sender, $p$ is recipient, $p$ depends on the communication, or $p$ does not depend on the communication).
    In each subcase, the proof follows the same pattern, so as a representative case, we detail when $p$ is the recipient of the communication, i.e., $G' = s \mto p \{i\<S_i\> \sdot G'_i\}_{i \in I}$.
    Recall
    \begin{align*}
        \lmed{p}{}{G'} &= \crt{p}{s} \gets \big\{i{:}~ \ol{\ci{\mu}{p}} \puts i \cdot {(\ol{\crt{p}{q}} \puts i)}_{q \in \deps} \cdot \crt{p}{s}(v) \sdot \ol{\ci{\mu}{p}}[w] \cdot (v \fwd w \| \lmed{p}{}{G'_i}) \big\}_{i \in I},
        \tag{\text{\Cref{alg:router} \cref{line:rtrRecv}}}
        \\
        \lmed{s}{}{G'} &= \ci{\mu}{p} \gets \big\{i{:}~ \ol{\crt{s}{p}} \puts i \cdot {(\ol{\crt{s}{q}} \puts i)}_{q \in \deps} \cdot \ci{\mu}{s}(v) \sdot \ol{\crt{s}{p}}[w] \cdot (v \fwd w \| \lmed{s}{}{G'_i}) \big\}_{i \in I},
        \tag{\text{\Cref{alg:router} \cref{line:rtrSend}}}
        \\
        \mcl{O}_{G'} &= \ci{\mu}{s} \triangleright \{i{:}~ \ol{\ci{\mu}{p}} \triangleleft i \cdot {(\ol{\ci{\mu}{q}} \triangleleft i)}_{q \in \deps} \cdot \ci{\mu}{s}(v) \sdot \ol{\ci{\mu}{p}}[w] \cdot (v \fwd w \| \mcl{O}_{G'_i}) \}_{i \in I}.
        \tag{\text{\Cref{alg:medium} \cref{line:mdmComm}}}
    \end{align*}
    Let $\B = \B(G',\B_0)$.
    We have $\B = \B_0 \cup \bigcup_{i \in I}(\B_1^i \cup \B(G'_i,\B_2^i))$ with $\B_1^i$ and $\B_2^i$ as defined above.
    Take any $(P,Q) \in \B$; we distinguish cases depending on the subset of $\B$ to which $(P,Q)$ belongs:
    \begin{itemize}
        \item
            $(P,Q) \in \B_0$.
            By \Cref{l:startTrans}, we have $E[\hub_G] \xRightarrow{\tilde{\alpha}} C[\lmed{p}{}{G'}] \Rightarrow P$ and $E[\mcl{O}] \xRightarrow{\tilde{\alpha}} D[\mcl{O}_{G'}] \Rightarrow Q$, where $C$ and $D$ do not contain an output or selection on $\ci{\mu}{p}$.

            Suppose $P \xrightarrow{\alpha} P'$; we need to exhibit a matching weak transition from $Q$.
            By assumption, there are no outputs or selections on $\ci{\mu}{p}$ in $C$ and $D$.
            Since there are no outputs or selections on $\ci{\mu}{p}$ in $C$, by definition of $\lmed{p}{}{G'}$, we need only consider two cases for $\alpha$:
            \begin{itemize}
                \item
                    $\alpha = \tau$.
                    We have $Q \Rightarrow Q$, so $Q \xRightarrow{\tau} Q$.
                    Since $C[\lmed{p}{}{G'}] \Rightarrow P'$ and $D[\mcl{O}_{G'}] \Rightarrow Q$, we have $(P',Q) \in \B_0 \subseteq \B$.

                \item
                    $\alpha = \ol{\ci{\mu}{p}} \puts j$ for some $j \in I$.
                    To enable this transition, which originates from $p$'s router, somewhere in the $\tau$-transitions between $C[\lmed{p}{}{G'}]$ and $P$ the label $j$ was received on $\crt{p}{s}$, sent by the router of $s$ on $\crt{s}{p}$.
                    For this to happen, the label $j$ was received on $\ci{\mu}{s}$, sent from the context on $\ci{s}{\mu}$.
                    Since $\hub_G$ and $\mcl{O}$ are embedded in the same context, the communication of $j$ between $\ci{s}{\mu}$ and $\ci{\mu}{s}$ can also take place after a number of $\tau$-transitions from $D[\mcl{O}_{G'}]$, after which the selection of $j$ on $\ci{\mu}{p}$ becomes enabled.
                    Hence, since there are no outputs or selection on $\ci{\mu}{p}$ in $D$, we have $Q \Rightarrow Q_0 \xrightarrow{\ol{\ci{\mu}{p}} \puts j} Q'$.
                    We have $D[\mcl{O}_{G'}] \Rightarrow Q_0$, so $(P,Q_0) \in \B_0$.
                    Since $P \xrightarrow{\ol{\ci{\mu}{p}} \puts j} \Rightarrow P'$ and $Q_0 \xrightarrow{\ol{\ci{\mu}{p}} \puts j} \Rightarrow Q'$, we have $(P',Q') \in \B_1^j \subseteq \B'$.
            \end{itemize}

            Now suppose $Q \xrightarrow{\alpha} Q'$; we need to exhibit a matching weak transition from $P$.
            Again, we need only consider two cases for $\alpha$:
            \begin{itemize}
                \item
                    $\alpha = \tau$.
                    Analogous to the similar case above.

                \item
                    $\alpha = \ol{\ci{\mu}{p}} \puts j$ for some $j \in I$.
                    To enable this transition, which originates from the orchestrator, somewhere in the $\tau$-transitions between $D[\mcl{O}_{G'}]$ and $Q$ the label $j$ was received on $\ci{\mu}{s}$, sent from the context on $\ci{s}{\mu}$.
                    Hence, this communication can also take place after a number of transitions from $E[\hub_G]$, where the label is received by the router of $s$.
                    After this, from $C[\lmed{p}{}{G'}]$, the router of $s$ forwards $j$ to $p$'s router (communication between $\crt{s}{p}$ and $\crt{p}{s}$), enabling the selection of $j$ on $\ci{\mu}{p}$ in $p$'s router.
                    Hence, since there are no outputs or selections in $C$, we have $P \Rightarrow P_0 \xrightarrow{\ol{\ci{\mu}{p}} \puts j} P'$.
                    We have $C[\lmed{p}{}{G'}] \Rightarrow P_0$, so $(P_0,Q) \in \B_0$.
                    Since $P_0 \xrightarrow{\ol{\ci{\mu}{p}} \puts j} \Rightarrow P'$ and $Q \xrightarrow{\ol{\ci{\mu}{p}} \puts j} \Rightarrow Q'$, we have $(P',Q') \in \B_1^j \subseteq \B$.
            \end{itemize}

        \item
            \begin{sloppypar}
                $(P,Q) \in \B_1^j$ for some $j \in I$.
                We have $E[\hub_G] \xRightarrow{\tilde{\alpha}} C[\lmed{p}{}{G'}] \Rightarrow P_0 \xrightarrow{\ol{\ci{\mu}{p}} \puts j} \Rightarrow P$ and ${E[\mcl{O}_G] \xRightarrow{\tilde{\alpha}} D[\mcl{O}_{G'}] \Rightarrow Q_0 \xrightarrow{\ol{\ci{\mu}{p}} \puts j} \Rightarrow Q}$ where $(P_0,Q_0) \in \B_0$.
                Since we have already observed the selection of $j$ on $\ci{\mu}{p}$ from both the hub and the orchestrator, we know that the routers of $p$ and $s$ are in branch $j$, and similarly the orchestrator is in branch $j$.
            \end{sloppypar}

            Suppose $P \xrightarrow{\alpha} P'$.
            To exhibit a matching weak transition from $Q$ we only need to consider two cases for $\alpha$:
            \begin{itemize}
                \item
                    $\alpha = \tau$.
                    We have $Q \xRightarrow{\tau} Q$, and $P_0 \xrightarrow{\ol{\ci{\mu}{p}} \puts j} \Rightarrow P'$ and $Q_0 \xrightarrow{\ol{\ci{\mu}{p}} \puts j} \Rightarrow Q$, so $(P',Q) \in \B_1^j \subseteq \B$.

                \item
                    $\alpha = \ol{\ci{\mu}{p}}[y]$ for some $y$.
                    The observed output of some $y$ on $\ci{\mu}{p}$ must originate from $p$'s router.
                    This output is only enabled after receiving some $v$ over $\crt{p}{s}$, which must be sent by the router of $s$ over $\crt{s}{p}$.
                    The output by the router of $s$ is only enabled after receiving some $v$ over $\ci{\mu}{s}$, sent by the context over $\ci{s}{\mu}$.
                    Since the hub and the orchestrator are embedded in the same context, the communication of $v$ from $\ci{s}{\mu}$ to $\ci{\mu}{s}$ can also occur (or has already occurred) for the orchestrator.
                    After this, the output of $y$ over $\ci{\mu}{p}$ is enabled in the orchestrator, i.e., $Q \Rightarrow Q_1 \xrightarrow{\ol{\ci{\mu}{p}}[y]} Q'$.
                    We have $Q_0 \xrightarrow{\ol{\ci{\mu}{p}} \puts j} \Rightarrow Q_1$, so $(P,Q_1) \in \B_1^j$.
                    Since $P \xrightarrow{\ol{\ci{\mu}{p}}[y]} \Rightarrow P'$ and $Q_1 \xrightarrow{\ol{\ci{\mu}{p}}[y]} \Rightarrow Q'$, we have $(P',Q') \in \B_2^j$.
                    By definition, $\B_2^j \subseteq \B(G'_j,B_2^j) \subseteq \B$, so $(P',Q') \in \B$.
            \end{itemize}

            Now suppose $Q \xrightarrow{\alpha} Q'$.
            To exhibit a matching weak transition from $P$ we only need to consider two cases for $\alpha$:
            \begin{itemize}
                \item
                    $\alpha = \tau$.
                    Analogous to the similar case above.

                \item
                    $\alpha = \ol{\ci{\mu}{p}}[y]$ for some $y$.
                    The observed output of some $y$ on $\ci{\mu}{p}$ must originate from the orchestrator.
                    This output is only enabled after receiving some $v$ over $\ci{\mu}{s}$, sent by the context of $\ci{s}{\mu}$.
                    Since the hub and the orchestrator are embeded in the same context, the communication of $v$ from $\ci{s}{\mu}$ to $\ci{\mu}{s}$ can also occur (or has already occurred) for the router of $s$.
                    After this, the router of $s$ sends another channel $v'$ over $\crt{s}{p}$, received by $p$'s router on $\crt{p}{s}$.
                    This enables the output of $y$ on $\ci{\mu}{p}$ by $p$'s router, i.e., $P \Rightarrow P_1 \xrightarrow{\ol{\ci{\mu}{p}}[y]} P'$.
                    We have $P_0 \xrightarrow{\ol{\ci{\mu}{p}} \puts j} \Rightarrow P_1$, so $(P_1,Q) \in \B_1^j$.
                    Since $P \xrightarrow{\ol{\ci{\mu}{p}}[y]} \Rightarrow P'$ and $Q_1 \xrightarrow{\ol{\ci{\mu}{p}}[y]} \Rightarrow Q'$, we have $(P',Q') \in \B_2^j$.
                    As above, this implies that $(P',Q') \in \B$.
            \end{itemize}

        \item
            For some $j \in I$, $(P,Q) \in \B(G'_j,\B_2^j)$.
            The thesis follows from proving that $\B(G'_j,\B_2^j)$ is a weak bisimulation.
            For this, we want to appeal to the coinduction hypothesis, so we have to show that $\B_2^j = \<G \xrightarrow{\beta_1} \ldots \xrightarrow{\beta_k} G' \xrightarrow{s \> p:j\<S_j\>} G'_j, (E[\hub_G],E[\mcl{O}]), p\>$.
            We prove that $(P_2,Q_2) \in \B_2^j$ if and only if $(P_2,Q_2) \in \<G \xrightarrow{\beta_1} \ldots \xrightarrow{\beta_k} G' \xrightarrow{s \> p:j\<S_j\>} G'_j, (E[\hub_G],E[\mcl{O}]), p\>$, i.e., we prove both directions of the bi-implication:
            \begin{itemize}
                \item
                    Take any $(P_2,Q_2) \in \B_2^j$.
                    We have $E[\hub_G] \xRightarrow{\tilde{\alpha}} C[\lmed{p}{}{G'}] \Rightarrow P_0 \xrightarrow{\ol{\ci{\mu}{p}} \puts j} \Rightarrow P_1 \xrightarrow{\ol{\ci{\mu}{p}}[y]} \Rightarrow P_2$ and $E[\mcl{O}] \xRightarrow{\tilde{\alpha}} D[\mcl{O}_{G'}] \Rightarrow Q_0 \xrightarrow{\ol{\ci{\mu}{p}} \puts j} \Rightarrow Q_1 \xrightarrow{\ol{\ci{\mu}{p}}[y]} \Rightarrow Q_2$, where $(P_0,Q_0) \in \B_0$ and $(P_1,P_1) \in \B_1^j$.

                    By definition, somewhere during the transitions from $C[\lmed{p}{}{G'}]$ to $P_1$, we find $C'[\lmed{p}{}{G'_j}]$, which may then further reduce by $\tau$-transitions towards $P_2$.
                    As soon as we do find $C'[\lmed{p}{}{G'_j}]$, the output on $\ci{\mu}{p}$ is available, and the selection on $\ci{\mu}{p}$ has already occurred or is still available.
                    Because they are asynchronous actions, we can observe the selection and output on $\ci{\mu}{p}$ as soon as they are available, before further reducing $p$'s router.
                    Hence, we can observe $C[\lmed{p}{}{G'}] \Rightarrow \xrightarrow{\ol{\ci{\mu}{p}} \puts j} \Rightarrow \xrightarrow{\ol{\ci{\mu}{p}}[y]} \Rightarrow C''[\lmed{p}{}{G'_j}] \Rightarrow P_2$, i.e., \[{E[\hub_G] \xRightarrow{\tilde{\alpha}} C[\lmed{p}{}{G'}] \xRightarrow{{\ol{\ci{\mu}{p}} \puts j} \, {\ol{\ci{\mu}{p}}[y]}} C''[\lmed{p}{}{G'_j}] \Rightarrow P_2}.\]
                    By definition, $\lmed{p}{}{G'_j}$ has no output or selection on $\ci{\mu}{p}$ available, so there are no outputs or selections on $\ci{\mu}{p}$ in $C''$.

                    By a similar argument, we can observe $D[\mcl{O}_{G'}] \Rightarrow \xrightarrow{\ol{\ci{\mu}{p}} \puts j} \Rightarrow \xrightarrow{\ol{\ci{\mu}{p}}[y]} \Rightarrow D''[\mcl{O}_{G'_j}] \Rightarrow Q_2$, i.e., $E[\mcl{O}] \xRightarrow{\tilde{\alpha}} D[\mcl{O}_{G'}] \xRightarrow{{\ol{\ci{\mu}{p}} \puts j} \, {\ol{\ci{\mu}{p}}[y]}} D''[\mcl{O}_{G'_j}] \Rightarrow Q_2$.
                    Also in this case, there are no outputs or selections on $\ci{\mu}{p}$ in $D''$.

                    By assumption and definition, \[(C''[\lmed{p}{}{G'}],D''[\mcl{O}_{G'}]) \in \B_0 = \<G \xrightarrow{\beta_1} \ldots \xrightarrow{\beta_k} G', (E[\hub_G],E[\mcl{O}]), p\>.\]
                    Hence, by definition, $(P_2,Q_2) \in \<G \xrightarrow{\beta_1} \ldots \xrightarrow{\beta_k} G' \xrightarrow{s \> p:j\<S_j\>} G'_j, (E[\hub_G],E[\mcl{O}]), p\>$.

                \item
                    Take any $(P,Q) \in \<G \xrightarrow{\beta_1} \ldots \xrightarrow{\beta_k} G' \xrightarrow{s \> p:j\<S_j\>} G'_j, (E[\hub_G],E[\mcl{O}]), p\>$.
                    By definition, there are $(P',Q') \in \<G \xrightarrow{\beta_1} \ldots \xrightarrow{\beta_k} G', (E[\hub_G],E[\mcl{O}]), p\>$ such that $P' \xRightarrow{\ol{\ci{\mu}{p}} \puts j \, \ol{\ci{\mu}{p}}[y]} C[\lmed{p}{}{G'}] \Rightarrow P$ and $Q' \xRightarrow{\ol{\ci{\mu}{p}} \puts j \, \ol{\ci{\mu}{p}}[y]} D[\mcl{O}_{G'}] \Rightarrow Q$.
                    Since, $\B_0 = \<G \xrightarrow{\beta_1} \ldots \xrightarrow{\beta_k} G', (E[\hub_G],E[\mcl{O}]), p\>$, by definition $(P,Q) \in \B_2^j$.
                    \qedhere
            \end{itemize}
    \end{itemize}
\end{proof}


\section{Routers in Action}
\label{s:routersInAction}

We demonstrate our router-based analysis of global types by means of several examples.
First, in \secref{ss:intrl} and \secref{ss:deleg} we consider two simple protocols: they illustrate the different components of our approach, and our support for delegation and interleaving.
Then in \secref{ss:GAuth} we revisit  the authorization protocol $G_\sff{auth}$ from \Cref{s:introduction} to illustrate how our analysis supports also more complex protocols featuring also non-local choices and recursion.

\subsection{Delegation and Interleaving}
\label{ss:intrl}

We illustrate our analysis by considering a global type with delegation and interleaving, based on an example by Toninho and Yoshida~\cite[Ex.~6.9]{journal/toplas/ToninhoY18}.
Consider the global type:
\begin{align*}
    G_{\sff{intrl}} := p \mto q {:} 1\<{!}\sff{int} \sdot \bullet\> \sdot r \mto t {:} 2\<\sff{int}\> \sdot p \mto q {:} 3 \sdot \bullet
\end{align*}
Following Toninho and Yoshida~\cite{journal/toplas/ToninhoY18}, we define implementations of the roles of the four participants ($p,q,r,t$) of $G_{\sff{intrl}}$ using three processes ($P_1$, $P_2$, and $P_3$):
$P_2$ and $P_3$ implement the roles of $q$ and $r$, respectively, and $P_1$ interleaves the roles of $p$ and $t$ by sending a channel $s$ to $q$ and receiving an \sff{int} value $v$ from $r$, which it should forward to $q$ over $s$.
\begin{align*}
    P_1
    &:= \ol{\ci{p}{\mu}} \puts 1 \cdot \ol{\ci{p}{\mu}}[s] \cdot (\ci{t}{\mu} \gets \{2{:}~ \ci{t}{\mu}(v) \sdot \ol{s}[w] \cdot v \fwd w \} \| \ol{\ci{p}{\mu}} \puts 3 \cdot \ol{\ci{p}{\mu}}[z] \cdot \0 )
    \\
    &\vdash \ci{p}{\mu}{:}~ \oplus^0 \big\{1{:}~ \msgprop{{!}\sff{int} \sdot \bullet} \tensor^1 \oplus^8 \{3{:}~ \bullet \tensor^9 \bullet \} \big\},
    ~\ci{t}{\mu}{:}~ \&^6 \{2{:}~ \bullet \parr^7 \bullet \}
    \\[4pt]
    P_2
    &:= \ci{q}{\mu} \gets \{1{:}~ \ci{q}{\mu}(y) \sdot y(x) \sdot \ci{q}{\mu} \gets \{3{:}~ \ci{q}{\mu}(u) \sdot \0 \} \}
    \vdash \ci{q}{\mu}{:}~ \&^2 \big\{1{:}~ \dual{\msgprop{{!}\sff{int} \sdot \bullet}} \parr^3 \&^{10} \{3{:}~ \bullet \parr^{11} \bullet \} \big\}
    \\[4pt]
    P_3
    &:= \ol{\ci{r}{\mu}} \puts 2 \cdot \ol{\ci{r}{\mu}}[\bm{33}] \cdot \0
    \vdash \ci{r}{\mu}{:}~ \oplus^4 \{2{:}~ \bullet \tensor^5 \bullet \}
\end{align*}
where `$\bm{33}$' denotes a closed channel endpoint representing the number ``$33$''.

To prove that $P_1$, $P_2$, and $P_3$ correctly implement $G_\sff{intrl}$, we compose them with the routers synthesized from $G_\sff{intrl}$.
For example, the routers for $p$ and $t$, to which $P_1$ will connect, are as follows (omitting curly braces for branches on a single label):
\begin{align*}
    \rtr_p
    &= \ci{\mu}{p} \gets 1 \sdot \crt{p}{q} \puts 1 \cdot \ci{\mu}{p}(s) \sdot \ol{\crt{p}{q}}[s'] \cdot (s \fwd s' \| \ci{\mu}{p} \gets 3 \sdot \crt{p}{q} \puts 3 \cdot \ci{\mu}{p}(z) \sdot \ol{\crt{p}{q}}[z'] \cdot (z \fwd z' \| \0))
    \\
    \rtr_t
    &= \crt{t}{r} \gets 2 \sdot \ci{\mu}{t} \puts 2 \cdot \crt{t}{r}(v) \sdot \ol{\ci{\mu}{t}}[v'] \cdot (v \fwd v' \| \0)
\end{align*}

We assign values to the priorities in $\msgprop{{!}\sff{int} \sdot \bullet} = \bullet \tensor^\pri \bullet$ to ensure that $P_1$ and $P_2$ are well-typed;
assigning $\pri = 8$ works, because the output on $s$ in $P_1$ occurs after the input on $\ci{t}{\mu}$ (which has priority 6--7) and the input on $y$ in $P_2$ occurs before the second input on $\ci{q}{\mu}$ (which has priority 10--11).

The types assigned to $\ci{p}{\mu}$ and $\ci{t}{\mu}$ in $P_1$ coincide with $(G_\sff{intrl} \onto^0 p)$ and $(G_\sff{intrl} \onto^0 t)$, respectively (cf.\ \defref{d:locproj}).
Therefore, by \Cref{t:routerTypes}, the process $P_1$ connect to the routers for $p$ and $t$ $\nu{\ci{p}{\mu} \ci{\mu}{p}} \nu{\ci{t}{\mu} \ci{\mu t}{}} (P_1 \| \rtr_p \| \rtr_t)$ is well-typed.
Similarly, $\nu{\ci{q}{\mu} \ci{\mu}{q}} (P_2 \| \rtr_q)$ and $\nu{\ci{r}{\mu} \ci{\mu}{r}} (P_3 \| \rtr_r)$ are well-typed.

The composition of these routed implementations results in the following network:
\begin{align*}
    N_\sff{intrl} := \hspace{-.5em} \begin{array}{c}
        \nu{\crt{p}{q} \crt{q}{p}} \nu{\crt{p}{r} \crt{r}{p}}
        \\
        \nu{\crt{p}{t} \crt{t}{p}} \nu{\crt{q}{r} \crt{r}{q}}
        \\
        \nu{\crt{q}{t} \crt{t}{q}} \nu{\crt{r}{t} \crt{t}{r}}
    \end{array} \hspace{-.5em} \left( \hspace{-.5em} \begin{array}{l}
            \phantom{{} \| {}} \nu{\ci{p}{\mu} \ci{\mu}{p}} \nu{\ci{t}{\mu} \ci{\mu}{t}} (P_1 \| \rtr_p \| \rtr_t)
            \\
            {} \| \nu{\ci{q}{\mu} \ci{\mu}{q}} (P_2 \| \rtr_q)
            \\
            {} \| \nu{\ci{r}{\mu} \ci{\mu}{r}} (P_3 \| \rtr_r)
    \end{array}\right)
\end{align*}
We have $N_\sff{intrl} \in \sys(G_\sff{intrl})$ (cf.\ \defref{d:networks}), so, by \Cref{t:globalDlFree}, $N_\sff{intrl}$ is deadlock free and, by \Cref{t:completeness} and \Cref{t:soundness}, it correctly implements $G_\sff{intrl}$.

\subsection{Another Example of Delegation}
\label{ss:deleg}

Here, we further demonstrate our support for interleaving, showing how a participant can delegate the rest of its interactions in a protocol.
The following global type formalizes a protocol in which a Client ($c$) asks an online Password Manager ($p$) to login with a Server~($s$):
\begin{align*}
    G_\sff{deleg} := c \mto p {:} \sff{login} \<S\> \sdot G'_\sff{deleg}
\end{align*}
where
\begin{align*}
    S
    &:= {!}({?} \sff{bool} \sdot \bullet) \sdot S'
    \\
    S'
    &:= \& \{ \sff{passwd}{:}~ {?} \sff{str} \sdot \oplus \{ \sff{auth}{:}~ {!} \sff{bool} \sdot \bullet \} \}
    \\
    G'_\sff{deleg}
    &:= c \mto s {:} \sff{passwd} \<\sff{str}\> \sdot s \mto c {:} \sff{auth} \<\sff{bool}\> \sdot \bullet
\end{align*}
Here $S'$ expresses the type of $\rtr_c$'s channel endpoint $\ci{\mu}{c}$.
This means that we can give implementations for $c$ and $p$ such that $c$ can send its channel endpoint $\ci{c}{\mu}$ to $p$, after which $p$ logs in with $s$ in $c$'s place, forwarding the authorization boolean received from $s$ to $c$.
Giving such implementations is relatively straightforward, demonstrating the flexibility of our global types and analysis using APCP and routers.

Using local projection, we can compute a type for $c$'s implementation to safely connect with its router
\begin{align*}
    G_\sff{deleg} \onto^0 c = \oplus^{0} \{ \sff{login}{:}~ \msgprop{S} \tensor^{1} (G'_\sff{deleg} \onto^4 c) \}
\end{align*}
where
\begin{align*}
    \msgprop{S}
    &= (\bullet \parr^{\pri} \bullet) \tensor^{\kappa} \msgprop{S'}
    \\
    \msgprop{S'}
    &= \&^{\pi} \{ \sff{passwd}{:}~ \bullet \parr^{\rho} \oplus^{\delta} \{ \sff{auth}{:}~ \bullet \tensor^{\phi} \bullet \} \}
    \\
    G'_\sff{deleg} \onto^4 c
    &= \oplus^{4} \{ \sff{passwd}{:}~ \bullet \tensor^{5} \&^{10} \{ \sff{auth}{:}~ \bullet \parr^{11} \bullet \} \}
\end{align*}
Notice how $\dual{\msgprop{S'}} = G'_\sff{deleg} \onto^4 c$, given the assignments $\pi = 4, \rho = 5, \delta = 10, \phi = 11$.

We can use these types to guide the design of a process implementation for $c$.
Consider the process:
\begin{align*}
    Q := \ol{\ci{c}{\mu}} \puts \sff{login} \cdot \ol{\ci{c}{\mu}}[u] \cdot \ol{u}[v] \cdot (u \fwd \ci{c}{\mu} \| v(a) \sdot \0) \vdash \emptyset; \ci{c}{\mu}{:}~ G_\sff{deleg} \onto^0 c
\end{align*}
This implementation is interesting: after the first exchange in $G_\sff{deleg}$---sending a fresh channel $u$ (to $p$)---$c$~sends another fresh channel $v$ over $u$; then, $c$ delegates the rest of its exchanges in $G'_\sff{deleg}$ by forwarding all traffic on $\ci{c}{\mu}$ over $u$; in the meantime, $c$ awaits an authorization boolean over $v$.

Again, using local projection, we can compute a type for $p$'s implementation to connect with its router:
\begin{align*}
    G_\sff{deleg} \onto^0 p = \&^{2} \{ \sff{login}{:}~ \dual{\msgprop{S}} \parr^{3} \bullet \}
\end{align*}
We can then use it to type the following implementation for $p$:
\begin{align*}
    P := \ci{p}{\mu} \gets \left\{ \hspace{-.5em} \begin{array}{rl}
            \sff{login}{:}
            & \ci{p}{\mu}(\ci{c}{\mu}) \sdot \ci{c}{\mu}(v)
            \\
            & {} \sdot \ol{\ci{c}{\mu}} \puts \sff{passwd} \cdot \ol{\ci{c}{\mu}}[\bm{pwd123}]
            \\
            & {} \cdot \ci{c}{\mu} \puts \{ \sff{auth}{:}~ \ci{c}{\mu}(a) \sdot \ol{v}[a'] \cdot a \fwd a' \}
    \end{array} \hspace{-.5em} \right\} \vdash \emptyset; \ci{p}{\mu}{:}~ G_\sff{deleg} \onto^0 p
\end{align*}
In this implementation, $p$ receives a channel $\ci{c}{\mu}$ (from $c$) over which it first receives a channel $v$.
Then, it behaves over $\ci{c}{\mu}$ according to $c$'s role in $G'_\sff{deleg}$.
Finally, $p$ forwards the authorization boolean received from $s$ over $v$, effectively sending the boolean to $c$.

Given an implementation for $s$, say $S \vdash \emptyset; \ci{s}{\mu}{:}~ G_\sff{deleg} \onto^0 s$, what remains is to assign values to the remaining priorities in $\msgprop{S}$: assigning $\pri = 12, \kappa = 4$ works.
Now, we can compose the implementations $P$, $Q$ and $S$ with their respective routers and then compose these routed implementations together to form a deadlock free network of $G_\sff{deleg}$.
This way, e.g., the router for $c$ is as follows (again, omitting curly braces for branches on a single label):
\begin{align*}
    \rtr_c
    &= \ci{\mu}{c} \gets \sff{login} \sdot \crt{c}{p} \puts \sff{login} \cdot \ci{\mu}{c}(u) \sdot \ol{\crt{c}{p}}[u'] \cdot (
    \\
    &\phantom{{}={}} \quad u \fwd u' \| \ci{\mu}{c} \gets \sff{passwd} \sdot \crt{c}{s} \puts \sff{passwd} \cdot \ci{\mu}{c}(v) \sdot \ol{\crt{c}{s}}[v'] \cdot (
    \\
    &\phantom{{}={}} \qquad v \fwd v' \| \crt{c}{s} \gets \sff{auth} \sdot \ci{\mu}{c} \puts \sff{auth} \cdot \crt{c}{s}(w) \sdot \ol{\ci{\mu}{c}}[w'] \cdot (w \fwd w' \| \0)))
\end{align*}
Interestingly, the router is agnostic of the fact that the endpoint $u$ it receives over $\ci{\mu}{c}$ is in fact the opposite endpoint of the channel formed by $\ci{\mu}{c}$.

\subsection{The Authorization Protocol in Action}
\label{ss:GAuth}

Let us repeat $G_\sff{auth}$ from \Cref{s:introduction}:
\begin{align*}
    G_{\sff{auth}} = \mu X \sdot s \mto c \left\{ \begin{array}{@{}l@{}}
            \sff{login} \sdot c \mto a {:} \sff{passwd}\<\sff{str}\> \sdot a \mto s {:} \sff{auth}\<\sff{bool}\> \sdot X,
            \\
            \sff{quit} \sdot c \mto a {:} \sff{quit} \sdot \bullet
    \end{array} \right\}
\end{align*}
The relative projections of $G_\sff{auth}$ are as follows:
\begin{align*}
    G_\sff{auth} \wrt (s,a)
    &= \mu X \sdot s \snd c \left\{ \begin{array}{@{}l@{}}
            \sff{login} \sdot \gskip \sdot a {:} \sff{auth}\<\sff{bool}\> \sdot X,
            \\
            \sff{quit} \sdot \gskip \sdot \bullet
    \end{array} \right\}
    \\
    G_\sff{auth} \wrt (c,a)
    &= \mu X \sdot c \rcv s \left\{ \begin{array}{@{}l@{}}
            \sff{login} \sdot c {:} \sff{passwd}\<\sff{str}\> \sdot \gskip \sdot X,
            \\
            \sff{quit} \sdot c {:} \sff{quit} \sdot \bullet
    \end{array} \right\}
    \\
    G_\sff{auth} \wrt (s,c)
    &= \mu X \sdot s \left\{ \begin{array}{@{}l@{}}
            \sff{login} \sdot \gskip^2 \sdot X,
            \\
            \sff{quit} \sdot \gskip \sdot \bullet
    \end{array} \right\}
\end{align*}

\begin{figure}[!ht]
    \begin{mdframed}
        \begin{align*}
            \rtr_c
            &= \mu X(\ci{\mu}{c}, \crt{c}{s}, \crt{c}{a}) \sdot \crt{c}{s} \gets \left\{ \hspace{-.5em} \begin{array}{ll}
                    \sff{login}{:}
                    & \ol{\ci{\mu}{c}} \puts \sff{login} \cdot \ol{\crt{c}{a}} \puts \sff{login} \cdot \crt{c}{s}(u) \sdot \ol{\ci{\mu}{c}}[u']
                    \\
                    & {} \cdot (u \fwd u' \| \ci{\mu}{c} \gets \left\{ \hspace{-.5em} \begin{array}{ll}
                            \sff{passwd}{:}
                            & \ol{\crt{c}{a}} \puts \sff{passwd} \cdot \ci{\mu}{c}(v) \sdot \ol{\crt{c}{a}}[v']
                            \\
                            & {} \cdot (v \fwd v' \| X\call{\ci{\mu}{c}, \crt{c}{s}, \crt{c}{a}})
                    \end{array} \hspace{-.5em} \right\} ),
                    \\
                    \sff{quit}{:}
                    & \ol{\ci{\mu}{c}} \puts \sff{quit} \cdot \ol{\crt{c}{a}} \puts \sff{quit} \cdot \crt{c}{s}(w) \sdot \ol{\ci{\mu}{c}}[w']
                    \\
                    & {} \cdot (w \fwd w' \| \ci{\mu}{c} \gets \left\{ \hspace{-.5em} \begin{array}{ll}
                            \sff{quit}{:}
                            & \ol{\crt{c}{a}} \puts \sff{quit} \cdot \ci{\mu}{c}(z) \sdot \ol{\crt{c}{a}}[z']
                            \\
                            & \cdot (z \fwd z' \| \0)
                    \end{array} \hspace{-.5em} \right\} )
            \end{array} \hspace{-.5em} \right\}
            \\
            &\vdash \hspace{-.5em} \begin{array}[t]{l}
                \ci{\mu}{c}{:}~
                \mu X \sdot \oplus^{2} \left\{ \hspace{-.5em} \begin{array}{ll}
                        \sff{login}{:}
                        & \bullet \tensor^{3} \&^{4} \{ \sff{passwd}{:}~ \bullet \parr^{5} X \},
                        \\
                        \sff{quit}{:}
                        & \bullet \tensor^{3} \&^{4} \{ \sff{quit}{:}~ \bullet \parr^{5} \bullet \}
                \end{array} \hspace{-.5em} \right\} = \dual{(G_\sff{auth} \onto^0 c)},
                \\
                \crt{c}{s}{:}~
                \mu X \sdot \&^{1} \{ \sff{login}{:}~ \bullet \parr^{2} X, \sff{quit}{:}~ \bullet \parr^{2} \bullet \} = \relprop{c}{s}{0}{G_\sff{auth} \wrt (c,s)},
                \\
                \crt{c}{a}{:}~
                \mu X \sdot \oplus^{2} \left\{ \hspace{-.5em} \begin{array}{ll}
                        \sff{login}{:}
                        & \oplus^{5} \{ \sff{passwd}{:}~ \bullet \tensor^{6} X \},
                        \\
                        \sff{quit}{:}
                        & \oplus^{5} \{ \sff{quit}{:}~ \bullet \tensor^{6} \bullet \}
                \end{array} \hspace{-.5em} \right\} = \relprop{c}{a}{0}{G_\sff{auth} \wrt (c,a)}
            \end{array}
            \\[3pt] \cline{1-2}
            \\[-15pt]
            \rtr_s
            &= \mu X(\ci{\mu}{s}, \crt{s}{c}, \crt{s}{a}) \sdot \ci{\mu}{s} \gets \left\{ \hspace{-.5em} \begin{array}{ll}
                    \sff{login}{:}
                    & \ol{\crt{s}{c}} \puts \sff{login} \cdot \ol{\crt{s}{a}} \puts \sff{login} \cdot \ci{\mu}{s}(u) \sdot \ol{\crt{s}{c}}[u']
                    \\
                    & {} \cdot (u \fwd u' \| \crt{s}{a} \gets \left\{ \hspace{-.5em} \begin{array}{rl}
                            \sff{auth}{:}
                            & \ol{\ci{\mu}{s}} \puts \sff{auth} \cdot \crt{s}{a}(v) \sdot \ol{\ci{\mu}{s}}[v']
                            \\
                            & {} \cdot (v \fwd v' \| X\call{\ci{\mu}{s}, \crt{s}{c}, \crt{s}{a}})
                    \end{array} \hspace{-.5em} \right\} ),
                    \\
                    \sff{quit}{:}
                    & \ol{\crt{s}{c}} \puts \sff{quit} \cdot \ol{\crt{s}{a}} \puts \sff{quit} \cdot \ci{\mu}{s}(v) \sdot \ol{\crt{s}{c}}[v']
                    \\
                    & \cdot (v \fwd v' \| \0)
            \end{array} \hspace{-.5em} \right\}
            \\
            &\vdash \hspace{-.5em} \begin{array}[t]{l}
                \ci{\mu}{s}{:}~
                \mu X \sdot \&^{0} \{ \sff{login}{:}~ \bullet \parr^{1} \oplus^{10} \{ \sff{auth}{:}~ \bullet \tensor^{11} X \}, \sff{quit}{:}~ \bullet \parr^{1} \bullet \} = \dual{(G_\sff{auth} \onto^0 s)},
                \\
                \crt{s}{c}{:}~
                \mu X \sdot \oplus^{1} \{ \sff{login}{:}~ \bullet \tensor^{2} X, \sff{quit}{:}~ \bullet \tensor^{2} \bullet \} = \relprop{s}{c}{0}{G_\sff{auth} \wrt (s,c)},
                \\
                \crt{s}{a}{:}~
                \mu X \sdot \oplus^{1} \{ \sff{login}{:}~ \&^{9} \{ \sff{auth}{:}~ \bullet \parr^{10} X \}, \sff{quit}{:}~ \bullet \} = \relprop{s}{a}{0}{G_\sff{auth} \wrt (s,a)}
            \end{array}
            \\[3pt] \cline{1-2}
            \\[-15pt]
            \rtr_a
            &= \mu X(\ci{\mu}{a}, \crt{a}{c}, \crt{a}{s}) \sdot \crt{a}{s} \gets \left\{ \hspace{-.5em} \begin{array}{l}
                    \sff{login}{:}~ \ol{\ci{\mu}{a}} \puts \sff{login}
                    \\
                    {} \cdot \crt{a}{c} \gets \left\{ \hspace{-.5em} \begin{array}{l}
                            \sff{login}{:}
                            \\
                            \crt{a}{c} \gets \left\{ \hspace{-.5em} \begin{array}{l}
                                    \sff{passwd}{:}~ \ol{\ci{\mu}{a}} \puts \sff{passwd} \cdot \crt{a}{c}(u) \sdot \ol{\ci{\mu}{a}}[u']
                                    \\
                                    {} \cdot \left(\begin{array}{l}
                                        u \fwd u'
                                        \\
                                        {} \| \ci{\mu}{a} \gets \left\{ \hspace{-.5em} \begin{array}{l}
                                                \sff{auth}{:}
                                                \\
                                                \ol{\crt{a}{s}} \puts \sff{auth} \cdot \ci{\mu}{a}(v) \sdot \ol{\crt{a}{s}}[v']
                                                \\
                                                {} \cdot (v \fwd v' \| X\call{\ci{\mu}{a}, \crt{a}{c}, \crt{a}{s}})
                                        \end{array} \hspace{-.5em} \right\}
                                    \end{array} \hspace{-.5em} \right)
                            \end{array} \hspace{-.5em} \right\},
                            \\
                            \sff{quit}{:}~ \error{\ci{\mu}{a},\crt{a}{c},\crt{a}{s}}
                    \end{array} \hspace{-.5em} \right\},
                    \\
                    \sff{quit}{:}~ \ol{\ci{\mu}{a}} \puts \sff{quit}
                    \\
                    {} \cdot \crt{a}{c} \gets \left\{ \hspace{-.5em} \begin{array}{l}
                            \sff{login}{:}~ \error{\ci{\mu}{a},\crt{a}{c},\crt{a}{s}},
                            \\
                            \sff{quit}{:}~ \crt{a}{c} \gets \{ \sff{quit}{:}~ \ol{\ci{\mu}{a}} \puts \sff{quit} \cdot \crt{a}{c}(w) \sdot \ol{\ci{\mu}{a}}[w'] \cdot (w \fwd w' \| \0) \},
                    \end{array} \hspace{-.5em} \right\}
            \end{array} \hspace{-.5em} \right\}
            \\
            &\vdash \begin{array}[t]{l}
                \ci{\mu}{a}{:}~
                \mu X \sdot \oplus^{2} \left\{ \hspace{-.5em} \begin{array}{ll}
                        \sff{login}{:}
                        & \oplus^{6} \{ \sff{passwd}{:}~ \bullet \tensor^{7} \&^{8} \{ \sff{auth}{:}~ \bullet \parr^{9} X \} \},
                        \\
                        \sff{quit}{:}
                        & \oplus^{6} \{ \sff{quit}{:}~ \bullet \tensor^{7} \bullet \}
                \end{array} \hspace{-.5em} \right\} = \dual{(G_\sff{auth} \onto^0 a)},
                \\
                \crt{a}{c}{:}~
                \mu X \sdot \&^{2} \left\{ \hspace{-.5em} \begin{array}{ll}
                        \sff{login}{:}
                        & \&^{5} \{ \sff{passwd}{:}~ \bullet \parr^{6} X \} ,
                        \\
                        \sff{quit}{:}
                        & \&^{5} \{ \sff{quit}{:}~ \bullet \parr^{6} \bullet \}
                \end{array} \hspace{-.5em} \right\} = \relprop{a}{c}{0}{G_\sff{auth} \wrt (a,c)},
                \\
                \crt{a}{s}{:}~
                \mu X \sdot \&^{1} \{ \sff{login}{:}~ \oplus^{9} \{ \sff{auth}{:}~ \bullet \tensor^{10} X \} , \sff{quit}{:}~ \bullet \} = \relprop{a}{s}{0}{G_\sff{auth} \wrt (a,s)}
            \end{array}
        \end{align*}
    \end{mdframed}

    \caption{Routers synthesized from $G_{\sff{auth}}$.}
    \label{f:GauthRouters}
\end{figure}

The typed routers synthesized from $G_\sff{auth}$ are given in \Cref{f:GauthRouters}.
Let us explain the behavior of $\rtr_a$, the router of $a$.
$\rtr_a$ is a recursive process on recursion variable $X$, using the endpoint for the implementation $\ci{\mu}{a}$ and the endpoint for the other routers $\crt{a}{c}$ and $\crt{a}{s}$ as context.
The initial message in $G_\sff{auth}$ from $s$ to $c$ is a dependency for $a$'s interactions with both $s$ and $c$.
Therefore, the router first branches on the first dependency with $s$: a label received over $\crt{a}{s}$ (\sff{login} or \sff{quit}).
Let us detail the \sff{login} branch.
Here, the router sends \sff{login} over $\ci{\mu}{a}$.
Then, the router branches on the second dependency with $c$: a label received over $\crt{a}{c}$ (again, \sff{login} or \sff{quit}).
\begin{itemize}
    \item
        In the second \sff{login} branch, the router receives the label \sff{passwd} over $\crt{a}{c}$, which it then sends over $\ci{\mu}{a}$.
        The router then receives an endpoint (the password) over $\crt{a}{c}$, which it forwards over $\ci{\mu}{a}$.
        Finally, the router receives the label \sff{auth} over $\ci{\mu}{a}$, which it sends over $\crt{a}{s}$.
        Then, the router receives an endpoint (the authorization result) over $\ci{\mu}{a}$, which it forwards over $\crt{a}{s}$.
        The router then recurses to the beginning of the loop on the recursion variable $X$, passing the endpoints $\ci{\mu}{a},\crt{a}{s},\crt{a}{c}$ as recursive context.
    \item
        In the \sff{quit} branch, the router is in an inconsistent state, because it has received a label over $\crt{a}{c}$ which does not concur with the label received over $\crt{a}{s}$.
        Hence, the router signals an alarm on its endpoints $\ci{\mu}{a},\crt{a}{s},\crt{a}{c}$.
\end{itemize}

Notice how the typing of the routers in \Cref{f:GauthRouters} follows \Cref{t:routerTypes}: for each $p \in \{c,s,a\}$, the endpoint $\ci{\mu}{p}$ is typed with local projection (\defref{d:locproj}), and for each $q \in \{c,s,a\} \setminus \{p\}$ the endpoint $\crt{p}{q}$ is typed with relative projection (Defs.\ \labelcref{d:relproj,d:relprop}).

\begin{figure}[t]
    \begin{mdframed}
        \vspace{-1em}
        \begin{align*}
            & \mdm{\{c,s,a\}}{G_\sff{auth}}
            \\
            &\hspace{1em} = \hspace{-.5em} \begin{array}[t]{l}
                \mu X(\ci{\mu}{c}, \ci{\mu}{s}, \ci{\mu}{a})
                \\
                {} \sdot \ci{\mu}{s} \gets \left\{ \hspace{-.5em} \begin{array}{ll}
                        \sff{login}{:}
                        &\hspace{-.7em} \ol{\ci{\mu}{c}} \puts \sff{login} \cdot \ol{\ci{\mu}{a}} \puts \sff{login} \cdot \ci{\mu}{s}(u) \sdot \ol{\ci{\mu}{c}}[u']
                        \\
                        &\hspace{-.7em} {} \cdot (u \fwd u' \| \ci{\mu}{c} \gets \left\{ \hspace{-.5em} \begin{array}{ll}
                                \sff{auth}{:}
                                &\hspace{-.7em} \ol{\ci{\mu}{a}} \puts \sff{passwd} \cdot \ci{\mu}{c}(v) \sdot \ol{\ci{\mu}{a}}[v']
                                \\
                                &\hspace{-.7em} {} \cdot (v \fwd v' \| \ci{\mu}{a} \gets \left\{ \hspace{-.5em} \begin{array}{ll}
                                        \sff{auth}{:}
                                        &\hspace{-.7em} \ol{\ci{\mu}{s}} \puts \sff{auth} \cdot \ci{\mu}{a}(w) \sdot \ol{\ci{\mu}{s}}[w']
                                        \\
                                        &\hspace{-.7em} {} \cdot (w \fwd w' \| X\call{\ci{\mu}{c}, \ci{\mu}{s}, \ci{\mu}{a}})
                                \end{array} \hspace{-.5em} \right\} )
                        \end{array} \hspace{-.5em} \right\} ),
                        \\
                        \\
                        \sff{quit}{:}
                        &\hspace{-.7em} \ol{\ci{\mu}{c}} \puts \sff{quit} \cdot \ol{\ci{\mu}{a}} \puts \sff{quit} \cdot \ci{\mu}{s}(z) \sdot \ol{\ci{\mu}{c}}[z']
                        \\
                        &\hspace{-.7em} {} \cdot (z \fwd z' \| \ci{\mu}{c} \gets \{ \sff{quit}{:}~ \ol{\ci{\mu}{a}} \puts \sff{quit} \cdot \ci{\mu}{c}(y) \sdot \ol{\ci{\mu}{a}}[y'] \cdot (y \fwd y' \| \0) \} )
                \end{array} \hspace{-.5em} \right\}
            \end{array}
            \\
            &\hspace{1em} \vdash \ci{\mu}{c}{:}~ \ol{G_\sff{auth} \onto^0 c}, \ci{\mu}{s}{:}~ \ol{G_\sff{auth} \onto^0 s}, \ci{\mu}{a}{:}~ \ol{G_\sff{auth} \onto^0 a}
        \end{align*}
    \end{mdframed}

    \caption{Orchestrator synthesized from $G_\sff{auth}$ (cf.\ \defref{d:medium}).}
    \label{f:GauthMedium}
\end{figure}

Consider again the participant implementations given in \Cref{ex:impl}: $P$ implements the role of $c$, $Q$ the role of $s$, and $R$ the role of $a$.
Notice that the types of the channels of these processes coincide with relative projections:
\begin{align*}
    P
    &\vdash \emptyset; \ci{c}{\mu}{:}~ G_\sff{auth} \onto^0 c
    &
    Q
    &\vdash \emptyset; \ci{s}{\mu}{:}~ G_\sff{auth} \onto^0 s
    &
    R
    &\vdash \emptyset; \ci{a}{\mu}{:}~ G_\sff{auth} \onto^0 a
\end{align*}
Let us explore how to compose these implementations with their respective routers.
The order of composition determines the network topology.
\begin{description}
    \item[Decentralized]
        By first composing each router with their respective implementation, and then composing the resulting routed implementations, we obtain a decentralized topology:
        \[
            N_\sff{auth}^\sff{decentralized} := \hspace{-.5em} \begin{array}{r}
                \nu{\crt{c}{s} \crt{s}{c}}
                \\
                \nu{\crt{c}{a} \crt{a}{c}}
                \\
                \nu{\crt{s}{a} \crt{a}{s}}
            \end{array} \hspace{-.5em} \left( \hspace{-.5em} \begin{array}{l}
                    \phantom{{} \| {}} \nu{\ci{\mu}{c} \ci{c}{\mu}}
                    \\
                    {} \| \nu{\ci{\mu}{s} \ci{s}{\mu}}
                    \\
                    {} \| \nu{\ci{\mu}{a} \ci{a}{\mu}}
                \end{array} \hspace{-.5em} \hspace{-.5em} \begin{array}{l}
                    ( \rtr_c \| P )
                    \\
                    ( \rtr_s \| Q )
                    \\
                    ( \rtr_a \| R )
                \end{array}
            \right)
        \]
        This composition is in fact a  {network of routed implementations} of $G$ (cf.\ \defref{d:networks}), so \Cref{t:completeness,t:soundness,t:globalDlFree} apply: we have $N_\sff{auth}^\sff{decentralized} \in \sys(G_\sff{auth})$, so $N_\sff{auth}^\sff{decentralized}$ behaves as specified by $G_\sff{auth}$ and is deadlock free.

    \item[Centralized]
        By first composing the routers, and then composing the connected routers with each implementation, we obtain a centralized topology:
        \[
            N_\sff{auth}^\sff{centralized} := \hspace{-.5em} \begin{array}{c}
                \nu{\ci{\mu}{c} \ci{c}{\mu}}
                \\
                \nu{\ci{\mu}{s} \ci{s}{\mu}}
                \\
                \nu{\ci{\mu}{a} \ci{a}{\mu}}
            \end{array} \hspace{-.5em} \left( \hspace{-.5em} \begin{array}{l}
                \hspace{-.5em} \begin{array}{c}
                    \nu{\crt{c}{s} \crt{s}{c}}
                    \\
                    \nu{\crt{c}{a} \crt{a}{c}}
                    \\
                    \nu{\crt{s}{a} \crt{a}{s}}
                \end{array} \hspace{-.5em} \left( \hspace{-.5em} \begin{array}{l}
                        \phantom{{} \| {}} \rtr_c
                        \\
                        {} \| \rtr_s
                        \\
                        {} \| \rtr_a
                \end{array} \hspace{-.5em} \right) \hspace{-.5em} \begin{array}{l}
                    {} \| P
                    \\
                    {} \| Q
                    \\
                    {} \| R
                \end{array} \hspace{-.5em}
            \end{array} \hspace{-.5em} \right)
        \]
        Note that the composition of routers is a \emph{hub of routers} (\defref{d:hub}).
        Consider the composition of $P$, $Q$ and $R$ with the orchestrator of $G_\sff{auth}$ (given in \Cref{f:GauthMedium}):
        \[
            N_\sff{auth}^\sff{orchestrator} := \hspace{-.5em} \begin{array}{c}
                \nu{\ci{\mu}{c} \ci{c}{\mu}}
                \\
                \nu{\ci{\mu}{s} \ci{s}{\mu}}
                \\
                \nu{\ci{\mu}{a} \ci{a}{\mu}}
            \end{array} \hspace{-.5em} \left(
                \mdm{\{c,s,a\}}{G_\sff{auth}} \hspace{-.5em} \begin{array}{l}
                    {} \| P
                    \\
                    {} \| Q
                    \\
                    {} \| R
                \end{array} \hspace{-.5em}
            \right)
        \]
        By \Cref{t:mediumBisim}, the hub of routers and the orchestrator of $G_\sff{auth}$ are weakly bisimilar (\defref{d:weakBisim}).
        Hence, $N_\sff{auth}^\sff{centralized}$ and $N_\sff{auth}^\sff{orchestrator}$ behave the same.
\end{description}
Since each of $N_\sff{auth}^\sff{top}$ with $\sff{top} \in \{\sff{decentralized}, \sff{centralized}, \sff{orchestrator}\}$ is typable in empty contexts, by \Cref{t:globalDlFree}, each of these compositions is deadlock free.
Moreover, $N_\sff{auth}^\sff{decentralized}$ and $N_\sff{auth}^\sff{centralized}$ are structurally congruent, so, by \Cref{t:completeness,t:soundness}, they behave as prescribed by $G_\sff{auth}$.
Finally, by \Cref{t:mediumBisim}, $N_\sff{auth}^\sff{centralized}$ and $N_\sff{auth}^\sff{orchestrator}$ are bisimilar, and so $N_\sff{auth}^\sff{orchestrator}$ also behaves as prescribed by $G_\sff{auth}$.

\section{Related Work}
\label{s:relwork}

\subparagraph*{Types for Deadlock Freedom}

Our decentralized analysis of global types is related to type systems that ensure deadlock freedom for multiparty sessions with delegation and interleaving~\cite{conf/concur/BettiniCDLDY08,conf/coordination/PadovaniVV14,journal/mscs/CoppoDYP16}.
Unlike these works, we rely on a type system for \emph{binary} sessions which is simple and enables an expressive analysis of global types.
Coppo \etal~\cite{conf/concur/BettiniCDLDY08,conf/coordination/CoppoDPY13,journal/mscs/CoppoDYP16} give type systems for multiparty protocols, with asynchrony and support for interleaved sessions by tracking of mutual dependencies between them; as per Toninho and Yoshida~\cite{journal/toplas/ToninhoY18}, our example in \Cref{ss:intrl} is typable in APCP but untypable in their system.
Padovani \etal~\cite{conf/coordination/PadovaniVV14} develop a type system that enforces liveness properties for multiparty sessions, defined on top of a $\pi$-calculus with labeled communication.
Rather than global types, their type structure follows approaches based on \emph{conversation types}~\cite{journal/tcs/CairesV10}.
Toninho and Yoshida~\cite{journal/toplas/ToninhoY18} analyze binary sessions, leveraging on deadlock freedom results for multiparty sessions to extend Wadler's CLL~\cite{conf/icfp/Wadler12} with cyclic networks.
Their process language is synchronous and uses replication rather than recursion.
We note that their Examples~6.8 and~6.9 can be typed in APCP (cf.\ \secref{ss:intrl});
a detailed comparison between their extended CLL and APCP is interesting future work.

\subparagraph*{MPST and Binary Analyses of Global Types}

There are many works on MPST and their integration into programming languages; see~\cite{journal/csur/HuttelLVCCDMPRT16,journal/ftpl/AnconaBB0CDGGGH16} for surveys.
Triggered by flawed proofs of type safety and limitations of usual theories, Scalas and Yoshida~\cite{conf/popl/ScalasY19} define a meta-framework of multiparty protocols based on local types, without global types and projection. Their work has been a source of inspiration for our developments; we address similar issues by adopting relative types, instead of cutting ties with global types.

As already mentioned, Caires and P\'{e}rez~\cite{conf/forte/CairesP16} and Carbone \etal~\cite{conf/concur/CarboneLMSW16} reduce the analysis of global types to binary session type systems based on intuitionistic and classical linear logic, respectively.
Our routers strictly generalize the centralized mediums of Caires and P\'{e}rez (cf.\ \secref{ss:mediums}).
We substantially improve over the expressivity of the decentralized approach of Carbone \etal based on coherence, but reliant on encodings into centralized arbiters; for instance, their approach does not support the example from Toninho and Yoshida~\cite{journal/toplas/ToninhoY18} we discuss in \secref{ss:intrl}.
Also, Caires and P\'{e}rez support neither recursive global types nor asynchronous communication, and neither do Carbone \etal.

Scalas \etal~\cite{conf/ecoop/ScalasDHY17} leverage on an encoding of binary session types into \emph{linear types}~\cite{conf/ppdp/DardhaGS12,journal/acmpls/KobayashiPT99} to reduce multiparty sessions to processes typable with linear types, with applications in Scala programming.
Their analysis is decentralized but covers processes with synchronous communication only; also, their deadlock freedom result is limited with respect to ours: it does not support interleaving, such as in the example in \secref{ss:intrl}.

\subparagraph*{Monitoring through MPST}

Our work and the works discussed so far all consider the verification of implementations of multiparty protocols through static type checking.
Bocchi \etal~\cite{journal/tcs/BocchiCDHY17} use a \emph{dynamic} approach: communication between implementations is enacted by \emph{monitors}, which are derived from the global type to prevent protocol violations.
In their approach, Bocchi \etal rely on the traditional workflow for MPST: projection onto binary session types based on the merge operation.
Interestingly, Bocchi \etal's semantics relies on \emph{routing}, which is similar in spirit, but not in details, to our routers: their routing approach abstracts away from the actual network structure, while our routers enable the concrete realization of a decentralized network structure.
We also note that Bocchi \etal's monitors, based on finite state machines, live on the level of semantics, while our routers, $\pi$-calculus processes, live on the same level as implementations.
The theory by Bocchi \etal has resulted in the development of tools for a practical application of monitoring in Python~\cite{journal/fmsd/DemangeonHHNY15}, including an extension to real-time systems~\cite{journal/fac/NeykovaBY17}.

\subparagraph*{Other Approaches to Multiparty Protocols}

In a broader context, Message Sequence Charts (MSCs) provide graphical specifications of multiparty protocols.
Alur \etal~\cite{journal/tcs/AlurEY05} and Abdallah \etal~\cite{journal/ssm/AbdallahHJ15} study the decidability of model-checking properties such as implementability of MSC Graphs and High-level MSCs (HMSCs) as Communicating FSMs (CFSMs).
Genest \etal~\cite{journal/jcss/GenestMSZ06} study the synthesis of implementations of HMSCs as CFSMs; as we do, they use extra synchronization messages in some cases.
We follow an entirely different research strand: our analysis is type-based and targets well-formed global types that are implementable by design.
We note that the decidability of key notions for MPST (such as well-formedness and typability) has been addressed in~\cite{journal/acm/HondaYC16}.

Collaboration diagrams are another visual model for communicating processes (see, e.g.~\cite{journal/soca/BultanF08}).
Salaün \etal~\cite{journal/tsc/SalaunBR12} encode collaboration diagrams into the LOTOS process algebra~\cite{report/Brinksma89} to enable model-checking~\cite{conf/cav/GaravelMLS07},  realizability checks for synchronous and asynchronous communication, and synthesis of participant implementations.
Their implementation synthesis is reminiscent of our router synthesis, and also adds extra synchronization messages to realize otherwise unrealizable protocols with non-local choices.

\section{Conclusion}
\label{s:concl}

We have developed a new analysis of multiparty protocols specified as global types.
One distinguishing feature of our analysis is that it accounts for multiparty protocols implemented by arbitrary process networks, which can be centralized (as in orchestration-based approaches) but also decentralized (as in choreography-based approaches).
Another salient feature is that we can ensure both protocol conformance (protocol fidelity, communication safety) and deadlock freedom, which is notoriously hard to establish for protocols/implementations involving delegation and interleaving.
To this end, we have considered asynchronous process implementations in APCP, the typed process language that we introduced in~\cite{report/vdHeuvelP21B}.
Our analysis enables the transference of correctness properties from APCP to multiparty protocols.
We have illustrated these features using the authorization protocol $G_{\sff{auth}}$ adapted from Scalas and Yoshida~\cite{conf/popl/ScalasY19} as a running example; additional examples further justify how our approach improves over previous analyses (cf.\ \Cref{s:routersInAction}).

Our analysis of multiparty protocols rests upon three key innovations:
\emph{routers}, which enable global type analysis as decentralized networks;
\emph{relative types} that capture protocols between pairs of participants;
\emph{relative projection}, which admits global types with non-local choices.
In our opinion, these notions are interesting on their own.
In particular, relative types shed new light on more expressive protocol specifications than usual MPST, which are tied to notions of local types and merge/subtyping.

There are several interesting avenues for {future work}.
Comparing relative and merge-based well-formedness would continue the tread of new projections of global types (cf.\ \appref{a:compare} for initial findings).
We would also like to develop a type system based on relative types, integrating the logic of routers into a static type checking that ensures deadlock freedom for processes.
Finally, we are interested in developing practical tool support based on our findings.
For this latter point, following \cite{conf/popl/JiaGP16}, we would  like  to first formalize a theory of runtime monitoring based on routers, which can already be seen as an elementary form of \emph{choreographed monitoring} (cf.\ \cite{book/FrancalanzaPS18}).

\paragraph{Acknowledgments}

We are grateful to the anonymous reviewers for their constructive feedback and suggestions, which were enormously helpful to improve the presentation.
Research partially supported by the Dutch Research Council (NWO) under project No. 016.Vidi.189.046 (Unifying Correctness for Communicating Software).

\phantomsection
\addcontentsline{toc}{section}{References}
\bibliographystyle{plainurl}
\bibliography{refs}


\clearpage
\appendix

\section{Comparing Merge-based Well-for\-med\-ness and Relative Well-for\-med\-ness}
\label{a:compare}

It is instructive to examine how the notion of well-formed global types induced by our relative projection compares to \emph{merge-based} well-formedness, the notion induced by (usual) local projection~\cite{journal/lmcs/HuBYD12,book/CarboneYH09}.

Before we recall the definition of merge-based well-formedness, we define the projection of global types to local types.
Local types express one particular participant's perspective of a global protocol.
Although $\gskip$ is not part of standard definitions of local types, we include it to enable a fair comparison with relative types.

\begin{definition}[Local types]\label{d:loctype}
    \emph{Local types} $L$ are defined as follows, where the $S_i$ are the message types from \defref{d:globtypes}:
    \[
        L ::= {?}p \{i\<S\> \sdot L\}_{i \in I} \sepr {!}p \{i\<S\> \sdot L\}_{i \in I} \sepr \mu X \sdot L \sepr X \sepr \bullet \sepr \gskip \sdot L
        \tag*{\lipicsEnd}
    \]
\end{definition}
The local types ${?}p \{i\<S_i\> \sdot L_i\}_{i \in I}$ and ${!}p \{i\<S_i\> \sdot L_i\}_{i \in I}$ represent receiving a choice from $p$ and sending a choice to $p$, respectively.
All of $\bullet$, $\mu X \sdot L$, $X$, and $\gskip$ are just as before.

Instead of external dependencies, the projection onto local types relies on an operation on local types called \emph{merge}.
Intuitively, merge allows combining overlapping but not necessarily identical receiving constructs.
This is one main difference with respect to our relative projection.

\begin{definition}[Merge of Local Types]\label{d:locmerge}
    For local types $L_1$ and $L_2$, we define $L_1 \merge L_2$ as the \emph{merge} of $L_1$ and $L_2$:
    \begin{align*}
        \gskip \sdot L_1 \merge \gskip \sdot L_2
        &:= L_1 \merge L_2
        &
        \bullet \merge \bullet
        &:= \bullet
        \\
        \mu X \sdot L_1 \merge \mu X \sdot L_2
        &:= \mu X \sdot (L_1 \merge L_2)
        &
        X \merge X
        &:= X
        \\
        {!}p \{i\<S_i\> \sdot L_i\}_{i \in I} \merge {!}p \{i\<S_i\> \sdot L_i\}_{i \in I} &:= {!}p \{i\<S_i\> \sdot L_i\}_{i \in I}
        \\
        {?}p \{i\<S_i\> \sdot L_i\}_{i \in I} \merge {?}p \{j\<S'_j\> \sdot L'_j\}_{j \in J}
        &:= {?}p \left( \hspace{-.5em} \begin{array}{l}
                \phantom{{} \cup {}} \{i\<S_i\> \sdot L_i \}_{i \in I \setminus J}
                \\
                {} \cup \{j\<S'_j\> \sdot L'_j \}_{j \in J \setminus I}
                \\
                {} \cup \{k\<S_k \merge S'_k\> \sdot (L_k \merge L'_k) \}_{k \in I \cap J}
        \hspace{-.5em} \end{array} \right)
    \end{align*}
    The merge between message types $S_1 \merge S_2$ corresponds to the identity function.
    If the local types do not match the above definition, their merge is undefined.
    \lipicsEnd
\end{definition}

We can now define local projection based on merge:

\begin{definition}[Merge-based Local Projection]\label{d:mergeLocproj}
    For global type $G$ and participant $p$, we define $G \under p$ as the \emph{merge-based local projection} of $G$ under $p$:
    \begin{align*}
        \bullet \under p
        &:= \bullet
        &
        (\gskip \sdot G) \under p
        &:= \gskip \sdot (G \under p)
        &
        X \under p
        &:= X
        \\
        (\mu X \sdot G) \under p
        &:= \mathrlap{ \begin{cases}
                \bullet
                & \text{if $G \under p = \gskip^\ast \sdot \bullet$ or $G \under p = \gskip^\ast \sdot X$}
                \\
                \mu X \sdot (G \under p)
                & \text{otherwise}
        \end{cases} }
        \\
        (r \mto s \{i\<U_i\> \sdot G_i\}_{i \in I}) \under p
        &:= \mathrlap{ \begin{cases}
            {?}r \{i\<U_i\> \sdot (G_i \under p)\}_{i \in I}
            & \text{if $p = s$}
            \\
            {!}s \{i\<U_i\> \sdot (G_i \under p)\}_{i \in I}
            & \text{if $p = r$}
            \\
            \gskip \sdot (\merge_{i \in I} (G_i \under p))
            & \text{otherwise}
        \end{cases} }
        \\
        (G_1 \| G_2) \under p
        &:= \mathrlap{ \begin{cases}
            G_1 \under p
            & \text{if $p \in \part(G_1)$ and $p \notin \part(G_2)$}
            \\
            G_2 \under p
            & \text{if $p \in \part(G_2)$ and $p \notin \part(G_1)$}
            \\
            \bullet
            & \text{if $p \notin \part(G_1) \cup \part(G_2)$}
        \end{cases} }
        \tag*{\lipicsEnd}
    \end{align*}
\end{definition}

\begin{definition}[Merge Well-Formedness]\label{d:mwf}
    A global type $G$ is \emph{merge well-formed} if, for every $p \in \part(G)$, the merge-based local projection $G \under p$ is defined.
    \lipicsEnd
\end{definition}

The classes of relative and merge-based well-formed global types overlap: there are protocols that can be expressed using dependencies in relative types, as well as using merge in local types.
Interestingly, the classes are \emph{incomparable}: some relative well-formed global types are not merge-based well-formed, and vice versa.
We now explore these differences.

\subsection{Relative Well-Formed, Not Merge Well-Formed}\label{as:rwfNotMwf}

The merge of local types with outgoing messages of different labels is undefined.
Therefore, if a global type has communications, e.g., from $s$ to $a$ with different labels across branches of a prior communication between $b$ and $a$, the global type is not merge well-formed.
In contrast, such global types can be relative well-formed, because the prior communication may induce a dependency.
Similarly, global types with communications with different participants across branches of a prior communication are never merge well-formed, but may be relative well-formed.
The following example demonstrates a global type with messages of different labels across branches of a prior communication:

\begin{example*}
    We give an adaptation of the two-buyer-seller protocol in which Seller ($s$) tells Alice ($a$) to pay or not, depending on whether Bob ($b$) tells $a$ to buy or not.
    \begin{align*}
        G_\sff{rwf} := b \mto a \left\{ \begin{array}{@{}l@{}}
                \sff{ok} \sdot s \mto a {:} \sff{pay}\<\sff{int}\> \sdot \bullet,
                \\
                \sff{cancel} \sdot s \mto a {:} \sff{cancel} \sdot \bullet
        \end{array} \right\}
    \end{align*}
    This protocol is relative well-formed, as the relative projections under every combination of participants are defined.
    Notice how there is a dependency in the relative projection under $s$ and $a$:
    \begin{align*}
        G_\sff{rwf} \wrt (s,a) = a \rcv b \left\{ \begin{array}{@{}l@{}}
                \sff{ok} \sdot s {:} \sff{pay}\<\sff{int}\> \sdot \bullet,
                \\
                \sff{cancel} \sdot s {:} \sff{cancel} \sdot \bullet
        \end{array} \right\}
    \end{align*}
    However, we do not have merge well-formedness: the merge-based local projection under $s$ is not defined:
    \begin{align*}
        G_\sff{rwf} \under s = \gskip \sdot ({!}a {:} \sff{pay}\<\sff{int}\> \sdot \bullet \merge {!}a {:} \sff{cancel} \sdot \bullet)
        \tag*{\lipicsEnd}
    \end{align*}
\end{example*}

\subsection{Merge Well-Formed, Not Relative Well-Formed}

For a communication between, e.g., $a$ and $b$ to induce a dependency for subsequent communications between other participants, at least one of $a$ and $b$ must be involved.
Therefore, global types where communications with participants other than $a$ and $b$ have different labels across branches of a prior communication between $a$ and $b$ are never relative well-formed.
In contrast, merge can combine the reception of different labels, so such global types may be merge well-formed---as long as the sender is aware of which branch has been taken before.
The following example demonstrates such a situation, and explains how such global types can be modified to be relative well-formed:

\begin{example*}
    Consider a variant of the two-buyer-seller protocol in which Seller ($s$) invokes a new participant, Mail-service ($m$), to deliver the requested product.
    In the following global type, Bob ($b$) tells Alice ($a$) of its decision to buy or not, after which $b$ sends the same choice to $s$, who then either invokes $m$ to deliver the product or not:
    \begin{align*}
        G_\sff{mwf} := b \mto a \left\{ \begin{array}{@{}l@{}}
                \sff{ok} \sdot b \mto s {:} \sff{ok} \sdot s \mto m {:} \sff{deliver}\<\sff{str}\> \sdot \bullet,
                \\
                \sff{quit} \sdot b \mto s {:} \sff{quit} \sdot s \mto m {:} \sff{quit} \sdot \bullet
        \end{array} \right\}
    \end{align*}
    $G_\sff{mwf}$ is merge well-formed: the merge-based local projections under all participants are defined.
    Notice how the two different messages from $s$ are merged in the merge-based local projection under $m$:
    \begin{align*}
        G_\sff{mwf} \under m = \gskip^2 \sdot {?}s \left\{ \begin{array}{@{}l@{}}
                \sff{deliver}\<\sff{str}\> \sdot \bullet,
                \\
                \sff{quit} \sdot \bullet
        \end{array} \right\}
    \end{align*}
    $G_\sff{mwf}$ is not relative well-formed: the relative projection under $s$ and $m$ is not defined.
    The initial exchange between $b$ and $a$ cannot induce a dependency, since neither of $s$ and $m$ is involved.
    Hence, the relative projections of both branches must be identical, but they are not:
    \begin{align*}
        \gskip \sdot s {:} \sff{deliver}\<\sff{str}\> \sdot \bullet \neq \gskip \sdot s {:} \sff{quit} \sdot \bullet
    \end{align*}

    We recover relative well-formedness by modifying $G_\sff{mwf}$: we give $s$ the same options to send to $m$ in both branches of the initial communication:
    \begin{align*}
        G'_\sff{mwf} := b \mto a
        \left\{ \begin{array}{@{}l@{}}
                \sff{ok} \sdot b \mto s {:} \sff{ok} \sdot s \mto m \{ \sff{deliver}\<\sff{str}\> \sdot \bullet,\quad \sff{quit} \sdot \bullet \},
                \\
                \sff{quit} \sdot b \mto s {:} \sff{quit} \sdot s \mto m \{ \sff{deliver}\<\sff{str}\> \sdot \bullet,\quad \sff{quit} \sdot \bullet \}
        \end{array} \right\}
    \end{align*}
    The new protocol is still merge well-formed, but it is now relative well-formed too; the relative projection under $s$ and $m$ is defined:
    \begin{align*}
        G'_\sff{mwf} \wrt (s,m) = \gskip^2 \sdot s \left\{ \begin{array}{@{}l@{}}
                \sff{deliver}\<\sff{address}\> \sdot \bullet,
                \\
                \sff{quit} \sdot \bullet
        \end{array} \right\}
    \end{align*}
    This modification may not be ideal, though, because $s$ can quit the protocol even if $b$ has ok'ed the transaction, and that $s$ can still invoke a delivery even if $b$ has quit the transaction.
\end{example*}

\end{document}